\documentclass{article}
\usepackage{graphicx}
\usepackage{color}\usepackage{amsmath, amsthm}
\usepackage{amssymb}
\usepackage{rotating}
\usepackage{epsfig}
\usepackage{verbatim}
\usepackage{rotating}
\usepackage{a4}
\newcommand{\Scri}{\mathcal{I}}
\newcommand{\abs}[1]{\left| #1 \right|}
\input amssym.def
\input amssym.tex
\newtheorem{notation}{Notation}[section]
\newtheorem{theorem}{Theorem}[section]

\newtheorem{corollary}{Corollary}[section]
\newtheorem{proposition}{Proposition}[section]
\newtheorem{lemma}{Lemma}[section]
\newtheorem*{theorem*}{Theorem}
\def\ben{\begin{equation}}
\def\een{\end{equation}}

\def\bea{\begin{eqnarray}}
\def\eea{\end{eqnarray}}

\title{Stability and decay-rates for the five-dimensional Schwarzschild metric
  under biaxial perturbations}
\author{Gustav Holzegel}

\begin{document}
\maketitle
\begin{abstract}
In this paper we prove the non-linear asymptotic stability of the
five-dimensional Schwarzschild metric under biaxial vacuum perturbations. 
This is the statement that the evolution of 
$\left(SU\left(2\right) \times U\left(1\right)\right)$-symmetric vacuum
perturbations of initial data for the five-dimensional 
Schwarzschild metric finally converges in a suitable sense to a member
of the Schwarzschild family. It constitutes the first result proving the
existence of non-stationary vacuum black holes 
arising from asymptotically flat initial data dynamically approaching
a stationary solution. In fact, we show quantitative rates of approach.
The proof relies on vectorfield multiplier estimates, which are used in
conjunction with a bootstrap argument to establish polynomial 
decay rates for the radiation on the perturbed spacetime. Despite 
being applied here in a five-dimensional context, 
the techniques are quite robust and may admit applications to various 
four-dimensional stability problems.
\end{abstract}
\tableofcontents

\section{Introduction}
The existence of black holes features among the most 
fundamental predictions of general relativity. In the 
appropriate mathematical language of the theory, these 
objects correspond to solutions of the Einstein equations
\begin{equation} \label{Einstein}
R_{\mu \nu} - \frac{1}{2} R g_{\mu \nu} = 8 \pi T_{\mu \nu}
\end{equation}
possessing a regular event horizon and a complete 
null-infinity. General relativity admits an 
initial-value formulation suggesting that the 
appropriate setup to study black holes is in evolution 
from initial data. In this context, the main objective 
is to determine whether the maximal development associated 
to given data admits a complete null-infinity and 
a regular event horizon.

Some important special black hole solutions (hence their initial-data) 
are known in closed form. They are static or stationary, with the 
well-known Schwarzschild and Kerr family of solutions 
amongst them, which are believed to play crucial roles as 
``final states'' in gravitational collapse. It is fundamental for our
understanding of the theory to investigate the stability of 
these explicit solutions, that is to say the global 
structure of the evolution arising from initial data 
close (in an appropriate sense) to that 
of the known reference solution. Due to the complexity 
of this non-linear problem, most rigorous studies 
have been focussed on special symmetry classes. 
Specifically, a paramount problem of black 
hole physics, the full non-linear stability of the Kerr-solution, 
remains open to date.

A model in which both the global spacetime 
structure associated to the evolution of general initial data and 
the stability of certain solutions in particular
have been mathematically understood previously is that of 
the self-gravitating scalar field under spherical symmetry. 
The assumption of spherical symmetry casts the Einstein 
equations as a 1+1 dimensional system of PDEs, the
inclusion of a massless scalar field being the simplest 
way to circumvent Birkhoff's theorem.\footnote{Birkhoff's theorem
  implies that spherically symmetric vacuum solutions 
are either Minkowskian or Schwarzschildean.} In the context of this model, 
Christodoulou \cite{Christodoulou} proved that generic 
initial data either disperse, i.e.~asymptote to 
Minkowski space for late times, or collapse to regular 
black holes. His seminal work was extended by 
Dafermos and Rodnianski \cite{DafRod}, who 
proved that the development of initial data collapsing to 
black holes in fact approaches a Schwarzschild-metric on the 
exterior of the black hole at a sufficiently fast polynomial rate. 
These decay rates \cite{DafRod} of the scalar field 
were first suggested on a heuristic level by Price \cite{Price}, and are thought 
to be sharp. It is remarkable that \cite{DafRod} 
is a ``large data'' result. The initial data need not be assumed 
close to Schwarzschildean; all initial data containing a trapped
surface are shown to approach a Schwarzschild metric. 
\subsection{The model}
An alternative model allowing the study of gravitational collapse \emph{in vacuo}
was recently proposed by Bizon et al.~\cite{Bizon}. To understand
their idea we recall that, in view of the four-dimensional Birkhoff's 
theorem, gravitational collapse in vacuo ($T_{\mu \nu}=0$ 
in (\ref{Einstein})) cannot be studied under spherical symmetry. 
In axisymmetry on the other hand, the Einstein equations 
no longer reduce to a system of $1+1$ dimensional PDEs and 
the resulting problem does not seem tractable with current mathematical
techniques. The way out of this dilemma suggested by \cite{Bizon} is 
to study the Einstein vacuum equations under
$SU\left(2\right)$-symmetry in \emph{five} dimensions. This 
is motivated by the following observation: The analogue of spherical 
symmetry in four dimensions, i.e.~an $SO\left(3\right)$ action on 
an orbital two-sphere, is clearly an 
$SO\left(4\right)\cong\left(SU\left(2\right)_L \times
SU\left(2\right)_R\right) \slash \mathbb{Z}^2$ action on a 3-sphere in
five dimensions. However, via the latter
isomorphism there exist subgroups of $SO\left(4\right)$, for instance
$SU\left(2\right)_L$ and $\left(SU\left(2\right)_L \times U\left(1\right)_R \right) \slash \mathbb{Z}^2$ which still act transitively on the 3-sphere.\footnote{The subscripts $L$ and $R$ stand for the left and the 
right action respectively.} 
Consequently, even within the class of the smaller
symmetry-groups (commonly called triaxial- or 
biaxial- Bianchi IX depending on the subgroup to which one restricts) 
the Einstein equations reduce to a 
system of 1+1 dimensional PDEs. Moreover, Birkhoff's 
theorem is evaded by the introduction of one or two 
(in the triaxial case) dynamical degrees of freedom arising 
from the reduced symmetry. 

In the biaxial case this degree of freedom is manifest 
in a certain function $B$, which 
geometrically speaking corresponds to the ``squashing'' of the 
three sphere. $B$ is normalized such that it is zero for the
Schwarzschild-Tangherlini metric. From the point of view of the analysis
it can be understood as the analogue of the massless scalar field in four 
dimensions. The Einstein equations (\ref{Einstein}) 
imply the following non-linear wave equation 
for the squashing field $B$
\begin{equation} \label{nonlinB}
\square_g B = - \frac{4}{3r^2} \left(e^{-8B}-e^{-2B}\right)   \, .
\end{equation}

In \cite{Bizon} the model outlined was investigated numerically,  
suggesting that small initial data will disperse, 
whereas large data will collapse to black holes, approaching 
some Schwarzschild-Tangherlini black hole for large 
times. The mathematical study of the model was initiated shortly 
thereafter by M.~Dafermos in collaboration with the 
present author. In \cite{DafHol}, the following 
statement\footnote{Actually, it follows from a stronger 
statement proven in \cite{DafHol}.} was proven:
\begin{theorem*} 
Consider a triaxial-symmetric initial data set
$\left(\Sigma, g, K\right)$, which is close in an
appropriate norm\footnote{See \cite{DafHol} for the 
precise definition.} to an initial data set 
$\left(\Sigma, g_S, K_S\right)$ evolving to the five
dimensional Schwarzschild-Tangherlini solution of mass $M$. 
Let the squashing fields $B_1$, $B_2$ which are identically 
zero for the five-dimensional Schwarzschild metric, be of compact 
support on the initial hypersurface. Let $\mathcal{Q}$ 
be the Lorentzian quotient of the future Cauchy 
development of the data. Then $\mathcal{Q}$ contains 
a subset with Penrose diagram:
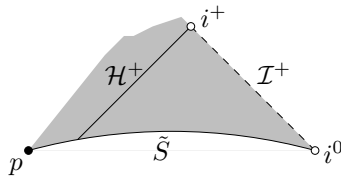
\[
\input{geg.pstex_t}
\]
It particular, the quotient of the maximal development of 
the set $\left(\Sigma, g, K\right)$ admits a complete
null-infinity with final Bondi mass $\mathcal{M}_f$ close 
to $M$, and a regular event horizon $\mathcal{H}^+$ on 
which the Penrose inequality $r^2 \leq 2 \mathcal{M}_f$ 
holds. Here $r$ is the area-radius function. 
\end{theorem*} 
The above theorem can be paraphrased as stating that 
perturbations of Schwarz\-schild-Tangherlini initial 
data again collapse to regular black holes close 
to the original Schwarzschild black hole. This result 
was termed \emph{orbital stability} of the five-dimensional
Schwarzschild metric in \cite{DafHol} 
and generated the first vacuum black hole solutions arising from
asymptotically flat initial data that are not 
stationary.\footnote{Solutions with a 
future complete, but not past complete, $\Scri^+$ 
have been constructed previously by Chru\'sciel \cite{PTC}, by 
solving a certain parabolic problem.}

Crucial for the proof of the above theorem is the 
existence of good monotonicity properties for a 
function $m\left(u,v\right)$, called the Hawking 
mass, defined in (\ref{Hawkmass}). It converges to the ADM mass defined at the 
asymptotically flat end. It is shown to satisfy 
$\partial_u m \leq 0$ and $\partial_v m \geq 0$ 
on the domain of outer communications, leading to 
an a-priori bound for the total mass fluctuation 
on the spacetime in terms of the initial data. 

%Moreover, 
%the quantity $\mu = \frac{2m}{r^2}$ 
%satisfies $\mu < 1$ on the black hole exterior, 
%$J^+\left(\tilde{S}\right) \cap J^-\left(\Scri^+\right)$.
%
%
%
\subsection{The main theorem}
Orbital stability provides of course certain 
control over the global structure of the solution. 
Nevertheless, it leaves the details of the late-time 
behaviour unclear. In particular, solutions could 
exhibit unexpected features at late times with the 
squashing field $B$ oscillating in some complicated 
manner and the geometry thus never settling down. 
This problem is finally addressed in the 
present paper. By proving appropriate 
decay-rates we will show that the squashing field does decay 
for late times and hence that perturbations converge 
to another member of the Schwarzschild-Tangherlini 
family. 
\subsubsection{The statement}
The main result is
\begin{theorem} \label{asymptoticstab}
Consider a biaxial-symmetric initial data set
$\left(\Sigma, g, K\right)$, which is close in the sense of 
the previous theorem to an initial data set 
$\left(\Sigma, g_S, K_S\right)$ whose maximum development 
is the five dimensional Schwarzschild-Tangherlini solution 
of mass $M$.  Let $\pi : \mathcal{M} \rightarrow \mathcal{Q}$ denote
the projection map of the maximal development of $\left(\Sigma, g,
K\right)$ to the 
two-dimensional Lorentzian quotient space $\mathcal{Q}$ 
and let $\tilde{S}=\pi\left(\Sigma\right)$. 
Fix a curve of constant area radius, $r=r_K$, 
away from the horizon, intersecting $\tilde{S}$ 
at $P$ as depicted below.
\begin{figure}[h!]
\[ %\label{theofig}
\input{theorem.pstex_t}
\]
\end{figure}
Assume furthermore that the initial data slice $\tilde{S}$ coincides
for $r \geq r_K$ with an integral curve of the globally defined 
vectorfield $\nabla r$ on $\mathcal{Q}$ and that the 
data is Schwarzschildean outside a compact set, i.e.~that 
the squashing field $B$ is of compact support.

Define regular coordinates $\left(u,v\right)$ on the subset 
$J^+\left(\tilde{S} \cap \{r \geq r_K \}\right) \cap J^-\left(\Scri^+\right)$ 
of the Penrose diagram arising by the previous Theorem
as follows. Let the point $R$, determined by the intersection of the curve $r^2=4m$ with $\tilde{S}$, have coordinates $u=v=\sqrt{M}$. Set $r_{,v}=\frac{1}{2} \left(1-\mu\right)$, with $\mu =
\frac{2m}{r^2}$, along the null-ray $\overline{PQ}$ 
and $r_{,u} = -\frac{1}{2}$ along null-infinity. In
these coordinates $u \rightarrow \infty$ along 
null-infinity as $i^+$ is approached. The horizon 
$\mathcal{H}^+$ is parametrized as $(\infty, v)$. 
Define $t=\frac{v+u}{2}$ and $r^\star=\frac{v-u}{2}$. 

Then there exists a dimensionless constant $\delta >
0$, depending only on the geometry of $\tilde{S}$ such that 
if the field $B$ satisfies
\begin{equation} \label{initassump}
M^{-\frac{3}{4}}\left[r^\frac{3}{2} |B| + r^\frac{5}{2}\Big|\frac{B_{,u}}{r_{,u}}\Big| + r^\frac{5}{2}\Big|\frac{B_{,v}}{r_{,v}}\Big| \right] \leq \delta 
\end{equation}
on $\tilde{S} \cap \{r \geq r_K\}$ and
\begin{eqnarray} \label{EKBintro}
% E^K_B \left(\tilde{S} \cap \{r \geq r_K\}\right) = 
\frac{1}{M} \int_{\tilde{S} \cap \{r \geq r_K\}} \Big[u^2 \left(\partial_u B\right)^2 + v^2
  \left(\partial_v B\right)^2 +\left(u^2+v^2\right)
  \left(-r_{,u}\right) B^2\Big] \frac{1}{\Omega} dvol_3 \leq \delta^2 \, ,
\end{eqnarray}
as well as
\begin{equation} \label{initassump2}
M^{-\frac{3}{4}}\left[r^\frac{3}{2} |B| +
  r^\frac{5}{2}\Big|\frac{B_{,u}}{r_{,u}}\Big| \right] \leq \delta 
\end{equation}
on the ray $v= v\left(P\right) \cap \{r \leq r_K\}$, then 
the squashing function $B$ satisfies
\begin{equation} \label{dechoz}
|B| + \sqrt{M} |B_{,v}| + \sqrt{M} \Big|\frac{B_{,u}}{r_{,u}}\Big| \leq \frac{C\sqrt{M}}{v_+} \textrm{ \ \ \ for $r \leq r_K$ }
\end{equation}
where $v_+ = \max\left(1,v\right)$,
\begin{equation} \label{deccent}
|B| + \sqrt{M} |B_{,v}| + \sqrt{M} \Big|\frac{B_{,u}}{r_{,u}}\Big| \leq \frac{C\sqrt{M} }{t} \textrm{ \ \ \ for $r \geq r_K$ }
\end{equation}
\begin{equation} \label{decr}
|B| \leq \frac{C \ M^\frac{3}{4}}{r^\frac{3}{2}} \textrm{\ \ \ \ for $r \geq r_K$}
\end{equation}
on $\mathcal{D} = J^+\left(\tilde{S} \cap \{r \geq r_K\}\right) \cap
J^-\left(\mathcal{I}^+\right)$ for a dimensionless constant $C$ (which depends on the choice of $r_K$) 
computable from the initial data.
\end{theorem}
We will refer to this result as the \emph{asymptotic stability} 
of the Schwarzschild-Tangherlini solution. {\bf In particular, 
Theorem \ref{asymptoticstab} produces the first dynamical 
vacuum solutions arising from asymptotically flat initial data 
and converging to stationary black holes for late times.} 

\subsubsection{Remarks}
Restricting $\tilde{S}$ to coincide with a $\nabla r$ 
integral curve for $r\geq r_K$ is justified by Cauchy stability 
and the fact that the global properties of the Penrose diagram 
are already known by the orbital stability
result of the previous theorem. It has been assumed to avoid some 
clumsy notation in the proof. 

Cauchy stability also justifies 
stating the smallness assumptions  (\ref{initassump}), (\ref{EKBintro}) 
and (\ref{initassump2}) on the slice 
\begin{equation} \label{sigrK}
\tilde{S}_{r_K} = \left(\tilde{S} \cap \{r \geq r_K\}\right) \cup \Big( \{ v =
v\left(P\right)\} \cap \{r \leq r_K \} \Big) \, .
\end{equation} 
instead of $\tilde{S}$.\footnote{The smallness assumption
(\ref{initassump2}) easily translates into 
an appropriate smallness assumption on $\tilde{S}$, depending on the
geometry of $\tilde{S}$ for $r \leq r_K$, after extending 
the coordinate system to all of $J^+\left(\tilde{S}\right) 
\cap J^-\left(\Scri^+\right)$.} 
The advantage of doing it this way 
is that (\ref{initassump}) and (\ref{initassump2}) 
do not depend on the choice of double-null coordinates on the 
Penrose diagram.\footnote{This will become useful later 
because the bootstrap argument applied in the proof requires 
the definition of different coordinate systems.} 
Assumption (\ref{EKBintro}) on the
contrary depends on the choice of coordinates. However, since both the
$u$ and the $v$ coordinate are easily shown to be 
finite in the region where $B$ is supported 
on $\tilde{S} \cap \{r \geq r_K\}$ in the given coordinate
system, assumption (\ref{EKBintro}) is automatically satisfied 
if we choose the $\delta$ in (\ref{initassump}) small enough since $B$
and is assumed to be of \emph{compact} support initially. Hence it
could be dropped by making $\delta$ even smaller. We have nevertheless 
included (\ref{EKBintro}) for conceptual reasons which 
will become apparent later in the proof.\footnote{The quantity
(\ref{EKBintro}) is related to a boundary term in 
the vectorfield multiplier estimate associated with the 
vectorfield $K$.} The condition (\ref{EKBintro}) would also be required if 
one eventually drops the assumption of compact support, for 
then (\ref{EKBintro}) imposes conditions on the decay of the 
fields near infinity.

Factors of $\sqrt{M}$ have been inserted 
in all formulae to make constants dimensionless. 
\subsection{Summary of the proof}
Before we embark upon an outline and a discussion of the 
proof, it is perhaps illuminating to compare and contrast 
the situation with the proof of  Price's law \cite{DafRod} for a self-gravitating spherically symmetric
scalar field in $3+1$ dimensions. It turns 
out that the techniques developed in the latter paper to derive decay rates
do not generalize to the system 
under consideration.
The underlying reason can be traced back to two 
crucial estimates applied in \cite{DafRod}. The first of these, 
which allows one to extract decay directly 
from the horizon, relies heavily on the homogeneity of the 
non-linear wave equation satisfied by 
the field $\phi$ in the scalar field model. 
The second estimate is made possible by the existence of an 
almost Riemann invariant, a quantity admitting better 
decay properties than the scalar field $\phi$ itself, which can be 
exploited to derive uniform decay of the energy in the area radius $r$. 
This decay played an important role in conjunction with the pigeonhole 
principle completing the argument in \cite{DafRod}.

In the five-dimensional case there is no almost Riemann 
invariant and hence no apparent analogue to obtain decay in $r$
for the energy in the asymptotic region. Moreover, the wave
equation (\ref{nonlinB}) satisfied by the dynamical field 
$B$ has an inhomogeneous part, which in particular appears
in the redshift estimate. These obstacles necessitate 
a very different approach to proving decay. The path we choose 
here is based on exploiting energy currents arising from 
vectorfield multipliers. This method was already central 
in the proof of the non-linear stability of Minkowski 
space \cite{ChristKlei} and has recently been applied at the linear level
in the black hole context for the first time \cite{DafRod2}. In the
latter paper, decay rates for a scalar field satisfying the homogeneous linear
wave equation on a four-dimensional Schwarzschild spacetime are
proven.\footnote{Clearly, this is the associated linear 
problem to the model of the self-gravitating scalar-field. Most
notably, it can be treated without any symmetry 
assumptions on the scalar field, cf.~\cite{DafRod2}.} 
Key to establishing decay, at least away 
from the horizon, is the application of a so-called Morawetz 
vectorfield. A careful analysis reveals that the 
decay-rates can be generalized to the linear problem 
associated with the non-linear problem studied here, namely 
the analysis of the linearized version of the wave 
equation (\ref{nonlinB}) on a fixed Schwarzschild-Tangherlini 
background. What is more, the method of compatible currents being 
very geometric and robust in nature in fact carries over to the non-linear 
problem suggesting that the decay rates 
(\ref{dechoz}), (\ref{deccent}), (\ref{decr}) 
may be established for the non-linear problem as well. 
However, in contrast to the linear case several non-linear error-terms 
now enter the various estimates, which cannot be controlled a-priori. 
This requires the introduction of a bootstrap argument 
to be applied in conjunction with the estimates obtained from 
the method of compatible currents.

It is noteworthy that the paper provides the first application 
of compatible currents techniques in a (non-linear) black 
hole context. The argument presented here
is generally more robust than that of \cite{DafRod} but is of course restricted 
to small data. More precisely, the method presented is expected to be 
appropriate to eventually address non-linear problems without 
symmetry, most famously the non-linear stability of the
Kerr-solution. In particular, since the technique is not 
bound to any dimension one should be able to reprove a 
version of ``Price's law'' \cite{DafRod} for small initial data 
along the lines of the present paper. 

%After a brief summary of the vectorfield method (\ref{vfm}), 
%I will outline the necessary bootstrap assumptions (\ref{btai}) 
%and finally describe how they are retrieved with the help 
%of the vectorfield method  (\ref{btri}).
%
%
\subsubsection{Compatible currents} \label{vfm}
The basic idea behind the exploitation of energy currents 
based on vectorfield multipliers is quite simple. 
We construct a Lagrangian whose field equation 
generates the non-linear wave equation (\ref{nonlinB}) 
satisfied by the squashing field $B$. The canonical 
energy momentum tensor $T_{\mu \nu}$ can be contracted 
with a vectorfield $V^\mu$ to produce
a one-form $P_\nu=T_{\mu \nu}V^\mu$. Finally, Stokes' theorem 
 relates the spacetime (or ``bulk'') integral 
of the divergence $\nabla^\nu P_\nu$ over a certain 
region to integrals along its boundary. This leads to the identity
\begin{equation} \label{bvfi}
\int_{\partial \mathcal{D}} P^\mu n_\mu = \int_{\mathcal{D}} \nabla_\mu P^\mu =
\int_{\mathcal{D}} \left[T_{\mu \nu} \pi^{\mu \nu} + V^\mu \nabla^\nu
  T_{\mu \nu} \right] \, .
\end{equation}
where $\pi^{\mu \nu} = \frac{1}{2}\left(\nabla^\mu V^\nu 
+ \nabla^\nu V^\mu\right)$ is the
deformation tensor of the vectorfield $V$.
One possible application of (\ref{bvfi}) is
to estimate a future boundary integral from the past 
boundary and the spacetime-term. On the other hand, for
some vector fields we will estimate a bulk-term from the 
boundary terms. The power of the method arises from an 
interplay between the identities associated with different 
vectorfields adapted to the geometry of particular regions. 
It is crucial that due to the Lagrangian structure both the 
boundary and the bulk term of (\ref{bvfi}) only depend on the 
$1$-jet of $B$. Suitably applied, the method ultimately produces weighted 
$L^2$-bounds on the fields from which pointwise bounds on the fields follow 
in the standard manner.
\subsubsection{The bootstrap} \label{btai} 
Before any bootstrap assumptions can be specified, 
coordinates have to be defined on the Penrose diagram. This 
turns out to be a rather subtle issue, 
intimately related to the bootstrap argument
itself. The crucial observation is that the 
coordinates have to be normalized \emph{to the future} of 
the bootstrap region, in order to capture the decay 
for late times in the estimates.\footnote{This is reminiscient of
  the situation in Christodoulou-Klainerman's proof of the stability of Minkowski space 
  \cite{ChristKlei}.} This is realized as 
follows. Consider the integral curves of the vectorfield 
$\nabla r$, foliating the black hole exterior.\footnote{Note that 
for convenience, we have assumed in Theorem \ref{asymptoticstab} 
that the initial data are also defined on such a curve, at least 
up to its intersection with a curve $r=r_K$.}
Each of these curves also 
intersects the curve of fixed area radius  $r=2\sqrt{M_f}$ 
(with $M_f$ being the final Bondi mass the latter is 
comfortably away from the horizon). Hence we can 
associate a \emph{geometric time} to 
any $\nabla r$ integral curve by using the affine 
parameter along the curve $r=2\sqrt{M_f}$.
\begin{figure}[h!]
\[
\input{introcor2.pstex_t}
\]
\caption{The choice of coordinates.}
\end{figure}
Now for each such ``time" $\tilde{\tau}$ on the curve, we 
construct a coordinate system $\mathcal{C}_{\tilde{\tau}}$ (depending
on $\tilde{\tau}$!) on the black hole exterior by the 
following procedure. We find the point $A$ on the $\nabla r$ curve associated to $\tilde{\tau}$, where $r^2=4m$. The $r=const$ curve through $A$ will intersect the data at some point $D$. The affine length from $A$ to $D$ along that curve defines the \emph{coordinate time} at $A$.\footnote{See definition (\ref{timedef}). We add a factor of $\sqrt{M}$ in order to avoid dividing by zero when we state decay in $t$.} The actual coordinate system $(u,v)$ is finally defined by imposing that
$t=\frac{u+v}{2}=T$ holds on the integral curve of the vectorfield $\nabla r$ 
starting at the point $A$, at least up to the point $B$ 
where the integral curve intersects a certain constant $r_K$-curve, fixed once 
and for all, which is chosen to lie close to the horizon. Moreover we 
set $r^\star=\frac{v-u}{2}=0$ at the point $A$ and 
$r_{,v} = \frac{1}{2}\left(1-\mu\right)$ on $\overline{BE}$. 
There is some choice to complete the coordinate 
system by specifying $r_{,u}$ on $\overline{BC}$. For most practical
calculations we will use Eddington-Finkelstein type
coordinates, setting $\nu = -\frac{1}{2}\left(1-\mu\right)$ 
on $\overline{BC}$. In any case, the bounds 
proven will be manifestly independent 
on the choice of $u$ coordinate on $\overline{BC}$. 

An important issue immediately arising from the way we define
the coordinates is that the notion of a constant $t$ slice 
differs in the different coordinate systems depending on the choice 
of $\tilde{\tau}$. (Of course the analogous statement holds for the notion of 
timelike surfaces of constant $r^\star$.) Nevertheless, 
we will show that the coordinate 
systems remain uniformly close to each other in a suitable sense, in particular that 
the $t_{\tilde{\tau}}$ coordinate of the initial data slice between
$r_K$ and the support radius is always 
close to $\sqrt{M}$, however large we choose $\tilde{\tau}$. A
detailed analysis is given in section \ref{stabcorsec}.

Every $\tilde{\tau}$ defines a $T$, which in turn defines 
a region $\mathcal{A}\left(T\right)$ depicted in Figure \ref{bootsreg}. It is 
the region, enclosed by the $t=T$-curve up to some point $B^\prime$ with 
coordinates $(T,r^\star_K)$, the null-line $v=T+r^\star_{K}$ linking $B^\prime$ 
with the horizon, a horizon piece, the null-line $v=2\sqrt{M}+r^\star_K$,
the $t=2\sqrt{M}$ piece and the $u=u_0$ null-line on which the field $B$ is 
identically zero by the assumption of compact
support\footnote{Note that this null-line has a geometric significance
  by the assumption of compact support. The exact value of
  $u_0$ will depend on the coordinate system chosen.}. Here $r^\star_K
= \sup_{t < T}
  r^\star\left(t,r_K\right)$. Another curve, $r^\star=r^\star_{cl}$,
  located to the right of $r^\star=r^\star_K$ will also be introduced
  and fixed. We now choose a small 
constant $c$ and define the bootstrap region to be the region 
associated to the largest time $\tilde{\tau}_B$, such that for any
$\sqrt{M} \leq \tilde{\tau} \leq \tilde{\tau}_B$  the 
following ``statement $\mathcal{P}$'' holds in the associated 
region $\mathcal{A}\left(T\left(\tilde{\tau}\right)\right)$ 
in the coordinate system $C_{\tilde{\tau}}$:
\begin{enumerate}
\item In the subregion $\{r^\star \geq r^\star_K \} \cap \mathcal{A}\left(T\right) $, the area radius satisfies 
\begin{eqnarray} \label{efrel}
\Big| r^\star - \left[ r\left(t,r^\star\right) + \sqrt{\frac{M_A}{2}} \left(\log
\left(\frac{r\left(t,r^\star\right)-\sqrt{2M_A}}{r\left(t,r^\star\right)+\sqrt{2M_A}}\right)
+ p \right) \right] \Big| < c \sqrt{M}
\end{eqnarray}
with 
\begin{equation} 
p = - 2\sqrt{2} - \log \frac{2-\sqrt{2}}{2+\sqrt{2}}
\end{equation}
and $M_A$ defined to be the Hawking mass at
the point $\left(T,r^\star=0\right)$.\footnote{The reader should note
  that in Schwarzschild with $M=M_A$ the left hand side of
  (\ref{efrel}) is identically zero. The coordinate $r^\star$ is then
  the so called Regge-Wheeler tortoise coordinate.} \label{boot1}
\item We have\footnote{This assumption states in particular that the
  initial data slice is both near and to the past of the
  bootstrap region. It ensures that the bootstrap region does not move
  away from the data.}
\begin{equation} 
\frac{1}{2} \sqrt{M} < \sup_{\tilde{S} \cap \{r^\star \geq r^\star_K\} \cap
  \{ u \geq u_0\} } t < \frac{3}{2}\sqrt{M} \, .
\end{equation} 
\label{boot1b}
\item the weighted energy $E^K_B$ defined in (\ref{ekbfir}) 
satisfies $E^K_B\left(\tilde{T}\right) < c M $ on all 
arcs $\{ t = \tilde{T} < T \}
  \cap \{r^\star \geq r^\star_K\} \subset \mathcal{A}\left(T\right)$. \label{boot2} 
\item the energy-flux satisfies $m\left(u_{hoz},v_2\right) 
- m\left(u_{hoz},v_1\right) < \frac{c M^2}{\left(v_{1+}\right)^2}$ 
for any $v_1 \leq v_2$ along the part of the horizon located in 
$\mathcal{A}\left(T\right)$, where $v_{i+} = \max\left(1,v_i\right)$. \label{boot3} 
\item \begin{equation} 
m\left(u_{r^\star_{cl}},v\right) - m\left(u_{hoz},v\right) < \frac{c M^2}{v_+^2}
\end{equation}
holds in $\mathcal{A}\left(T\right)$. Here $v_+ =\max(1,v)$. \label{boot4} 
%on all $v=constant$ null-lines in the region 
%$r^\star \leq r^\star_{cl}$. 
\item the integral bound 
\begin{equation} \label{intebound}
\tilde{F}^Y_B = \int r^3 \frac{\left(B_{,u}\right)^2}{\Omega^2} du < \frac{C_L M^2}{v_+^2} \textrm{\ \ \ for \ \ \ $C_L = \sup_{r^\star \geq r^{\star}_{cl}} \frac{1}{1-\mu}$}
\end{equation}
holds along lines of constant $v$ in the region 
$\{r^\star \leq r^\star_{cl} \} \cap \{ u \leq T-r^\star\left(T,r_K\right) \} \cap \mathcal{A}\left(T\right)$, 
corresponding to a decay of energy as measured by local observers near the
horizon.\footnote{That is to say the quantity $\tilde{F}^Y_B$ measures exactly the energy which is not seen by the Hawking energy at the horizon.} \label{boot5} 
\end{enumerate}
We define the set
\begin{equation} 
A = \Big\{ \tilde{\tau} \in \left[\sqrt{M}, \infty\right) 
\  \ \Big| \ \  \mathcal{P}_{T\left(\hat{\tau}\right)} 
\textrm{\ \ \ holds in
  $\mathcal{A}\left(T\left(\hat{\tau}\right)\right)$ 
\textrm{for all $\hat{\tau} \leq \tilde{\tau}$} \Big\} }
\subset \left[\sqrt{M}, \infty\right)  \ \ \ \ \ \,
\end{equation}
which will be shown to be open, closed and non-empty. This 
implies that the statement $\mathcal{P}$ holds on the 
entirety of the black hole exterior. 
The decay rates of Theorem \ref{asymptoticstab} 
follow immediately after proving that the coordinate 
systems used in the bootstrap converge to one which is 
close to the one asserted by Theorem \ref{asymptoticstab}.

The openness of the set $A$ follows from a 
straightforward continuity argument. The difficult 
part in closing the bootstrap therefore is to ``improve''
the statement $\mathcal{P}$ on the 
closure of the set $\mathcal{A}\left(T\right)$.
\subsubsection{Closing the bootstrap} \label{btri}
The third bootstrap assumption is shown to 
imply $\frac{1}{\left(t_i\right)^2}$ decay
of the energy-flux on the arcs 
$\{ t = t_i \} \cap \{ r^\star \geq r^\star_K \} \cap \{ u \geq
\frac{1}{11}t_{i} \}$, from which pointwise bounds 
on the field $B$ and its $v$-derivative are obtained. 
Additionally, strong decay of $B$ in the area radius $r$  
can be extracted from the boundedness of $E_B^K$. 
The assumptions also provide sufficient control over 
the coordinate functions at late times. In particular one 
determines the relation between the 
area radius $r\left(u,v\right)$ and the coordinate 
$r^\star=\frac{v-u}{2}$, at least in the region where 
$r^\star \geq r^\star_K$. For late times this relation converges  
to the well-known formula expressing the area radius 
$r$ in terms of the tortoise coordinate $r^\star$ of 
the five-dimensional Schwarzschild metric 
as captured by bootstrap assumption \ref{boot1}. It follows 
in particular that the value of $r$ does not change much 
(the corrections are shown to be of order $\frac{1}{t}$) along a
$r^\star=const$-curve in the region $r^\star \geq r^\star_K$, 
allowing us to go back and forth between the 
two in the course of the paper. Moreover, bootstrap 
assumption \ref{boot1} is improved.

Various constant $r$- and constant $r^\star$-curves in the region $r \geq r_K$ will 
play a crucial role, since certain integrands arising from the 
method of compatible currents admit good signs in appropriate
regions.\footnote{By bootstrap assumption \ref{boot1} constant $r$ and constant $r^\star$-curves are close to one another in that region.}  
The $r^\star=r^\star_{cl}$-curve, occurring in the 
bootstrap for instance (along which $r \approx r_{cl}$ by the previous 
remarks), is determined by various requirements defined later 
but is in any case located to the right of the aforementioned $r=r_K$. The latter 
curve on the other hand, can and will be chosen close to the horizon 
providing a source of smallness in the bootstrap argument.
A second source of smallness arises from Cauchy stability:
After picking some $r_K$ we can choose a very 
late time $t_0$ up to which the fields are still small 
and after which terms like $\frac{C\left(r_K\right)}{t}$, with $C\left(r_K\right)$ a constant depending on the choice of $r_K$, are small.
\\

We now turn to various energy currents arising from vectorfield 
multipliers and describe how the bootstrap is closed. The remarkable 
properties admitted by the Hawking mass for the system 
under consideration manifest themselves in 
the identity (\ref{bvfi}) for the vectorfield
\begin{equation}
T = \frac{4r_{,v}}{\Omega^2} \partial_u -  \frac{4r_{,u}}{\Omega^2} \partial_v  \, .
\end{equation}
The spacetime-term associated to the $T$-energy identity vanishes and one obtains a
relation between boundary-terms, which are precisely 
the associated energy fluxes. The monotonicity 
of the Hawking mass equips all boundary terms with 
signs when applied in the region\footnote{From the vectorfield point of view this follows from the fact that $T$ is timelike, that the normal to the region is non-spacelike and the positivity properties of $T_{\mu \nu}$. Cf. (\ref{bvfi}).}
(cf.~figure \ref{regionD})
\begin{eqnarray} \label{regDintro}
{}^{u_H}\mathcal{D}^{{r^\star_{cl}},u_J}_{[t_1,t_2]} := &&\Big(\{t_1
\leq t \leq t_2 \} \cap \{{r^\star} \geq r^\star_{cl} \} \cap \{ u \geq
u_J \}\Big) \nonumber \\ &\cup& \Big\{ \{t_1+r^\star_{cl} \leq v \leq t_2 + r^\star_{cl} \} \cap \{r^\star \leq r^\star_{cl} \} \cap \{ u \leq u_H \} \Big\} \, .
\end{eqnarray}
Such regions arise from a dyadic decomposition of the 
bootstrap region between $t_0$ and $T$ with $t_{i+1}=1.1t_i$
playing a crucial role later in the argument.

It can be shown that
the boundary-terms associated to the vectorfield
\begin{equation}
X = f\left(r^\star\right) \left( \partial_u - \partial_v \right) \, ,
\end{equation}
for some 
carefully chosen bounded function $f$,
are controlled by the energy-flux (i.e.~the $T$ boundary-terms) 
and the integral bound (\ref{intebound}) when applied in 
the region (\ref{regDintro}). The function $f$ is in turn chosen 
such that the spacetime-term of $X$ admits a positive sign. In
conjunction with the bootstrap assumptions this results in 
a $\frac{1}{\left(t_i\right)^2}$-decay bound for a 
positive spacetime integral in the dyadic region 
${}^{t_{i+1}-r^\star_{cl}}\mathcal{D}^{{r^\star_{cl}},\frac{1}{10}t_i}_{[t_i,t_{i+1}]}$,
which will prove useful in controlling the spacetime integrals of
other vectorfields. 

Close to the horizon, in a
characteristic rectangle $\left[u_1=t_1-r^\star_{cl}, u_2=u_{hoz}\right] \times
\left[v_1=t_1+r^\star_{cl}, v_2=t_2+r^\star_{cl}\right]$ associated to
the dyadic region 
${}^{u_{hoz}}\mathcal{D}^{r^\star_{cl}, u_J}_{\left[t_1,t_2\right]}$, 
we will apply the vectorfield 
\begin{equation}
Y = \frac{\alpha\left(r^\star\right)}{\Omega^2} \partial_u +
\beta\left(r^\star\right) \partial_v
\end{equation} 
for appropriately chosen functions $\alpha$ and $\beta$ (cf.~the bold
rectangle in Figure \ref{clobo}). The strategy is to control the 
future-null boundary integrals from the past boundary- 
and the associated spacetime term.\footnote{Physically, the
  boundary terms of the $Y$ vectorfield correspond to the energy 
flux as measured by a local observer near the horizon.}
The integrand of the latter contains a part admitting a good sign, 
which can be used in combination with the spacetime 
term of $X$ to control the remaining spacetime term of $Y$. 
Moreover, one ingoing boundary-term being located completely 
in the region $r^\star \geq r^\star_{cl}$, is always controlled 
by the energy flux and hence decays like $\frac{1}{t^2}$. Applying the 
identity in the characteristic rectangle with the bottom being  
$v_0=t_0+r^\star_{cl}$, where 
an appropriate smallness assumption holds by Cauchy stability, 
and the top being $v=\tilde{v}$ for any 
$v_0 \leq \tilde{v} \leq T + r^\star_K$  immediately 
yields uniform boundedness for both the boundary terms 
and the good spacetime term of $Y$. 
The argument can be improved by a pigeonhole 
principle applied in every characteristic rectangle. Namely, 
one extracts from the good spacetime term of $Y$ a ``good $F_Y$-slice", 
i.e.~a slice on which the local $Y$-energy density 
decays like $\frac{1}{v_i}$ times the good spacetime term 
plus a contribution from the energy in the 
region $r^\star \geq r^\star_{cl}$. This is depicted as  
the dotted line in Figure \ref{clobo} below.\footnote{Alternatively one can
  extract a ``good $F_T$-slice'' on which the $T$-energy flux is
  improved. This will come in handy later.} Applying the 
vectorfield identity for $Y$ again in a region 
with the good slice as its past-boundary, one exports the 
$\frac{1}{v_i}$-decay to all dyadic rectangles. Iterating the
procedure one obtains $\frac{C}{\left(v_i\right)^2}$
decay for all boundary-terms and the good spacetime term of $Y$.
The decay of the $Y$ boundary terms leads to the pointwise bound
$|r^\frac{3}{2} \frac{B_{,u}}{r_{,u}}| \leq \frac{C}{v}$ in the region 
$r^\star \leq r^\star_{cl}$, which can be exported to 
the region $r^\star_{cl} \leq r^\star \leq \frac{9}{10}t$ 
using the energy estimate and the decay in the central region. 

With the pointwise bound on $r^\frac{3}{2} \frac{B_{,u}}{r_{,u}}$ 
at our disposal, we can finally make use of the Morawetz vectorfield  
\begin{equation}
K = \frac{\left(u+a\right)^2}{M} \partial_u + \frac{\left(v-a\right)^2}{M} \partial_v \, 
\end{equation}
for a constant $a$.\footnote{This suitably chosen constant 
defines the origin of the vectorfield.} As mentioned previously, 
its application is necessitated by the lack of an almost 
Riemann invariant and it proves 
crucial in the derivation of decay rates away from the 
horizon. The vectorfield identity for the region 
${}^{u_H=\tilde{T}-r^\star_{K}}\mathcal{D}^{{r^{\star}_K},u_0}_{[t_0,\tilde{T}]}$
associated to any $\tilde{T} \leq T$ and some 
large $t_0$ relates a future 
boundary term to a past boundary term, a horizon-term 
and their associated spacetime term.

The boundary terms on the $\tilde{T}$-arc contain 
``good''-terms which are precisely the strongly weighted 
energies $E^K_B\left(\tilde{T}\right)$ of the 
second bootstrap assumption and error-terms. The 
vectorfield identity is now 
exploited so as to estimate this ``good" term on the 
future arc in terms of \emph{all} other terms entering the
 identity. These latter quantities are in turn shown 
to be small or of good sign, which will finally 
improve assumption \ref{boot2}. 
To derive the smallness for the various terms, 
it will be necessary to subdivide the domain of 
integration and to apply different estimates in 
each region, carefully taking the geometry of the black 
hole into account.\footnote{It is here where the pointwise bound on
$r^\frac{3}{2} \frac{B_{,u}}{r_{,u}}$ established earlier enters.}
It should be emphasized that these estimates belong to 
the most subtle ones in the paper. They make
crucial use of the monotonicity manifest in the Raychaudhuri
equations (\ref{eom1}) and (\ref{eom2}), and exploit an exponential 
decay associated with the redshift very close to the horizon 
by introducing an intermediate region between $r^\star=r^\star_K$ 
and the horizon. 

For the boundary terms, there 
are two sources from which the smallness is finally 
obtained: One is the choice of the curve $r=r_K$, which 
can be chosen very close to the horizon. 
The other stems from the choice of a late time 
$t_0$ up to which the initial data has only changed 
by an amount as small as we may wish by Cauchy stability and after 
which the good decay estimates, i.e.~the 
weight of $\frac{1}{t_0}$ carries over.

To establish smallness for the spacetime term appearing 
in the $K$-vector-identity, on the other hand, a further 
argument is needed. This term consists of a ``main''-term, 
which is the one that appears
in the linear case, and error-terms. The error-terms 
can be dealt with very analogously to the treatment of the 
error-boundary terms. The main term is shown to 
admit a good sign for $r^\star \leq r^\star_{cl}$ and 
for some $r^\star \geq R^\star$ for some $R^\star$. The remaining piece 
in the central region is divided into dyadic 
regions, $t_{j+1}=1.1t_j$. Each $K$-integral of 
such a dyadic region can be controlled by $t_{j+1}$ 
times the spacetime integral of the vectorfield $X$
in that region. Since the $X$-bulk term decays 
like $\frac{1}{\left(t_{j+1}\right)^2}$ as outlined above, 
summing up the dyadic regions yields smallness for 
the main $K$-spacetime-term (arising from the large time $t_0$, 
where we start the dyadic decomposition). This improves 
bootstrap assumption \ref{boot2}.

With the third bootstrap assumption being improved 
on all arcs $\tilde{T} \leq T$ it follows that the 
decay of the energy has been improved on all 
arcs.\footnote{This is a consequence of the previously 
mentioned fact that the expression for $E^K_B$ contains strong 
weights from which the decay can be extracted.} 
As a corollary, the same 
decay is obtained through any achronal 
hypersurface lying completely in the region 
$r^\star_K \leq r^\star \leq \frac{9}{10}t$.  
\begin{figure}[h!] 
\[
\input{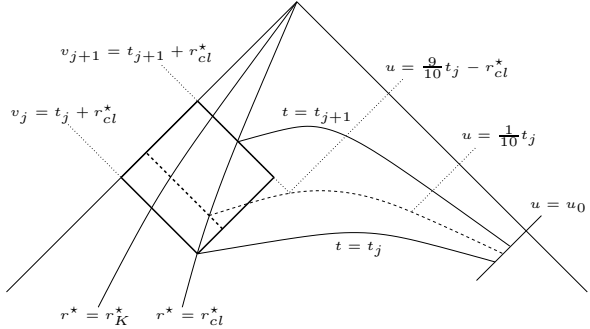}
\]
\caption{Closing the bootstrap.} \label{clobo}
\end{figure}
In the final step we find in each dyadic rectangle a 
``good $F_T$-slice'' on which the energy flux is improved 
to $\frac{\epsilon}{\left(v_i\right)^2}$, very analogous to 
finding a ``good $F_Y$-slice'' as described above. Combining it with 
the improved decay on the associated arc (cf.~the dotted slice 
in Figure \ref{clobo}), the domain 
of dependence property improves the bootstrap 
assumptions \ref{boot3} and \ref{boot4}. Additionally,
we can finally find a good $F_Y$-slice in each 
characteristic rectangle (improving 
assumption \ref{boot5} on that slice), which in conjunction 
with the energy decay now being improved to $\frac{\epsilon}{v^2}$ 
everywhere in $r^\star \leq \frac{9}{10}t$, can  
be exported to all $v$-slices. Hence assumption \ref{boot5} is also
retrieved with a better constant. This completes the proof 
that the set $A$ is indeed closed and the main theorem 
follows in view of the previous remarks.

It should be noted that the decay rate that can 
be extracted in this argument is limited by the weights appearing in 
the $K$ vectorfield, i.e.~by the decay in the 
central region.\footnote{Clearly, better decay in the
  central region could immediately be exported to the horizon by a
  reiteration of the pigeonhole principle in conjunction with the
  vectorfield $Y$.} In particular, we cannot 
derive the stronger decay $\frac{1}{v^{3-\epsilon}}$ 
near the horizon obtained in \cite{DafRod} for the massless
scalar field. It is an interesting question whether other methods 
can improve the decay rates proven in this paper.

\subsection{Outline of the paper}
We start by introducing the biaxial Bianchi IX 
model and some notation  (section \ref{bibibasic})
before defining the aforementioned future-normalized
coordinate system $C_{\tilde{\tau}}$ in section \ref{Coordinates}. 
Various a-priori bounds, which can be obtained 
without invoking the main bootstrap argument and turn out to be helpful 
at many stages of the paper are derived in section \ref{Basic}. An 
important point to keep in mind, however, is that the decay of the 
energy in the area radius \emph{cannot} be obtained by 
these methods due to the lack of an almost Riemann
invariant for the model under consideration. 
The method of compatible currents is explained in more detail 
in section \ref{compatiblecur}, 
where moreover the relevant identities associated with 
the regions considered later are derived. 
In particular, the Hawking mass is recovered as a potential of 
a certain vectorfield-current (section \ref{HawkT}). After 
defining the bootstrap 
assumptions (section \ref{bootstrap}), various bounds for the fields 
are derived from them and the stability of the 
coordinate systems $\mathcal{C}_{\tilde{\tau}}$ defined in section
\ref{Coordinates} is established (section \ref{analboot}).  The 
identities associated to 
the vectorfields $Y$ and $X$ are analyzed in sections \ref{Ysection} 
and \ref{Xsection}. Here a somewhat lengthy argument is 
pursued to construct the 
function $f$ implicit in the vectorfield $X$, which finally 
ensures that its spacetime term admits a positive sign. 
Section \ref{XconY} reveals how to control the weighted energies 
produced by $Y$ near the horizon with the help of the 
vectorfield $X$. The relevant version of the pigeonhole 
principle is also explained at this stage. Finally, in section \ref{vecKsec}
the Morawetz vectorfield $K$ is introduced and the 
necessary estimates to control the various error-integrals, 
as outlined in the introduction, are performed. 
Everything is put together in section \ref{closboot}, where 
the bootstrap is closed. The paper finishes with some final 
remarks and open questions.
\section{Biaxial Bianchi IX} \label{bibibasic}
The class of biaxial Bianchi IX metrics was introduced in
\cite{Bizon}. We recall that these spacetimes are topologically 
$\mathcal{M} = \mathcal{Q} \times SU(2)$, where $\mathcal{Q}$ 
is a two-dimensional manifold and that global 
coordinates $(u,v)$ can be found on $\mathcal{Q}$ expressing 
the metric of $\mathcal{M}$ in the form
\begin{equation} \label{metansatz}
g = -\Omega^2\left(u,v\right) du dv + \frac{1}{4} r^2\left(u,v\right) \left(e^{2B\left(u,v\right)} \left(\sigma_1^2 +
  \sigma_2^2 \right) + e^{-4B\left(u,v\right)} \sigma_3^2 \right)
\end{equation}
where $B$ and $r$ are functions $\mathcal{Q} \rightarrow \mathbb{R}$
and the $\sigma_i$ form a basis of left invariant one-forms on
$SU(2)$. Note that if $B=0$, the symmetry is enhanced to
$\left(SU\left(2\right)_L \times SU\left(2\right)_R\right) \slash \mathbb{Z}^2 =
SO\left(4\right)$ and the metric reduces to 
the five dimensional Schwarzschild-Tangherlini metric in view of 
the higher dimensional
version of a well-known theorem due to Birkhoff.\footnote{Note also 
the relation
  between the familiar round metric ($d\omega_{\mathbb{S}^3}^2$) and 
the bi-invariant metric on $S^3$, $d\omega_{\mathbb{S}^3}^2 =
\frac{1}{4} \left(\sigma_1^2+ \sigma_2^2+\sigma_3^2 \right)$.} In this
sense, $B$ is the dynamical degree of freedom ruling the model. 
See \cite{DafHol} for a more detailed discussion.

The vacuum Einstein equations for the above model reduce to a system of
$1+1$ dimensional PDEs on the quotient manifold $\mathcal{Q}$:
\begin{equation} \label{eom1}
\partial_u \left( \Omega^{-2} \partial_u r \right) =
-\frac{2r}{\Omega^2}\left((B_{,u})^2\right),
\end{equation}
\begin{equation} \label{eom2}
\partial_v \left( \Omega^{-2} \partial_v r \right) = 
-\frac{2r}{\Omega^2}\left((B_{,v})^2\right),
\end{equation}
\begin{equation} \label{revol}
r_{,uv} = -\frac{1}{3} \frac{\Omega^2 \rho}{r} - \frac{2
  r_{,u}r_{,v}}{r} = -\frac{\Omega^2}{r^3}m -
  \frac{1}{3}\frac{\Omega^2}{r} \left(\rho-\frac{3}{2}\right),
\end{equation}
\begin{equation} \label{omegaevol}
\partial_u \partial_v \log \Omega = \frac{\Omega^2 \rho}{2r^2} + 
\frac{3}{r^2} r_{,u} r_{,v} - 3 \left(B_{,v}\right)\left(B_{,u}\right)
 = \frac{3 \Omega^2}{2r^4} m +
\frac{\Omega^2}{2r^2}\left(\rho-\frac{3}{2}\right) - 3 \frac{\theta
  \zeta}{r^3} \, ,
\end{equation}
\begin{equation} \label{Bevol}
B_{,uv} = -\frac{3}{2}\frac{r_{,u}}{r} B_{,v}
-\frac{3}{2}\frac{r_{,v}}{r} B_{,u}+ \frac{\Omega^2}{3r^2} \left(e^{-8B}-e^{-2B}\right) \, .
\end{equation}
Here we have defined the quantity\footnote{The quantity $\rho$ is related to the scalar curvature of the group orbit by $R=\frac{4}{r^2}\rho$.}
\begin{equation} \label{scalcurv}
\rho = 2e^{-2B} -\frac{1}{2}e^{-8B} \leq \frac{3}{2}  \, ,
\end{equation}
with the inequality following from elementary calculus. Equality holds
if and only if $B=0$. Note that the non-linear wave equation 
(\ref{Bevol}) can be written as (\ref{nonlinB}) with $\square$ 
being the d'Alembertian of the metric (\ref{metansatz}).

A remarkable feature of the above system is the
existence of a function $m\left(u,v\right)$ called the \emph{Hawking
mass} and defined by
\begin{equation} \label{Hawkmass}
m = \frac{r^2}{2} \left(1+\frac{4r_{,u}r_{,v}}{\Omega^2} \right)
\, .
\end{equation}
Since the inequalities $r_{,u} < 0$ and $r_{,v} \geq 0$ were 
shown \cite{DafHol} to hold everywhere on the black hole exterior\footnote{They hold on the initial data for small
  perturbations of Schwarzschild-Tangherlini and are seen to
  be preserved by equations (\ref{eom1}) and (\ref{eom2}).}, 
the Hawking mass has the following
monotonicity properties there:
\begin{equation} \label{dum}
\partial_u m = -4r^3 \frac{\lambda}{\Omega^2} \left(B_{,u}\right)^2
+ r \nu \left(1-\frac{2}{3}\rho \right) \leq 0 \, ,
\end{equation}
\begin{equation} \label{vum}
\partial_v m = -4r^3 \frac{\nu}{\Omega^2} \left(B_{,v}\right)^2
+ r \lambda \left(1-\frac{2}{3}\rho\right) \geq 0 \, .
\end{equation}
This allows the derivation of energy estimates for the field $B$,
which plays an important role at all stages of the present paper. 
The existence of these estimates was already an essential ingredient 
in the proof of the orbital stability \cite{DafHol}.

We conclude this section recalling
some notation introduced in \cite{DafHol}. We set
\begin{equation}
\lambda = r_{,v} \textrm{ \ \ \ \ \ } \nu = r_{,u} \textrm{ \ \ \ \ \
} \zeta= r^\frac{3}{2} B_{,u} \textrm{ \ \ \ \ \ } \theta =
r^\frac{3}{2} B_{,v} \, 
\end{equation}
and introduce the quantities
\begin{equation} \label{kapgamdef}
\kappa = \frac{\lambda}{1-\mu} = \frac{\Omega^2}{-4\nu} \textrm{ \ \ \
  \ and \ \ \ \ } \gamma = \frac{-\nu}{1-\mu} = \frac{\Omega^2}{4\lambda}
\end{equation}
satisfying 
\begin{equation} \label{kapevol}
\kappa_{,u} = \kappa \left(\frac{2}{r^2} \frac{\zeta^2}{\nu} \right)
\, , 
\end{equation}
\begin{equation} \label{gammaevol}
\gamma_{,v} = \gamma \left(\frac{2}{r^2}
\frac{\theta^2}{\lambda}\right) \, ,
\end{equation}
as well as the auxiliary quantities 
\begin{equation} \label{deltaB}
\varphi_1(B) = \left(1-\frac{2}{3}\rho\right) -8B^2 \textrm{ \ \ \ \ \
  and \ \ \ \ \ } \varphi_2(B) =
  \frac{4B}{3}\left(e^{-8B}-e^{-2B}\right)+8B^2 \, ,
\end{equation}
both of order $B^3$. The volume element associated to
(\ref{metansatz}) is 
\begin{equation}
dVol = \sqrt{g} du dv dw = r^3 \frac{\Omega^2}{2} du dv
dA_{\mathcal{S}^3} = r^3 \Omega^2 dt dr^\star dA_{\mathcal{S}^3}
\end{equation}
where in the standard Euler-coordinates on $SU(2)$ (cf. \cite{DafHol})  
\begin{equation}
dA_{\mathcal{S}^3} = \frac{1}{8} \sin \theta d{\bf \omega} =
\frac{1}{8} \sin \theta d\theta d\phi d\psi  \textrm{ \ \ \ and hence
  \ \ \ } \int dA_{\mathbb{S}^3} = 2\pi^2\, .
\end{equation}
The monotonicity of the Hawking mass justifies the definitions
\begin{equation}
m_{min}, m_{max} \textrm{ \ \ \ for the minimal and maximal Hawking mass.}
\end{equation}
Furthermore, the quantity $M_f$ will denote the final Bondi mass 
and $M$ the mass of the perturbed Schwarzschild solution. 
The mass $M$ determines the scaling of the problem and 
I have normalized all quantities appropriately using factors of $M$. 
In particular, ``smallness'' always refers to dimensionless quantities. 

We write $C\left(\epsilon\right)$ for a constant
satisfying $\lim_{\epsilon \rightarrow 0} C \left(\epsilon\right)
= 0$. The notation $a \sim b$ is used if there exist 
uniform constants $c_1,c_2$ with $c_1 \leq \frac{a}{b} \leq c_2$.  Finally, we define
\begin{equation}
v_+ = \max\left(1,v\right)   \textrm{ \ \ \ and \ \ \ } v_{i+} = \max\left( 1,v_i \right) \, .
\end{equation}

\section{Choice of coordinates}  \label{Coordinates}
As mentioned in the introduction, the choice of coordinates is already
a rather delicate issue for the problem under consideration. 
Although the final result does not depend on the choice of 
coordinates, the bootstrap-techniques applied in the proof 
require the coordinates to be normalized to the future of the 
bootstrap region introduced in section \ref{bootstrap}. 
If on the contrary one normalized the coordinates on the initial data, 
one would not be able to obtain the improved decay of the 
fields at late times from the estimates, roughly speaking 
because contributions from the initial data, which have not 
yet decayed, enter the estimates. This necessitates, after 
a purely geometric definition of ``time'' for $\nabla r$ integral curves 
on the black hole exterior, the introduction 
of a different coordinate system 
$\mathcal{C}_{\tilde{\tau}}=(u_{\tilde{\tau}},v_{\tilde{\tau}})$ defined  with 
respect to every such ``time'' $\tilde{\tau}$. All such coordinate systems $\mathcal{C}_{\tilde{\tau}}$ are 
defined on the set
\begin{equation} \label{calD}
\mathcal{D} = J^-\left(\Scri^+\right) \cap J^+\left(\tilde{S}_{r_K}\right)
\end{equation}
of the black hole exterior. In section 9 we 
shall exploit the bootstrap assumptions to establish that 
-- in a certain region -- these coordinate systems are uniformly 
close to each other in a suitable sense. It should be observed that 
the coordinate systems $\mathcal{C}_{\tilde{\tau}}$ are 
different from the coordinate system asserted in 
Theorem \ref{asymptoticstab}. In the last section of the paper 
we will show that the coordinate system $\mathcal{C}_{\tilde{\tau}}$ 
for $\tilde{\tau} \rightarrow \infty$ is close to the one 
asserted by Theorem \ref{asymptoticstab}.

We begin by considering the family of $\nabla r$ integral 
curves starting out from some $r=r_K$-curve which 
is chosen close to the horizon\footnote{The choice will 
provide a source of smallness later in the bootstrap argument.} 
such that still $1-\mu \geq \tilde{c}>0$ holds for a 
small $\tilde{c}$, and ending at spacelike infinity $i^0$. These curves
foliate $\mathcal{D} \cap \{r \geq r_K\}$. Moreover, every curve 
admits a unique point where $r = 2\sqrt{m}$. 
We pick any such curve and label the corresponding point by $A$. 
Denote the mass at $A$ by $m_A$ and consider the 
curve $r^2 = 4m_A$ going 
through $A$ and intersecting the initial data at some point $D$. 
Let $\tau_{AD}$ be the affine length of the constant $r$ curve (with 
tangent vector normalized to one) connecting $A$ and $D$. Finally, define 
\begin{equation} \label{timedef}
T = \sqrt{M} + \frac{\tau_{AD}}{\sqrt{1-\frac{2m_A}{r^2}}} = \sqrt{M} + \sqrt{2} \ \tau_{AD}
\end{equation}
to be the time associated to the $\nabla r$ curve under consideration. In this way 
we can assign a notion of time to any $\nabla r$ integral curve. Considering next 
the curve $r^2=4M_f$ with affine parameter $\tilde{\tau}$ starting 
from the initial data, we obtain a map
\begin{equation} \label{thetamap}
\vartheta : \left[0,\infty\right) \ni \tilde{\tau} \mapsto T \in \left[\sqrt{M},\infty\right) \,
\end{equation}
which is defined by taking $\tilde{\tau}$ to the time associated to 
the $\nabla r$ integral curve which intersects the curve $r^2 = 4M_f$ 
at $\tilde{\tau}$. The map $\vartheta$ is easily seen to be continuous and surjective. 

For every $\tilde{\tau}$ a coordinate system $(u,v)$ is defined as follows. 
Let $A$ have coordinates $(u,v)=(T=\vartheta\left(\tilde{\tau}\right),T=\vartheta\left(\tilde{\tau}\right))$. Set $\kappa=\gamma=\frac{1}{2}$ along 
the $\nabla r$ integral curve up to $r=r_K$. Since
\begin{equation}
\nabla r \left(u+v\right) = \left(\nabla r\right)^u + \left(\nabla
r\right)^v = \frac{2}{\Omega^2} \left(-\nu - \lambda\right) =
\frac{2\left(1-\mu\right)}{\Omega^2} \left(\gamma - \kappa\right) \, ,
\end{equation}
we have that $t=\frac{u+v}{2}$ (thus defining $t$) is indeed 
equal to the constant $T$ on the $\nabla r$ integral curve through $A$. 
Moreover $r^\star = \frac{v-u}{2}$ is equal to $0$ at $A$.
\begin{figure}[h!]
\[
\input{coors2b.pstex_t}
\]
\caption{The choice of coordinates.} \label{codiag}
\end{figure}
Let the $\nabla r$ curve defining the coordinate system 
intersect $r=r_K$ at $B$. We erect the constant $u=u_B$-ray 
to the past of $B$ and set $\kappa=\frac{1}{2}$ there. 
The coordinate system is completed by specifying 
the $u$-coordinate on $\overline{BC}$. 
We set $\nu = - \frac{1}{2}\left(1-\mu\right)$ on
 $\overline{BC}$. This might send the horizon to $u=\infty$, namely if 
$1-\mu=0$ at $C$.\footnote{Of course, one 
does not expect this to be the case generically.} We will see that 
in these coordinates $v \rightarrow \infty$ at $\Scri^+$. The
 coordinates thus defined will be refered to as Eddington Finkelstein
coordinates. We also allow ourselves to move freely between the 
coordinates $(u,v)$ and $\left(t=\frac{v+u}{2},
r^\star=\frac{v-u}{2}\right)$. 

Clearly if $\tilde{\tau}=0$, then the associated integral curve 
coincides for $r\geq r_K$ with the curve on which the initial data is 
defined, and $t=\sqrt{M}$ defines the initial-data slice in $r \geq
r_K$. Note that in any coordinate system associated to some
$\tilde{\tau}>0$, a slice on which 
$t=constant$ does in general \emph{not} agree with 
a $\nabla r$-slice. However, 
once we have introduced the bootstrap assumptions,
we will be able to show that the two slices mentioned 
remain uniformly close to each other in $r^\star \geq r^\star_K$ for any 
$\tilde{\tau} > 0$. This argument is postponed to 
section \ref{StabCord}. Here we only introduce

\begin{notation} \label{tnotation}
Let $t^{\tilde{\tau}_A}_{\tilde{\tau}_B}$ denote the $t$-coordinate, measured
in the coordinate system defined by $\tilde{\tau}_A$, of 
the point defined by the intersection of the $\nabla r$ 
integral curve determined by $\tilde{\tau}_B$ and the curve 
$r^2=4m\left(\vartheta\left(\tilde{\tau}_A\right),r^\star=0\right)$.
\end{notation}

We conclude with a remark on the differentiability of the coordinate systems. 
Due to the ``cusp'' at the point $B$ the coordinate system is only $C^1$: 
The quantities $\kappa$ and $\gamma$ (and by definition (\ref{kapgamdef}) the first derivatives of the 
area radius function $r\left(u,v\right)$) are clearly continuous. The second 
derivatives $r_{,vv}$ and $r_{,uu}$ however, are discontinuous at the point $B$. 
This could be avoided by applying an appropriate smooth 
interpolating function in a small neighborhood around 
the point $B$. However, we will see later that the 
bootstrap involves only first derivatives (and hence 
continuous quantities) and that the regularity suffices 
to close the bootstrap.
\section{Basic estimates} \label{Basic}
In this section we are going to show that given an appropriate
smallness assumption on the field $B$, 
namely (\ref{initassump}) and (\ref{initassump2}), 
the field and its derivatives remain small on the entire $\mathcal{D}$. 
Since this ``first round'' is independent of the 
main bootstrap argument, it provides a good way-in to familiarize 
oneself with the basic estimates applied in different regions of 
the black hole exterior. The bounds in this section will be proven
in the coordinate system $\mathcal{C}_{\tilde{\tau}}$ associated to
\emph{any} $\tilde{\tau} \geq 0$.\footnote{Note that if
  $\tilde{\tau}=0$ the coordinates are normalized on initial data.} In this
context it is crucial that the smallness assumptions
(\ref{initassump}) and (\ref{initassump2}) are manifestly independent
of the coordinate choice. From \cite{DafHol} we recall that
\begin{equation} \label{inisma}
1 - \frac{m_{min}}{m_{max}} < \epsilon\left(\delta\right) \textrm{ \ \ \ \ with $\lim_{\delta \rightarrow 0} \epsilon\left(\delta\right) = 0 $}
\end{equation}
can be chosen arbitrarily small by an appropriate assumption 
on the initial data. We will abbreviate $\epsilon\left(\delta\right)$ 
by $\epsilon$ in the following. In view of the 
monotonicity-properties of the Hawking mass ((\ref{dum}) and (\ref{vum})) the 
mass difference between any two points cannot exceed
$m_{max} \cdot \epsilon\left(\delta\right)$. We note
\begin{lemma}
If (\ref{inisma}) holds, then on the horizon we have
\begin{equation} \label{gohoz}
0 \leq 1-\mu \leq \frac{2m_{max}}{r^2} \epsilon
\end{equation}
\end{lemma}
\begin{proof}
From \cite{DafHol} we have both $1-\mu \geq 0$ on the black hole exterior, as well as the Penrose inequality $1-\frac{2M_f}{r^2} \leq 0$ holding on the horizon with $M_f$ the final Bondi mass. Combining this with (\ref{inisma}) immediately yields (\ref{gohoz}).
\end{proof}
\begin{corollary} \label{hozfluct}
The area radius $r$ satisfies
\begin{equation}
|r_+-r_-| \leq \frac{4m_{max}}{\sqrt{m_{min}}} \epsilon 
\textrm{\ \ \ on $\mathcal{H}^+$ } 
\end{equation}
with $r_\pm$ being the maximal (minimal) value of $r$ on the horizon.
\end{corollary}
\begin{corollary} \label{rcurvcol}
For any given $\eta > 0$ we can choose the $\delta$ of the initial 
data so small that for some $r=r_K$ curve located 
completely in $\mathcal{D}$ the estimate
\begin{equation}
r_K-r \leq \eta
\end{equation}
holds in $r \leq r_K$.
\end{corollary}
For the estimates in this section only we will explicitly couple the location of 
the $r=r_K$ curve to the smallness of the initial data. 
In particular we define the curve $r_K$ by 
\begin{equation} \label{etadef}
1-\frac{2m_{max}}{\left(r_{K}\right)^2} = \epsilon^\frac{1}{3}
\end{equation}
It follows easily that the maximum $r$ difference in the region $r \leq r_K$ satisfies
\begin{equation} \label{Rdiff}
\Delta r \leq r_K - r_{-} \leq 3\frac{m_{max}}{2\sqrt{2m_{min}}} \epsilon^\frac{1}{3} \, .
\end{equation}
%
%
% \subsection{Smallness and decay in $r$ of $B$, $\frac{\zeta}{\nu}$, $\theta$}
\begin{proposition} \label{simpsmall}
In any coordinate system $\mathcal{C}_{\tilde{\tau}}$ and with the assumptions 
of Theorem \ref{asymptoticstab} on the initial data we have
\begin{equation} \label{wantb}
 r |B| + \sqrt{r} |\theta| + M^{\frac{1}{4}}\Big|\frac{\zeta}{\nu}\Big| \leq \sqrt{M} \cdot C\left(\delta\right) 
\end{equation}
everywhere in $\mathcal{D}$. Moreover the
coordinate function $\kappa$ satisfies
\begin{equation}
\Big|\kappa - \frac{1}{2}\Big| \leq C\left(\delta\right)
\end{equation} 
everywhere in $\mathcal{D}$. Here the constant 
$C\left(\delta\right)$ can be made arbitrarily small 
by choosing the $\delta$ of the initial data sufficiently small.
\end{proposition}
Proposition \ref{simpsmall} will immediately follow from Propositions
\ref{Bprop}-\ref{zetnuprop} (plus their associated Corollaries) 
proven in the remainder of the section, each of them establishing 
the bounds in different regions of the black hole exterior. 
Note that the radial decay of $B$ promised by Proposition 
\ref{simpsmall} is weaker than that of Theorem \ref{asymptoticstab}.
%
%
%
%
%\subsection{Proof of Proposition \ref{simpsmall}}
\begin{proposition} \label{Bprop}
In the region $\mathcal{D}$ we have
\begin{equation} \label{Bbound}
|B \left(u,v\right)| \leq C_1\left(\epsilon\left(\delta\right),\delta\right) \,.
\end{equation}
where $C_1\left(\epsilon\left(\delta\right), \delta \right)$ can 
be made arbitrarily small by choosing the $\delta$ in
(\ref{initassump}) 
small enough.
\end{proposition}

\begin{proof}
We integrate from the initial data to any point in the region $r \geq
r_K$ and estimate as follows
\begin{eqnarray} \label{eqB2}
| B\left(u,v\right) | &\leq&  \delta \left(\frac{\sqrt{M}}{r}\right)^\frac{3}{2}
\left(u_{data},v\right) 
+ \int_{u_{data}}^{u}
  \frac{\zeta}{r^\frac{3}{2}} du \nonumber \\ &\leq&  \delta \left(\frac{\sqrt{M}}{r}\right)^\frac{3}{2}
   + \sqrt{\int_{u_{data}}^{u}
  \frac{\zeta^2 \left(1-\mu \right)}{-\nu} du} \sqrt{\int_{u_{data}}^{u}
  \frac{-\nu}{\left(1-\mu\right)r^3} du} \nonumber \\ 
&\leq&  \delta \left(\frac{\sqrt{M}}{r}\right)^\frac{3}{2}  + \frac{1}{\sqrt{2}} 
\sqrt{m_{max} - m_{min}} \sup_{r \geq r_K}
  \left(\frac{1}{\sqrt{1-\mu}}\right) \frac{1}{r} \nonumber \\ &\leq&
  \delta \left(\frac{\sqrt{M}}{r}\right)^\frac{3}{2} +
  \epsilon^\frac{1}{3}\sqrt{\frac{m_{max}}{r^2}} \, ,
\end{eqnarray}
which proves (\ref{Bbound}) in that region. \\
Next we turn to the region $\mathcal{D} \cap \{ r \leq r_K \}$.
We choose a constant $C > 2M^\frac{3}{4} {r_-^{-\frac{3}{2}}}\delta$ such that $|B| \leq C$ still implies that
\begin{equation}
\frac{3}{2} - \rho = \frac{3}{2} -
\left(2e^{-2B}-\frac{1}{2}e^{-8B}\right) \leq \frac{3}{2} \frac{m}{r^2}
\end{equation}
in $\mathcal{D} \cap \{ r \leq r_K \} $. For $r_K$ sufficiently close
to the horizon it is easily seen that $C=\frac{1}{10}$ is good
enough. Define the region
\begin{equation}
\mathcal{R} = \Bigg\{\left(u,v\right) \in \mathcal{D} \cap \{ r \leq r_K
\}  \ \ : \ \ |B\left(\bar{u},\bar{v}\right) | < C \textrm{\ \ for all \ }
  \left(\bar{u},\bar{v}\right) \in J^{-}\left(u,v\right)\Bigg\}
\end{equation}
which is clearly open and non-empty. We are going to apply a
bootstrap argument: Pick a point in the closure of
$\mathcal{R}$, where $B \leq C$ by continuity. We are
going to improve this bound by showing that in the causal past of that point
$B$ is in fact smaller than $\frac{C}{2}$. By continuity it follows 
that the set $\mathcal{R}$ is also closed, hence must constitute 
the entirety of 
$\mathcal{D} \cap \{ r \leq r_K \}$. The argument proceeds in 
two steps. First we make use of the redshift
estimate integrating the equation
\begin{equation} \label{zetanuueq}
\left(\frac{\zeta}{\nu}\right)_{,v} = -\frac{3}{2} \frac{\theta}{r} -
\frac{4}{3} \frac{\kappa}{\sqrt{r}} \left(e^{-8B}-e^{-2B}\right) -
\frac{\zeta}{\nu} \left(\frac{4\kappa}{r^3} m +
\frac{4\kappa}{3r}\left(\rho-\frac{3}{2}\right) \right)
\end{equation}
from the initial data yielding
\begin{eqnarray} \label{rseq}
\frac{\zeta}{\nu} \left(u,v\right) &=& \frac{\zeta}{\nu}\left(u,v_i\right) e^{-\int_{v_i}^v \left[\frac{4\kappa}{r^3} m +
\frac{4\kappa}{3r}\left(\rho-\frac{3}{2}\right) \right]
  \left(u,\bar{v}\right) d\bar{v}}  \\ &+& \int_{v_i}^v e^{-\int_{\bar{v}}^v \left[\frac{4\kappa}{r^3} m +
\frac{4\kappa}{3r}\left(\rho-\frac{3}{2}\right) \right]
  \left(u,\hat{v}\right) d\hat{v}} \left[-\frac{3}{2} \frac{\theta}{r} -
\frac{4}{3} \frac{\kappa}{\sqrt{r}} \left(e^{-8B}-e^{-2B}\right)
\right]\left(u,\bar{v}\right) d\bar{v} \nonumber
\end{eqnarray}
and hence
\begin{eqnarray} \label{redshift}
\Big| \frac{\zeta}{\nu} \left(u,v\right)\Big| &\leq&  \delta \frac{\left(\sqrt{M}\right)^\frac{3}{2}}{r_-} + \frac{3}{2} \sqrt{\int_{v_i}^v e^{-\int_{\bar{v}}^v
    \left[\frac{4\kappa}{r^3} m\right] \left(u,\hat{v}\right)
    d\hat{v}} \kappa\left(u,\bar{v}\right) d\bar{v}}\sqrt{\int_{v_i}^v
  \frac{\theta^2}{\kappa r^2} \left(u,\bar{v}\right)d\bar{v}}
\nonumber \\
&\phantom{XX}+& \frac{4}{3} \sup_{r \leq r_K} \left[|e^{-8B}-e^{-2B}|
  \frac{1}{\sqrt{r}}
\right] \int_{v_i}^v e^{-\int_{\bar{v}}^v
    \left[\frac{2\kappa}{r^3} m\right] \left(u,\hat{v}\right)
    d\hat{v}} \kappa \left(u,\bar{v}\right) d\bar{v} \nonumber \\
&\leq& \delta \frac{\left(\sqrt{M}\right)^\frac{3}{2}}{r_-} + 
\frac{3}{2} \sqrt{\epsilon} \sqrt{\frac{m_{max}
    \left(r_K\right)^3}{4m_{min} r_{-}^2}} + \frac{2}{3}
    \frac{\left(r_K\right)^3}{\sqrt{r_-} m_{min}}
     \left(e^{8C}+e^{2C}\right) \equiv \tilde{C} M^\frac{1}{4}
     \, . \nonumber \\
\end{eqnarray}
It follows that $|\frac{\zeta}{\nu}|$ is bounded (but note that the
last term might not be small) in that region. In the second 
step we integrate from the $r=r_K$
curve, on which $B$ is small by (\ref{eqB2}), or the initial data 
to obtain
\begin{equation}
B\left(u,v\right) - B\left(u_R,v\right) = \int_{u_{r_K}}^u
\frac{\zeta}{r^\frac{3}{2}} \left(\bar{u},v\right) d\bar{u} 
\end{equation}
and use the previous bound (\ref{eqB2})
\begin{eqnarray}
|B\left(u,v\right)| &\leq& \delta\frac{\sqrt{M}^\frac{3}{2}}{r^\frac{3}{2}} +
  \epsilon^\frac{1}{3} \sqrt{\frac{m_{max}}{r^2}} + \tilde{C} M^\frac{1}{4}  
  \int_{u_{r_K}}^u \frac{-\nu}{r^\frac{3}{2}} \nonumber \\ &\leq& 
 \delta\frac{\sqrt{M}^\frac{3}{2}}{r^\frac{3}{2}} +\epsilon^\frac{1}{3}
  \sqrt{\frac{m_{max}}{r^2}} +
  \frac{\tilde{C}}{2} \frac{M^\frac{1}{4}}{\left(r_-\right)^\frac{5}{2}}
  \left(\left(r_K\right)^2 - \left(r_-\right)^2 \right) \, . \nonumber 
\end{eqnarray}
Now because the $r$ difference is given by (\ref{Rdiff}) 
in the region under consideration, we have indeed shown that 
$B$ is smaller than $\frac{C}{2}$ in $\mathcal{R}$ for an 
appropriate choice of $\epsilon$. By continuity the 
set $\mathcal{R}$ is also closed. Hence 
$\mathcal{R} = \mathcal{D} \cap \{r \leq r_K \}$.
\end{proof}

\begin{corollary} \label{Bcor}
In $r \geq r_K$ we have that
\begin{equation}
|B \left(u,v\right)| \leq \sqrt{M} \frac{C_2\left(\epsilon,\delta\right)}{r} \,.
\end{equation}
\end{corollary}
\begin{proof}
This is the statement of (\ref{eqB2}). 
\end{proof}
It is instructive to compare the $\frac{1}{r}$-decay of 
Corollary \ref{Bcor} with the analogous estimate derived 
for the massless scalar field in four dimensions \cite{DafRod}. 
In the latter case, one obtained by the above method 
$\frac{1}{\sqrt{r}}$-decay. There existed an 
almost Riemann invariant, i.e.~a 
certain combination of the field and its derivatives, however, 
admitting better decay properties than the 
field or its derivatives alone. Via this quantity, it was 
possible to improve the decay in $r$ of the field
itself to $\frac{1}{r}$, which was it turn sufficient to 
extract energy decay in $r$. In five dimensions there 
is no almost Riemann invariant and energy decay in $r$ 
will only be obtained from the application of the 
Morawetz vectorfield $K$ in the context of the 
bootstrap argument pursued later.
\begin{corollary} \label{zetnucor}
In the region $\mathcal{R} = \mathcal{D} \cap \{
r \leq r_K \}$ we have
\begin{equation} \label{zenusmall}
\Big| \frac{\zeta}{\nu} \Big| \leq  M^\frac{1}{4} C_3\left(\epsilon,\delta \right)   \, .
\end{equation}
\end{corollary}
\begin{proof}
This follows from revisiting the red-shift estimate (\ref{redshift})
above, this time improving the estimate for the 
$\left(e^{-8B}-e^{-2B}\right)$- term by Proposition \ref{Bprop}, 
which implies that $|e^{-8B}-e^{-2B}|$ is $\epsilon$-small. In this way
we obtain a smallness factor for all the terms involved in (\ref{redshift}).
\end{proof}

\begin{proposition}
In $\mathcal{D}$ we have
\begin{equation} \label{kapclo1}
\Big|\kappa - \frac{1}{2}\Big| \leq  C_4\left(\epsilon,\delta\right) \, .
\end{equation}
\end{proposition}
\begin{proof}
Integrating (\ref{kapevol}) from the $t=T \cap \{ r \geq r_K \}$ 
surface to any point in the region $ \{ r \geq r_K \}$ yields
\begin{equation}
\kappa\left(u,v\right) = \kappa\left(u_T,v\right)\exp
\left(\int_{u_T}^u \frac{2}{r^2} \frac{\left(1-\mu\right) \zeta^2}
{\left(\nu\right)\left(1-\mu\right)} \right) \left(\bar{u},v\right) d\bar{u} \,.
\end{equation}
If the point under consideration lies to the future of the $t=T$
 hypersurface ($u \geq u_T$), the upper bound
$\kappa \leq \frac{1}{2}$ follows from monotonicity, 
whereas the lower bound is obtained via
\begin{eqnarray} \label{helbk}
|\kappa\left(u,v\right)| &\geq& \frac{1}{2}\exp\left(\sup_{r \geq r_K} \left(\frac{2}{r^2
 \left(1-\mu\right)}\right) \int_{u_T}^u
 \frac{\zeta^2\left(1-\mu\right)}{\nu}
 \left(\bar{u},v\right)d\bar{u}\right) \nonumber \\
&\geq& \frac{1}{2} \exp\left(\frac{-4m_{max} \epsilon^\frac{2}{3}}{r_K^2}\right) \, .
\end{eqnarray}
On the other hand, integrating (\ref{kapevol}) from 
the null line $u=T-r^\star\left(T,r_K\right)$ downwards the lower 
bound follows from monotonicity and the upper one 
by using (\ref{zenusmall}) 
\begin{equation} \label{heubk}
|\kappa\left(u,v\right)| \leq \frac{1}{2}
  \exp\left(\left(C_3\left(\epsilon,\delta\right)\right)^2\right) \exp
 \left(-\frac{2\sqrt{M}}{r\left(u,v\right)}+\frac{2\sqrt{M}}{r_K}\right) \leq \frac{1}{2} + \tilde{c}_4\left(\epsilon,\delta\right) \, .
\end{equation}
Since the $r$-difference in the region $r \leq r_K$ is $\epsilon^\frac{1}{3}$-small 
by (\ref{Rdiff}), we obtain the desired
upper bound for $\kappa$ in particular on all of $r=r_K$. 

Now any point located in the past of the $t=T$ hypersurface 
and satisfying $r \geq r_K$ can be reached by integrating 
(\ref{kapevol}) from either $t=T$ or from $r=r_K$ where 
the upper bound (\ref{heubk}) has already been established. 
The lower bound for $\kappa$ at such a point follows from
monotonicity, the upper one from
\begin{eqnarray}
|\kappa\left(u,v\right)| &\leq& \left(\frac{1}{2} + \tilde{c}_4\left(\epsilon,\delta\right)\right) \exp\left(\sup_{r \geq r_K} \left(\frac{2}{r^2
 \left(1-\mu\right)}\right) \int_{u_T}^u
 \frac{\zeta^2\left(1-\mu\right)}{-\nu}
 \left(\bar{u},v\right)d\bar{u}\right) \nonumber \\
&\leq& \left(\frac{1}{2} + \tilde{c}_4\left(\epsilon,\delta\right)\right)
 \exp\left(\frac{4m_{max}\epsilon^\frac{2}{3}}{r_K^2}\right) \, .
\end{eqnarray}
To extend the estimates to the entire region $r \leq r_K$ we integrate
(\ref{kapevol}) from the
$r=r_K$ curve (on which the lower bound (\ref{heubk}) and the upper bound (\ref{helbk}) has been established) to the horizon. Again the upper bound follows from monotonicity and for the lower one we write
 \begin{equation}
\kappa\left(u,v\right) = \kappa\left(u_0,v\right)\exp
\left(\int_{u_0}^u \frac{2\nu}{r^2} \frac{\zeta^2}
{\left(\nu\right)^2} \right) \left(\bar{u},v\right) d\bar{u} 
\end{equation}
and estimate, using (\ref{zenusmall})
\begin{equation}
|\kappa\left(u,v\right)| \geq \frac{1}{2}
 \exp\left(\frac{-4m_{max}\epsilon^\frac{2}{3}}{\left(r_K\right)^2}\right) \ \exp\left(\left(C_3\left(\epsilon,\delta\right)\right)^2\right) \exp
 \left(-\frac{2\sqrt{M}}{r\left(u,v\right)}+\frac{2\sqrt{M}}{r_K}\right) \, .
\end{equation}
Taking again (\ref{Rdiff}) into account, we obtain the
lower bound for $\kappa$ also in that region.
\end{proof}
With the bound on $\kappa$ established we also have good control over the 
quantity $\lambda=\kappa\left(1-\mu\right)$. In particular $\lambda <
1$ everywhere and $\lambda$ becomes very small (perhaps zero) at the
horizon. In particular, it follows that 
\begin{equation}
|\theta| < \frac{|\theta|}{\kappa\left(1-\mu\right)} = \frac{|\theta|}{\lambda}
\end{equation}
holds everywhere on $\tilde{S} \cap \{r\geq r_K\}$ and hence the
$\frac{\theta}{\lambda}$-part of the smallness condition
(\ref{initassump}) implies smallness for $\theta$ as well. With this
in mind we can prove
\begin{proposition}
In $\mathcal{D}$ we have
\begin{equation}
|\theta| \leq C_5\left(\epsilon,\delta\right) \sqrt{\frac{M}{r}}
\end{equation}
\end{proposition}
\begin{proof}
We rewrite equation (\ref{Bevol}) as
\begin{equation} \label{dut}
\partial_u \theta = -\frac{3}{2} \frac{\lambda \zeta}{r} +
\frac{\Omega^2}{3\sqrt{r}} \left(e^{-8B}-e^{-2B}\right)
\end{equation}
and integrate it from $\tilde{S} \cap \{r \geq r_K\}$ to any point in $\mathcal{D}$. 
We note that for $|B|$ small we can find a constant $K$ such that
\begin{equation} \label{Khelp}
\left(e^{-8B} - e^{-2B}\right)^2 \leq K \left(1 - \frac{2}{3}\rho\right)
\end{equation}
holds. This constant approaches $\frac{9}{2}$ as $|B|$ goes to zero. We 
then estimate
\begin{eqnarray}
|\theta\left(u,v\right) | \leq 
M^\frac{3}{4} \frac{\delta}{r}\left(u_{data},v\right) 
+ \frac{3}{2}
 \sqrt{\int_{u_{data}}^{u} \frac{\zeta^2 \left(1-\mu\right)}{-\nu}
  d\bar{u}}\sqrt{\int_{u_{data}}^{u}\frac{\kappa
 \lambda \left(-\nu\right)}{r^2} \left(\bar{u},v\right) d\bar{u}}
 \nonumber \\ 
+ \sup\left(\frac{4}{3}\kappa\right) \sqrt{K} \sqrt{\int_{u_{data}}^{u} r \left(-\nu\right)\left(1-\frac{2}{3}\rho\right) d\bar{u}}\sqrt{\int_{u_{data}}^{u} \frac{-\nu}{r^2}\left(\bar{u},v\right)d\bar{u}}
\nonumber \\
\leq M^\frac{3}{4}\frac{\delta}{r} + \sqrt{\epsilon}
 \frac{\sqrt{m_{max}}}{\sqrt{r}} + 8\sqrt{\epsilon}
 \frac{\sqrt{M}}{\sqrt{r}} \, .
\end{eqnarray}
\end{proof}
Finally, we extend the bound on $\frac{\zeta}{\nu}$ to the region $r
\geq r_K$.
\begin{proposition} \label{zetnuprop} 
We have
\begin{equation} \label{zetnuevery}
\Big| \frac{\zeta}{\nu} \Big| \leq C_6\left(\epsilon,\delta\right)
\end{equation}
in all of $\mathcal{D}$.
\end{proposition}
\begin{proof}
Integrate equation (\ref{zetanuueq}) from the $r=r_K$-curve, where
$|\frac{\zeta}{\nu}| \leq C_3 \left(\epsilon,\delta\right)$ by
Corollary \ref{zetnucor} out to infinity. 
Note that due to the estimate proven for the field $B$ in Corollary
\ref{Bcor} we may achieve (choosing $\delta$ small enough) that
\begin{equation}
\frac{3}{2}-\rho \leq \frac{3m}{r^2} 
\end{equation}
holds in the region $r \geq r_K$. Using again (\ref{Khelp}) we can follow 
the string of estimates
\begin{eqnarray}
\Big| \frac{\zeta}{\nu} \left(u,v\right) \Big| &\leq& \Big|
\frac{\zeta}{\nu} \left(u,v_{r_K}\right) \Big| + \frac{3}{2}
\sqrt{\int_{v_{r_K}}^v \frac{\theta^2}{\kappa} \left(u,\bar{v}\right)
  d\bar{v}} \sqrt{\int_{v_{r_K}}^v \frac{\lambda}{\left(1-\mu\right)r^2}
  \left(u,\bar{v}\right)  d\bar{v}} \nonumber \\
&\phantom{X}& +\frac{4}{3} \sqrt{\int_{v_{r_K}}^v
    \left(e^{-8B}-e^{-2B}\right)^2 r \left(u,\bar{v}\right) d\bar{v}}
\sqrt{\int_{v_{r_K}}^v \frac{\kappa^2}{r^2} \left(u,\bar{v}\right) d\bar{v} }
\nonumber \\
&\leq& \Big|
\frac{\zeta}{\nu} \left(u,v_{r_K}\right) \Big| +
\frac{3}{2}\sqrt{\epsilon} \frac{2}{\epsilon^\frac{1}{6}} \frac{\sqrt{m_{max}}}{\sqrt{r_K}}
 \nonumber \\ 
&+& \frac{4 \sqrt{K}}{3\sqrt{2}} \sup_{r \geq r_K} \left(
\frac{1}{\sqrt{\lambda}}\right) \sqrt{\int_{v_{r_K}}^v
    \left(1-\frac{2}{3}\rho \right) r \lambda  d\bar{v}}\sqrt{\int_{v_{r_K}}^v
  \frac{\lambda}{\left(1-\mu\right) r^2} d\bar{v}}  \nonumber \\
&\leq& \Big|
\frac{\zeta}{\nu} \left(u,v_{r_K}\right) \Big| +
\frac{3\sqrt{m_{max}}}{\sqrt{r_K}}\epsilon^\frac{1}{3} 
+ \frac{4 \sqrt{K}}{3\sqrt{2}} \sup_{r \geq r_K} \left(
\frac{1}{\sqrt{\kappa}\left(1-\mu\right)}\right) \frac{\sqrt{m_{max}}}{\sqrt{r_K}}
\sqrt{\epsilon}
\nonumber \\
&\leq& M^\frac{1}{4} C_3 \left(\epsilon,\delta\right) +
\frac{3\sqrt{m_{max}}}{\sqrt{r_K}}\epsilon^\frac{1}{3} 
+ 2\frac{4 \sqrt{K}}{3\sqrt{2}} \frac{2\sqrt{M}}{\sqrt{r_K}} \epsilon^\frac{1}{6}
\end{eqnarray}
to conclude the result.
\end{proof}
So far we have shown that $r B$, $\frac{\zeta}{\nu}$, and $\sqrt{r} \theta$ are
small and that $\kappa$ is everywhere close to $\frac{1}{2}$ for the
perturbed spacetime. Estimates for some higher 
derivative quantities will be required later. However, since 
all bounds can be considerably improved once the bootstrap 
assumptions have been introduced, we postpone the derivation of 
further pointwise estimates to section \ref{pointwiseE}. Here we only note
\begin{proposition} \label{omvclo}
On $\mathcal{D}$ we have, independent of the coordinate system $\mathcal{C}_{\tilde{\tau}}$, the bound
\begin{equation} \label{quibv}
\Big| \frac{\Omega_{,v}}{\Omega} - \frac{m}{r^3} \Big| \leq \frac{1}{\sqrt{M}} C_7\left(\epsilon\right) 
\end{equation}
\end{proposition}
\begin{proof}
From the fact that $\kappa=\frac{1}{2}$ on $\{ t=T \} \cap \{r^\star \geq r^\star\left(T,r_K\right)\}$ (hence $\kappa_{,r^\star}$ = 0 there) and on $\{u=T-r^\star\left(T,r_K\right)\} \cap \{ t \leq T\}$ (hence $\kappa_{,v}=0$ on this null-line) the bound (\ref{quibv}) follows on these sets. We can obtain the quantity $\frac{\Omega_{,v}}{\Omega}$ at any point on $\mathcal{D}$ by integrating equation (\ref{omegaevol}) from the aforementioned set to the desired point.  Inserting the estimates of Proposition \ref{simpsmall} gives (\ref{quibv}) everywhere. 
\end{proof}

{\bf Remark: } The quantity $\frac{\Omega_{,v}}{\Omega}$ is
discontinuous at the point $B$ in the coordinate 
system $\mathcal{C}_{\tilde{\tau}}$. This discontinuity 
is propagated along the null-line $v=v\left(B\right)$ when 
integrating the quantity $\partial_u \frac{\Omega_{,v}}{\Omega}$ 
(which is continuous! (cf. \ref{omegaevol})) in $u$ (cf.~also Appendix 
\ref{reggree}).
\\

We conclude the section with a useful bound for 
the quantity $\gamma$ in the region 
$\mathcal{D} \cap \{t \leq T\} \cap \{ r \geq r_K \}$. 
\begin{proposition} \label{gamma}
In $\mathcal{D} \cap \{t \leq T\} \cap \{ r \geq r_K \}$ we have in the
coordinate system $\mathcal{C}_{\tilde{\tau}}$
\begin{equation} \label{gamclo1}
C_8 \left(\epsilon,\delta\right) \leq \gamma - \frac{1}{2} \leq 0
\end{equation}
\end{proposition}
\begin{proof}
Integrate (\ref{gammaevol}) from the $t=T$-slice 
in the past-direction. By monotonicity $\gamma \leq
\frac{1}{2}$ is obvious. The other direction is derived from 
\begin{equation}
\gamma \left(u,v\right) = \gamma \left(u, v_T\right) \exp
\left(\int_{v}^{v_T}  -\frac{2}{r^2} \frac{\theta^2}{\lambda}
\left(u,\bar{v}\right) d\bar{v} \right)
\end{equation}
and the estimate
\begin{eqnarray}
\gamma \left(u,v\right) &\geq& \frac{1}{2} \exp\left[-\left(\sup_{\{r \geq r_K \}
  \cap \{t \leq T \}} \frac{2}{r^2
  \left(1-\mu\right)}\right) \int_{v}^{v_T} \frac{\theta^2}{\kappa}
  \left(u,\bar{v}\right) d\bar{v}\right] \nonumber \\
&\geq& \frac{1}{2} \exp
  \left[C_8\left(\epsilon\right)\right] \, ,
\end{eqnarray}
which follows by choosing the mass fluctuation small enough. 
\end{proof}
We close the section by emphasizing once more that 
the bounds proven in this section are independent 
of the particular coordinate system used, i.e.~of how large 
we choose ${\tilde{\tau}}$ (and hence $T$). In this context it is important 
that the smallness assumptions  (\ref{initassump}) and (\ref{initassump2}) are
invariant under a change of coordinates. 
\section{Compatible currents} \label{compatiblecur}
\subsection{The basic identity}
Varying the Lagrangian
\begin{equation} \label{Lagrangian}
\mathcal{L} = \frac{1}{2} g^{\mu \nu} \partial_\mu B
\partial_\nu B + \frac{1}{2r^2} \left(1-\frac{2}{3}\rho \right)
\end{equation}
with respect to $B$ leads to the non-linear wave 
equation (\ref{nonlinB}) satisfied by the field $B$. 
We associate to (\ref{Lagrangian}) the energy momentum tensor 
\begin{equation} \label{enmom}
T_{\mu \nu} = \partial_\mu B \partial_\nu B - \frac{1}{2} g_{\mu \nu}
\left(\partial B \right)^2 - \frac{1}{2r^2} g_{\mu \nu}
\left(1-\frac{2}{3}\rho \right) 
\end{equation}
satisfying the equation
\begin{equation}
\nabla^\mu T_{\mu \nu} = \frac{1}{r^3} \left(1-\frac{2}{3}\rho \right)
\nabla_\nu r \, .
\end{equation}
Given any vectorfield $V$ we can define its deformation tensor
\begin{equation}
\pi^{\mu \nu}_V = \frac{1}{2} \left(\nabla^\mu V^\nu +
\nabla^\nu V^\mu \right) 
\end{equation}
and the vector
\begin{equation} \label{palpha}
P^\alpha = g^{\alpha \beta} T_{\beta \delta} V^\delta \, .
\end{equation}
The method of compatible currents is based on the following
basic identity for an arbitrary vector field $V$:
\begin{equation} \label{intid}
-\nabla_\alpha P^\alpha = - \left(T_{\alpha \beta}
 \pi_V^{\alpha \beta} + \left(\nabla^\beta T_{\alpha \beta}\right) 
 V^\alpha \right) \, .
\end{equation}
\subsection{Useful formulae}
In $\left(u,v\right)$-coordinates the components 
of the energy momentum tensor (\ref{enmom}) read
\begin{eqnarray}
T_{uu} &=& \left(\partial_u B\right)^2 \, , \nonumber \\
T_{vv} &=& \left(\partial_v B\right)^2  \, , \nonumber \\
T_{uv} &=& -\frac{1}{2r^2} g_{uv} \left(1-\frac{2}{3}\rho \right) =
\frac{1}{4r^2} \Omega^2 \left(1-\frac{2}{3}\rho \right)  \, , \nonumber \\
T_{ij} &=& = - \frac{1}{2} g_{ij}
\left(\partial^\alpha B \partial_\alpha B +
\frac{1}{r^2}\left(1-\frac{2}{3}\rho \right) \right)  \, .
%\nonumber \\ 
%g^{ij} T_{ij} &=& -\frac{3}{2} \left(\partial^\alpha B \partial_\alpha
%B + \frac{1}{r^2}\left(1-\frac{2}{3}\rho \right) \right)  \, .
\end{eqnarray}
The vectorfields $V$ used in this paper have $u$ and $v$ components 
only and will furthermore depend only on these two variables. 
For such vectorfields we compute the components of their deformation tensor:
\begin{eqnarray}
\pi^{uu} &=& \frac{4}{\Omega^2} \partial_v
\left(\frac{V_v}{\Omega^2}\right) \, ,
\nonumber \\ 
\pi^{vv} &=& \frac{4}{\Omega^2} \partial_u
\left(\frac{V_u}{\Omega^2}\right)  \, ,
\nonumber \\ 
\pi^{uv} &=& \frac{2}{\left(\Omega^2\right)^2} \left(\partial_v V_u +
\partial_u V_v\right)  \, , \nonumber \\ 
g_{ij} \pi^{ij} &=& -\frac{6}{r} \left(\frac{\nu}{\Omega^2} V_v +
\frac{\lambda}{\Omega^2}V_u \right) \, .
\end{eqnarray}
Finally, the following explicit formulae for the contraction
\begin{equation}
T_{\mu \nu} \pi^{\mu \nu} = T_{uu} \pi^{uu} + T_{vv} \pi^{vv} +
2T_{uv} \pi^{uv} + T_{ij} \pi^{ij}
\end{equation}
will be useful:
\begin{eqnarray}
T_{\mu \nu} \pi^{\mu \nu} &=& \frac{4}{\Omega^2} \Bigg(\left(\partial_u
B\right)^2 \partial_v \left(\frac{V_v}{\Omega^2}\right) + \left(\partial_v
B\right)^2 \partial_u \left(\frac{V_u}{\Omega^2}\right)  + \nonumber \\
&& \phantom{XXXXX} \frac{1}{4r^2}
\left(\partial_v V_u + \partial_u V_v \right) 
\left(1-\frac{2}{3}\rho \right)  \Bigg) \nonumber \\
&& +\frac{3}{r} \left(\frac{\nu}{\Omega^2} V_v + \frac{\lambda}{\Omega^2}V_u
\right) \left(\partial^\alpha B \partial_\alpha B +
\frac{1}{r^2}\left(1-\frac{2}{3}\rho \right) \right) 
\end{eqnarray}
and
\begin{eqnarray} \label{basicintegrand}
-T_{\mu \nu} \pi^{\mu \nu} - V^\nu \nabla^\mu T_{\mu \nu} = \nonumber 
\\ -\frac{4}{\Omega^2} \Bigg(\left(\partial_u
B\right)^2 \partial_v \left(\frac{V_v}{\Omega^2}\right) + \left(\partial_v
B\right)^2 \partial_u \left(\frac{V_u}{\Omega^2}\right)  + \frac{1}{4r^2}
\left(\partial_v V_u + \partial_u V_v \right) 
\left(1-\frac{2}{3}\rho \right)  \Bigg) \nonumber \\ 
-\frac{3}{r} \left(\frac{\nu}{\Omega^2} V_v + \frac{\lambda}{\Omega^2}V_u
\right) \left(\partial^\alpha B \partial_\alpha B\right) - \frac{1}{r^3}\left(\frac{\nu}{\Omega^2} V_v + \frac{\lambda}{\Omega^2}V_u
\right) \left(1-\frac{2}{3}\rho \right)  \, .
\end{eqnarray}
%
%
%
%+
%
%
%
%
\subsection{Basic regions}
In the course of the paper we shall apply the 
basic vectorfield identity (\ref{intid}) for
different vector fields in adapted regions of the black hole exterior. 
Here the relevant formulae arising from (\ref{intid}) 
for these regions are derived.\footnote{Since the coordinate system is only
piecewise $C^2$, the justification of these formulae, which are easily
derived formally, requires some care. A detailed discussion can be found
in Appendix \ref{reggree}.}

\subsubsection{Characteristic Rectangles} 
Writing out the identity (\ref{intid}) for a 
null-rectangle $\mathcal{R} = \left[u_1,u_2\right] 
\times \left[v_1,v_2\right]$ yields
\begin{eqnarray}
-\int_{vol} \nabla_\alpha P^\alpha &=& - \int_{\mathbb{S}^3} \int_{v_1}^{v_2}
 \int_{u_1}^{u_2} \frac{1}{\sqrt{g}} \partial_\alpha \left(\sqrt{g}
 P^\alpha \right) \sqrt{g} du dv d{\bf \omega} \nonumber \\ &=&
 - \int_{\mathbb{S}^3}\int_{v_1}^{v_2}
 \int_{u_1}^{u_2} \left[ \partial_u \left(\sqrt{g}
 P^u \right) +  \partial_v \left(\sqrt{g}
 P^v \right)\right] du dv d{\bf \omega} \, .
\end{eqnarray}
Defining the bulk term
\begin{equation}
I^V_B = - \int_{vol\left(\mathcal{R}\right)} \left(T_{\alpha \beta}
 \pi_V^{\alpha \beta} + \left(\nabla^\beta T_{\alpha \beta}\right)
 V^\alpha \right) \frac{\Omega^2}{2}r^3 du dv dA_{\mathbb{S}^3}
\end{equation}
and the boundary terms
\begin{eqnarray}
F^V_B \left(\left[u_1,u_2\right] \times \{v\} \right) &=&
-\int_{\mathbb{S}^3} \int_{u_1}^{u_2} \sqrt{g} P^v
\left(\bar{u},v\right) d\bar{u}d{\bf \omega} \nonumber \\ &=&
2\pi^2 \int_{u_1}^{u_2} 
\left[r^3 \left(\partial_u B\right)^2 V^u + \frac{r \Omega^2}{4} 
\left(1-\frac{2}{3}\rho \right)V^v \right] du \, ,
\end{eqnarray}
\begin{eqnarray}
F^V_B \left( \{u\} \times \left[v_1,v_2\right] \right) &=&
-\int_{\mathbb{S}^3} \int_{v_1}^{v_2} \sqrt{g} P^u
\left(u,\bar{v}\right) d\bar{v} d{\bf \omega}\nonumber \\ &=&
2\pi^2 \int_{v_1}^{v_2} 
\left[r^3 \left(\partial_v B\right)^2 V^v + \frac{r \Omega^2}{4} 
\left(1-\frac{2}{3}\rho \right)V^u \right] dv  \, ,
\end{eqnarray}
we find the identity
\begin{eqnarray}
F^V_B \left( \{u_2\} \times \left[v_1,v_2\right] \right) + F^V_B
\left(\left[u_1,u_2\right] \times \{v_2\} \right) \nonumber \\ 
= I^V_B \left(\mathcal{R}\right) + F^V_B \left( \{u_1\} \times
\left[v_1,v_2\right] \right) + F^V_B \left(\left[u_1,u_2\right] \times \{v_1\} \right)
\end{eqnarray}

\subsubsection{The region
  ${}^{u_H}\mathcal{D}^{{r^\star_g},u_J}_{[t_1,t_2]}$}
Another important region is 
\begin{eqnarray} \label{Dregion}
{}^{u_H}\mathcal{D}^{{r^\star_g},u_J}_{[t_1,t_2]} :=&& \Big(\{t_1
\leq t \leq t_2 \} \cap \{r^\star \geq r^\star_g \} \cap \{ u_J \leq u \leq u_H \}\Big)
\nonumber \\
&\cup& \Big(\{ \left(u,v\right) \in \left[t_1-{r^\star_g},u_H\right]
\times \left[t_1+{r^\star_g},t_2 + {r^\star_g} \right] \}\Big)
\end{eqnarray}
\begin{figure}[h!]
\[
\input{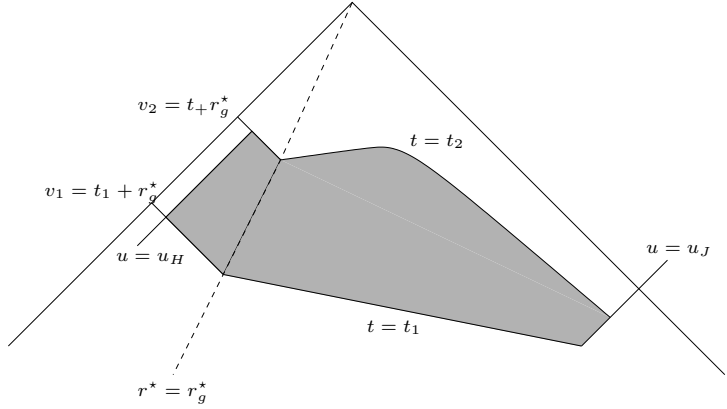}
\]
\caption{The region ${}^{u_H}\mathcal{D}^{{r^\star_g},u_J}_{[t_1,t_2]}$.} \label{regionD}
\end{figure}
for which one finds the basic identity 
\begin{equation}
\hat{I}^V_B\left({}^{u_H}\mathcal{D}^{{r^\star_g},u_J}_{[t_1,t_2]}\right)
= \hat{F}^V_B\left(t_2\right) -
\hat{F}^V_B\left(t_1\right) + \hat{H}_{u_H} - \hat{J}_{u_J} \, ,
\end{equation}
with the bulk term
\begin{equation}
\hat{I}^V_B\left({}^{u_H}\mathcal{D}^{{r^\star_g},u_J}_{[t_1,t_2]}\right) = \int_{{}^{u_H}\mathcal{D}^{{r^\star_g},u_J}_{[t_1,t_2]}} \left(-T_{\mu
  \nu} \pi^{\mu \nu}_V - \left(\nabla^\mu T_{\mu \nu}\right)V^v\right) dVol
\end{equation}
and the boundary terms
\begin{eqnarray}
\frac{1}{2\pi^2} \hat{F}^V \left(t\right) = &&\int_{{r^\star_g}}^{t-u_J} -P^t
\left(t,r^\star\right) \Omega^2 r^3 dr^\star \nonumber \\ &&+
\int_{t-{r^\star_g}}^{u_H} \left[r^3 \left(\partial_u B\right)^2 V^u +
  \frac{r\Omega^2}{4} \left(1-\frac{2}{3}\rho\right) V^v \right] du \, ,
\end{eqnarray}
where
\begin{equation}
-P^t = \frac{V^u}{2\Omega^2}
 \left[2\left(\partial_u B\right)^2 +
 \frac{\Omega^2}{2r^2} \left(1-\frac{2}{3}\rho\right) \right] +  \frac{V^v}{2\Omega^2}
 \left[2\left(\partial_v B\right)^2 +
 \frac{\Omega^2}{2r^2} \left(1-\frac{2}{3}\rho\right)\right]
\end{equation}
and
\begin{equation}
\frac{1}{2\pi^2} \hat{H}^V_{u_H} = \int_{v_1}^{v_2} \left[r^3
  \left(\partial_v B \right)^2 V^v + \frac{r \Omega^2}{4}
  \left(1-\frac{2}{3} \rho\right) V^u \right] \left(u_{hoz},v\right) dv 
\end{equation}
and 
\begin{equation}
\frac{1}{2\pi^2} \hat{J}^V_{u_J} = \int_{2t_1-u_J}^{2t_2-u_J} \left[r^3
  \left(\partial_v B \right)^2 V^v + \frac{r \Omega^2}{4}
  \left(1-\frac{2}{3} \rho\right) V^u \right]\left(u_J,v\right) dv \, .
\end{equation}
For the region under consideration we will also need to apply Green's
identity to a term of the form $D \cdot \Box \left(B^2\right)$ for
some function $D$.\footnote{The formula derived here is a-priori valid only for
  $D \in C^2$. However it also holds for a $D$ admitting less
  regularity, as is shown explicitly in Appendix \ref{reggree}, where
  we demonstrate that for the cases where (\ref{basicgreen}) 
is applied in the paper (equations (\ref{DKreg}) 
and (\ref{cgfd})), $D$ indeed satisfies these requirements.}
\begin{eqnarray} \label{basicgreen}
\hat{I}^V_B \left({}^{u_H}\mathcal{D}^{{r^\star_g},u_J}_{[t_1,t_2]}\right) = ... &+& \int
\left[\left(\Box B^2 \right) D \right]dVol = ... + \int
\left[ B^2 \left( \Box D \right)\right] dVol \nonumber \\ &+& G\left(t_2\right) -
G\left(t_1\right) + N \left(t_2\right) - N \left(t_1\right) +
H^G_{u_H} - J^G_{u_J}
\end{eqnarray}
where
\begin{equation}
\frac{1}{2\pi^2} G \left(t\right) = \int_{{r^\star_g}}^{t-u_0} \left[B^2 \partial_t
  D - D \partial_t B^2 \right] r^3 \left(t,r^\star \right) dr^\star \, ,
\end{equation}
\begin{equation}
\frac{1}{2\pi^2} N \left(t\right) = \int_{t-{r^\star_g}}^{u_H} \left[B^2 \partial_u
  D - D \partial_u B^2 \right] r^3 \left(u, t+{r^\star_g} \right) du \, ,
\end{equation}
\begin{equation}
\frac{1}{2\pi^2} H^G_{u_H} = \int_{t_1+{r^\star_g}}^{t_2+{r^\star_g}} \left[B^2 \partial_v
  D - D \partial_v B^2 \right] r^3 \left(u_{H}, v \right) dv \, ,
\end{equation}
\begin{equation}
\frac{1}{2\pi^2} J^G_{u_J} = \int_{t_1+{r^\star_g}}^{t_2+{r^\star_g}} \left[B^2 \partial_v
  D - D \partial_v B^2 \right] r^3 \left(u_{J}, v \right) dv \, .
\end{equation}
We then define the renormalized bulk term
\begin{equation}
I^V_B \left({}^{u_H}\mathcal{D}^{{r^\star_g},u_J}_{[t_1,t_2]}\right) := ... + \int_{{}^{u_H}\mathcal{D}^{{r^\star_g},u_J}_{[t_1,t_2]}}
\left[ B^2 \left( \Box D \right)\right] dVol \, ,
\end{equation}
for which the identity
\begin{equation} \label{finidg}
I^V_B \left({}^{u_H}\mathcal{D}^{{r^\star_g},u_J}_{[t_1,t_2]}\right) = F^V_B \left(t_2\right) - F^V_B
\left(t_1\right)+ H_{u_H} - J_{u_J}
\end{equation}
with
\begin{equation}
F^V_B \left(t\right) = \hat{F}^V_B \left(t\right) - G\left(t\right) -
N\left(t\right) \, ,
\end{equation}
\begin{equation}
H^V_{u_H} = \hat{H}^V_{u_H} - H^G_{u_H} \, ,
\end{equation}
\begin{equation}
J^V_{u_J} = \hat{J}^V_{u_J} - J^G_{u_J}
\end{equation}
holds. Note that for $u_J=u_0$, the boundary terms $J_{u_J}$ all
vanish, because $B$ does not have any support on $u=u_0$ by the domain
of dependence property.

Finally, for future reference we also define the subregion
\begin{equation} \label{Bsubreg}
\mathcal{B}^{{r^\star_g},R^\star_g}_{[t_1,t_2]} =
	{}^{u_H}\mathcal{D}^{{r^\star_g},u_J}_{[t_1,t_2]} \cap
	\{ r^\star_g \leq r^\star \leq R^\star_g \}
\end{equation}
and the slice
\begin{equation} \label{slicedef}
\Sigma_{\bar{t}} = \Big(\{ t=\bar{t} \} \cap \{ r^\star \geq r^\star_{cl}
\}\Big)  \cup
\Big(\{ v = \bar{t}+ r^\star_{cl} \} \cap \{ r^\star \leq
r^\star_{cl}\}\Big) \, .
\end{equation}
\section{The vectorfield $T$ and the Hawking mass} \label{HawkT}
Recall that the Hawking mass $m$ defined in (\ref{Hawkmass}) 
satisfies (\ref{dum}) and (\ref{vum}). The one-form $dm$ is 
closed and by simple connectedness of the Penrose
diagram, exact. It follows that energy is conserved. This fact can also
be seen from the integral identity (\ref{intid}) applied to the 
the vector-field
\begin{equation}
T = \frac{4\lambda}{\Omega^2} \partial_u -  \frac{4\nu}{\Omega^2}
\partial_v \, .
\end{equation}
If we apply the identity (\ref{intid}) in
the region ${}^{u_H}\mathcal{D}^{{r^\star_g},u_J}_{[t_1,t_2]}$, 
energy conservation translates into the following relation 
between the boundary terms:
\begin{equation}
F^T_B\left(t_2\right) = F^T_B\left(t_1\right) -
H^T_{u_H} + J^T{\left(u_J\right)} \, ,
\end{equation}
where
\begin{eqnarray} \label{enbt}
F^T_B \left(t\right) = \int_{t-r^\star_{cl}}^{u_H} \left[4r^3
  \lambda \frac{\left(B_{,u}\right)^2}{\Omega^2} - r
  \nu \left(1-\frac{2}{3} \rho \right) \right]
  \left(u,t+r^\star_{cl}\right) du \nonumber \\
+\int_{r^\star_{cl}}^{t-u_0} \Bigg(r^3 \frac{\left(B_{,v}\right)^2}{\kappa} + 4r^3 \frac{\lambda}{\Omega^2} \left(B_{,u}\right)^2
+ r \left(\lambda-\nu\right) \left(1-\frac{2}{3}\rho \right)
  \Bigg)\left(t,r^\star \right) dr^\star \, ,
\end{eqnarray}
\begin{eqnarray} \label{enht}
H^T_{u_H} = \int_{v_1}^{v_2} \left[r^3
  \frac{\left(B_{,v}\right)^2}{\kappa} + r \lambda \left(1-\frac{2}{3}
  \rho \right) \right] \left(u_{H},v\right) dv \, ,
\end{eqnarray}
\begin{eqnarray} \label{enjt}
J^T{\left(u_J\right)} = \int_{2t_1-u_J}^{2t_2-u_J} \left[r^3
  \frac{\left(B_{,v}\right)^2}{\kappa} + r \lambda \left(1-\frac{2}{3}
  \rho \right) \right] \left(u_J,v\right) dv \, .
\end{eqnarray}
We will sometimes use the notation $E\left(\Sigma\right)$, for the
energy flux through an achronal slice $\Sigma$.
\section{The bootstrap} \label{bootstrap}
The bootstrap is intimately related to the choice of coordinate systems defined in section \ref{Coordinates}. We will use the notation introduced in that section. 
\subsection{The bootstrap region and the statement $\mathcal{P}$}
Let 
\begin{equation}
a = \sqrt{\frac{M}{2}} \left[-3\sqrt{2} - \log \left(\frac{2-\sqrt{2}}{2+\sqrt{2}}\right)\right]
\end{equation}
and $c$ be some small constant. Define
\begin{equation}
S = t\partial_t + \left(r^\star-a\right) \partial_{r^\star} \textrm{ \ \ \ \ \ \ } 
\underline{S} = t\partial_{r^\star} + \left(r^\star-a\right)
\partial_t
\end{equation}
and the quantity
\begin{eqnarray} \label{ekbfir}
E^K_B \left(t\right) = \frac{2\pi^2}{M} 
 \int_{r^\star_{K}}^{t-u_0} \Bigg[
  \left(1+2\nu\right) \left(\left(SB\right)^2 + \left(\underline{S}B \right)^2 \right)
  \nonumber \\
  + \left(-2\nu \right) \Bigg(\left(SB + \frac{3}{2} \frac{r^\star-a}{r} B
  \right)^2 + \frac{\left(r^\star-a\right)^2}{r^2}B^2
  \Bigg) \nonumber \\
  + \left(-2\nu \right) \Bigg(
 \left(\underline{S}B + \frac{3}{2} \frac{t}{r} B
  \right)^2 +  \frac{t^2}{r^2}B^2 \Bigg) \Bigg] r^3 dr^\star 
\end{eqnarray}
with
\begin{equation}
r^\star_K = \sup_{t < T} r^\star\left(t,r_K\right) \, . \label{rstarK}
\end{equation}
\begin{figure}[h!]
\[
\input{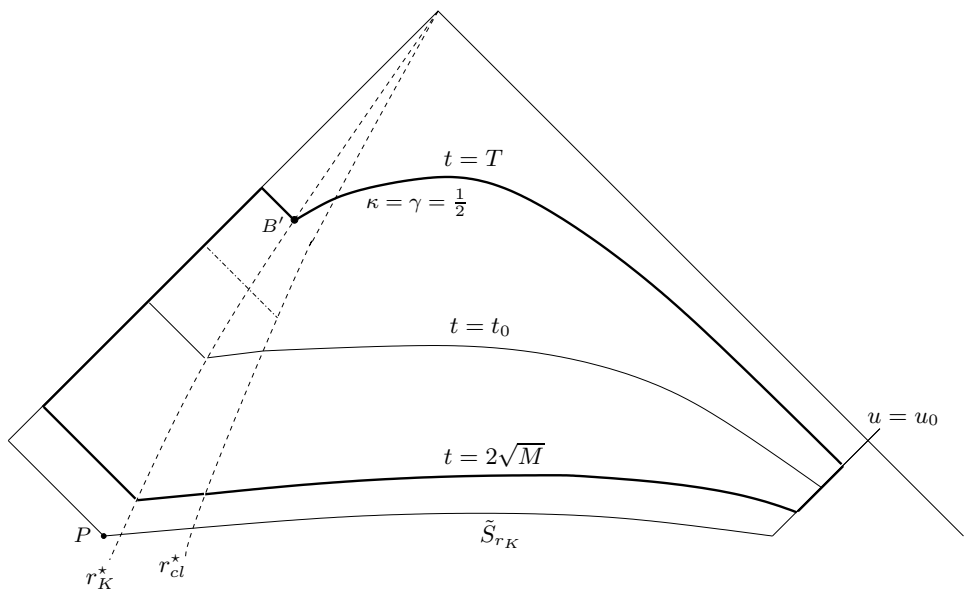}
\]
\caption{The bootstrap region.} \label{bootsreg}
\end{figure}
To each $\tilde{\tau}$ we associate the region 
$\mathcal{A}\left(T\left(\tilde{\tau}\right)\right)={}^{u_{hoz}}\mathcal{D}^{r^\star_{K},u_0}_{[2\sqrt{M},T)}$
  (hence defining the $T$ in (\ref{rstarK})).

We define the statement $\mathcal{P}_{T\left(\tilde{\tau}\right)}$ associated 
to a region $\mathcal{A}\left(T\left(\tilde{\tau}\right)\right)$ to 
be\footnote{We will sometimes abbreviate $T\left(\tilde{\tau}\right)$ by $T$, 
reminding the reader that any $T$ arises from $\tilde{\tau}$ 
as described in section \ref{Coordinates}.}
\begin{enumerate}
\item In the subregion $\{r^\star \geq r^\star_K \} \cap \mathcal{A}\left(T\right) $, the area radius satisfies 
\begin{eqnarray} \label{fibor}
\Big| r^\star - \left[ r\left(t,r^\star\right) + \sqrt{\frac{M_A}{2}} \left(\log
\left(\frac{r\left(t,r^\star\right)-\sqrt{2M_A}}{r\left(t,r^\star\right)+\sqrt{2M_A}}\right)
+ p \right) \right] \Big| < c \sqrt{M}
\end{eqnarray}
with 
\begin{equation} \label{pdef}
p = - 2\sqrt{2} - \log \frac{2-\sqrt{2}}{2+\sqrt{2}}
\end{equation}
and $M_A$ defined to be the Hawking mass at
the point $\left(T,r^\star=0\right)$.
\item We have
\begin{equation} \label{tass}
\frac{1}{2} \sqrt{M} < \sup_{\tilde{S} \cap \{r^\star \geq r^\star_K\} \cap
  \{ u \geq u_0\} } t < \frac{3}{2}\sqrt{M} \, .
\end{equation}
\item The weighted energy-density (\ref{ekbfir}) satisfies
\begin{equation} \label{Kass}
\frac{1}{M} E^K_B\left(\tilde{T}\right) < c \textrm{ \ \ \ \ on all arcs $\{ 2\sqrt{M}
  \leq t = \tilde{T} < T \} \cap \{r^\star \geq r^\star_K\} \cap
  \mathcal{A}\left(T\right)$ } \, .
\end{equation}
\item 
The energy-flux satisfies 
\begin{equation} \label{hozass}
m\left(u_{hoz},v_2\right) - m\left(u_{hoz},v_1\right) < c \ M \frac{M}{\left(v_{1+}\right)^2}
\end{equation}
for any $v_1 \leq v_2$ along the part of the horizon located in
$\mathcal{A}\left(T\right)$ and
\item 
\begin{equation} \label{ceilass}
m\left(u_{r^\star_{cl}},v\right) - m\left(u_{hoz},v\right) < c \ M \frac{M}{v_+^2}
\end{equation}
holds in $\mathcal{A}\left(T\right)$ for an $r^\star_{cl}$ defined in the subsection below.
% on all $v=constant$ null-lines in the region $r^\star \leq r^\star_{cl}$. 
\item The integral bound 
\begin{equation} \label{intebound2}
\tilde{F}^Y_B = \int r^3 \frac{\left(B_{,u}\right)^2}{-\nu} du < C_L M \frac{M}{v_+^2} \textrm{\ \ \ \ \ for \ \ $C_L = \sup_{r^\star \geq r^\star_{cl}} \frac{1}{1-\mu}$}
\end{equation}
holds along lines of constant $v$ in the region 
$\{r^\star \leq r^\star_{cl} \} \cap \{ u \leq T-r^\star\left(T,r_K\right) \} \subset \mathcal{A}\left(T\right)$, corresponding to a decay of energy for local observers near the horizon. 
\end{enumerate}

Finally, we define the set
\begin{equation} \label{Adef}
A = \Big\{ \tilde{\tau} \in \left[\sqrt{M}, \infty\right) \  \ \Big| \ \  \mathcal{P}_{T\left(\hat{\tau}\right)} \textrm{\ \ \ holds in $\mathcal{A}\left(T\left(\hat{\tau}\right)\right)$ \textrm{for all $\hat{\tau} \leq \tilde{\tau}$} \Big\} }
\subset \left[\sqrt{M}, \infty\right)  \ \ \ \ \ \,
\end{equation}
Note that the lower bound on $\tilde{\tau}$ ensures that $T >
2\sqrt{M}$ (cf. (\ref{timedef})). The following key-Theorem will 
close the bootstrap and is easily seen to imply 
the decay rates of Theorem \ref{asymptoticstab}. It will only be
proven at the end of the paper. 
\begin{theorem} \label{setS}
The set $A$ is non-empty, open and closed.
\end{theorem}
A few remarks are in order. The first two bootstrap assumptions ensure
that the different coordinate systems $\mathcal{C}_{\tilde{\tau}}$ do
not move too far away from one another, at least in the region
$r^\star \geq r^\star_K$. The first controls the deviation of the
relation between the coordinate $r^\star$ and the area radius $r$ from
the familiar relation between the Regge-Wheeler coordinate and the
area radius in the Schwarzschild metric. In particular, for
Schwarzschild the left hand side of (\ref{fibor}) is zero. The second
assumption ensures that the bottom of the bootstrap region (the
$t=2\sqrt{M}$ slice) does not move away too much from the
geometrically defined initial data (and is moreover always located to
the future of the data). In other words, the coordinates of $\tilde{S} \cap \{r^\star \geq r^\star_K\} \cap \{u \geq u_0\}$ are similar in all coordinate systems $\mathcal{C}_{\tilde{\tau}}$. 

The open-part of Theorem \ref{setS} follows from a simple continuity argument:

\begin{proposition}
The set $A$ defined in (\ref{Adef}) is open.
\end{proposition}
\begin{proof}
We observe that the integral $E^K_B\left(t\right)$ 
and in fact all the quantities appearing in statement $\mathcal{P}$ 
of the bootstrap assumptions depend continuously on the 
choice of $\tilde{\tau}$. 
%Indeed, changing $\tilde{\tau}$ slightly 
%will alter the notion of a $t=const$ slice and 
%of $r^\star_{K}$, $r^\star_{cl}$-curves
%(and hence of $E^K_B\left(t\right)$, for example) but it will do so in 
%a continuous manner. Hence, given any 
%$\tilde{\tau} \in A$ we can find an $\epsilon$-neighbourhood 
%around $\tilde{\tau}$ in $A$ such that the statement 
%$\mathcal{P}$ still holds with constant $c$ for 
%any $\hat{\tau}$ in that neighbourhood.  
\end{proof}
One should note in this context that all bootstrap assumptions involve only first derivatives of the fields and the area radius, and hence only continuous quantities (cf.~the remarks on the differentiability of the coordinate systems at the end of section \ref{Coordinates}).

The hard part of Theorem \ref{setS} consists in showing that $A$ is closed. This will be accomplished by improving the constants appearing in the inequalities of the bootstrap assumptions. 
\subsection{The choice of $r^\star_{cl}$} \label{rstarcldef}
In this subsection we define the quantity $r^\star_{cl}$ with respect to 
the coordinate system associated to the bootstrap region. Clearly, 
the location of $r^\star_{cl}$ will change between different coordinate
systems when the bootstrap region is altered. However, by bootstrap 
assumption \ref{boot1}, it will always stay close to a geometrically 
defined curve of constant $r$, which is determined below.

By Propositions \ref{omvclo} and \ref{simpsmall} we know 
that on $\mathcal{D}$ the bound
\begin{equation}
\sqrt{M} \Big|\frac{\Omega_{,v}}{\Omega} - \frac{m}{r^3}\Big| + \Big|\kappa-\frac{1}{2}\Big| \leq C\left(\epsilon\right)
\end{equation}
holds in any coordinate system $\mathcal{C}_{\tilde{\tau}}$. For any 
small number $\psi > 0$ we can hence choose the initial data small 
enough such that there exists an $r_Y < \frac{3}{2} \sqrt{M}$ satisfying
\begin{equation} \label{psicond1}
\max_{r \leq r_Y} \Big[\log \frac{r_Y}{r_-},
  r\frac{\Omega_{,v}}{\Omega} - \frac{1}{2}, 1-\mu\Big] < \psi \, .
\end{equation} 
Here Corollary \ref{rcurvcol} has been used for the bound on the first factor.
By bootstrap assumption \ref{boot1} the curve $r^\star_Y := \inf_{t
  < T} r^\star\left(t,r_Y\right)$ is always close to the geometrically
  defined curve $r_Y$. Hence we can additionally impose that
\begin{equation} \label{psicond2}
\frac{\sqrt{M}}{-r^\star_Y} < \psi
\end{equation}
holds. Next we are going to determine how small $\psi$ has to be. We 
define two functions $\alpha\left(r^\star\right)$ and 
$\beta\left(r^\star\right)$ in the coordinate system associated with
the bootstrap region as follows. 
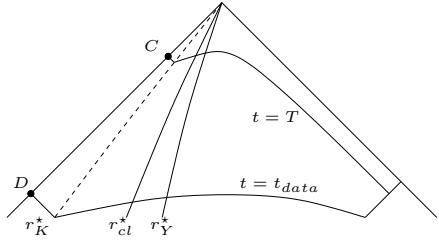
\[
\input{alphabet2.pstex_t}
\]
The function $\alpha$ which is supported only for $r^\star \leq
-\frac{1}{2}\sqrt{M}$ is everywhere non-negative and defined by setting 
$\alpha\left(r^\star_C=\frac{T+r^\star_{K}-u_{hoz}}{2}\right)=1$ and
\begin{equation} \label{alphacond}
\alpha^\prime \left(r^\star\right) = \left\{ \begin{array}{lll}
0 & \textrm{ for $r^\star \leq r^\star_C$} \\
\frac{1}{\sqrt{M}} \tilde{\chi}\left(r^\star\right) & \textrm{ in $\left[r^\star_C,r^\star_K\right]$} \\
\frac{M^\frac{1}{4}}{\left(\sqrt{M}+|r^\star|\right)^{\frac{3}{2}}} & \textrm{in $[r^\star_K,
    r^\star_Y]$} \, .
\end{array} \right.
\end{equation}
with $M=m\left(T,r^\star=0\right)$ and $\tilde{\chi}$ a smooth positive interpolating function. In particular $\alpha=1$ on $\overline{DC}$.

The non-negative function $\beta$, again with support 
only for $r^\star \leq -\frac{1}{2}\sqrt{M}$, is defined by setting  $\beta\left(r^\star_D=\frac{2\sqrt{M}+r^\star_{K}-u_{hoz}}{2}\right) = 0$ and imposing that
\begin{equation} \label{bpcon}
\frac{24}{r\left(t,r^\star\right)} \Omega^2\left(t,r^\star\right)
\geq \beta^\prime \geq \frac{18}{r\left(t,r^\star\right)} \Omega^2\left(t,r^\star\right) 
\end{equation}
in all of $r^\star \leq r^\star_Y$. We can estimate the value of
$\beta$ on $r^\star_Y$ by
\begin{equation}
\beta\left(r^\star_Y\right) = 0 + \int_{D}^{r^\star_Y}
\beta_{,r^\star} dr^\star \leq  \int_{D}^{r^\star_Y}
24 \frac{\Omega^2}{r} dr^\star \leq  \int_{D}^{r^\star_Y}
24 \frac{\gamma \kappa}{\gamma +\kappa} \partial_{r^\star} \log r dr^\star
\leq 12 \log \frac{r_Y}{r_-} \nonumber \, .
\end{equation}
Hence $\beta$ remains controlled by the $r$-fluctuation 
in $r^\star \leq r^\star_Y$ and hence small by 
choosing $\psi$ above suitably small. Note that $\alpha$ and $\beta$ are 
in particular supported away from the curve $r^\star=0$. 

We finally choose the $\psi$ of (\ref{psicond1}), (\ref{psicond2}) 
so small that the inequalities
\begin{equation} \label{condY1}
\left(4\alpha \frac{\Omega_{,v}}{\Omega}r - \alpha^\prime r \right) > \max\Big[2
\left(\frac{1}{4\kappa} \alpha-\beta \lambda \right)^2, \frac{r}{4\sqrt{M}} \alpha, \frac{\kappa\left(1-\mu\right) r}{\sqrt{M}} \Big] \, ,
\end{equation}
\begin{equation} \label{condY2}
\alpha \geq \kappa \left(4\beta \lambda + 2r \beta^\prime + 8r \beta \frac{\Omega_{,v}}{\Omega}\right) + \max\Big[\frac{\kappa r \beta}{2\sqrt{M}}, \frac{r}{2\sqrt{M}}\Big] \, ,
\end{equation}
\begin{equation} \label{condY3}
\left(- \frac{r^\star -a}{r} \right) \Big[24\mu r \frac{\Omega_{,v}}{\Omega} + \left(1-\mu\right)\left(-70\kappa - 36\kappa \mu \right)\Big]  > 45 \, .
\end{equation}
hold in the region $r^\star \leq r^\star_Y$ 
and set $r^\star_{cl} = r^\star_Y - 2\sqrt{M}$.
\\ \\
{\bf Remark: } The constant $\psi$ and the corresponding $r_Y$ (and
the upper bound on initial data) can easily be computed explicitly
and is fixed once and for all. \emph{In particular it does not depend on the
size of the bootstrap region and the coordinate system that comes
along with it.} The curve $r^\star=r^\star_Y$ and 
hence $r^\star=r^\star_{cl}$ is then also fixed and always 
close to $r_Y$ by bootstrap assumption \ref{boot1} and the 
fact that $r_K$ is chosen much closer to the horizon than $r_Y$. 
Smallness for the bootstrap on the other hand, will be 
exploited via the $r^\star_K$-curve and by choosing the 
initial data even smaller to ``beat the constants'' which 
are introduced by the choice of $r^\star_{cl}$. 
\subsection{Cauchy stability}
For the closed-part we will have to improve the constant $c$ in 
the statement $\mathcal{P}$ (i.e.~the
bounds (\ref{Kass}-\ref{intebound2})) in the region 
$\mathcal{A}\left(T\right)$. The argument constitutes the body of the
paper. In this context we note that within the process of improving the
bootstrap assumptions there will be two sources of smallness. The 
first arises from the fact that $r=r_K$ can be chosen very close 
to the horizon. The second is obtained by selecting a $\nabla r$-slice 
belonging to some large $\tilde{\tau}_0$ (and hence large associated time $t_0$) up to which Cauchy stability holds by a suitable smallness assumption on the data. This is expressed precisely by the following
\begin{proposition} \label{Cauchystab}
For any small $\eta > 0$, $\tilde{\delta} > 0$, and 
any large $\tilde{\tau}_0$ (hence large 
associated time $T_0=\vartheta\left(\tilde{\tau}_0\right)$, 
with $\vartheta$ defined in (\ref{thetamap})) we can find 
an $r_K$ and a $\delta > 0$ such that the following statement 
is true: If the smallness assumptions 
(\ref{initassump}) and (\ref{initassump2}) of Theorem \ref{asymptoticstab} hold for $\delta$, then
\begin{enumerate}
\item The curve $r=r_K$ away from the horizon satisfies $ r_K^2 - r_-^2 < \eta$
\item in the coordinate system defined by
$\tilde{\tau}_\bullet \in [0,\tilde{\tau}_0]$ the $t$-coordinate 
of the subset $\tilde{S} \cap \{r \geq r_K\} \cap \{ u \geq u_0 \}$ of the initial data satisfies
\begin{equation} \label{tinb}
|t-\sqrt{M}| < \tilde{\delta} \sqrt{M}.
\end{equation}
\item In the coordinate system defined by
  $\mathcal{C}_{\tilde{\tau}_0}$, the statement $\mathcal{P}$ holds
  with constant $\tilde{\delta}$ (instead of $c$) in
  the region ${}^{u_{hoz}}\mathcal{D}^{r^\star_K, u_0}_{[2\sqrt{M},
  T_0]}$ and moreover, the pointwise bound
\begin{equation}
|B| + M^{-\frac{1}{4}} \Big|\frac{\zeta}{\nu}\Big| +
 M^{-\frac{1}{4}}|\theta| \leq \sqrt{M} \frac{\tilde{\delta}}{v_+}
\end{equation}
holds on any slice $\Sigma_{t}$ (cf.~(\ref{slicedef})) 
for $2\sqrt{M} \leq t \leq T_0$.
\end{enumerate}
\end{proposition}
\begin{proof} The first assertion is the statement of 
Corollary \ref{rcurvcol}. For the second statement 
consider the coordinate system 
$\mathcal{C}_{T=\vartheta\left(\tilde{\tau}_\bullet\right)}$ 
for a given $\tilde{\tau}_\bullet \in [0,\tilde{\tau}_0]$. The 
vectorfield $\nabla r$ introduced in section
 \ref{Coordinates} can be expressed in the 
associated $(t,r^\star)$ coordinates
\begin{equation}
\nabla r = \frac{1}{4\kappa \gamma} \left[\left(\gamma-\kappa\right)
  \partial_t + \left(\gamma+\kappa\right)
  \partial_{r^\star} \right] \,
\end{equation}
as can the vectorfield $\nabla_{\perp} r$ which is defined to be
orthogonal to $\nabla r$ and whose integral curves are 
the curves of constant area radius $r$:
\begin{equation}
\nabla_{\perp} r = \frac{1}{4\kappa \gamma} \left[\left(\gamma+\kappa\right)
  \partial_t + \left(\gamma-\kappa\right)
  \partial_{r^\star} \right] \, .
\end{equation}
The rescaled vectorfields
\begin{equation} \label{RGvector}
R = \frac{1}{\sqrt{1-\mu}} \nabla r \textrm{ \ \ \ and \ \ \ }  G = \frac{1}{\sqrt{1-\mu}} \nabla_\perp r
\end{equation} \
satisfy the orthonormality relations
\begin{equation}
g\left(R,R\right) = 1 \textrm{ \ \ \ and \ \ \ } g\left(G,G\right) =
-1 \textrm{ \ \ \ and \ \ \ } g\left(R,G\right)=0 \, .
\end{equation}
Let ${\varrho}$ be the affine parameter along $R$ and $\tau$ 
the affine parameter along $G$. In the following, 
we frequently refer to figure \ref{codiag} 
of section \ref{Coordinates}. At the point $A$ we 
have $t=T_\bullet=\sqrt{M}+\sqrt{2}\tau_{AD}$ by 
definition. We would like to estimate the value of $t$ 
at the point $D$ and compare it to $1$, which is the value 
of $t$ if $\tilde{\tau}_\bullet=0$, $T_\bullet=\sqrt{M}$ and the coordinates are
defined on initial data. The rate at which $t$ changes in affine
parameter along the integral curve of $\nabla_\perp r$ going through $A$ 
is given by
\begin{equation} \label{upsa}
\frac{dt}{d\tau} = \frac{1}{4\kappa \gamma} \left(\kappa + \gamma
\right) \frac{1}{\sqrt{1-\mu}} \, .
\end{equation}
We will integrate (\ref{upsa}) from $\tau=0$ to 
$\tau=\tau_{AD}$ with initial 
condition $T_0 = \sqrt{M} +\frac{\tau_{AD}}{\sqrt{1-\mu_A}}$ 
at $A$. By Proposition \ref{simpsmall} and \ref{gamma} the estimates 
\begin{equation}
|\kappa + \gamma - 1 | \leq C\left(\epsilon\right) \textrm{ \ \ \ and \ \ \ }
 \Big| \frac{1}{\sqrt{1-\mu}} -\sqrt{2} \Big| \leq C\left(\epsilon\right)
\end{equation}
hold along the curve. Given the fixed $\tilde{\tau}_0$ we choose the initial data so small that $\tilde{\tau}_0 \cdot C\left(\epsilon\right)$ is as small as we may wish. Hence $\tau_{AD}\cdot C\left(\epsilon\right)$ is small for any $\tau_{AD}\left(\tilde{\tau}_\bullet\right)$ with $\tilde{\tau}_\bullet \leq \tilde{\tau}_0$.\footnote{Note that $\tilde{\tau}_\bullet$ is close to $\tau_{AD}$, since the curves $r^2=4M_F$ and $r^2=4m_A$ converge to one another for the initial data going to zero} With these choices the estimate
\begin{equation}
| t \left(D\right) - \sqrt{M} | \leq C\left(\epsilon\right)
\end{equation}
simply follows from integrating (\ref{upsa}). 

In a completely analogous fashion, by considering $\frac{dr^\star}{d\tau}$ and using that $|\gamma - \kappa| \leq C\left(\epsilon\right)$ we can show that the point $r^\star=0$ on the initial data is close to $r=2\sqrt{M}$: $| r\left(t_{data},0\right)- 2\sqrt{M} | \leq C\left(\epsilon\right) \sqrt{M}$. 

Before we finally estimate how $t$ changes along the integral curve of 
$R$ through $D$ (i.e.~the location of the initial data), we 
derive a rough estimate for the relation of $r$ and $r^\star$. 
Consider the vectorfield
\begin{equation}
L = \frac{1}{\Omega} \partial_{r^\star} \, , 
\end{equation}
whose integral curves are the curves of constant $t$. The coordinate
$r^\star$ changes along such a curve (affine parameter $l$) by
\begin{equation}
\frac{dr^\star}{dl} = \frac{1}{\Omega} = \frac{1}{\sqrt{4\kappa \gamma \left(1-\mu\right)}} \, .
\end{equation}
Integrating from $r^\star=0$, where $r \approx 2\sqrt{M}$ outwards to infinity noting 
that $1-\mu > \frac{4}{9}$ and that both $\kappa$ and $\gamma$ 
are close to $\frac{1}{2}$ in the region under consideration, we obtain
\begin{equation} \label{rstargrowth}
\frac{9}{10}l \leq r^\star \leq \frac{9}{4\sqrt{2}} l \, .
\end{equation}
On the other hand, the area radius changes according to
\begin{equation}
\frac{dr}{dl} = \frac{1}{\Omega}\left(\lambda - \nu\right) = \frac{1}{2}\frac{\sqrt{1-\mu}}{\sqrt{\gamma \kappa}} \left(\kappa+\gamma\right)
\end{equation}
leading to the estimate
\begin{equation} \label{rgrowth}
2\sqrt{M} - C\left(\epsilon\right) \sqrt{M} + \frac{2}{5} l \leq r \leq 2\sqrt{M} + C\left(\epsilon\right) \sqrt{M} +
\frac{11}{10}l \, .
\end{equation}
Combining (\ref{rstargrowth}) and (\ref{rgrowth}) yields the relation
\begin{equation} \label{rstarrb}
\frac{9}{11} r - c_1 \leq r^\star \leq \frac{45}{8\sqrt{2}} r  \textrm{ \ \ \ \ with $c_1 = \frac{9}{11}\left(2\sqrt{M}-C\left(\epsilon\right) \sqrt{M}\right)$} 
\end{equation}
along any curve of constant time in the region $r \geq 2\sqrt{M}$.  In particular, if a quantity decays
in $r$ in the asymptotic region, it decays in $r^\star$ as well.

Finally, we can consider the integral curve of $R$ through $D$ on which
 the initial data is defined. We want to prove that 
the value of $t$ does not change much along that curve (at least 
up to the area radius $\tilde{R}$ where the support ends).
First we show that the horizon is a finite length of affine
parameter along $\nabla r$ away from $D$. 
Namely, since the $r$-component of the 
vectorfield $R$ is given by $R^r = \sqrt{1-\mu}$, we have the equation
\begin{equation}
\frac{dr}{d\varrho} = \sqrt{1-\mu} \, .
\end{equation} 
Starting at $r\left(0\right)=2\sqrt{m_A}$ and integrating inwards to the
point where the curve intersects the horizon we find
\begin{eqnarray}
0 \leq -\varrho_{hoz} &=& \int_{r_{hoz}}^{2\sqrt{m_A}} dr \frac{1}{
  \sqrt{1-\mu}} =  \int_{r_{hoz}}^{2\sqrt{m_A}} dr 2 \cdot
\partial_r \left(\sqrt{1-\mu}\right) \frac{r^3}{4m-2r m_{,r}}
dr \nonumber \\ &\leq& \left(4\sqrt{m_A} + C\left(\delta\right)\right)
  \left(\frac{1}{\sqrt{2}} - C\left(\delta\right)\right) 
\end{eqnarray}
for some small $\delta$. On the other hand, we can integrate 
outwards from $D$ along $\nabla r$ to a point where $r=\tilde{R}$. 
From (\ref{rgrowth}) we know that the affine parameter is 
controlled by the $r$ value along the curve, hence for large $\tilde{R}$
\begin{equation}
\varrho \leq \frac{6}{5}\tilde{R} \, .
\end{equation}
Finally, $t$ changes along the curve according to 
\begin{equation} \label{talongnabr}
\frac{dt}{d\varrho} = \frac{1}{4\gamma \kappa} \left(\gamma -
\kappa\right) \frac{1}{\sqrt{1-\mu}} \leq 0 \, .
\end{equation}
Within $[r_K, 2\sqrt{m_A}]$ and $[2\sqrt{m_A}, \tilde{R}]$ 
we can use the pointwise bound
\begin{equation}
\left(\gamma - \kappa\right) \frac{1}{\sqrt{1-\mu}}  \leq
C\left(\epsilon\right) 
\end{equation}
following from the results of section \ref{Basic}
and choose $\epsilon$ (hence the initial data) so small that 
$C\left(\epsilon\right)$ exceeds the support radius $\tilde{R}$:
\begin{equation}
|t_D - t| \leq C\left(\epsilon\right) \frac{6}{5} \tilde{R} 
\leq C\left(\epsilon\right) \, .
\end{equation}
In this way we can make the difference in $t$ small in the region between 
$r=r_K$ and $r=\tilde{R}$ on the $\nabla r$ integral curve.  \\

\begin{comment}
{\bf Remark:} For the previous estimate the smallness clearly 
depends on the support radius. We can improve this once we 
have derived the \emph{decay} estimate for $\kappa-\gamma$ 
from the bootstrap assumptions in 
Proposition (\ref{kapgamdecr}). \\ 
\end{comment}

The pointwise bound of statement $3$ follows directly
from Proposition \ref{simpsmall} together with the 
fact that the quantity $v$ is finite
in the region under consideration.

For (\ref{fibor}) of statement $\mathcal{P}$ we observe 
that on $\{t=T_0\} \cap \{r^\star \geq
r^\star_K\}$ we have $\partial_t r = 0$ by definition. From 
$\partial_{r^\star} r = \left(\kappa + \gamma\right)
\left(1-\mu\right)$ we derive, using Propositions \ref{simpsmall} 
and \ref{gamma}, the estimate $1$ of statement $\mathcal{P}$ on
$t=T_0$ for an arbitrary good constant by a suitable smallness 
assumption on the data. However, along a curve of constant 
$r^\star \geq r^\star_K$, the value of $r$ changes only by an amount 
which can be made small by suitable choice of initial data, as is 
seen from the estimate 
\begin{equation}
|r\left(t_b,r^\star\right) - r\left(t_a,r^\star\right)| = \Big|\int_{t_a}^{t_b}
 \left(\lambda+\nu\right) dt \Big| \leq \int_{2\sqrt{M}}^{T_0}
 \left(1-\mu\right) \left(\kappa - \gamma\right) dt \leq
 C\left(\epsilon\right) \cdot T_0
\end{equation}
is small for any $t_a, t_b \in \left[2\sqrt{M},T_0\right]$ if the data 
is small enough.

The second bootstrap assumption has been dealt with in 
statement $2$ of Proposition \ref{Cauchystab} already.

The third bootstrap assumption involves integrals over compact
intervals with the integrand containing $B$ and its derivatives. The 
integral is small on $t=2\sqrt{M}$ by assumption (\ref{initassump})
and Cauchy stability. Again from Cauchy stability it 
follows that $E^K_B$ will stay as small as we may wish up 
to the chosen $T=\vartheta\left(\tilde{\tau}_0\right)$ slice if 
we only chose the data small enough. This is perhaps most easily
  seen directly from the fact that $u$ and $v$ are always finite 
in the region of integration, and taking into account the 
pointwise bounds on $B$, $\frac{\zeta}{\nu}, \theta$ established in Proposition
  \ref{simpsmall}. Put together it follows that the quantity
$E^K_B$ can be made smaller than $\tilde{\delta}$ for a finite 
$t$ slice by an appropriate assumption on the data. 

The bootstrap assumptions
involving the energy can be satisfied by choosing the data
sufficiently small (recall the a-priori bound on the mass
fluctuation (\ref{inisma})). Finally, assumption (\ref{intebound2}) 
follows from the pointwise bound on $\frac{\zeta}{\nu}$ (cf. Proposition
\ref{simpsmall}) and realizing that integrating the 
quantity $\nu$ in $u$ yields a finite result.  
Hence, in the coordinate system defined by $\tilde{\tau}_0$, all 
inequalities in the statement $\mathcal{P}$ can be 
brought to hold with constant $\tilde{\delta}$ in 
the region ${}^{u_{hoz}}\mathcal{D}^{r^\star_K, u_0}_{[2\sqrt{M},
  T_0]}$.
\end{proof}
\begin{corollary} \label{empty}
The set $A$ defined in (\ref{Adef}) is non-empty.
\end{corollary}
\begin{proof}
By statement $2$ of Proposition \ref{Cauchystab} for any
$\tilde{\tau}_{\bullet} \leq \tilde{\tau}_0$ the coordinates of a point
in the associated region
$\mathcal{A}\left(\vartheta\left(\tilde{\tau}_\bullet\right)\right)$ will
be close to the coordinates of the same point in the coordinate system
defined by $\tilde{\tau}_0$. Hence the statement $\mathcal{P}$ holds with
constant $\tilde{\delta}$ in
$\mathcal{A}\left(\vartheta\left(\tilde{\tau}_\bullet\right)\right)$ for
all $\tilde{\tau}_\bullet \leq \tilde{\tau}_0$ by choosing $\delta$
small enough. Therefore $[\sqrt{M},\tilde{\tau}_0] \subset A$.
\end{proof}
In order to be useful in conjunction with the bootstrap, statement 
$3$ of Proposition \ref{Cauchystab} has to hold in any coordinate 
system $\mathcal{C}_{\tilde{\tau}}$ associated to a $\tilde{\tau} >
\tilde{\tau}_0$ with $\tilde{\tau} \in A$. The argument is postponed to 
Proposition \ref{Cauchystab2}, after we have derived appropriate decay
bounds from the bootstrap assumptions in the next section.

Proposition \ref{Cauchystab} also provides us with two sources of 
smallness. In particular it justifies the following algebra 
for constants:
\begin{equation}
C\left(r^\star_{cl}\right) \eta = \tilde{\delta}
\end{equation}
\begin{equation}
\frac{C\left(r^\star_K\right)}{t_0} = \tilde{\delta}
\end{equation}
Namely, after we have chosen $\psi>0$ (cf.~(\ref{psicond1}) and
(\ref{psicond2})) to determine $r^\star_{cl}$, we can choose $\eta$ so
small that it ``beats'' any constant depending on $r^\star_{cl}$, and
finally $t_0$ so large that  $\frac{1}{\sqrt{M}}
\frac{C\left(r^\star_K\right)}{t_0}$ is as small as we may wish. (Of
course, the restrictions on the initial data get stronger and stronger 
in this process.) Consequently, everywhere that the formulation ``we choose $t_0$ so large that''  is used in the paper, we always have an application of Proposition \ref{Cauchystab} in mind.
\section{Analyzing the bootstrap assumptions} \label{analboot}
In this section we are going to derive certain decay bounds for the energy, the squashing field and some other quantities. These estimates will be useful for late times, i.e.~they are to be understood in conjunction with Proposition \ref{Cauchystab} where we can choose such a late time. 
The time $t_0$ up to which Cauchy stability holds is chosen in particular so large that for $t \geq t_0$ we have $v \sim t$ in the region $r^\star_{cl} \leq r^\star \leq \frac{9}{10}t$ and that $v \sim t \sim r^\star$ in the region $r^\star \geq \frac{9}{10}t$. Moreover $v_0=t_0+r^\star_{cl} >> \sqrt{M}$. 
All statements about decay in this section are then valid in the subregion $\{t \geq t_0\} \cap \{v \geq v_0\}$ of the bootstrap region.

\subsection{Energy decay}
From assumption (\ref{Kass}), we can directly derive 
$\frac{1}{t^2}$ decay of the energy in certain regions for late times.
\begin{proposition} \label{EKcontrol}
On a hypersurface of constant $t$ we have the bounds
\begin{eqnarray}
\frac{2\pi^2}{M} \int_{r^\star_{K}}^{t-u_0}
\left(-2\nu\right) \frac{\left(r^\star-a\right)^2}{r^2}B^2 r^3 dr^\star
\leq E^K_B \left(t\right) \, ,
\end{eqnarray}
\begin{eqnarray} \label{Bhw}
\frac{2\pi^2}{M} \int_{r^\star_{K}}^{t-u_0} 
\left(-2\nu\right)\frac{t^2}{r^2}B^2 r^3 dr^\star 
\leq E^K_B \left(t\right)
\end{eqnarray}
and 
\begin{eqnarray} \label{derhw}
 \frac{2\pi^2}{M}\int_{r^\star_{K}}^{t-u_0} \Big(\left(u+a\right)^2 \left(\partial_u
 B\right)^2 + \left(v-a\right)^2 \left(\partial_v
 B\right)^2 \Big) r^3 dr^\star \leq 2 E^K_B \left(t\right) \, .
\end{eqnarray}
\end{proposition}
\begin{proof}
The first two bounds follow directly from (\ref{ekbfir}). For the last
inequality note that $2\left(\partial_u B\right)^2 \left(u+a\right)^2
+ 2\left(\partial_v B\right)^2 \left(v-a\right)^2 = 
\left(SB\right)^2 + \left(\underline{S}B\right)^2$ and 
\begin{eqnarray}
\left(SB\right)^2 + \left(\underline{S}B\right)^2 &=& \left(1+2\nu\right)
\left(\left(SB\right)^2 + \left(\underline{S}B\right)^2 \right) +
\left(-2\nu\right) \left(\left(SB\right)^2 +
\left(\underline{S}B\right)^2 \right) \nonumber \\
&\leq& \left(1+2\nu\right)
\left(\left(SB\right)^2 + \left(\underline{S}B\right)^2 \right)
\nonumber \\ 
&&+4 \left(-2\nu\right) \Bigg( \left(SB + \frac{3\left(r^\star-a\right)}{2r} B \right)^2 +
\left(\underline{S}B + \frac{3t}{2r} B \right)^2 \Bigg) \nonumber \\ 
&&+ 3\left(-2\nu\right) \left(\frac{\left(r^\star-a\right)^2}{r^2}B^2 + \frac{t^2}{r^2}B^2 \right)
\end{eqnarray}
and that we control all the terms on the right hand side separately by (\ref{ekbfir}).
\end{proof}
The following proposition is an immediate application of the 
latter and allows us to estimate the energy flux through certain slices for
late times.
\begin{proposition} \label{decfromk}
Let $(r_1^\star,t_1)$, $(\tilde{r}_1^\star,t_1)$ be such that
$t_1 - \tilde{r}_1^\star + a \geq \sqrt{M}$ and $t_1 + r_1^\star - a \geq \sqrt{M}$ and let
additionally $r_1^\star \geq r^\star_{K}$. Then we have
\begin{equation}
m \left(\tilde{r}_1^\star,t_1\right) - m
\left(r_1^\star,t_1\right) \leq 
3M \Bigg(\left(t_1-\tilde{r}^\star_1+a \right)^{-2}
E^K_B \left(t_1\right) + \left(t_1+r^\star_1-a \right)^{-2}
E^K_B \left(t_1\right) \Bigg) \, . \nonumber
\end{equation}
\end{proposition}
\begin{proof}
\begin{eqnarray}
m \left(\tilde{r}_1^\star,t_1\right) - m
\left(r_1^\star,t_1\right) = \int_{r_1^\star}^{\tilde{r}^\star_1}
\partial_{r^\star} m dr^\star = \int_{r_1^\star}^{\tilde{r}^\star_1}
\Big(-\partial_{u} m+ \partial_{v} m \Big) dr^\star
\nonumber \\ 
\leq \int_{r_1^\star}^{\tilde{r}^\star_1} \left(
\frac{\left(B_{,u}\right)^2}{\gamma} r^3 +
\frac{\left(B_{,v}\right)^2}{\kappa} r^3 + B^2 r
\left(\lambda-\nu\right) \left(8 + \frac{\varphi_1\left(B\right)}{B^2}\right)
\right) dr^\star \nonumber \\ \leq
\left(t_1-\tilde{r}^\star_1+a\right)^{-2} 
\int_{r_1^\star}^{\tilde{r}^\star_1} \left(\left(u+a\right)^2
\frac{\left(B_{,u}\right)^2}{\gamma} - \nu \left(u+a\right)^2 \frac{B^2}{r^2} \left(8 + \frac{\varphi_1\left(B\right)}{B^2}\right)\right) 
r^3 dr^\star  \nonumber \\ + \left(t_1+r^\star_1-a\right)^{-2} 
\int_{r_1^\star}^{\tilde{r}^\star_1} \left(\left(v-a\right)^2
 \frac{\left(B_{,v}\right)^2}{\kappa} + \lambda \left(v-a\right)^2 \frac{B^2}{r^2} \left(8 + \frac{\varphi_1\left(B\right)}{B^2}\right) \right) r^3 dr^\star \nonumber \\ \leq 3M \Bigg(\left(t_1-\tilde{r}^\star_1+a\right)^{-2}
E^K_B \left(t_1\right) + \left(t_1+r^\star_1-a \right)^{-2}
E^K_B \left(t_1\right) \Bigg) \, , \nonumber
\end{eqnarray}
where we have used (\ref{dum}), (\ref{vum}) and Proposition
\ref{EKcontrol} as well as the bounds (\ref{kapclo1}) and
(\ref{gamclo1}). 
\end{proof}
The previous proposition can be combined with the bootstrap
assumptions (\ref{ceilass}) and (\ref{hozass}). The fact that
energy is conserved then immediately yields decay 
for any achronal slice in a certain subregion 
of $\mathcal{A}\left(T\right)$ as elaborated in the following
\begin{proposition} \label{decayfromK}
In the bootstrap-region $\mathcal{A}\left(T\right)$ the energy flux 
through any achronal surface 
\begin{equation}
S \subset \mathcal{A}\left(T\right) \cap \{
 r^\star \leq \frac{10}{11}t \} 
\end{equation}
with $v_{-} = \min_{v} S \geq t_0+r^\star_{cl} \geq \sqrt{M}$ and $\min_t S \geq t_0$ 
satisfies
\begin{equation} \label{fluxdec}
E\left(S\right) \leq M^2\frac{C\left(c\right)}{v_-^2} \, .
\end{equation}
\end{proposition}
\begin{figure}[h!]
\[
\input{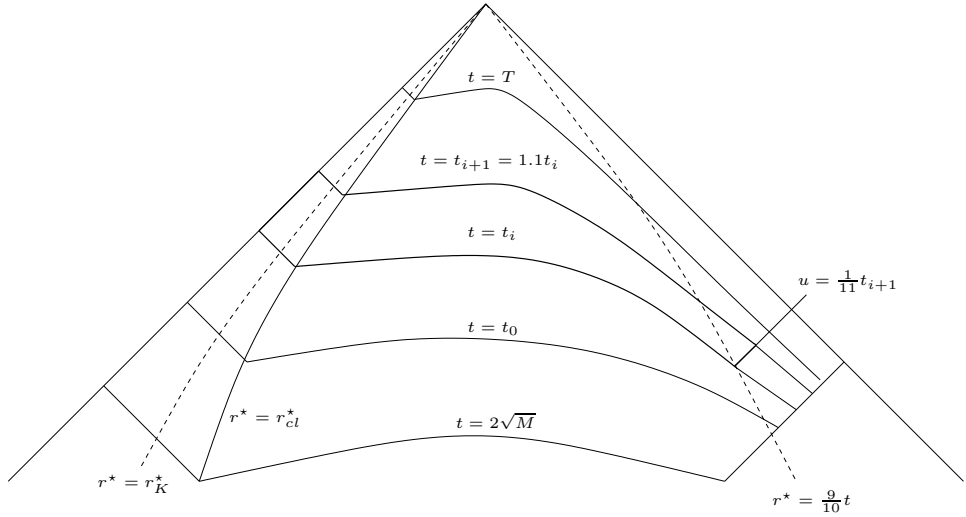}
\]
\caption{Energy decay from $K$.} \label{tlk}
\end{figure}
\begin{proof}
Dyadically decompose the region $\mathcal{A}\left(T\right)$ into
regions 
${}^{u_{hoz}}\mathcal{D}_{[t_i,t_{i+1}]}^{r^\star_{cl},u_0}$ 
with $t_{i+1}=1.1t_i$.\footnote{\label{dyadicexplain} This decomposition implies that the width of each region is of the size of the $t$ coordinate it is at. It should be noted that this decomposition may not fit exactly, i.e.~the last of these dyadic tubes may have a smaller width. To keep the notation reasonably clean this fact is always to be understood implicitly. The results derived for each dyadic region in the paper are of course independent of the fact that the last region may be smaller.} Proposition \ref{decfromk} applied 
to any slice $t_{i+1}$ with
$r_1^\star=r^\star_{K}$ and $\tilde{r}_1^\star=\frac{10}{11}t_{i+1}$ 
yields (for late times, i.e.~when $t-r^\star_{cl}+a \approx t$, which is the case for 
$t \geq t_0$, cf.~Proposition \ref{Cauchystab})
\begin{equation}
m\left(t_{i+1}, r^\star=\frac{10}{11}t_{i+1} \right) - m\left(t_{i+1},
r^\star_{K} \right)  \leq M^2\frac{C\left(c\right)}{t_i^2}
\end{equation} 
Combining this decay in the central region with the energy decay at
the horizon (bootstrap assumptions (\ref{hozass}) and (\ref{ceilass})) 
we find from energy conservation that the energy must decay 
like $\frac{1}{v_-^2}$ through any achronal slice in the 
region where $r^\star \leq \frac{10}{11} t$. This shows
(\ref{fluxdec}), noting that for large times $t \geq t_0$ we have $t \sim v$ in the region $r^\star_{cl} \leq r^\star \leq \frac{9}{10}t$. 
Note in particular that we have this 
decay of energy flux through the regions 
${}^{u_{hoz}}\mathcal{D}_{[t_i,t_{i+1}]}^{r^\star_{cl},u=\frac{1}{11}t_{i+1}}$
for large $t_i$ (cf.~Figure \ref{tlk}, where such a region is depicted). 
\end{proof}
\subsection{Decay estimates for $\kappa$ and $\gamma$}
The following proposition establishes appropriate decay bounds on $\kappa$ and
$\gamma$ sufficient to improve the estimate  (\ref{fibor}) for the relation between 
$r^\star$ and $r$ in the central region in the next section. 
\begin{proposition} \label{kgom}
In the region $\mathcal{A}\left(T\right) \cap \{ r^\star \leq
r^\star_{cl} \} \cap \{v\geq v_0\}$ we have
\begin{equation} \label{kgclose0}
\Big|\kappa-\frac{1}{2}\Big| \leq C_L\frac{M}{v^2} \, .
\end{equation}
In the region $\mathcal{A}\left(T\right) 
\cap \{ r^\star \geq r^\star_K \} \cap \{
r^\star \leq \frac{9}{10}t \}$ we have
\begin{equation} \label{kgclose1}
\frac{1}{2} \leq \kappa \leq \frac{1}{2} + C_L \left(2+c\right) \frac{M}{t^2} 
\end{equation}
\begin{equation} \label{kgclose2}
\frac{1}{2} \geq \gamma  \geq \frac{1}{2} - C_K \ c \frac{M}{t^2}
\end{equation}
with $C_K=\sup_{r^\star \geq r^\star_{K}} \frac{1}{1-\mu}$ and $C_L=\sup_{r^\star \geq r^\star_{cl}} \frac{1}{1-\mu}$.
\end{proposition}
\begin{proof}
Integrating equation (\ref{kapevol})
from the set $\{ u=T-r^\star\left(T,r_K\right)\} \cup \Big(\{t=T\}
\cap \{r^\star\left(T,r_K\right) \leq r^\star \leq r^\star_{cl}
\}\Big)$, where $\kappa=\frac{1}{2}$ by definition, to any point in
the region $r^\star \leq r^\star_{cl}$ 
yields after inserting bootstrap assumption (\ref{intebound2})
\begin{equation} \label{kbrst}
\Big|\kappa \left(t,r^\star_{cl}\right)-\frac{1}{2}\Big| \leq 2C_L \frac{M}{v^2}
\end{equation}
in that region establishing (\ref{kgclose0}). 
We can obtain $\kappa$ at any point in the 
remaining region $\mathcal{A}\left(T\right) \cap \{ r^\star \geq r^\star_{cl} \} \cap \{r^\star \leq \frac{9}{10}t \}$ by integrating from the set
$L=\Big\{\{t=T\} \cap \{r^\star \geq r^\star_{cl}\}\Big\} \cup \{
r^\star=r^\star_{cl}\}$  on which either $\kappa$ is equal to
$\frac{1}{2}$ or satisfies the estimate (\ref{kbrst}), to the desired
point. An application of Proposition \ref{decayfromK} then yields
(\ref{kgclose1}) in the region $\mathcal{A}\left(T\right) \cap \{
r^\star_K \leq r^\star \leq \frac{9}{10}t \}$ as follows:
\begin{eqnarray}
\kappa\left(t,r^\star\right) &=& \kappa\left(u_L,v\right) \exp
\left(-\int^{u_L}_{t-r^\star} \frac{2}{r^2} \frac{\zeta^2}{\nu}
\left(\bar{u},v\right)d\bar{u} \right) \nonumber \\
&\leq& \kappa\left(u_L,v\right) \exp
\left(\sup \left[\frac{2}{r^2\left(1-\mu\right)}\right]\int^{u_L}_{t-r^\star} \frac{4\lambda}{\Omega^2} \zeta^2
\left(\bar{u},v\right)d\bar{u}\right)
\nonumber \\
&\leq& \left[\frac{1}{2}+2C_L\frac{M}{\left(t+r^\star\right)^2} \right]
\left(1+c C_L \frac{M}{t^2}\right) \, .
\end{eqnarray}
%Note that the value of the constant is determined by the value of
%$\frac{1}{1-\mu}$ on $r^\star=r^\star_{cl}$.
%
%
%
\begin{comment}
Finally, consider 
a point $P=\left(u,v\right)$ in the region 
$\mathcal{A}\left(T\right) \cap \{r^\star \geq \frac{9}{10}t \}$
and integrate $\kappa$ along a const $v$ line from the appropriate 
point on the $r^\star=\frac{9}{10}t$ curve (or the $t=T$-curve if the
point is accessible along constant $v$ from there) to $P$ to find the
bound 
\begin{equation}
\kappa\left(t,r^\star\right) \leq  \frac{1}{2} + \frac{C}{t^2} + \frac{\epsilon}{r^2} 
\end{equation}
in the region $r^\star \geq \frac{9}{10}t$. The estimate then follows
from the fact (cf.~equation \ref{rstarrb}) 
that $r \sim r^\star$ and $r^\star \geq \frac{9}{10}t$
also hold there. This establishes (\ref{kgclose1}).
\end{comment}
%
%
%

For the estimate (\ref{kgclose2}), we first note that $\gamma= \frac{1}{2}$
on $\{t=T\} \cap \{r \geq r_K\}$. On the
$r^\star=\frac{9}{10}t$ curve we can obtain (\ref{kgclose2})
by integrating (\ref{gammaevol})
from $\{t=T\} \cap \{\frac{9}{10}T \leq r^\star \leq
T-u_0\}$ downwards to any point in the region
$\mathcal{A}\left(T\right) \cap \{ r^\star \geq \frac{9}{10}t \}$.  
We use that $\lambda$ is bounded below and $|\theta|\leq
\frac{C\left(\epsilon\right)\sqrt{M}}{\sqrt{r}}$ in 
the integration region, both following from
Proposition \ref{simpsmall}, to obtain
\begin{equation} \label{immb}
\gamma \geq  \frac{1}{2} - \frac{C\left(\epsilon\right)M}{r^2}
\end{equation}
there. Since $r$ is controlled by $r^\star$ 
(cf.~equation (\ref{rstarrb})) and $r^\star \geq \frac{9}{10}t$ 
in the region under consideration, we find 
the bound (\ref{kgclose2}) in the region
$\mathcal{A}\left(T\right) \cap \{ r^\star \geq \frac{9}{10}t \}$, in
particular on the $r^\star = \frac{9}{10}t$-curve.
Finally, the value of $\gamma$ at any point in the remaining region 
$\mathcal{A}\left(T\right) \cap \{ r^\star_K \leq r^\star \leq
\frac{9}{10}t \}$ can be obtained by integrating
(\ref{gammaevol}) in $v$ from some point of the 
set $L^\prime=\Big\{\{t=T\} \cap \{r^\star \leq
\frac{9}{10}t\}\Big\} \cup \{r^\star=\frac{9}{10}t\}$ 
(on which $\gamma$ already satisfies (\ref{kgclose2})). 
Using the decay of the energy flux we arrive at 
(\ref{kgclose2}) in the remaining region:
\begin{eqnarray}
\gamma\left(t,r^\star\right) &=& \gamma\left(u,v_{L^\prime}\right) \exp
\left(-\int^{v_{L^\prime}}_{t+r^\star} \frac{2}{r^2} \frac{\theta^2}{\lambda}
\left(u,\bar{v}\right)d\bar{v} \right) \nonumber \\
&\geq& \gamma\left(u,v_{L^\prime}\right) \exp
\left(-\sup \left[\frac{2}{r^2\left(1-\mu\right)}\right]
\int^{v_{L^\prime}}_{t+r^\star} \frac{\theta^2}{\kappa}
\left(u,\bar{v}\right)d\bar{v} \right)
\nonumber \\
&\geq&  \frac{1}{2} \left(1-2C_K \ c \frac{M}{t^2}\right) \, .
\end{eqnarray}
\end{proof}
From the proof of the $\gamma$-estimate we deduce:
\begin{corollary} \label{betterc}
In the region $r^\star \geq r^\star_{cl}$, 
the estimate (\ref{kgclose2}) holds with the constant $C_K$ 
replaced by $C_L$. 
\end{corollary}
In the asymptotic region $t$ is like $r$ and the bounds extend:
\begin{corollary} \label{kapgamdecr}
In $\mathcal{A}_T \cap \{ r^\star \geq \frac{9}{10}t \} \cap \{v\geq v_0\}$ we have
\begin{equation}
\frac{1}{2} \leq \kappa \leq \frac{1}{2} + 2 C_L \left(2+c\right)
\frac{M}{r^2} \textrm{ \ \ \ \ \ and \ \ \ \ \ }
\frac{1}{2} \geq \gamma \geq \frac{1}{2} + \frac{C\left(\epsilon\right)}{r^2}
\end{equation}
\end{corollary}
\begin{proof}
The bound for $\gamma$ is the statement of (\ref{immb}). To obtain the
bound for $\kappa$ integrate (\ref{kapevol}) from $r^\star =
\frac{9}{10}$ to the asymptotic region of $\mathcal{A}\left(T\right)$ 
in $u$ using that $t \sim r^\star \sim r$ in the region $r^\star \geq
\frac{9}{10}t$ (cf. again (\ref{rstarrb})) and that the energy 
estimate holds in the region under consideration. Note again that $r$ could be replaced by $t$ in that region.
\end{proof}
\subsection{Stability of the coordinate systems} \label{stabcorsec}
\subsubsection{The relation between $r^\star$ and $r$.}
We are now in a position to derive an estimate for the relation 
between the coordinate $r^\star=\frac{v-u}{2}$ and the 
function $r\left(u,v\right)$. This estimate in conjunction 
with Proposition \ref{Cauchystab} will automatically 
improve bootstrap assumption \ref{boot1}, 
which -- modulo the error-term -- expresses precisely the relation of 
the tortoise coordinate $r^\star$ to the area radius in 
the five-dimensional Schwarzschild metric. For this section we will use $C(\left(r_K,c\right)$ to denote a constant which depends on the weight of $\frac{1}{1-\mu}$ on $r=r_K$ and on the parameter $c$ in the bootstrap assumptions.
\begin{proposition} \label{rstarrdec}
The estimate
\begin{eqnarray} \label{rstarr}
\Big| r^\star - \left[ r\left(t,r^\star\right) + \sqrt{\frac{M_A}{2}} \left(\log
\left(\frac{r\left(t,r^\star\right)-\sqrt{2M_A}}{r\left(t,r^\star\right)+\sqrt{2M_A}}\right)
+ p \right) \right] \Big| \leq C\left(r_K,c\right)\frac{M}{t}
\end{eqnarray}
with $p$ defined in (\ref{pdef}) and $M_A$ the Hawking mass at
the point $\left(T,r^\star=0\right)$, holds in the 
region $\mathcal{A}\left(T\right) \cap \{ r^\star \geq r^\star_K \}$.
\end{proposition}
\begin{proof}
The estimates 
\begin{equation} \label{lam+nu}
|\partial_{t} r| = |\lambda + \nu |\leq C\left(r_K,c\right)\frac{M}{t^2}
 \, ,
\end{equation}
\begin{equation} \label{lam-nu}
\partial_{r^\star} r  = \lambda - \nu \leq \left(1-\mu\right) + C\left(r_K,c\right)\frac{M}{t^2}
\leq \left(1-\frac{2M_A}{r^2}\right) + C\left(r_K,c\right) \frac{M}{t^2}
\end{equation}
in the region $\mathcal{A}\left(T\right) \cap \{ r^\star \geq r^\star_K \} \cap \{
r^\star \leq \frac{9}{10}t \}$ are a direct consequence of 
Proposition \ref{kgom}. Since also $r\left(T,r^\star=0\right) = 2\sqrt{M_A}$, 
the relation (\ref{rstarr}) follows in the region
$\mathcal{A}\left(T\right) \cap \{r^\star \geq r^\star_K \} \cap
\{r^\star \leq \frac{9}{10}t \}$. An application of Corollary
\ref{kapgamdecr} finally extends the bound to the remaining region,
$r^\star \geq \frac{9}{10}t$. 
\end{proof}
This means that in the region 
$\mathcal{A}\left(T\right) \cap \{ r^\star \geq r^\star_K \}$ we can go
back and forth from $r$ to $r^\star$ with an error-term of
$\frac{1}{t}$, which is small at late times. In analogy with 
Corollary \ref{betterc} we also have
\begin{corollary}
In the region $\mathcal{A}\left(T\right) \cap \{ r^\star \geq
r^\star_{cl} \}$ the estimate (\ref{rstarr}) holds with 
constant $C\left(r_{cl},c\right)$ replacing $C\left(r_K,c\right)$.
\end{corollary} 
\subsubsection{Stability of constant $t$ slices}
\label{StabCord}
In this section we are going to study the relation of the different
coordinate systems $\mathcal{C}_{\tilde{\tau}}=\left(u_{\tilde{\tau}},v_{\tilde{\tau}}\right)$ associated with 
different $\tilde{\tau} \in A$ (cf.~Section \ref{Coordinates}). 
Instead of the smallness estimates entering the proof 
of Proposition \ref{Cauchystab}, we will now exploit the decay 
estimates for the quantities $\kappa$ and $\gamma$
derived in Proposition \ref{kgom}. Recall Notation \ref{tnotation}. 
\begin{proposition} \label{stabt}
Let $\tilde{\tau}_A \in A$. In view of 
Corollary \ref{empty} assume $\tilde{\tau}_A \geq
\tilde{\tau_0}$. Then in the coordinate 
system  $\mathcal{C}_{\tilde{\tau}_A}=\left(u_{\tilde{\tau}_A},v_{\tilde{\tau}_A}\right)$ associated to $\tilde{\tau}_A$ the $t$-coordinate of 
the initial data slice satisfies the bound 
\begin{equation}
\sup_{\tilde{S} \cap \{r^\star \geq r^\star_K\} \cap \{u \leq u_0\}} |t - \sqrt{M}| \leq C\left(\epsilon\right) \, . \label{sc1}
\end{equation}
\end{proposition}
\begin{proof}
By Proposition \ref{Cauchystab} the statement (\ref{sc1}) 
already holds up to $\tilde{\tau}_0$ by 
a suitable smallness assumption on the initial data in all coordinate 
systems $\mathcal{C}_{\tilde{\tau}}$ with $\tilde{\tau} \leq \tilde{\tau}_0$. 
Consider now a coordinate system $\mathcal{C}_{\tilde{\tau}_A}$ 
for a $\tilde{\tau}_A \geq \tilde{\tau}_0$. Recall the 
vectorfields $G$ and $R$ defined in (\ref{RGvector}).
In the following, we again frequently refer to 
figure \ref{codiag} of section \ref{Coordinates}. 
At the point $A$ we have $t^{\tilde{\tau}_A}_{\tilde{\tau}_A}=T$ 
by definition. We would like to estimate the value $t^{\tilde{\tau}_A}_{0}$ 
at the point $D$ and compare it to $\sqrt{M}$, which is the value 
of $t$ if $\tilde{\tau}=0$ and the coordinates are defined on the 
initial data. The rate at which $t$ changes in affine
parameter $\tau$ along the integral curve of $\nabla_\perp r$ is given
by (\ref{upsa}). We first integrate (\ref{upsa}) along the curve $r^2=4m_A$, 
from $A$ to the point $A^\prime$, which is defined to be 
on the $\nabla r$ slice associated with $\tilde{\tau}_0$. 
Using the decay estimates
\begin{equation}
|\kappa + \gamma - 1 | \leq C_L\left(2+c\right) \frac{M}{t^2} \textrm{ \ \ \ and \ \ \ }
 \Big| \frac{1}{\sqrt{1-\mu}} -\frac{1}{\sqrt{1-\mu_A}} \Big| \leq C\left(c\right) 
\frac{M}{t^2}
\end{equation}
which hold along the curve by Proposition \ref{kgom} (and its Corollaries), we 
obtain an estimate
\begin{equation}
T-t^{\tilde{\tau}_A}_{\tilde{\tau}_0} \leq \frac{1}{\sqrt{1-\frac{2m_A}{r^2}}}\left(\tau_{A} - \tau_{A^\prime} \right) + C\left(c,r^\star_{cl}\right)\frac{M}{t^{\tilde{\tau}_A}_{\tilde{\tau}_0}}
\end{equation}
where the last term is small and the constant $C\left(c, r^\star_{cl}\right)$ depends on the weight of $\frac{1}{1-\mu}$ on $r^\star_{cl}$. Using the definition of 
$T=\vartheta\left(\tilde{\tau}_{A}\right) = \sqrt{M} +\frac{\tau_{A}}{\sqrt{1-\mu_A}}$ 
we derive
\begin{equation}
\Big| t^{\tilde{\tau}_A}_{\tilde{\tau}_0} - \frac{\tau_{A^\prime}}{\sqrt{1-\frac{2m_{A}}{r^2}}} -\sqrt{M} \Big| \leq C\left(c,r^\star_{cl}\right)\frac{M}{t^{\tilde{\tau}_A}_{\tilde{\tau}_0}}
\end{equation}
and with the bootstrap assumption on the energy
\begin{equation} \label{dect}
\Big| t^{\tilde{\tau}_A}_{\tilde{\tau}_0} - \frac{\tau_{A^\prime}}{\sqrt{1-\frac{2m_{A^\prime}}{r^2}}} -\sqrt{M}\Big| \leq C\left(c,r^\star_{cl}\right)\frac{M}{t^{\tilde{\tau}_A}_{\tilde{\tau}_0}}
\end{equation}
In the second step we integrate (\ref{upsa}) from $A^\prime$ 
to $D$. In this region we can use the smallness estimates 
for $\kappa$ and $\gamma$ as in the proof of 
Proposition \ref{Cauchystab} obtaining 
\begin{equation} \label{smat}
\Big| t^{\tilde{\tau}_A}_{\tilde{\tau}_0} - t^{\tilde{\tau}_A}_{0} \Big| \leq \frac{1}{\sqrt{1-\frac{2m_{A^\prime}}{r^2}}}\left(\tau_{A^\prime}  \right) + \sqrt{M} C\left(\epsilon\right) \, ,
\end{equation}
where we used the fact $\frac{1}{\sqrt{M}} C\left(\epsilon\right) \cdot \tau_{A^\prime}$ is small by a suitable choice of the initial data, which in turn follows from the smallness of $\frac{1}{\sqrt{M}}C\left(\epsilon\right) \tilde{\tau}_0$ by Proposition \ref{Cauchystab} and the estimate $|\tau_{A^\prime} - \tilde{\tau}_0| \leq \sqrt{M} C\left(\epsilon\right)$.
%(that $\frac{1}{\sqrt{M}} C\left(\epsilon\right) \cdot \tilde{\tau}_0$ is 
%small follows from 
%Proposition \ref{Cauchystab}). 
Putting together the 
estimates (\ref{dect}) and (\ref{smat}) we obtain
\begin{equation}
\Big| \sqrt{M} + \frac{\tau_{A^\prime}}{\sqrt{1-\frac{2m_{A^\prime}}{r^2}}} - t^{\tilde{\tau}_A}_{0}\Big| \leq \frac{\tau_{A^\prime}}{\sqrt{1-\frac{2m_{A^\prime}}{r^2}}} + C\left(c,r^\star_{cl}\right)\frac{M}{t^{\tilde{\tau}_A}_{\tilde{\tau}_0}} + \sqrt{M} C\left(\epsilon\right)
\end{equation}
from which it follows (choosing $\tilde{\tau}_0$ large enough and the
initial data suitably small) that the $t$-coordinate at $D$ is 
close to $\sqrt{M}$. In the second step, which is 
identical to the one in Proposition \ref{Cauchystab}, one 
finally shows that $t$ only changes by $C\left(\epsilon\right)$ along 
the $\nabla r$-curve through $D$ on which the initial data is defined. 
\end{proof}

\begin{corollary} \label{boot2impr}
Bootstrap assumption $2$ is improved.
\end{corollary}

One easily generalizes the previous proposition 
to the statement that $t$ does not change 
much along a $\nabla r$ integral curve located in the 
bootstrap region:
\begin{proposition}
With the assumptions of Proposition \ref{stabt}, the $t$-coordinate along the
$\nabla r$ integral curve associated to 
$\sqrt{M} \leq \tilde{\tau}_i \leq \tilde{\tau}_A$ satisfies
\begin{equation}
\sup_{\left(\nabla r\right)_{\tilde{\tau}_i} \cap \{r^\star \geq r^\star_K\}}
|t-\vartheta\left(\tilde{\tau}_i\right)| \leq C\left(\epsilon\right)
\end{equation}
in the coordinate 
system  $\mathcal{C}_{\tilde{\tau}_A}=\left(u_{\tilde{\tau}_A},v_{\tilde{\tau}_A}\right)$.
\end{proposition}
\begin{proof}
Repeat the proof of the previous Proposition, now integrating only up to
the $\nabla r$ integral curve associated with $\tilde{\tau}_i$. In the 
second step, when integrating equation (\ref{talongnabr}) along
the $\nabla r$ integral curve, one again uses the smallness estimate for
$\kappa - \gamma$ in $[r^\star_K,\tilde{R}^\star]$. However 
the decay estimate 
\begin{equation}
\left(\gamma - \kappa\right) \frac{1}{\sqrt{1-\mu}}  \leq 
\tilde{\epsilon}\frac{M^\frac{3}{4}}{r^\frac{3}{2}}
\end{equation}
following from Proposition (\ref{kapgamdecr}) can now be used 
in the region $[\tilde{R}, \infty)$. The $\tilde{\epsilon}$ 
arises because $\frac{M^\frac{1}{4}}{\sqrt{r}}$ is small in
$\left[\tilde{R},\infty\right)$. Inserting that the affine parameter $\varrho$ 
is proportional to $r$ (cf. (\ref{rgrowth})), one concludes that
\begin{equation}
0 \leq -\frac{dt}{d\varrho} \leq \hat{\epsilon} \frac{M^\frac{3}{4}}{\varrho^\frac{3}{2}}
\end{equation}
and hence the change in $t$ along any $\nabla r$ integral 
curve is also small within the region $[\tilde{R}, \infty)$. 
\end{proof}
\begin{proposition} \label{Cauchystab2}
Statement $3$ of Proposition \ref{Cauchystab} holds in any coordinate system
$\mathcal{C}_{\tilde{\tau}}$ for $\tilde{\tau} \geq \tilde{\tau}_0$
and $\tilde{\tau} \in A$.  
\end{proposition}
\begin{proof}
Bootstrap assumption $1$ and the previous Proposition implies that the
location of the region
${}^{u_{hoz}}\mathcal{D}^{r^\star_K, u_0}_{[2\sqrt{M},
  T_0]}$ only changes slightly between the different coordinate
systems. In particular, the $v$ coordinate of the region ${}^{u_{hoz}}\mathcal{D}^{r^\star_K, u_0}_{[2\sqrt{M},
  T_0]}$ is uniformly bounded in the different coordinate systems, as
is the $t$ coordinate for $r^\star \geq r^\star_K$. Hence 
if statement $3$ of Proposition \ref{Cauchystab} holds in the
coordinate system $\mathcal{C}_{\tilde{\tau}_0}$ it also holds in the
coordinate system $\mathcal{C}_{\tilde{\tau}}$ for $\tilde{\tau} \geq
\tilde{\tau}_0$ and $\tilde{\tau} \in A$.\footnote{The
  $\delta$ of Proposition \ref{Cauchystab} may have to be chosen
  slightly smaller but the change is uniform in $\tilde{\tau}_A$ and 
hence the size of the bootstrap region!}
\end{proof}
Finally, we conclude from Proposition \ref{rstarrdec} 
\begin{corollary} \label{boot1impr}
Bootstrap assumption $1$ is improved.
\end{corollary}
\begin{proof}
We apply Proposition \ref{Cauchystab}, i.e.~we choose $t_0$ large 
such that $C\left(r_K,c\right)\frac{M}{t_0}$ is very small 
(in particular smaller than $\frac{c}{4}$) 
and the initial data so small that the bootstrap assumptions hold 
with constant $\frac{c}{2}$ at $t=t_0$. Then for $t > t_0$ the estimate of 
Proposition \ref{rstarrdec} takes over and improves the 
constant $c$ in (\ref{fibor}).
\end{proof}
\subsection{Pointwise bounds} \label{pointwiseE}
In this subsection we derive pointwise decay bounds on the 
squashing field $B$ and its derivatives, as well as on some higher
order quantities. The key idea is that these bounds hold up to some
large time $t_0$ by Cauchy stability (cf.~Propositions
\ref{Cauchystab} and \ref{Cauchystab2}).\footnote{The location of $t=t_0$ might change slightly from coordinate system to coordinate system but the change is uniformly controlled by $C\left(\epsilon\right)$ as has just been established in section \ref{StabCord}.} After that time the energy decay derived from the bootstrap assumptions in Proposition \ref{decayfromK} ensures appropriate decay estimates for the fields. 
\subsubsection{The squashing field and its derivatives}
\begin{proposition} \label{pointBcent}
The pointwise bound
\begin{equation} \label{pointwiseB}
|B\left(t,r^\star \right)| \leq
 \sqrt{C_L} \ C\left(c\right) \frac{\sqrt{M}}{t}
\end{equation}
holds everywhere in $\mathcal{A}\left(T\right) 
\cap \{r^\star_{cl} \leq r^\star \leq \frac{9}{10}t \}$. 
\end{proposition} 
\begin{proof}
The estimate (\ref{pointwiseB}) holds for $t \in [2\sqrt{M},t_0]$ 
(for some large but finite $t_0$) 
by Proposition \ref{Cauchystab2} with an appropriate choice of 
the initial data. For $[t_0,T]$ we integrate
out in the $u$-direction from the set 
$L=\{u=u_0\} \cup \left(\{t=t_0\} \cap \{r^\star \geq
\frac{9}{10}t_0\right) \}$, where either 
$B \equiv 0$ by the assumption of compact support 
or the bound (\ref{pointwiseB}) holds by Cauchy 
stability, to the $r^\star=\frac{9}{10}t$ curve: 
\begin{eqnarray}
B\left(t, r^\star=\frac{9}{10}t \right) &=& B\left(u_L,v=\frac{19}{10}t\right) + \int_{u_L}^{\frac{1}{10}t}
B_{,u} \left(u,v=\frac{19}{10}t\right) du \\ &\leq& \tilde{\delta} \frac{\sqrt{M}}{v} +
\sqrt{\int_{u_L}^{\frac{1}{10}t}\frac{4\lambda}{\Omega^2} \zeta^2 du
}\sqrt{\int_{u_L}^{\frac{1}{10}t} \frac{-4\kappa \nu}{4r^3 \lambda}
  du} \leq C\left(\epsilon\right) \frac{\sqrt{M}}{t} \nonumber 
\end{eqnarray}
since $r \sim r^\star \sim t$ on the curve. Note also that 
along a line of constant $v$ in the region 
$r^\star \geq \frac{9}{10}t$ we have $v \sim t$.
Integrating out further from a point
$\left(t,\frac{9}{10}t\right)$ on 
the $r^\star=\frac{9}{10}t$-curve along the slice $t=const$ we obtain
\begin{eqnarray}
|B| &\leq& C\left(\epsilon\right)\frac{\sqrt{M}}{t} + \int^{\frac{9}{10}t}_{r^\star_{cl}}
\left|\partial_{r^\star} B\right| dr^\star \nonumber \\ &\leq& C\left(\epsilon\right)\frac{\sqrt{M}}{t} + \sqrt{
  \int^{\frac{9}{10}t}_{r^\star_{cl}} \left(\partial_{r^\star}
  B\right)^2 r^3 dr^\star}\sqrt{
  \int^{\frac{9}{10}t}_{r^\star_{cl}} \left[-\frac{\partial}{\partial
    r^\star} \frac{1}{r^2} \right]\left(
  \frac{1}{2\frac{\partial r}{\partial r^\star}}\right)dr^\star}
\nonumber \\ &\leq& C\left(\epsilon\right)\frac{\sqrt{M}}{t} + C\left(c\right)\frac{\sqrt{M}}{t} \sup
\sqrt{\frac{1}{2\frac{\partial r}{\partial r^\star}}} 
\end{eqnarray} 
yields (\ref{pointwiseB})
in the whole region $r^\star_{cl} \leq r^\star \leq \frac{9}{10}t$ since $\sup \sqrt{\frac{1}{2\frac{\partial r}{\partial r^\star}}} \leq \frac{4}{5}C_L$ in the region $r^\star_{cl}$.
\end{proof}

Recall that in the region $r^\star \geq \frac{9}{10}t$
we were able to derive $\frac{1}{r}$-decay 
of the field $B$ without involving the bootstrap assumptions 
(cf. Corollary \ref{Bcor}). The next Proposition shows that the 
boundedness of the quantity $E^K_B$ improves this decay considerably:
\begin{proposition} \label{decayinr}
In the region $\mathcal{X} = \mathcal{A}\left(T\right) \cap \{r^\star \geq 3\sqrt{M} \} \cap \{u \geq \sqrt{M}
\}$ we have
\begin{equation}
|B| \leq C\left(c\right)\frac{M}{r^\frac{3}{2}u^\frac{1}{2}} \, .
\end{equation}
\begin{proof}
Choose a point $(t,r^\star)$ in the region $\mathcal{X}$ and a point 
$(t,\tilde{r}^\star)$ in the central region ($r^\star_{cl} \leq
\tilde{r}^\star \leq 3\sqrt{M} $). We have
\begin{eqnarray}
r^3 B^2 \left(t,r^\star\right) = r^3B^2 \left(t,\tilde{r}^\star\right)
+ \int_{\tilde{r}^\star}^{r^\star} \partial_{r^\star} \left(r^3
B^2\right) dr^\star \, .
\end{eqnarray}
By Proposition \ref{pointBcent},  $|B\left(t,\tilde{r}^\star\right)|
\leq \frac{C}{t}$. Moreover, by Cauchy-Schwarz
\begin{eqnarray}
\int_{\tilde{r}^\star}^{r^\star} \partial_{r^\star} \left(r^3
B^2\right) dr^\star \leq 2 \sqrt{\int_{\tilde{r}^\star}^{r^\star}
  \left(\partial_{r^\star} B\right)^2 r^3 u^2 dr^\star}\frac{1}{t}
\sqrt{\int_{\tilde{r}^\star}^{r^\star} t^2 B^2 r^3 \frac{1}{u^2}
  dr^\star} \nonumber \\ + \frac{1}{t^2}
\int_{\tilde{r}^\star}^{r^\star} t^2 B^2 r^2 dr^\star \, .
\end{eqnarray}
We can finally insert the inequalities (\ref{derhw}) and (\ref{Bhw}) to find
\begin{eqnarray}
B^2 \left(t,r^\star\right) \leq C_L \ C\left(c\right) \frac{M \left(3\sqrt{M}\right)^3}{r^3 t^2} + 2C\left(c\right) \frac{M^2}{tr^2 u} + C\left(c\right) \frac{M^2}{r^2t^2} \, .
\end{eqnarray}
Noting that in the region $u \geq 1$ we have for large times $t \geq
\frac{r}{2}$ yields the desired result.
\end{proof}
\end{proposition}
For the region region $u_0 \leq
u \leq \sqrt{M}$ we can follow the same proof replacing $u^2$ by $u^2 + M$
(to avoid dividing by zero) to obtain
\begin{proposition} \label{decayinr2}
In the region $\widetilde{\mathcal{X}} =  \mathcal{A}\left(T\right) \cap \{r^\star \geq 3\sqrt{M}\}$ we have
\begin{equation}
|B| \leq C\left(c\right)\frac{M^\frac{3}{4}}{r^\frac{3}{2}} \, .
\end{equation}
\end{proposition}
Having established better decay of $B$ from the bootstrap 
assumptions we can also derive better decay of $\theta$ 
via an auxiliary quantity $\Theta$, which is the analogue of the 
almost Riemann invariant in four dimensions. Note however that 
we cannot use $\Theta$ to improve the decay in 
$B$ itself (as in the four-dimensional case) but only 
in its derivatives, once better decay 
in $B$ has already been established.

\begin{lemma} \label{auxdecay}
On $\{r^\star=\frac{9}{10}t\}$, the quantity $\Theta = \theta + \frac{3}{2}\sqrt{r} \lambda B$ satisfies
\begin{equation}
|\Theta\left(u,v\right)| \leq C\left(c\right)\frac{M^\frac{3}{4}}{v_+} \label{Thetabound}
\end{equation}
\end{lemma}
\begin{proof}
Integrate the equation (recall definition (\ref{deltaB}))
\begin{equation} \label{Thetaevol}
\partial_u \Theta = B \left[\frac{35\lambda \nu}{4\sqrt{r}} - \frac{11\Omega^2}{2r^{\frac{5}{2}}}m\right] + \frac{1}{4} \frac{\Omega^2}{\sqrt{r}} \frac{\varphi_2\left(B\right)}{B} - \frac{\Omega^2}{2\sqrt{r}} B \left(\rho-\frac{3}{2}\right)
\end{equation}
from the set $L=\{u=u_0\} \cup \left(\{t=t_0\} \cap \{r^\star \geq \frac{9}{10}t_0\right) \}$, where either $\Theta \equiv 0$ by the assumption of compact support or the bound (\ref{Thetabound}) 
holds by Proposition \ref{Cauchystab} with constant $\tilde{\delta}$, to the $r^\star=\frac{9}{10}t$ curve. Since the right hand side of equation (\ref{Thetaevol}) satisfies
\begin{equation}
\Bigg|B \left[\frac{35\lambda \nu}{4\sqrt{r}} - \frac{11\Omega^2}{2r^{\frac{5}{2}}}m\right] + \frac{1}{4} \frac{\Omega^2}{\sqrt{r}} \frac{\varphi_2\left(B\right)}{B} - \frac{\Omega^2}{2\sqrt{r}} B \left(\rho-\frac{3}{2}\right)\Bigg| \leq C\left(c\right)\frac{M}{r^2}
\end{equation} 
in the region $r^\star \geq \frac{9}{10}t$ following in turn from the decay of $B$ derived in Proposition \ref{decayinr}, we obtain the estimate
\begin{equation}
|\Theta \left(u,v\right)| \leq \tilde{\delta} \frac{M^\frac{3}{4}}{v_+} + C\left(c\right)\frac{M^\frac{3}{4}}{r} \leq C\left(c\right)\frac{M^\frac{3}{4}}{v_+} \, .
\end{equation}
\end{proof}

\begin{corollary} \label{thetardec}
\begin{equation}
|\theta\left(u,v\right)| \leq C\left(c\right)\frac{M^\frac{3}{4}}{r}
\end{equation} 
holds in $r^\star \geq \frac{9}{10}t$.
\end{corollary}
\begin{proof}
Use Lemma \ref{auxdecay} and take into account Proposition \ref{decayinr}.
\end{proof}

\begin{proposition} \label{pointthcent}
The pointwise bound
\begin{equation} \label{pointwiseth}
|\theta\left(t,r^\star \right)| \leq C\left(c\right) \frac{M^\frac{3}{4}}{t} 
\end{equation}
holds everywhere in $\mathcal{A}\left(T\right) \cap \{r^\star_{cl} \leq r^\star \leq
\frac{9}{10}t \}$. 
\end{proposition} 
\begin{proof}
By the previous corollary
\begin{equation}
|\theta\left(u,v\right)| \leq C\left(c\right)\frac{M^\frac{3}{4}}{t} \label{thetabound}
\end{equation} 
holds on $r^\star = \frac{9}{10}t$.
We can integrate equation (\ref{dut}) from that curve 
or $t=t_0$, where the bound $|\theta| \leq \tilde{\delta} \frac{M^\frac{3}{4}}{v}$ 
holds by Proposition \ref{Cauchystab}, to any point in the region 
$\mathcal{A}\left(T\right) \cap \{r^\star \leq \frac{9}{10}t \}$:
\begin{equation}
\theta \left(u,v\right) = \theta\left(u_i,v\right) - \frac{3}{2} \int_{u_i}^u
\frac{\zeta}{r} \lambda d\bar{u} + \int_{u_i}^u \frac{\Omega^2}{3\sqrt{r}}
\left(e^{-8B}-e^{-2B}\right) d\bar{u}
\end{equation}
and hence
\begin{eqnarray}
|\theta \left(u,v\right)| &\leq& \tilde{\delta} \frac{M^\frac{3}{4}}{v} + \frac{3}{2} \sqrt{\int_{u_i}^u
\frac{4\zeta^2\lambda}{\Omega^2}d\bar{u}} \sqrt{\int_{u_i}^u
\frac{-4\kappa \nu \lambda}{r^2}d\bar{u}}  \\ 
&+& \frac{1}{3} \sqrt{\int_{u_i}^u
-4\kappa \nu  \left(e^{-8B}-e^{-2B}\right)^2 r d\bar{u}} \sqrt{\int_{u_i}^u
\frac{-4\kappa \nu}{r^2}d\bar{u}} \leq C\left(c\right)\frac{M^\frac{3}{4}}{v} \nonumber \, .
\end{eqnarray}
The energy estimate, Proposition \ref{decayfromK} and the fact 
that $v \sim t$ in the region $r^\star_{cl} \leq r^\star \leq
\frac{9}{10}t$ yields the desired result.
\end{proof}
With the previous Proposition and Proposition \ref{simpsmall}, Corollary \ref{thetardec} is easily 
extended to the entire bootstrap region:
\begin{corollary} \label{thetardece}
\begin{equation}
|\theta\left(u,v\right)| \leq C\left(c\right)\frac{M^\frac{3}{4}}{r}
\end{equation} 
holds in all of $\mathcal{A}\left(T\right)$.
\end{corollary}
Close to the horizon we have
\begin{proposition} \label{pointBthhoz}
The pointwise bounds
\begin{equation} \label{pointwiseBh} 
|B\left(t,r^\star_{cl} \right) | + |\theta \left(t,r^\star_{cl} \right)| \leq
2 \sqrt{C_L} \ C\left(c\right) \frac{\sqrt{M}}{v_+} \, .
\end{equation}
hold everywhere in $\mathcal{A}\left(T\right) \cap \{r^\star \leq
r^\star_{cl} \} \cap \{ u \leq T-r^\star\left(T,r_K\right) \}$.
\end{proposition}
\begin{proof}
The decay for $\theta$ was already obtained in the proof of
Proposition \ref{pointthcent}. From Proposition \ref{pointBcent}, 
we know that on $r^\star=r^\star_{cl}$ we have the decay
\begin{equation} \label{pbdh}
|B| \leq \sqrt{C_L} \ C\left(c\right) \frac{\sqrt{M}}{v_+}
\end{equation}
Consequently,
\begin{eqnarray}
B\left(u,v\right) &=& B\left(u_{r^\star_{cl}},v\right) 
+ \int^{u}_{u_{r^\star_{cl}}} \frac{\zeta}{r^\frac{3}{2}}
\left(\bar{u},v\right) d\bar{u} \nonumber \\ &\leq& \sqrt{C_L} \ C\left(c\right) \frac{\sqrt{M}}{v_+}
+ \sqrt{ \int^{u}_{u_{r^\star_{cl}}} \frac{\zeta^2}{-\nu}d\bar{u}}  \sqrt{ \int^{u}_{u_{r^\star_{cl}}} \frac{-\nu}{r^3} d\bar{u}}
\end{eqnarray}
and upon inserting bootstrap assumption (\ref{intebound2}) we obtain the result.
\end{proof}
\subsubsection{Higher order quantities}
In this subsection various bounds on the derivatives of 
the quantity $\Omega^2$ are proven. Since we have not yet 
established a pointwise bound on $\frac{\zeta}{\nu}$, the estimates 
will turn out to be suboptimal.\footnote{Such a pointwise bound could in principle be established via a bootstrap argument in the style of Proposition \ref{Bprop}, with the pointwise decay bound on $B$ (\ref{pbdh}) now entering the estimates.} However, they suffice to estimate 
certain error-terms in the $X$-vectorfield identity. An interplay 
between the $X$ and the $Y$-vectorfield will finally 
generate a pointwise bound on $\frac{\zeta}{\nu}$, which allows one
to optimize the estimates (cf.~Proposition \ref{omuomvcent}). 
In particular, the new decay will then suffice to control 
the error-terms occurring in the identity (\ref{bvfi}) for the vectorfield $K$.

The first step is to improve Proposition \ref{omvclo} to a 
decay bound. In the following $C\left(r^\star_{cl},c\right)$ 
denotes a constant whose weight is determined by $C_L$ and 
which also depends on the $c$ in the bootstrap assumptions. 

It should be emphasized again that the 
quantity $\frac{\Omega_{,v}}{\Omega}$ is only 
piecewise continuous, with a discontinuity spreading 
along the null line $v=T+r^\star\left(T,r_K\right)$. 
The estimates below are valid because the quantity 
$\partial_u \frac{\Omega_{,v}}{\Omega}$, which is 
integrated along null-lines, is continuous. The same 
considerations are valid for the quantity 
$\frac{\Omega_{,u}}{\Omega}$ whose discontinuity 
is along the null line $u=T-r^\star\left(T,r_K\right)$.
\begin{proposition} \label{omvcent}
In the region $\mathcal{A}\left(T\right) \cap \{ r^\star \geq r^\star_K \} \cap \{
r^\star \leq \frac{9}{10}t \}$ we have
%\begin{equation} \label{omv1}
%\Big|\frac{\Omega_{,v}}{\Omega} - \frac{m}{r^3}\Big| \leq \epsilon
%\end{equation}
the one-sided bound
\begin{equation} \label{omv2}
\frac{\Omega_{,v}}{\Omega} - \frac{m}{r^3} \leq
C\left(r^\star_{cl},c\right)\frac{\sqrt{M}}{t^2} \, .
\end{equation}
In $\mathcal{A}\left(T\right) \cap \{ r^\star \leq r^\star_{cl} \}$ we
have
\begin{equation} \label{omv3}
\frac{\Omega_{,v}}{\Omega} - \frac{m}{r^3} \leq
C\left(r^\star_{cl},c\right)\frac{\sqrt{M}}{v_+^2}
\end{equation}
\end{proposition}
\begin{proof}
Since $\kappa=\frac{1}{2}$ on the $u=T-r^\star\left(T,r_K\right)$ ray we have
\begin{equation}
\kappa_{,v}=0= 2\kappa \frac{\Omega_{,v}}{\Omega} - \kappa
  \left(\frac{4\kappa}{r^3} m + \frac{4}{3} \frac{\kappa}{r}
  \left(\rho-\frac{3}{2}\right) \right) 
\end{equation}
and, in view of Proposition \ref{pointBthhoz}, the estimate
\begin{equation}
\Big|\frac{\Omega_{,v}}{\Omega}\left(u=T-r^\star\left(T,r_K\right),v
\right)\Big| \leq
\frac{m}{r^3}\left(u=T-r^\star\left(T,r_K\right),v \right) +
C\left(r^\star_{cl},c\right)\frac{\sqrt{M}}{v_+^2} \, .
\end{equation} 
Moreover, on 
$\{t=T \} \cap \{r^\star\left(T,r_K\right) \leq r^\star \leq \frac{9}{10}T \}$ 
we have by the constancy of $\kappa$
\begin{equation}
\kappa_{,r^\star} = 0 = 2\kappa \frac{\Omega_{,v}}{\Omega} - \kappa
  \left(\frac{4\kappa}{r^3} m + \frac{4}{3} \frac{\kappa}{r}
  \left(\rho-\frac{3}{2}\right) \right) - \kappa \frac{2}{r^2}
  \frac{\zeta^2}{\nu} 
\end{equation}
and hence
\begin{equation} \label{gstv1}
\frac{\Omega_{,v}}{\Omega} - \frac{m}{r^3} \leq C\left(r^\star_{cl},c\right)\frac{\sqrt{M}}{v_+^2}
\end{equation}
following from the fact that $|B| \leq \frac{C}{v_+}$ in that region by Proposition \ref{pointBthhoz}. Note 
that the inequality (\ref{gstv1}) would also be two-sided 
if we had the analogous pointwise bound on $\frac{\zeta}{\nu}$.
Integrating (\ref{omegaevol}) downwards from the set 
$L = \{ u=T-r^\star\left(T,r_K\right) \} \cup \left(\{t=T \} 
\cap \{r^\star\left(T,r_K\right) \leq r^\star \leq r^\star_{cl}
\}\right)$ to the $r^\star=r^\star_{cl}$ curve yields
\begin{equation}
\frac{\Omega_{,v}}{\Omega} \left(u,v\right)= \frac{\Omega_{,v}}{\Omega}\left(u_L,v\right) + \int^u_{u_L} \left[-6\frac{\kappa m \nu}{r^4}-2\kappa \frac{\nu}{r^2} \left(\rho-\frac{3}{2}\right) - 3 \frac{\zeta \theta}{r^3} \right] \left(\bar{u},v\right)d\bar{u}
\end{equation}
and upon inserting the pointwise estimates on $B$ and $\theta$
(Proposition \ref{pointBthhoz}), bootstrap assumption \ref{boot4} 
and the estimate
\begin{equation}
\Big| \int^{u_L}_{u} \left[- 3 \frac{\zeta \theta}{r^3} \right] \left(\bar{u},v\right)d\bar{u} \Big|\leq 3 C\left(c\right)\frac{M^\frac{3}{4}}{v_+} \sqrt{ \int^u_{u_L}\frac{\zeta^2}{\Omega^2}d\bar{u} }\sqrt{ \int^u_{u_L} \frac{-4\kappa \nu}{r^6} d\bar{u}} \leq C\left(c\right) \sqrt{C_L} \frac{M^\frac{7}{4}}{r_-^\frac{5}{2}} \frac{1}{v_+^2}  \nonumber
 \end{equation}
for which (\ref{intebound2}) has been used, we finally find that
\begin{equation} 
\frac{\Omega_{,v}}{\Omega} - \frac{m}{r^3} \leq
C\left(r^\star_{cl},c\right)\frac{\sqrt{M}}{v_+^2}
\end{equation}
holds everywhere in $r^\star \leq r^\star_{cl}$ establishing
(\ref{omv3}). Starting from this curve or from the curve 
$\{t=T \} \cap \{r^\star_{cl} \leq r^\star \leq \frac{9}{10}t \}$
we can integrate (\ref{omegaevol}) further to any
point in the region $\mathcal{A}\left(T\right) \cap \{ r > r_K \} \cap \{
r^\star \leq \frac{9}{10}t \}$, this time using the energy 
estimate instead of (\ref{intebound2})
to obtain (\ref{omv2}). \\
\end{proof}

\begin{proposition} \label{omucent}
In the region $\mathcal{A}\left(T\right) \cap \{ r^\star \geq r^\star_K \} \cap \{
r^\star \leq \frac{9}{10}t \}$ 
\begin{equation} \label{omudec}
\abs{\frac{\Omega_{,u}}{\Omega} + \frac{m}{r^3}} \leq
\frac{C\left(\epsilon\right)}{t} \, ,
\end{equation}
in the region $\mathcal{A}\left(T\right) \cap \{ r^\star \leq r^\star_K \} \cap \{
u \leq T-r^\star_K \}$ 
\begin{equation} \label{omudec2}
0 > \frac{\Omega_{,u}}{\Omega} \geq - \frac{m}{r^3}
-\frac{C\left(\epsilon\right)}{v_+} \, .
\end{equation}
\end{proposition}
\begin{proof}
$\gamma=\frac{1}{2}$ on the set $L=\{t=T\} \cap \{
r^\star\left(T,r_K\right) \leq \frac{9}{10}T \}$. From
\begin{equation}
-\gamma_{,r^\star} = 0 = 2\gamma \frac{\Omega_{,u}}{\Omega} - \gamma
  \left(-\frac{4\gamma}{r^3} m - \frac{4}{3} \frac{\gamma}{r}
  \left(\rho-\frac{3}{2}\right) \right) - \gamma \frac{2}{r^2}
  \frac{\theta^2}{\lambda} 
\end{equation}
we derive using the decay estimates (\ref{pointwiseB}), (\ref{pointwiseBh}), (\ref{pointwiseth}) the bound
\begin{equation}
\Big|\frac{\Omega_{,u}}{\Omega} + \frac{m}{r^3} \Big| \leq
C\left(r_K,c\right)\frac{M}{r T^2} \, 
\end{equation}
on $L$. We write the evolution equation (\ref{omegaevol}) as
\begin{equation} \label{omegaevol2}
\partial_v \left(\frac{\Omega_{,u}}{\Omega}\right) = \gamma \left(6
m \frac{\lambda}{r^4} + \frac{2\lambda}{r^2}
\left(\rho-\frac{3}{2}\right) + 3 \frac{\theta}{\kappa}
\frac{\zeta}{\nu} \frac{\lambda}{r^3} \right)
\end{equation}
and integrate downwards in $v$. Using the estimates (\ref{kgclose1}),
(\ref{kgclose2}) and again the decay estimates for $B$ and $\theta$, the
error-terms are estimated
\begin{equation}
\Big| \int^{v_L}_{v} \left[3 \gamma \frac{\theta}{\kappa}
\frac{\zeta}{\nu} \frac{\lambda}{r^3} \right] \Big| \left(\bar{u},v\right)d\bar{v} \leq C\left(c\right)\frac{M^\frac{3}{4}}{v_+} \sup \Big| \frac{\zeta}{\nu} \Big|\int_v^{v_L} \frac{ \lambda}{r^3} d\bar{v} \leq \frac{C\left(\epsilon\right)}{v_+}
 \end{equation}
and
\begin{equation}
\Big| \int^{v_L}_{v} \gamma \frac{2\lambda}{r^2}
\left(\rho-\frac{3}{2}\right) \Big| \leq
\frac{C\left(\epsilon\right)}{v_+^2} \, .
\end{equation}
This establishes the estimate (\ref{omudec}) in a 
subregion ($u \geq \frac{1}{10}T$) of the region asserted in the
Proposition. For the remaining part, we derive the estimate
\begin{equation} \label{omrb}
\Big|\frac{\Omega_{,u}}{\Omega} + \frac{m}{r^3} \Big| \leq
C\left(\epsilon\right)\frac{M}{r^3}
\end{equation}
valid on $\{t=T\} \cap \{ r^\star \geq \frac{9}{10}T \}$ using the
decay of $B$, $\theta$ (Corollary \ref{thetardece} and Proposition
\ref{decayinr2}) in $r$. Integrating (\ref{omegaevol2})
downwards to any point in the region $\{r^\star \geq
\frac{9}{10}t\}$ using again the estimates 
for $B$ and $\theta$ one obtains (\ref{omrb}) in the 
entire region $\{r^\star \geq
\frac{9}{10}t\}$. Since $t \sim r^\star$ on 
the curve $r^\star = \frac{9}{10}t$ we obtain
\begin{equation}
\Big|\frac{\Omega_{,u}}{\Omega} + \frac{m}{r^3} \Big| \leq
C\left(\epsilon\right) \frac{M}{r t^2}
\end{equation}
on that curve. Finally, integrating (\ref{omegaevol2})
 from the $r^\star = \frac{9}{10}t$ curve downwards up 
to any point in the region $r^\star
\geq r^\star_K$ yields (\ref{omudec})
for the entire region asserted in the Proposition.

For the estimate (\ref{omudec2}) we integrate (\ref{omegaevol2}) 
from $r^\star=r^\star_K$ where $t \sim v$ downwards. Clearly the round
bracket on the right hand side of (\ref{omegaevol2}) is always 
positive and hence the upper bound of (\ref{omudec2}) follows
immediately. For the lower bound we use the 
estimate $\gamma \leq \frac{1}{2}$ available in the region under
consideration to estimate:
\begin{eqnarray}
\frac{\Omega_{,u}}{\Omega}\left(u,v\right) \geq
\frac{\Omega_{,u}}{\Omega}\left(u,v=u+2r^\star_K\right) 
- \int_{v}^{u+2r^\star_K} \frac{1}{2} \left(6
m \frac{\lambda}{r^4} + \frac{2\lambda}{r^2}
\left(\rho-\frac{3}{2}\right) + 3 \frac{\theta}{\kappa}
\frac{\zeta}{\nu} \frac{\lambda}{r^3}\right)
\nonumber \\
\geq -\frac{m}{r^3} \left(u,u+2r^\star_K\right) -
m\left(u,u+2r^\star_K\right)
\left(\frac{1}{r^3}\left(u,v\right)-\frac{1}{r^3}\left(u,u+2r^\star_K\right)\right)
- \frac{C\left(\epsilon\right)}{v_+}\nonumber \\
\geq -\frac{m}{r^3} \left(u,v\right) 
- \frac{C\left(\epsilon\right)}{v_+} \nonumber \, .
\end{eqnarray}
\end{proof}
We easily extend the bounds to the asymptotic region:
\begin{corollary} \label{omuomvdecr}
In the region $\mathcal{A}\left(T\right) 
\cap \{ r^{\star} \geq \frac{9}{10}t \} $ we have
\begin{equation}
\frac{\Omega_{,v}}{\Omega} - \frac{m}{r^3} \leq C\left(r^\star_{cl},c\right)
\frac{\sqrt{M}}{r^2} \textrm{ \ \ \ \ and \ \ \ \ }
\Big| \frac{\Omega_{,u}}{\Omega} + \frac{m}{r^3} \Big| \leq 
C\left(\epsilon\right)\frac{M}{r^3} \, .
\end{equation}
\end{corollary}
\begin{proof}
The first bound follows from integrating equation (\ref{omegaevol}) 
outwards from $r^\star = \frac{9}{10}t$ using that $r \sim t$ in that
region and the decay of the fields in $r$. The second bound was 
obtained in (\ref{omrb}).
\end{proof}
As seen in the proof of Propositions \ref{omvcent} 
and \ref{omucent} a pointwise \emph{decay} bound on
the quantity $\frac{\zeta}{\nu}$ would considerably improve the
estimates on the higher order quantities. We summarize this as
\begin{proposition} \label{omuomvcent}
Assume that
\begin{equation}
\Big|\frac{\zeta}{\nu}\Big| \leq C\frac{M^\frac{3}{4}}{v_+} \textrm{ \ \ \ holds in $r^\star
  \leq r^{\star}_{cl}$ \ \ \ \ \ \ and \ \ \ \ \ \ }
  \Big|\frac{\zeta}{\nu}\Big| \leq C\frac{M^\frac{3}{4}}{t} \textrm{ \
  \ \ in $\frac{9}{10}t \geq r^\star
  \geq r^{\star}_{cl}$ } \, .
\end{equation}
holds for a constant $C$ depending only on $r^\star_{cl}$. 
Then in the region $\{r^\star \geq r^\star_g\}$ for any 
$r^\star_{cl} \geq r^\star_g \geq r^\star_K$ 
we have the bounds
\begin{equation} \label{hepp1}
\Big|\frac{\Omega_{,v}}{\Omega} - \frac{m}{r^3}\Big| \leq
C\left(r^\star_{cl},c\right)\frac{\sqrt{M}}{t^2} \textrm{ \ \ \ \ and \ \ \ \ }
\Big| \frac{\Omega_{,u}}{\Omega} + \frac{m}{r^3} \Big| \leq 
C\left(r^\star_{g},c\right)\frac{\sqrt{M}}{t^2} \, .
\end{equation}
Moreover, the one-sided bound (\ref{omv3}) in the region 
$\mathcal{A}\left(T\right) \cap \{ r^\star \leq r^\star_{cl} \}$ 
is extended to
\begin{equation} 
\Big|\frac{\Omega_{,v}}{\Omega} - \frac{m}{r^3} \Big| \leq
C\left(r^\star_{cl},c\right)\frac{\sqrt{M}}{v_+^2} \, ,
\end{equation}  
and the bound (\ref{omudec2}) is refined to
\begin{equation}
0 > \frac{\Omega_{,u}}{\Omega} \left(u,v\right) \geq
-\frac{m}{r^3} \left(u,v\right) - C\left(r^\star_{cl},c\right)\frac{\sqrt{M}}{v_+^2} \, .
\end{equation}
in the region $\mathcal{A}\left(T\right) \cap
\{ r^\star \leq r^\star_{cl} \} \cap \{u \leq T-r^\star_K\} $.
\end{proposition}
\begin{proof}
Revisit the proof of Propositions \ref{omvcent} 
and \ref{omucent}. Note that the constant in the $u$-estimate 
of (\ref{hepp1}) improves by moving away from the horizon since it
depends on the weight of $\frac{1}{1-\mu}$ on $r^\star_g$.
\end{proof}
\section{The vectorfield $Y$} \label{Ysection}
Recall the functions $\alpha$ and $\beta$ defined in section \ref{bootstrap}.
Close to the horizon we are going to apply the vector field
\begin{equation}
Y = \frac{2\alpha\left(r^\star\right)}{\Omega^2} \partial_u + 2 \beta\left(r^\star\right) \partial_v
\end{equation}
for which
\begin{equation}
Y^u = \frac{2\alpha}{\Omega^2} \textrm { \ \ \ \ \ } Y^v = 2\beta \textrm
{ \ \ \ \ \ } Y_u = -\beta \Omega^2 \textrm { \ \ \ \ \ }
Y_v = -\alpha \, .
\end{equation}
The calculations will be carried out in the 
Eddington Finkelstein coordinates defined 
in section \ref{Coordinates}. From (\ref{basicintegrand}) we 
derive the identity
%
\begin{comment}
\begin{eqnarray}
-T_{\mu \nu} \pi^{\mu \nu} &=& 
-\frac{2\left(\partial_u B\right)^2}{\Omega^4}
 \left(4 \alpha \frac{\Omega_{,v}}{\Omega} 
- \alpha^\prime \right) - 2\beta^\prime
 \frac{\left(\partial_v B\right)^2}{\Omega^2} \nonumber \\ &+&
\frac{1}{\Omega^2 r^2} \left(1-\frac{2}{3}\rho\right)
 \left(-\frac{1}{2} \alpha^\prime + \frac{3\alpha \nu}{r} +
\frac{3\beta \lambda \Omega^2}{r} + \frac{1}{2}
\beta^\prime \Omega^2 + 2 \beta
\Omega^2 \frac{\Omega_{,v}}{\Omega}\right) \nonumber \\ &+&  \frac{12}{\Omega^2 r}
 \left(\kappa \alpha - \lambda \beta \right) \partial_u B \partial_v B
\end{eqnarray}
and note
\begin{equation}
-\left(\nabla^\beta T_{\beta \delta} \right)Y^\delta =
 -\frac{1}{r^3} \left(1-\frac{2}{3}\rho\right) 
\left(Y^u \partial_u r + Y^v \partial_v r\right) = -\frac{1}{r^3}
\left(1-\frac{2}{3}\rho\right) \left(2 \alpha \frac{\nu}{\Omega^2} + 2
\beta \lambda \right) 
\end{equation}
such that the spacetime term of the energy identity reads
\end{comment}
%
\begin{eqnarray}
-T_{\mu \nu} \pi^{\mu \nu} 
-\left(\nabla^\beta T_{\beta \delta} \right)Y^\delta =
-\frac{2\left(\partial_u B\right)^2}{\Omega^4}
 \left(4 \alpha \frac{\Omega_{,v}}{\Omega} 
- \alpha^\prime \right) - 2\beta^\prime
 \frac{\left(\partial_v B\right)^2}{\Omega^2} \nonumber \\ 
 + \frac{1}{\Omega^2 r^2} \left(1-\frac{2}{3}\rho\right)
 \left(-\frac{1}{2} \alpha^\prime + \frac{\alpha \nu}{r} +
\frac{\beta \lambda \Omega^2}{r} + \frac{1}{2}
\beta^\prime \Omega^2 + 2 \beta
\Omega^2 \frac{\Omega_{,v}}{\Omega}\right) \nonumber \\ +  \frac{12}{\Omega^2 r}
 \left(\frac{1}{4\kappa} \alpha - \lambda \beta \right) \partial_u B \partial_v B
\, .
\end{eqnarray}
In a characteristic rectangle $\mathcal{R} =
\left[u_1,u_2\right] \times \left[v_1,v_2\right]$ the identity
\begin{eqnarray} \label{Yid}
F^Y_B \left(\{u_2\} \times \left[v_1, v_2\right] \right) + F^Y_B
\left(\left[u_1, u_2\right] \times \{v_2\} \right) \\ = I^Y_B
\left(\mathcal{R}\right) + F^Y_B \left(\{u_1\} \times \left[v_1,
  v_2\right] \right)+ F^Y_B \left(\left[u_1, u_2\right] \times \{v_1\} \right)
\end{eqnarray}
follows, with the boundary-terms given by
\begin{equation} \label{bF1}
\frac{1}{2\pi^2} F^Y_B \left(\{u\} \times \left[v_1, v_2\right]\right) =
2 \int_{v_1}^{v_2}  \left(\beta \left(\partial_v
B\right)^2  + \frac{1}{4r^2}  \left(1-\frac{2}{3}\rho\right) \alpha
 \right) r^3 dv  \, ,
\end{equation}
\begin{equation} \label{bF2}
\frac{1}{2\pi^2} F^Y_B \left(\left[u_1, u_2\right] \times \{ v \} \right) =
2 \int_{u_1}^{u_2} \left(\frac{\alpha}{\Omega^2}
\left(\partial_u B\right)^2 + \frac{\beta \Omega^2}{4r^2}
\left(1-\frac{2}{3}\rho\right)  \right) r^3 du
\end{equation}
and the spacetime-term
\begin{eqnarray}
\frac{1}{2\pi^2} I^Y_B \left(\mathcal{R}\right) = \int_{v_1}^{v_2}\int_{u_1}^{u_2}
 \Bigg(-T_{\mu \nu} \pi^{\mu \nu} -
\left(\nabla^\beta T_{\beta \delta} \right) Y^\delta\Bigg) \frac{1}{2} \Omega^2 r^3 du dv
\nonumber \\ = \int_{v_1}^{v_2}\int_{u_1}^{u_2}
 \Bigg(-\Bigg[\frac{\left(\partial_u B\right)^2}{\Omega^2}
 \left(4 \alpha \frac{\Omega_{,v}}{\Omega} 
- \alpha^\prime \right) + \beta^\prime
 \left(\partial_v B\right)^2 \nonumber \\  +\frac{1}{2 r^2} \left(1-\frac{2}{3}\rho\right)
 \left(\frac{1}{2} \alpha^\prime - \frac{\alpha \nu}{r} -
\frac{\beta \lambda \Omega^2}{r} - \frac{1}{2}
\beta^\prime \Omega^2 - 2 \beta
\Omega^2 \frac{\Omega_{,v}}{\Omega}\right) \Bigg] + \nonumber \\
\frac{6}{r} \left(\frac{1}{4\kappa} \alpha - \lambda \beta \right)
\partial_u B \partial_v B \Bigg) r^3 du dv \, .
\end{eqnarray}
It will be useful to split the term into 
\begin{equation}
I^Y_B \left(\mathcal{R}\right) = -\tilde{I}^Y_B
\left(\mathcal{R}\right) + \widehat{I}^Y_B
\left(\mathcal{R}\right)
\end{equation}
where 
\begin{eqnarray} \label{posY}
\frac{1}{2\pi^2} \tilde{I}^Y_B \left(\mathcal{R}\right) = \int_{v_1}^{v_2}\int_{u_1}^{u_2} \Bigg[\frac{\left(\partial_u B\right)^2}{\Omega^2}
 \left(4 \alpha \frac{\Omega_{,v}}{\Omega} 
- \alpha^\prime \right) + \beta^\prime
 \left(\partial_v B\right)^2 \nonumber \\  +\frac{1}{2 r^2} \left(1-\frac{2}{3}\rho\right)
 \left(\frac{1}{2} \alpha^\prime - \frac{\alpha \nu}{r} -
\frac{\beta \lambda \Omega^2}{r} - \frac{1}{2}
\beta^\prime \Omega^2 - 2 \beta
\Omega^2 \frac{\Omega_{,v}}{\Omega}\right) \Bigg]r^3 du dv 
\end{eqnarray}
and
\begin{equation}
\frac{1}{2\pi^2}\widehat{I}^Y_B
\left(\mathcal{R}\right) = \int_{v_1}^{v_2}\int_{u_1}^{u_2}
 \Big(\frac{6}{r} \left(\frac{1}{4\kappa} \alpha - \lambda \beta \right)
\partial_u B \partial_v B \Big) r^3 du dv  \, .
\end{equation}
With the choices of the functions $\alpha$ and $\beta$ made in section \ref{bootstrap}, the integral
$\tilde{I}^Y_B\left(\mathcal{R}\right)$ is non-negative for $r^\star \leq
r^\star_Y$ and moreover, using (\ref{bpcon}),
\begin{eqnarray} \label{tildconhat}
&&\hat{I}^Y_B
\left(r^\star \leq r^\star_Y\right) = 2\pi^2 \int_{v_1}^{v_2}\int_{u_1}^{u_2}
\int_{\mathbb{S}^3} \Big(\frac{6}{r} \left(\frac{1}{4\kappa} \alpha - \lambda \beta \right)
\partial_u B \partial_v B \Big) r^3 du dv \nonumber \\ 
&\leq& 2\pi^2 \int_{v_1}^{v_2}\int_{u_1}^{u_2}
 \frac{1}{2} \left(\frac{2\left(\partial_u
   B\right)^2}{\Omega^2}\frac{\left(\frac{1}{4\kappa} \alpha-\lambda \beta\right)^2
   }{r}+ \frac{18}{r} \Omega^2
 \left(\partial_v B\right)^2\right) r^3 du dv 
 \nonumber \\ 
&\leq& 2\pi^2 \int_{v_1}^{v_2}\int_{u_1}^{u_2}
 \frac{1}{2}
 \left(\frac{\left(\partial_u
   B\right)^2}{\Omega^2}\left(4\alpha \frac{\Omega_{,v}}{\Omega} - \alpha^\prime \right) + \beta^\prime
 \left(\partial_v B\right)^2\right) r^3 du dv 
 \nonumber \\ &\leq&\frac{1}{2} \tilde{I}^Y_B\left(r^\star \leq r^\star_Y\right)
\end{eqnarray}
holds in $r^\star \leq r^\star_Y$. We conclude by rewriting 
identity (\ref{Yid}) for a characteristic
rectangle with one boundary being the horizon:
\begin{eqnarray} \label{hozid}
F^Y_B \left(\{u_{hoz}\} \times \left[v_1, v_2\right] \right) + F^Y_B
\left(\left[u_1, u_{hoz} \right] \times \{v_2\} \right) + \tilde{I}^Y_B
\left(\mathcal{R}\right) \nonumber \\ = \hat{I}^Y_B
\left(\mathcal{R}\right) + F^Y_B \left(\{u_1\} \times \left[v_1,
  v_2\right] \right)+ F^Y_B \left(\left[u_1, u_{hoz} \right] \times
\{v_1\} \right) \, .
\end{eqnarray}

\section{The vectorfield X} \label{Xsection}
All calculations in this section are performed in the Eddington Finkelstein coordinate system defined in section \ref{Coordinates}.
\subsection{The basic identity}
The vector field $X$ is defined as
\begin{equation}
X = 2f\left(r^\star \right) \partial_u - 2f\left(r^\star \right) \partial_v
\end{equation}
for some function $f$ chosen below and with $u_J$ satisfying $u_J \geq t_1-r^\star_{cl}$. It will be applied in the region 
\begin{equation} \label{basicXreg}
\mathcal{D}^{r^\star_{cl},u_J}_{[t_1,t_2]} :=
	{}^{u_H=t_2-r^\star_{cl}}\mathcal{D}^{r^\star_{cl}, u_J}_{[t_1,t_2]}
\end{equation}
for some $r^\star_{cl}$ also chosen below. We note
\begin{equation}
X^u = 2f \textrm{ \ \ \ \ \ } X^v = -2f \textrm{ \ \ \ \ \ } X_u =
f \Omega^2 \textrm{ \ \ \ \ \ } X_v = -f \Omega^2 \, .
\end{equation}
From now on primes will denote a derivative with respect to
$r^\star$, hence $\partial_v f \left(r^\star\right) = \frac{1}{2}
f^\prime$ and $\partial_u f \left(r^\star\right) = -\frac{1}{2}
f^\prime$. From (\ref{basicintegrand}) using 
\begin{equation}
-\left(\partial_{r^\star} B\right)^2 = 2 \partial_u B \partial_v B -
\left(\partial_u B\right)^2 - \left(\partial_v B\right)^2 
\end{equation}
we compute
\begin{eqnarray} \label{idXint}
-T_{\mu \nu} \pi^{\mu \nu} - \left(\nabla^\beta
 T_{\beta \delta} \right) X^\delta = \frac{2}{\Omega^2} f^\prime
 \left(\partial_{r^\star} B\right)^2 + \nabla^\alpha B \nabla_\alpha B
 \left(-f^\prime - \frac{3}{r} \left(\lambda-\nu\right) f \right) \nonumber
 \\ +  \frac{1}{r^2}\left(1-\frac{2}{3}\rho \right) \left(-f^\prime
 - \frac{\lambda-\nu}{r}f + \frac{\left(\Omega^2\right)_{,u} -
 \left(\Omega^2\right)_{,v}}{\Omega^2}f \right) \, .
\end{eqnarray}
With the boundary terms
\begin{eqnarray}
\frac{1}{2\pi^2} \widehat{F}_B^X \left(t_i\right) = -2\int_{r^\star_{cl}}^{t-u_J}
 f \partial_t B \partial_{r^\star} B
\left(t_i,r^\star\right) r^3 dr^\star \nonumber \\ 
+\int_{t-r^\star_{cl}}^{t_2-r^\star_{cl}} \left[r^3 \left(\partial_u B\right)^2 
\left(2f\right) + \frac{r\Omega^2}{4} 
\left(1-\frac{2}{3}\rho\right)\left(-2f\right) \right] du
\end{eqnarray}
and
\begin{equation}
\frac{1}{2\pi^2} \hat{H}^X_{u_H} = \int_{v_1}^{v_2} \left[r^3
  \left(\partial_v B \right)^2 \left(-2f\right) + \frac{r \Omega^2}{4}
  \left(1-\frac{2}{3} \rho\right) \left(2f\right) \right] dv \, ,
\end{equation}
\begin{equation}
\frac{1}{2\pi^2} \hat{J}^X_{u_J} = \int_{2t_1-u_J}^{2t_2-u_J} \left[r^3
  \left(\partial_v B \right)^2 \left(-2f\right) + \frac{r \Omega^2}{4}
  \left(1-\frac{2}{3} \rho\right) \left(2f\right) \right] dv \, ,
\end{equation}
one can state the identity
\begin{equation} \label{idst}
\int_{\mathcal{D}_{[t_1,t_2]}^{r^\star_{cl},u_J}} \left[-T_{\mu \nu} \pi^{\mu \nu} - \left(\nabla^\beta
 T_{\beta \delta} \right) X^\delta\right] dVol = \widehat{F}^X_B
 \left(t_1\right) - \widehat{F}^X_B
 \left(t_0\right) + \hat{H}^X_{u_H} - \hat{J}^X_{u_J} \, .
\end{equation}
Let us turn to the spacetime integral on the left of (\ref{idst}) with
the integrand being given by (\ref{idXint}).
In view of definition (\ref{deltaB}) we can write 
\begin{equation} \label{sigmaB}
\nabla^\alpha B \nabla_\alpha B = \frac{1}{2} \Box B^2 +
\frac{4B}{3r^2}\left(e^{-8B}-e^{-2B} \right) = \frac{1}{2} \Box B^2
-8\frac{B^2}{r^2} + \frac{1}{r^2}\varphi_2\left(B\right)
\end{equation}
and the integrand (\ref{idXint}) becomes
\begin{eqnarray} \label{cgfd}
-T_{\mu \nu} \pi^{\mu \nu} - \left(\nabla^\beta
 T_{\beta \delta} \right) X^\delta = \frac{2}{\Omega^2} f^\prime
 \left(\partial_{r^\star} B\right)^2 + \frac{1}{2} \Box B^2 
 \left(-f^\prime - \frac{3}{r} \left(\lambda-\nu\right) f \right) \nonumber
 \\ +  \frac{1}{r^2}\left(8B^2 \right) \left(
 2 \frac{\lambda-\nu}{r}f + \frac{\left(\Omega^2\right)_{,u} -
 \left(\Omega^2\right)_{,v}}{\Omega^2}f \right) \nonumber \\
+  \frac{1}{r^2}\left(\varphi_1\left(B\right) \right) \left(-f^\prime
 - \frac{\lambda-\nu}{r}f + \frac{\left(\Omega^2\right)_{,u} -
 \left(\Omega^2\right)_{,v}}{\Omega^2}f \right) \nonumber \\
\frac{\varphi_2\left(B\right)}{r^2} \left(-f^\prime - \frac{3}{r}
 \left(\lambda-\nu\right) f \right) \, .
\end{eqnarray}
Finally, we apply Green's theorem to the $\Box
B^2$-term.\footnote{cf.~the remarks in appendix \ref{reggree}}
Collecting the $B^2$-terms of (\ref{cgfd}) after the integration 
by parts we find
\begin{equation} \label{b2col}
-\frac{1}{2}B^2 \left[-32f \frac{\lambda-\nu}{r^3} +16\frac{f}{r^2}
 \frac{\left(\Omega\right)^2_{,v}-\left(\Omega\right)^2_{,u}}{\Omega^2} +
 \Box \left(f^\prime + \frac{3}{r} \left(\lambda-\nu\right) f \right)
 \right] \, .
\end{equation}
Since
\begin{equation}
\Box w(u,v) = \left(-\frac{4}{\Omega^2} \partial_u \partial_v -
\frac{6}{r} \frac{\nu}{\Omega^2} \partial_v - \frac{6}{r}
\frac{\lambda}{\Omega^2} \partial_u \right) w\left(u,v\right)
\end{equation}
and moreover $f$ depends only on $r^\star$ we arrive at
\begin{eqnarray}
\Box \left(f^\prime + \frac{3}{r} \left(\lambda-\nu\right)f \right) =
\frac{f^{\prime \prime \prime}}{\Omega^2} + \frac{6}{r}
\frac{\lambda-\nu}{\Omega^2} f^{\prime \prime} + f^\prime
\left[\frac{6}{\Omega^2} \left(\partial_{r^\star} \frac{\lambda-\nu}{r}\right) +
  \frac{9\left(\lambda-\nu\right)^2}{r^2\Omega^2} \right] \nonumber \\
f \left[-\frac{12}{\Omega^2} \partial_u \partial_v
  \left(\frac{\lambda-\nu}{r}\right) - \frac{18}{r}
  \frac{\nu}{\Omega^2}\partial_v
  \left(\frac{\lambda-\nu}{r}\right) - \frac{18}{r}
  \frac{\lambda}{\Omega^2}\partial_u
  \left(\frac{\lambda-\nu}{r}\right)  \right] \,. \nonumber
\end{eqnarray}
Computing the derivatives explicitly for the expression in the
square brackets of (\ref{b2col}) yields the identity
\begin{eqnarray}  \label{Cseven}
-32\frac{\lambda-\nu}{r^3} + \frac{16}{r^2}
\frac{\left(\Omega\right)^2_{,v}-\left(\Omega\right)^2_{,u}}{\Omega^2}
\nonumber \\ -\frac{12}{\Omega^2} \partial_u \partial_v
  \left(\frac{\lambda-\nu}{r}\right) - \frac{18}{r}
  \frac{\nu}{\Omega^2}\partial_v
  \left(\frac{\lambda-\nu}{r}\right) - \frac{18}{r}
  \frac{\lambda}{\Omega^2}\partial_u
  \left(\frac{\lambda-\nu}{r}\right)
\nonumber \\ =\left(\lambda-\nu\right) \left(-\frac{35}{r^3} - \frac{18\mu}{r^3}
\right) + \frac{1}{r^2}
\left(\frac{\Omega_{,v}}{\Omega}-\frac{\Omega_{u}}{\Omega}\right)\left(35+9\mu\right)+ \mathcal{I}_7 \left(B,\theta, \zeta \right)
\end{eqnarray}
with
\begin{eqnarray}  
{\cal{I}}_7 \left(B\right) &=& \frac{9}{r^4}\frac{\theta^2}{\kappa} +
\frac{36}{r^4} \frac{\lambda \zeta^2}{\Omega^2} +
\left(\rho-\frac{3}{2}\right)\frac{1}{r^3}
\left[-14\left(\lambda-\nu\right) + 8r\left(\frac{\Omega_{,v}}{\Omega}
  - \frac{\Omega_{,u}}{\Omega} \right)\right] \nonumber \\
&-& \frac{16}{r^\frac{7}{2}} 
\left(e^{-2B}-e^{-8B}\right) \left(\theta-\zeta\right) \, .
\end{eqnarray} 
We summarize the remaining error-terms as
\begin{equation}
{\cal{I}}_8 \left(B\right)= f\Bigg(
-\frac{\varphi_1\left(B\right)}{B^2r^2}\frac{\left(\Omega\right)^2_{,v}-\left(\Omega\right)^2_{,u}}{\Omega^2}  - 
\frac{\lambda-\nu}{r^3} \left(\frac{\varphi_1\left(B\right) +
  3\varphi_2\left(B\right)}{B^2}\right) \Bigg) -f^\prime
\frac{\varphi_1\left(B\right)+\varphi_2\left(B\right)}{B^2 r^2}  \nonumber
\end{equation}
and read off the pointwise estimate (cf.~Corollary \ref{Bcor})
\begin{equation} \label{error78}
\Big|{\cal{I}}_7 \left(B\right)\Big| + \Big|{\cal{I}}_8 \left(B\right)\Big| \leq C\left(\epsilon\right)\frac{\sqrt{M}}{r^4} \, .
\end{equation}
Taking care of the boundary terms arising from the application of
Green's identity we can finally state the identity (cf.~equation (\ref{finidg}))
\begin{equation} \label{Xbasicidentity}
I^X_B \left(\mathcal{\mathcal{D}}_{[t_1,t_2]}^{r^\star_{cl},u_J}\right) = F^X_B \left(t_2\right) - F^X_B
\left(t_1\right) + H^X_{u_H} - J^X_{u_J}
\end{equation}
with the renormalized bulk-term 
\begin{eqnarray} \label{IXB}
I_{B}^X \left(\mathcal{\mathcal{D}}_{[t_1,t_2]}^{r^\star_{cl},u_J}\right) = \int_{\mathcal{D}_{[t_1,t_2]}^{r^\star_{cl},u_J}} \Bigg\{\frac{2}{\Omega^2}
f^\prime \left(\partial_{r^\star} B\right)^2 + \nonumber \\
- \frac{B^2}{2} \Bigg[ \frac{f^{\prime \prime \prime}}{\Omega^2} +
  \frac{6}{r}\left(\frac{\lambda-\nu}{\Omega^2}\right) f^{\prime
    \prime}
+ f^\prime \left(\frac{6}{\Omega^2} \left(\partial_{r^\star} \frac{\lambda-\nu}{r} \right) +
  \frac{9}{r^2}\frac{\left(\lambda-\nu\right)^2}{\Omega^2} \right) \nonumber \\
+ f \Bigg(
\left(\lambda-\nu\right) \left(-\frac{35}{r^3} - \frac{18\mu}{r^3}
\right) + \frac{1}{r^2}
\left(\frac{\Omega_{,v}}{\Omega}-\frac{\Omega_{,u}}{\Omega}\right)\left(35+9\mu\right) \Bigg)\Bigg] \Bigg\} dVol \nonumber \\
+ \int_{\mathcal{D}_{[t_1,t_2]}^{r^\star_{cl},u_J}} B^2 \left[-\frac{1}{2}{\cal{I}}_7 \left(B\right) + {\cal{I}}_8 \left(B\right)\right] dVol \, ,
\end{eqnarray}
%which we write shorthand as
%\begin{eqnarray} \label{IXBsh}
%I_{B}^X \left(\mathcal{\mathcal{D}}_{[t_1,t_2]}^{r^\star_{cl},u_J}\right) = \int_{\mathcal{D}_{[t_1,t_2]}^{r^\star_{cl},u_J}} \Bigg\{\frac{2}{\Omega^2}
%f^\prime \left(\partial_{r^\star} B\right)^2 +  B^2 
%\left[\mathcal{F}-\frac{1}{2}C_7 \left(B\right) + C_8 \left(B\right) \right] 
%\end{eqnarray}
%with the obvious identification for $\mathcal{F}$
the new boundary terms
\begin{eqnarray} \label{XFterms}
F^X_B \left(t\right) &=& \hat{F}^X_B \left(t\right) - \int_{r^\star_{cl}}^{t-u_J} \int_{\mathbb{S}^3} 
\left(f^\prime + \frac{3}{r} \left(\lambda-\nu\right)f \right)
\left(\partial_t B \right) B \left(t,r^\star\right) r^3  dr^\star
dA_{\mathcal{S}^3} \nonumber \\
&+& \int_{r^\star_{cl}}^{t-u_J} \int_{\mathbb{S}^3} 
\frac{1}{2}\left(\partial_t \left(\frac{3}{r} \left(\lambda-\nu\right)\right)f
  \right)B^2 \left(t_1,r^\star\right) r^3 dr^\star dA_{\mathcal{S}^3}
  \nonumber \\
 &-& \int_{t-r^\star_{cl}}^\infty \int_{\mathbb{S}^3} 
\left(f^\prime + \frac{3}{r} \left(\lambda-\nu\right)f \right)
\left(\partial_u B \right) B \left(u,t+r^\star\right) r^3  du
dA_{\mathcal{S}^3} \nonumber \\
&+& \int_{t-r^\star_{cl}}^\infty \int_{\mathbb{S}^3} 
\frac{1}{2}\left(\partial_u \left(f^\prime + \frac{3}{r} \left(\lambda-\nu\right)\right)f
  \right)B^2 \left(u,t+r^\star_{cl}\right) r^3 du dA_{\mathcal{S}^3}
  \, ,  \nonumber \\
\end{eqnarray}
the horizon terms
\begin{eqnarray} \label{XHterms}
H^X_{u_H} = \hat{H}^X_{u_H}  &-&
\int_{t_1+r^\star_{cl}}^{t_2+r^\star_{cl}} \left[B\partial_vB \left(f^\prime
  + \frac{3}{r} \left(\lambda-\nu\right)f\right) \right]r^3
\left(u_{hoz},v\right) dv \nonumber \\
  &+& \int_{t_1+r^\star_{cl}}^{t_2+r^\star_{cl}} \left[\frac{B^2}{2} \partial_v \left(f^\prime + \frac{3}{r} \left(\lambda-\nu\right)f\right) \right]r^3
\left(u_{hoz},v\right) dv 
\end{eqnarray}
and the $J$-terms
\begin{eqnarray} \label{XJterms}
J^X_{u_J} = \hat{J}^X_{u_J}  &-&
\int_{2t_1-u_J}^{2t_2-u_J} \left[B\partial_vB \left(f^\prime
  + \frac{3}{r} \left(\lambda-\nu\right)f\right) \right]r^3
\left(u_{J},v\right) dv \nonumber \\
  &+& \int_{2t_1-u_J}^{2t_2-u_J} \left[\frac{B^2}{2} \partial_v
  \left(f^\prime + \frac{3}{r} \left(\lambda-\nu\right)f\right)
  \right]r^3 \left(u_{J},v\right) dv \, .
\end{eqnarray}
\subsection{Analysing the X-bulk-term}
\subsubsection{Borrowing from the derivative-term}
We would like the spacetime term (\ref{IXB}) to have a sign. 
To achieve this we borrow from the term containing a
derivative. Define
\begin{equation}
t^\prime = \frac{t_1+t_2}{2} \textrm{ \ \ \ \ and \ \ \ \ }
r^\star_{x} = r^\star_{cl} + \frac{t_1-t_2}{2} \, ,
\end{equation}
\begin{equation}
r^\star_x \left(t\right) = \left\{ \begin{array}{ll}
r^\star_{cl} + t_1 - t & \textrm{ for $t_1 \leq t \leq t^\prime$} \\
r^\star_{cl} + t - t_2 & \textrm{ for $t^\prime \leq t \leq t_2$} \, .
\end{array} \right.
\end{equation}
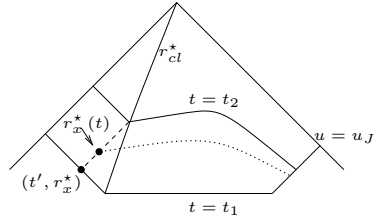
\[
\input{xborrow.pstex_t}
\]
and compute\footnote{Again care is necessary in the integration by
  parts because of the differentiability of the coordinate
  system. In any case it is sufficient to note that the integrands of
  the boundary terms are continuous and that the integrand of 
the bulk term is piecewise continuous (all terms except 
the $\xi^{\prime}$-term are in fact continuous everywhere).}
\begin{eqnarray}
\int_{t_1}^{t_2} dt \int_{r^\star_{x}\left(t\right)}^{t-u_J} \frac{f^\prime}{\Omega^2} \left(\partial_{r^\star}
B \right)^2 r^3 \Omega^2 dr^\star = \int_{t_1}^{t_2} dt  \int_{r^\star_{x}\left(t\right)}^{t-u_J} \frac{f^\prime}{\Omega^2} \left(\partial_{r^\star}
B + \xi B \right)^2 r^3 \Omega^2 dr^\star 
\nonumber \\
+\int_{t_1}^{t_2} dt \int_{r^\star_{x}\left(t\right)}^{t-u_J} B^2 \left(\frac{f^{\prime \prime}
  \xi}{\Omega^2} + \frac{\xi^\prime f^\prime}{\Omega^2} +
\frac{3}{r} \frac{\lambda-\nu}{\Omega^2}
f^\prime \xi \right) r^3 \Omega^2 dr^\star \nonumber
\\ 
-\int_{t_1}^{t_2} dt \int_{r^\star_{x}\left(t\right)}^{t-u_J} B^2\left(\frac{f^\prime}{\Omega^2}
\xi^2 \right) r^3 \Omega^2 dr^\star + J^{X,u_J}_{error}  +  H^{X,v_1=t_1+r^\star_{cl}}_{error}  + H^{X,u_2=t_2-r^\star_{cl}}_{error} \nonumber
\end{eqnarray}
for some function $\xi$ chosen in (\ref{xidef}). The
boundary-terms are
\begin{equation} \label{JXerr}
J^{X,u_J}_{error} = -\int_{t_1}^{t_2} f^\prime \xi B^2 r^3
\left(t,t-u_J\right) dt = -\int_{2t_1-u_J}^{2t_2-u_J} f^\prime \xi
B^2 r^3 \left(u_J,v\right) dv \, ,
\end{equation}
\begin{equation} \label{HXerr}
H^{X,v_1=t_1+r^\star_{cl}}_{error} =  \int_{t_1}^{t^\prime} dt
f^\prime \xi B^2 r^3 \left(t,v_1-t\right) =
\int_{t_1-r^\star_{cl}}^{t_2-r^\star_{cl}} f^\prime \xi B^2 r^3
\left(u,v_1\right) du 
\end{equation}
and 
\begin{equation} \label{HXerr2}
H^{X,u_2=t_2-r^\star_{cl}}_{error}  =  \int_{t^\prime}^{t_2} dt
f^\prime \xi B^2 r^3 \left(t,t-u_2\right) =
\int_{t_1+r^\star_{cl}}^{t_2+r^\star_{cl}}  f^\prime \xi B^2 r^3
\left(u_2,v\right)dv \, .
\end{equation}
To keep the notation clean we write $M=m(T,r^\star=0)$ in 
this section. For a sufficiently large constant $\sigma$ we define 
the shifted coordinate $x$
\begin{equation} \label{xshift}
x = r^\star - \sigma - \sqrt{M} \,.
\end{equation}
We choose
\begin{equation} \label{xidef}
\xi = \frac{3}{2} \frac{\lambda-\nu}{r} - \frac{nx}{x^2 + \sigma^2}
\end{equation}
for some $n \in \left(\frac{1}{2},\infty\right)$
%
%
%
\begin{comment}
Consequently
\begin{equation}
\xi^\prime = \frac{3}{2} \left(\frac{-2r_{,uv} -r_{,uu} +
  r_{,vv}}{r} - \frac{\left(\lambda-\nu\right)^2}{r^2} \right) - \frac{n}{\sigma^2+x^2} +
\frac{2nx^2}{\left(\sigma^2+x^2\right)^2} 
\end{equation}
\begin{equation}
\xi^2 = \frac{9}{4} \frac{\left(\lambda-\nu\right)^2}{r^2} - 3
\frac{\lambda-\nu}{r} \frac{nx}{\sigma^2 + x^2} 
+ \frac{n^2 x^2}{\left(\sigma^2+x^2\right)^2}  
\end{equation}
\begin{equation}
\frac{3}{r}\left(\lambda-\nu\right) \xi = \frac{9}{2}
\frac{\left(\lambda-\nu\right)^2}{r^2} - \frac{3}{r} \left(\lambda-\nu\right) 
\frac{nx}{\sigma^2 + x^2}
\end{equation}
\end{comment}
%
%
%
from which
\begin{equation}
-\xi^\prime + \xi^2 - \frac{3}{r}\left(\lambda - \nu\right) \xi = 
-\frac{9}{4} \frac{\left(\lambda - \nu \right)^2}{r^2} +
\frac{x^2 \left(n^2-n\right) + n\sigma^2}{\left(\sigma^2+x^2\right)^2} 
- \frac{3}{2} \left(\partial_{r^\star}
\left(\frac{\lambda-\nu}{r}\right)\right) \nonumber
\end{equation}
follows. Hence the integral (\ref{IXB}) can be expressed as
\begin{eqnarray} \label{IXBst}
I^X_B \left(\mathcal{\mathcal{D}}_{[t_1,t_2]}^{r^\star_{cl},u_J}\right) = \int_{\mathcal{\mathcal{D}}_{[t_1,t_2]}^{r^\star_{cl},u_J}} \frac{2f^\prime}{\Omega^2} 
\left(\partial_{r^\star} B + \xi B \right)^2 dVol \nonumber \\ 
- \frac{1}{2} \int_{\mathcal{D}_1} \Bigg[\frac{f^{\prime
      \prime \prime}}{\Omega^2} + f^{\prime \prime}
  \left(\frac{4nx}{\Omega^2\left(\sigma^2+x^2\right)}\right)+
  f^\prime \left(\frac{4x^2 \left(n^2-n\right) + 4n \sigma^2}{\left(\sigma^2+x^2\right)^2\Omega^2}\right) \nonumber \\
+f 
  \left[\left(\lambda-\nu\right) \left(-\frac{35}{r^3} - \frac{18\mu}{r^3}
\right) + \frac{1}{r^2}
\left(\frac{\Omega_{,v}}{\Omega}-\frac{\Omega_{,u}}{\Omega}\right)\left(35+9\mu\right) \right]  \Bigg] B^2 \ dVol \nonumber \\ + \int_{\mathcal{D}_1}
\left(-\frac{1}{2} {\cal{I}}_7 \left(B\right) + {\cal{I}}_8 \left(B\right) \right) B^2
dVol 
\nonumber \\
+ J^{X,u_J}_{error} + H^{X,v=t_1+r^\star_{cl}}_{error} + H^{X,u=t_2-r^\star_{cl}}_{error} 
\end{eqnarray}
which we will write shorthand as
\begin{equation}
I^X_B \left(\mathcal{\mathcal{D}}_{[t_1,t_2]}^{r^\star_{cl},u_J}\right) = \bar{I}^X_B
\left(\mathcal{\mathcal{D}}_{[t_1,t_2]}^{r^\star_{cl},u_J}\right) + J^{X,u_J}_{error} + H^{X,v=t_1+r^\star_{cl}}_{error} + H^{X,u=t_2-r^\star_{cl}}_{error} 
\end{equation}
where
\begin{eqnarray} \label{IXBstsh}
\bar{I}^X_B \left(\mathcal{\mathcal{D}}_{[t_1,t_2]}^{r^\star_{cl},u_J}\right) = \int \Big\{ \frac{2f^\prime}{\Omega^2} 
\left(\partial_{r^\star} B + \xi B \right)^2 + B^2 \left[F + f \cdot
  g -\frac{1}{2} {\cal{I}}_7 \left(B\right)
  + {\cal{I}}_8 \left(B\right) \right] \Big\} dVol 
\end{eqnarray}
with the identifications
\begin{equation}
F = -\frac{1}{2\Omega^2} \left(f^{\prime \prime \prime} +
\frac{4nxf^{\prime \prime}}{\sigma^2+x^2} + f^\prime \left(\frac{4x^2
  \left(n^2-n\right) + 4n
  \sigma^2}{\left(\sigma^2+x^2\right)^2}\right)\right) \, ,
\end{equation}
\begin{equation} \label{gterm}
g = -\frac{1}{2} \left[\left(\lambda-\nu\right)\left(-\frac{35}{r^3} - \frac{18\mu}{r^3}
\right) + \frac{1}{r^2}
\left(\frac{\Omega_{,v}}{\Omega}-\frac{\Omega_{u}}{\Omega}\right)\left(35+9\mu\right)\right] \, .
\end{equation}

\subsubsection{The choice of $f$} \label{rstarzero}
By Propositions \ref{kgom}, \ref{omvcent}  and \ref{omucent} we
have the bounds
\begin{equation}
\lambda - \nu = \left(1-\mu\right) + C\left(r^\star_{cl},c\right)
\frac{M}{t^2} \, ,
\end{equation}
\begin{equation}
\frac{\Omega_{,v}}{\Omega}-\frac{\Omega_{,u}}{\Omega} \leq \frac{\mu}{r} +
\frac{C\left(\epsilon\right)}{t} \, ,
\end{equation}
\begin{equation}
\Big|\frac{\Omega_{,v}}{\Omega}-\frac{\Omega_{,u}}{\Omega}\Big| \leq \frac{\mu}{r} + \frac{1}{\sqrt{M}}
C\left(\epsilon\right) 
\end{equation}
in the region $r^\star_{cl} \leq r^\star \leq \frac{9}{10}t$. This implies that 
\begin{equation} \label{goneside}
g \geq \frac{1}{2r^7} \left(35r^4-104M r^2 -108 M^2 \right) + \frac{C\left(\epsilon\right)}{t r^2}
\end{equation}
and that
\begin{equation}
\Big| g - \frac{1}{2r^7} \left(35r^4-104M r^2 -108 M^2 \right) \Big|
\leq \frac{C\left(\epsilon\right)}{\sqrt{M} r^2} \, ,
\end{equation}
where $M=m\left(T,r^\star=0\right)$.
It is apparent (note Proposition \ref{omvclo} in particular) 
that the expression $g$ is negative close to the
horizon, positive far away from it with a single zero at
some $r=r_{zero}\left(t\right)$ for which the estimate
\begin{equation}
\Big|r_{zero} - \sqrt{2M} \sqrt{\frac{1}{35} \left( 26+ \sqrt{1621} 
  \right)}\Big| \leq C\left(\epsilon\right) \sqrt{M}
\end{equation}
can be derived. This in turn implies an estimate (with error of order
$\frac{1}{t}$) for the $r^\star$ value along that curve 
via identification (\ref{rstarr}). It follows that 
$g$ changes sign in a small interval around $r^\star_{zero}$. 
For future calculation the estimate 
\begin{equation}
-\frac{1}{6}\sqrt{M} < r^\star_{zero} <  -\frac{1}{10}\sqrt{M}
\end{equation}
for $r^\star_{zero}$ will be sufficient.

We finally construct the function $f=f\left(x\right) =
f\left(r^\star-\sigma - \sqrt{M}\right)$ by setting
\begin{equation}
f\left(x_{zero}\right)=f\left(-\sigma - \sqrt{M} - r^\star_{zero}\right)=0
\end{equation}
and
\begin{equation}
f^\prime = \frac{M^{n-\frac{1}{2}}}{\left(x^2 + \sigma^2\right)^n} \, .
\end{equation}
Note that $f$ will be bounded for $n \in \left(\frac{1}{2},\infty\right)$. 
Later we will set $n=\frac{3}{2}$. We compute
\begin{equation}
f^{\prime \prime} = M^{n-\frac{1}{2}}\frac{-2xn}{\left(x^2+\sigma^2\right)^{n+1}} \, ,
\end{equation}
\begin{equation}
f^{\prime \prime \prime} = M^{n-\frac{1}{2}}
\frac{2n\left(x^2+2nx^2-\sigma^2\right)}{\left(x^2+\sigma^2\right)^{n+2}}
\end{equation}
to find
\begin{eqnarray}
F=M^{n-\frac{1}{2}}\frac{n}{\Omega^2}
\frac{x^2-\sigma^2}{\left(\sigma^2+x^2\right)^{n+2}}  \, .
\end{eqnarray}

We will now show that there exists a positive
constant $c\left(\sigma\right) > 0$ such that in (\ref{IXBstsh})
\begin{equation} \label{posFfg}
F + f \cdot g -\frac{1}{2}{\cal{I}}_7 \left(B\right) + {\cal{I}}_8 \left(B\right) \geq \frac{1}{r^3}c\left(\sigma\right)
\end{equation}
holds in the region of integration. Note that the above choice 
of $f$ ensures that $f \cdot g$ is positive everywhere,
except for an $\epsilon$-small correction-term. In particular, in
$\left[r^\star_{zero}-\frac{1}{10}\sqrt{M},
  r^\star_{zero}+\frac{1}{10}\sqrt{M}\right]$ we have
\begin{equation} \label{aroundzer}
f g \geq - \frac{1}{M^\frac{3}{2}} C\left(\epsilon\right)
\end{equation}
and outside of this set, using (\ref{error78}) and (\ref{goneside})
\begin{equation} \label{awayzer}
\frac{1}{8}f g -\frac{1}{2} {\cal{I}}_7 \left(B\right) + {\cal{I}}_8 \left(B\right) 
\geq \frac{c_2 \left(\sigma\right)}{r^3} \, .
\end{equation}
\subsubsection{Estimating $\bar{I}^X_B\left(\mathcal{D}_{[t_1,t_2]}^{r^\star_{cl},u_J}\right)$}
The aim is to establish (\ref{posFfg}). 
We will do the computations in the shifted $x$-coordinate (\ref{xshift}).
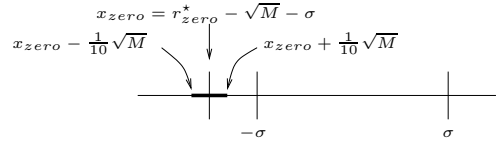
\[
\input{regXpos.pstex_t}
\]
\paragraph{The region $x \leq -\sigma$ and $x \geq \sigma$}
Note that  $F$ is already non-negative for  $x \leq -\sigma$ and $x \geq \sigma$.
 In the subinterval $\left[x_{zero}-\frac{1}{10}\sqrt{M},
    x_{zero}+\frac{1}{10}\sqrt{M}\right]$, the only subset where $f \cdot
 g$ might cause problems, we have the stronger bound
\begin{equation}
F \geq \frac{1}{M^\frac{3}{2}} c_1\left(\sigma\right) > 0 \, ,
\end{equation}
which upon combination with (\ref{aroundzer}) yields
\begin{equation}
F + f g -\frac{1}{2} {\cal{I}}_7 \left(B\right) + {\cal{I}}_8 \left(B\right) \geq
\frac{1}{r^3} c\left(\sigma\right)
\end{equation}
for the regions under consideration.
\paragraph{The region $[-\sigma,\sigma]$}
We shall show that we can dominate the term $F$ in the region
  $[-\sigma,\sigma]$ by the term $\frac{7}{8} f \cdot g$ in
  (\ref{IXBstsh}).\footnote{Note that we save $\frac{1}{8}fg$ to get
  the positivity of (\ref{awayzer}).} In conjunction with
  (\ref{awayzer}) this will yield the desired result (\ref{posFfg}). 
By Proposition \ref{kgom}
\begin{equation}
|\Omega^2-\left(1-\mu\right)| = |\left(4\gamma \kappa-1\right)
 \left(1-\mu\right)| \leq C\left(r^\star_{cl},c\right)\frac{M}{t^2}
\end{equation}
holds in the region $-\sigma \leq x \leq \sigma$ and also $f$ is
positive there.\footnote{Note that we can use the aforementioned 
Proposition because the region under consideration 
lies in $r^\star \geq r^\star_{cl}$.} 
In view of (\ref{goneside}) it suffices to show
\begin{equation} \label{inequality}
M^{n-\frac{1}{2}} \frac{n}{\left(1-\mu\right)}
\frac{\sigma^2-x^2}{\left(\sigma^2+x^2\right)^{2+n}} \leq \frac{1}{2r^3}
\left(35 - 104\frac{M}{r^2} - 108\frac{M^2}{r^4}\right) f \frac{7}{8}
\end{equation}
in the region $-\sigma \leq x \leq \sigma$. There we can estimate
\begin{equation}
f\left(x\right) = \int_{-\sigma-\sqrt{M}+r^\star_{zero}}^x f^\prime >
\int_{-\sigma-\sqrt{M}}^x f^\prime \geq M^{n-\frac{1}{2}}\frac{x+\sigma +
  \sqrt{M}}{\left(2\sigma^2 + 2 \sigma \sqrt{M} + M\right)^n}
\end{equation}
such that it suffices to establish
\begin{equation} \label{Xineq}
\frac{n}{\left(1-\mu\right)}
\frac{\sigma^2-x^2}{\left(\sigma^2+x^2\right)^{2+n}} < \frac{1}{2r^3}
\left(35 - \frac{104M}{r^2} - 108\frac{M^2}{r^4}\right) \frac{x+\sigma +
  \sqrt{M}}{\left(2\sigma^2 + 2 \sigma \sqrt{M} + M\right)^n}
\frac{7}{8} \, .
\end{equation}
\emph{Part I.} Consider first the region $-\sigma < x \leq -\frac{2}{3}
\sigma$ translating to $\sqrt{M} < r^\star \leq \frac{1}{3} \sigma
+ \sqrt{M}$. The lower bound on $r^\star$ implies a lower bound on
$r$ by identification (\ref{rstarr}). In particular, 
$r > \sqrt{M} \left(2+\frac{3}{5}\right)$ in that region and hence
\begin{equation}
1-\mu > \frac{7}{10} \textrm{ \ \ \ \ and \ \ \ \ } \frac{1}{1-\mu} < \frac{10}{7}
\end{equation}
as well as
\begin{equation}
35 - \frac{104M}{r^2} - 108\frac{M^2}{r^4} > 17 
\end{equation}
hold in the region under consideration. Consequently, 
for $\sigma$ sufficiently large we have
\begin{eqnarray} \label{lhs}
\frac{n}{\left(1-\mu\right)}
\frac{\sigma^2-x^2}{\left(\sigma^2+x^2\right)^{2+n}} &<& n
 \frac{10}{7} \frac{\sigma - x}{\sqrt{\sigma^2 + x^2}} \frac{\sigma +
  x}{\left(\frac{13}{9} \sigma^2\right)^{3/2+n}} \nonumber \\
&\leq& n \sqrt{2}
 \frac{10}{7} \left(\frac{9}{13}\right)^{\frac{3}{2}+n}
 \frac{x+\sigma}{\sigma^{3+2n}} 
\end{eqnarray}
The upper bound on $r^\star$ can be exploited for large $\sigma$ to give
\begin{equation}
r \leq \frac{5}{12}\sigma \, .
\end{equation} 
Again choosing $\sigma$ sufficiently large, this gives rise to the estimate
\begin{eqnarray} \label{rhs}
\frac{1}{2r^3}
\left(35 - \frac{104M}{r^2} - 108\frac{M^2}{r^4}\right) \frac{x+\sigma +
  \sqrt{M}}{\left(2\sigma^2 + 2 \sigma \sqrt{M} + M\right)^n}
\frac{7}{8} \nonumber \\ \geq \frac{17}{2}
\left(\frac{12}{5\sigma}\right)^3
\frac{x+\sigma}{\left(3\sigma^2\right)^{n}} \frac{7}{8} 
\geq \frac{116}{3^n}  \frac{x+\sigma}{\sigma^{2n+3}} \frac{7}{8} \, .
\end{eqnarray}
Comparing (\ref{rhs}) and (\ref{lhs}) for $n=\frac{3}{2}$ we have
shown the desired inequality (\ref{Xineq}) in the region 
under consideration. 
\\ \\
\emph{Part II.}
 Consider now the region $-\frac{2}{3}\sigma \leq x \leq \sigma$,
which translates to $\frac{1}{3}\sigma + \sqrt{M} \leq r^\star \leq 2\sigma
+ \sqrt{M}$. We can choose $\sigma$ so large that 
\begin{equation}
1-\mu \geq \frac{6}{7} 
\end{equation}
and 
\begin{equation}
\frac{r}{x+\sigma} \leq \left(\frac{7}{6}\right)^\frac{1}{3} 
\end{equation}
hold in the region under consideration.
We deduce that for large $\sigma$ 
\begin{equation}
\frac{7}{8}\frac{1}{2r^3} \left(35 - \frac{104m}{r^2} - 108\frac{m^2}{r^4}\right) 
\frac{x+\sigma + \sqrt{M}}{\left(2\sigma^2 + 2 \sigma \sqrt{M} +
  M\right)^n} \geq \frac{7}{8} 35\frac{2}{5}
\frac{x+\sigma}{\left(2 \sigma^2\right)^n \left(x+\sigma\right)^3}
\frac{6}{7} \nonumber
\end{equation}
and
\begin{equation}
\frac{n}{\left(1-\mu\right)}
\frac{\sigma^2-x^2}{\left(\sigma^2+x^2\right)^{2+n}} \leq \frac{7}{6}n
\frac{\left(\sigma-x\right)
\left(\sigma+x\right)}{\left(\sigma^2+x^2\right)^{2+n}} \, .  
\end{equation}
Hence it suffices to show that, for large $\sigma$, we have the 
inequality 
\begin{equation}
\frac{1}{9} n \left(2\sigma^2\right)^{n} \frac{\left(\sigma-x\right)
  \left(x+\sigma\right)^3}{\left(\sigma^2+x^2\right)^{2+n}} \leq 1 \, .
\end{equation}
For $n=\frac{3}{2}$ we obtain
\begin{equation}
\frac{1}{3}\sqrt{2} \frac{\left(\sigma-x\right)\left(x+\sigma\right)^3
  \sigma^3}{\left(\sigma^2+x^2\right)^\frac{7}{2}} < 1 \, ,
\end{equation}
which is shown to be true by elementary arguments. Namely, for $x<0$
we have
\begin{equation}
\frac{1}{3}\sqrt{2} \frac{\left(\sigma-x\right)\left(x+\sigma\right)^3
  \sigma^3}{\left(\sigma^2+x^2\right)^\frac{7}{2}} \leq
  \frac{1}{3}\sqrt{2}
  \frac{\left(\sigma-x\right)\sigma^6}{\left(\sigma^2+x^2\right)^\frac{7}{2}}
  \leq \frac{\sqrt{2}}{3}
  \frac{\sqrt{2\left(\sigma^2+x^2\right)}}{\left(\sigma^2+x^2\right)^\frac{7}{2}}
  \sigma^6 \leq \frac{2}{3} < 1
\end{equation}
and for $x \geq 0$ we have
\begin{eqnarray}
\frac{\sqrt{2}}{3} \frac{\left(\sigma-x\right)\left(x+\sigma\right)^3
  \sigma^3}{\left(\sigma^2+x^2\right)^\frac{7}{2}} \leq
  \frac{\sqrt{2}}{3} \frac{\sigma^4 \left(x+\sigma\right)^3
  }{\left(\sigma^2+x^2\right)^\frac{7}{2}} \leq \frac{\sqrt{2}}{3}
  \left(\frac{\sigma x +\sigma^2}{x^2+\sigma^2}\right)^3
  \frac{\sigma}{\sqrt{x^2 + \sigma^2}} \nonumber \\
\leq \frac{\sqrt{2}}{3}
  \left(\frac{\sigma x +\sigma^2}{x^2+\sigma^2}\right)^3 
\leq \frac{\sqrt{2}}{3} \left(\frac{1}{2}
  \left(1+\sqrt{2}\right)\right)^3 < \frac{2\sqrt{2}}{3} < 1 \, .
 \end{eqnarray}
This finally establishes that the integrand of $\bar{I}^X_B$ is
non-negative and vanishes if and only if $B=0$ everywhere.
\\ \\
{\bf Remark:} The minimum size of the constant $\sigma$ required for
the estimates above to work can be determined
explicitly. Choosing $\sigma = 40 \sqrt{M}$ for instance is large enough. 

\subsubsection{Summary}
We can write
\begin{equation} \label{summary}
\bar{I}^X_B \left(\mathcal{D}_{[t_1,t_2]}^{r^\star_{cl},u_J}\right) = I^X_B \left(\mathcal{D}_{[t_1,t_2]}^{r^\star_{cl},u_J}\right) -  J^{X,u_J}_{error} - H^{X,v=t_1+r^\star_{cl}}_{error} - H^{X,u=t_2-r^\star_{cl}}_{error} 
\end{equation}
and we have shown 
\begin{proposition} \label{B2bound}
The estimate
\begin{equation}
\int_{\mathcal{D}_{[t_1,t_2]}^{r^\star_{cl},u_J}} \frac{B^2}{r^3} dVol 
\leq C\left(\sigma\right) \bar{I}^X_B
\left(\mathcal{D}_{[t_1,t_2]}^{r^\star_{cl},u_J}\right) 
\end{equation}
holds.
\end{proposition}
In the following subsection we are going to estimate the 
error boundary terms on the right hand side of (\ref{summary}). 
\subsection{Controlling the error-boundary-terms}
\begin{lemma} \label{erixc}
The error-terms (\ref{JXerr})-(\ref{HXerr2}) satisfy
\begin{equation} \label{borrower}
|J^{X,u_J}_{error} | \leq C \left(\sigma\right) \left(m
 \left(u_J,2t_2-u_J\right) - m \left(u_J,2t_1-u_J\right) \right) 
\end{equation} \, ,
\begin{equation}
H^{X,v_1=t_1+r^\star_{cl}}_{error}\ \geq  0 
\end{equation}
and
\begin{equation} \label{err3}
H^{X,u_2=t_2-r^\star_{cl}}_{error} \geq 0 \, .
\end{equation}
\end{lemma}
\begin{proof}
The statement (\ref{borrower}) is immediate since we are away from the
horizon and both $\xi$ and $f^\prime$ decay sufficiently fast at infinity
to retrieve the correct powers of $r$ appearing in the energy. For the other 
two inequalities recall that $f^\prime \geq 0$ always and 
that $x \leq -\sqrt{M}$ in the region of integration and hence 
$\xi \geq 0$. 
\end{proof}
\subsection{Controlling the boundary-terms of
  $I^X_B\left(\mathcal{D}_{[t_1,t_2]}^{r^\star_{cl},u_J}\right)$}
%Recall the identity (\ref{Xbasicidentity}):
%\begin{equation} 
%I^X_B \left(\mathcal{\mathcal{D}}_{[t_1,t_2]}^{r^\star_{cl},u_J}\right) = F^X_B \left(t_2\right) - F^X_B
%\left(t_1\right) + H^X_{u_H} - J^X_{u_J}
%\end{equation}
The following Lemmata show that the boundary terms of the
vectorfield $X$ appearing in the vector field 
identity (\ref{Xbasicidentity}) are controlled by the
energy plus a contribution from the term $\tilde{F}^Y_B$ appearing as 
bootstrap assumption \ref{boot5}. Together with the results of 
Lemma \ref{erixc}, identity (\ref{summary}) ultimately yields 
an estimate for the positive spacetime integral $\bar{I}^X_B$, manifest 
in Proposition \ref{Xenergycor}.  
\begin{lemma} \label{Lembount}
We have, for any $q \in \mathbb{R}^+$,
\begin{eqnarray} \label{control}
|F_B^X \left(t\right)| &\leq& \left[C\left(r^\star_{cl}\right)+C_f q\right] E_F\left(t\right) + C_f r^2_{cl} B^2
\left(t,r^\star_{cl}\right) + \frac{2}{q} C_f \tilde{F}^Y_B \left(v=t+r^\star_{cl}\right) \nonumber \\
\end{eqnarray}
with
\begin{equation}
E_F\left(t\right) = m\left(t,t-u_J\right) - m\left(u_{H},
  t+r^\star_{cl} \right)  
\end{equation}
and $\tilde{F}^Y_B \left(v_1\right)$ being the quantity appearing as
 bootstrap assumption \ref{boot5}. Moreover 
$C_f=sup_{r \leq r^\star_{cl}} \left[f^\prime r\right] \ll 1$ 
is a small constant and $C\left(r^\star_{cl}\right)$ just depends on $r^\star_{cl}$.
\end{lemma}
\begin{proof} 
The $F$-boundary-terms arising from energy (\ref{enbt}) can be estimated
\begin{eqnarray} 
F^T_B \left(t\right) = \int_{t-r^\star_{cl}}^{t_2-r^\star_{cl}} \left[4r^3
  \lambda \frac{\left(B_{,u}\right)^2}{\Omega^2} - r
  \nu \left(1-\frac{2}{3} \rho \right) \right]
  \left(u,t+r^\star_{cl}\right) du \nonumber \\
+\int_{r^\star_{cl}}^{t-u_J} \Bigg(-4r^3 \frac{\nu}{\Omega^2}
\left(B_{,v}\right)^2 + 4r^3 \frac{\lambda}{\Omega^2} \left(B_{,u}\right)^2
+ r \left(\lambda-\nu\right) \left(1-\frac{2}{3}\rho \right)
  \Bigg)\left(t,r^\star \right) dr^\star  \nonumber \\
\geq \int_{t-r^\star_{cl}}^{t_2-r^\star_{cl}} \left[4r^3
  \lambda \frac{\left(B_{,u}\right)^2}{\Omega^2} - r
  \nu \left(1-\frac{2}{3} \rho \right) \right]
  \left(u,t+r^\star_{cl}\right) du \nonumber \\
+\frac{1}{C_L} \int_{r^\star_{cl}}^{t-u_J} \Bigg(\left(\partial_t B\right)^2 + \left(\partial_{r^\star} B\right)^2
+ \frac{B^2}{r^2}  \Bigg)r^3 \left(t,r^\star \right) dr^\star
\end{eqnarray}
The $F$-boundary terms arising from the basic identity for the vector
field $X$ (\ref{XFterms}) are
\begin{eqnarray} \label{bxterms}
\frac{1}{2\pi^2} F^X_B \left(t\right) = &-&2\int_{r^\star_{cl}}^{t-u_J}
 f \partial_t B \partial_{r^\star} B
\left(t_i,r^\star\right) r^3 dr^\star  \\ 
&-& \int_{r^\star_{cl}}^{t-u_J}
\left(f^\prime + \frac{3}{r} \left(\lambda-\nu\right)f \right)
\left(\partial_t B \right) B \left(t,r^\star\right) r^3  dr^\star
 \nonumber \\
&+& \int_{r^\star_{cl}}^{t-u_J}
\frac{1}{2}\left(\partial_t \left(\frac{3}{r} \left(\lambda-\nu\right)\right)f
  \right)B^2 \left(t_1,r^\star\right) r^3 dr^\star 
  \nonumber \\
&+&\int_{t-r^\star_{cl}}^{t_2-r^\star_{cl}} \left[r^3 \left(\partial_u B\right)^2 
\left(2f\right) + \frac{r\Omega^2}{4} 
\left(1-\frac{2}{3}\rho\right)\left(-2f\right) \right] du \nonumber \\
 &-& \int_{t-r^\star_{cl}}^{t_2-r^\star_{cl}}  
\left(f^\prime + \frac{3}{r} \left(\lambda-\nu\right)f \right)
\left(\partial_u B \right) B \left(u,t+r^\star\right) r^3  du
 \nonumber \\
&+& \int_{t-r^\star_{cl}}^{t_2-r^\star_{cl}}  
\frac{1}{2}\left(\partial_u \left(f^\prime + \frac{3}{r} \left(\lambda-\nu\right)f\right)
  \right)B^2 \left(u,t+r^\star_{cl}\right) r^3 du \, .
  \nonumber 
\end{eqnarray}

\paragraph{Spacelike $F^X$-terms}
We estimate the first three terms of (\ref{bxterms}) (with the index $sl$ denoting the restriction to the spacelike terms)
\begin{eqnarray} \label{ckkb}
\frac{|F_B^X \left(t\right)|}{2\pi^2}\Big|_{sl} \leq \int_{r^\star_{cl}}^{t-u_J} |f|
\left( \left(\partial_t B\right)^2 + \left(\partial_{r^\star}
B\right)^2 \right) r^3 dr^\star \nonumber \\ 
+ \int_{r^\star_{cl}}^{t-u_J}
\frac{f^\prime}{2} \left(\left(\partial_t B\right)^2 \sqrt{M} + \frac{B^2}{\sqrt{M}} \right) r^3 dr^\star
+\int_{r^\star_{cl}}^{t-u_J} |f| \frac{3}{2r} \left(\lambda-\nu\right)
\left(r \left(\partial_t B\right)^2 + \frac{B^2}{r}\right) r^3
dr^\star
\nonumber \\
+\int_{r^\star_{cl}}^{t-u_J} |f| \left[\frac{3}{2r^2}
  \left(\lambda-\nu\right) \left(\lambda+\nu\right) + \frac{3}{r}
  \left(r_{,vv}-r_{,uu}\right) \right] B^2 r^3 dr^\star \, .
\end{eqnarray}
Recall that $f^\prime$ is positive and that we have $|f| \leq C\left(\sigma\right)$ and
$f^\prime \leq C\left(\sigma\right) \frac{M}{r^3}$. Moreover
$\left(\lambda+\nu\right)$ is clearly bounded everywhere, as is
$r^2\frac{\Omega_{,v}}{\Omega}$ and $r^2\frac{\Omega_{,u}}{\Omega}$
(cf.~Corollary \ref{omuomvdecr}). 
Using the equations
\begin{equation} \label{rvvruu}
r_{,vv} = 2 \frac{\Omega_{,v}}{\Omega}\lambda - \frac{2}{r^2}\theta^2
\textrm{\ \ \ \ and \ \ \ \ } r_{,uu} = 2 \frac{\Omega_{,u}}{\Omega}\nu - \frac{2}{r^2}\zeta^2
\end{equation}
it becomes apparent that the terms in (\ref{ckkb}) 
are controlled by energy.
\paragraph{Null $F^X$-terms}
For the last three terms of (\ref{bxterms}), which only arise for the
$F^X\left(t_1\right)$-term by the geometry of the region, 
we observe that the first is manifestly controlled 
by the energy because $f$ is bounded along the rays 
of integration. For the second term one
estimates for any $q \in \mathbb{R}^+$
\begin{eqnarray}
\int_{t_1-r^\star_{cl}}^{t_2-r^\star_{cl}} f^\prime |\partial_u
B| |B| r^3 du &\leq& C_f \int_{t_1-r^\star_{cl}}^{t_2-r^\star_{cl}}\left[
 q r B^2 \left(-\nu\right)
+ \frac{1}{q} \frac{\left(B_{,u}\right)^2}{-\nu}r^3\right] du \nonumber \\ &\leq& C_f \left(q E_F +
\frac{1}{q} \tilde{F}^Y_B \left(v_1\right) \right) 
\end{eqnarray}
where bootstrap assumption (\ref{intebound2}) has been used, and
\begin{eqnarray}
\int_{t_1-r^\star_{cl}}^{t_2-r^\star_{cl}} 
\left(\frac{3}{r} \left(\lambda-\nu\right)|f| \right)
|\partial_u B| |B| \left(u,t+r^\star\right) r^3  du 
\nonumber \\ \leq
\int_{t_1-r^\star_{cl}}^{t_2-r^\star_{cl}} \left(\frac{3}{r} \left(\lambda-\nu\right)|f| \right)\frac{1}{2}\left(\frac{\left(\partial_u B \right)^2}{-\nu}\sqrt{M} +
\frac{1}{\sqrt{M}} B^2\left(-\nu\right) \right) r^3 \left(u,v_1\right) du \leq C 
\left(\sigma\right) E_F \, . \nonumber
\end{eqnarray}
Finally, for the third term we note that $x \leq -\sqrt{M}$ and hence $f^{\prime
  \prime} \geq 0$ in the integration region to estimate
\begin{eqnarray}
\int_{t_1-r^\star_{cl}}^{t_2-r^\star_{cl}}  
\frac{1}{2}\left[-\partial_u f^\prime -\left(\partial_u f\right) \frac{3}{r}
\left(\lambda-\nu\right) + |f| \Big|\partial_u\left(\frac{3}{r} \left(\lambda-\nu\right)\right)\Big|
  \right] B^2 \left(u,t+r^\star_{cl}\right) r^3 du \nonumber \\ 
\leq C_f r_{cl}^2 \ B^2 \left(t_1,r^\star_{cl}\right) +
C_f \int_{t_1-r^\star_{cl}}^{t_2-r^\star_{cl}} 2|B| |B_{,u}| r^2 du + C \left(m\left(u_1,v_1\right) - m\left(u_2,v_1\right)\right) \nonumber \\ 
\leq C_f r_{cl}^2 \ B^2 \left(t_1,r^\star_{cl}\right) + C_f q \cdot E_F
+ C_f \frac{1}{q} \tilde{F}^Y_B \left(v_1\right) + C E_F \, .
  \nonumber 
\end{eqnarray}
for any $q \in \mathbb{R}^+$.
\end{proof}
%
%
%
% far away piece
%
%
%
\begin{lemma} 
We have
\begin{equation} 
|J^X_{u_J}| \leq C E_J\left(u_J\right)
\end{equation}
for some constant $C$ and
\begin{equation}
E_J\left(u_J\right)= m \left(u_J,2t_2-u_J\right) - m
\left(u_J,2t_1-u_J\right) \, .
\end{equation}
\end{lemma}
\begin{proof}
For the terms in the first line 
of (\ref{XJterms}) apply the inequality $2 B B_{,v} \leq \frac{B^2}{r} 
+ r\left(B_{,v}\right)^2$ to retrieve the correct 
powers of $r$ in the energy. For the second line observe 
that $f^\prime$, $r_{,vv}$ and $r_{,uu}$ decay sufficiently fast in $r$.
\end{proof}
%
%
% horizon terms
%
%
\begin{lemma} \label{Lembounh}
We have
\begin{equation} 
H^X_{u_H} \leq C \ E_H\left(v_1,v_2\right) + C_f r_{cl}^2 B^2
  \left(t_2,r^\star_{cl}\right)
\end{equation}
for some constant $C$ and
\begin{equation}
E_H\left(v_1,v_2\right)= m\left(u_{hoz},v_2\right) - m\left(u_{hoz},
  v_1 \right) \, .
\end{equation}
\end{lemma}

\begin{proof}
The term $\hat{H}^X_{u_H}$ is manifestly controlled by the energy by
the fact that $f$ is bounded. For the second and third term in 
(\ref{XHterms}), we observe that the terms which are multiplied by a
$\lambda$- or a $\nu$-factor (or derivatives thereof) are controlled by the
energy. For the remaining terms we note that $x \leq -\sqrt{M}$ and
hence $f^{\prime \prime} > 0$ in this region and estimate (using a Hardy inequality)
\begin{eqnarray} \label{boho}
\int_{v_1}^{v_2} \left[|B| |\partial_v B| f^\prime +
  \frac{B^2}{2} \left(f^{\prime}\right)_{,v} \right]r^3 dv \nonumber
  \\ \leq C_f r_{cl}^2 B^2
  \left(t_2,r^\star_{cl}\right) + \frac{3}{2} \int_{v_1}^{v_2} B^2
  f^\prime r^2 \lambda dv + \int_{v_1}^{v_2} 2|B| |\partial_v
  B| f^\prime dv \nonumber \\ \leq C_f r_{cl}^2 B^2
  \left(t_2,r^\star_{cl}\right) + C_f E_H +  16\int_{v_1}^{v_2} \left(\partial_v B\right)^2
  \frac{\left(f^\prime\right)^2}{f^{\prime \prime}} r^3 dv +
  \int_{v_1}^{v_2} \left[ \frac{B^2}{4} \left(f^{\prime}\right)_{,v}
  \right]r^3 dv
\end{eqnarray}
from which we obtain
\begin{equation}
\int_{v_1}^{v_2} \left[|B| |\partial_v B| f^\prime +
  \frac{B^2}{4} \left(f^{\prime}\right)_{,v} \right]r^3 dv \leq  C_f r_{cl}^2 B^2
  \left(t_2,r^\star_{cl}\right) + E_H
\end{equation}
because
\begin{equation}
\frac{\left(f^\prime\right)^2}{f^{\prime \prime}} =
\frac{1}{-3x}\frac{M}{\sqrt{x^2+\sigma^2}}
\end{equation}
is small in the region under consideration
allowing us to estimate the derivative term in the last 
line of (\ref{boho}) by the energy.
\end{proof}

From the previous Lemmata and the identity (\ref{summary}) 
we conclude
\begin{proposition} \label{Xenergycor}
In the region $\mathcal{D}^{r^\star_{cl},u_J}_{[t_1,t_2]}$ we have for any $q \in \mathbb{R}^+$ 
\begin{eqnarray} \label{eqr2}
\bar{I}^X_B \left(\mathcal{D}^{r^\star_{cl},u_J}_{[t_1,t_2]}\right) &\leq&
\left[C\left(r^\star_{cl}\right) +C_f q\right] \Big[E_F\left(t_1\right) +
  E_F\left(t_2\right)\Big] + C\left(\epsilon\right) \frac{M^2}{t_1^2}\nonumber \\ &+& C E_H \left(v_1,v_2\right) + C E_J\left(u_J\right)   + 2C_f \frac{1}{q} \left(\tilde{F}^Y_B\left(v_1\right)\right) \, .
\end{eqnarray}
\end{proposition}
\begin{proof}
Add up the estimates of the previous Lemmata. Note that pointwise bounds for $B^2$ on the $r^\star=r^\star_{cl}$-curve  were obtained from the energy $E_f\left(t\right)$ in Proposition \ref{pointBcent}, including a small term $C\left(\epsilon\right) \frac{M^2}{t_1^2}$. Hence we can replace the terms appearing in Lemmata \ref{Lembount} and \ref{Lembounh} by the energy, such that (\ref{eqr2}) is finally obtained.
\end{proof}
\subsection{Controlling the time derivative from $I^X_B\left(\mathcal{B}\right)$}
For the statements of the next two propositions let $R^\star =
 -\frac{1}{3}\sqrt{M} < r^\star_{zero}$ and define the regions
\begin{equation}
\mathcal{B} = \{ t_0 \leq t \leq t_1 \} \cap \{ r^\star_{cl} \leq r^\star \leq R^\star \}
\end{equation}
and the slightly smaller region
\begin{equation}
\mathcal{B}_\varsigma = \{ t_0 \leq t \leq t_1 \} \cap \{ r^\star_{cl} +
  \varsigma \leq r^\star \leq R^\star - \varsigma \} \, .
\end{equation}
where $\varsigma \leq \frac{1}{6} \sqrt{M}$ is some positive number (in units of $\sqrt{M}$), say $\varsigma = \frac{1}{10} \sqrt{M}$. Define also a smooth cut off function $\chi$ supported in $\left[r^\star_{cl}, R^\star \right]$ and equal to $1$ for $\{ r^\star_{cl} + \varsigma \leq r^\star \leq R^\star - \varsigma \}$. Note that 
with the definition of $r^\star_{cl}$ (cf.~section \ref{rstarcldef}) we have $r^\star_{cl} + \varsigma <
r_Y^\star$ and also $R^\star - \varsigma \geq -\frac{1}{2}\sqrt{M}$.
\begin{proposition} \label{B2cor}
We have, for large $t_1$,
\begin{eqnarray} \label{ineq1}
\frac{1}{M^\frac{3}{2}} \int_{\mathcal{B}} B^2 dVol \leq
C\left(\sigma\right) \bar{I}^X_B
\left(\mathcal{B}^{r^\star_{cl},R^\star}_{[t_1,t_2]}\right)
\end{eqnarray}
and
\begin{eqnarray} \label{ineq2}
\frac{1}{\sqrt{M}} \int_{\mathcal{B}} \left(\partial_{r^\star} B\right)^2 dVol \leq
C\left(r^\star_{cl},\sigma\right) \bar{I}^X_B \left(
\mathcal{B}^{r^\star_{cl},R^\star}_{[t_1,t_2]}\right) \, .
\end{eqnarray}
Furthermore, we can control the time derivative 
\begin{eqnarray} \label{ineq3}
\frac{1}{\sqrt{M}}\int_{\mathcal{B}_\varsigma} \left(\partial_t B\right)^2 dVol &\leq&
C\left(r^\star_{cl},\sigma,\chi\right) \Big[ \bar{I}^X_B\left(\mathcal{B}^{r^\star_{cl},R^\star}_{[t_1,t_2]}\right) 
\nonumber \\ &+&  m \left(t_2,R^\star\right) - m
  \left(t_2,r^\star_{cl}\right) +  m \left(t_1,R^\star\right) - m
  \left(t_1,r^\star_{cl}\right) \Big]  \, .
\end{eqnarray}
\end{proposition}
\begin{proof}
Inequality (\ref{ineq1}) is the statement of Proposition \ref{B2bound}. 
Equation (\ref{ineq2}) follows from an application of the 
triangle inequality to the first term in
(\ref{IXBstsh}) and the previous bound on the $B^2$-integral, noting that $\frac{f^\prime}{\Omega^2}$ is bounded above and below in the region under consideration. Finally,
(\ref{ineq3}) is obtained via Green's identity\footnote{There are no
  problems with the differentiability here.}
\begin{eqnarray}
\int_\mathcal{B} \Box \chi \left(-\frac{1}{2} B^2\right) = 
\int_\mathcal{B} \chi \left(-\frac{1}{2}\Box B^2\right) 
- \int \chi B \partial_t B r^3
dr^\star dA_{\mathbb{S}^3} \Bigg|^{t=t_2}_{t=t_1} \, 
\end{eqnarray}
for the $\chi$ defined above, which can be written
\begin{eqnarray}
\int_\mathcal{B} \chi \frac{1}{\Omega^2} \left(\partial_t B\right)^2 dVol &=& \int_\mathcal{B} B^2 \left[-\frac{1}{2}\Box \chi + \frac{\chi}{r^2}
  \left(8 - \frac{\varphi_2\left(B\right)}{B^2} \right) \right] dVol\nonumber \\ &+&
  \int_\mathcal{B} \chi \left(\partial_{r^\star} B\right)^2 \frac{1}{\Omega^2} dVol + \int \chi B \partial_t B r^3 dr^\star
  dA_{\mathbb{S}^3} \Bigg|^{t=t_2}_{t=t_1} \nonumber \, .
\end{eqnarray}
The spacetime integrals on the right hand side are controlled by
 (\ref{ineq1}) and (\ref{ineq2}). For the boundary term in the 
second line we estimate
\begin{eqnarray}
\int \chi B \partial_t B \left(t, \bar{r}^\star \right) \bar{r}^3 d\bar{r}^\star
  dA_{\mathbb{S}^3} &\leq& \int_{r^\star_{cl}}^{R^\star}  \left(\frac{B^2}{\sqrt{M}} + \sqrt{M} \left(\partial_t
  B\right)^2 \right) \left(t, \bar{r}^\star \right) \bar{r}^3 d\bar{r}^\star
  dA_{\mathbb{S}^3} \nonumber \\ &\leq& \sqrt{M} C\left(r^\star_{cl}\right) \left(m \left(t,R^\star\right) - m
  \left(t,r^\star_{cl}\right) \right) \, .
\end{eqnarray}
Putting all this together we obtain 
\begin{eqnarray}
\frac{1}{\sqrt{M}}\int_\mathcal{B_\varsigma} \left(\partial_t
B\right)^2 dVol \leq  \frac{1}{\sqrt{M}}\int_\mathcal{B} \chi \frac{2}{\Omega^2} \left(\partial_t
B\right)^2 dVol
\nonumber \\ 
\leq C\left(r^\star_{cl},\sigma,\chi\right) \Big[ \bar{I}^X_B\left(\mathcal{B}^{r^\star_{cl},R^\star}_{[t_1,t_2]}\right) 
+  m \left(t_2,R^\star\right) - m
  \left(t_2,r^\star_{cl}\right) +  m \left(t_1,R^\star\right) - m
  \left(t_1,r^\star_{cl}\right) \Big]  \nonumber \, .
\end{eqnarray}
\end{proof}
We can summarize this as
\begin{proposition} \label{timecontrol}
The quantity
\begin{equation}
I_B \left(\mathcal{W}\right) = \int_{W} \left[\frac{1}{\sqrt{M}} \left( \partial_t B
  \right)^2 + \frac{1}{\sqrt{M}} \left( \partial_{r^\star} B \right)^2 + \frac{1}{M^\frac{3}{2}}B^2  \right] dVol
\end{equation}
satisfies
\begin{eqnarray}
I_B \left(\mathcal{B}_\varsigma\right) &\leq&
C\left(r^\star_{cl},\sigma,\chi\right) \Big[ \bar{I}^X_B\left(\mathcal{B}^{r^\star_{cl},R^\star}_{[t_1,t_2]}\right) 
\nonumber \\ &+&  m \left(t_2,R^\star\right) - m
  \left(t_2,r^\star_{cl}\right) +  m \left(t_1,R^\star\right) - m
  \left(t_1,r^\star_{cl}\right) \Big]  \, .
\end{eqnarray}
\end{proposition}
\begin{proof}
This is a direct consequence of Proposition \ref{B2cor}.
\end{proof}
\section{Combining $X$ and $Y$: Horizon Estimates} \label{XconY}
For the next two propositions recall the choice of
$r^\star_Y$ made in section \ref{Ysection}, which
implied in particular that $Y$ is supported only 
in $r^\star \leq -\frac{1}{2} \sqrt{M}$.

\subsection{Controlling $I^Y_B$ from $I^X_B$ and energy}
\begin{proposition} \label{YBcontrol}
Consider the characteristic rectangle $\mathcal{R} = \left[u_1,{u_{hoz}}\right]
  \times [v_1,v_2]$ together with the $r^\star=r^\star_{cl}$
  curve intersecting $\left(u_1,v_1\right)$. Define $u\left(v_2\right)$ by
  $r\left(u(v_2),v_2\right) = r^\star_{cl}$ and $r\left(u,v(u)\right) =
  r^\star_{cl}$. We have the estimates
\begin{eqnarray}
F^Y_B \left(\{ u_1 \} \times \left[v_1,v_2\right] \right) &\leq& C\left(r^\star_{cl}\right)
\left(m \left(u_1,v_2\right) - m \left(u_1,v_1\right) \right)
  \\
\int_{u(v_2)}^{u_{hoz}} \frac{\left(\partial_u B
  \right)^2}{\Omega^2}  du &\leq& C F_B^Y \left(\Big[u_1,
  u_{hoz}\Big] \times \{ v_2 \} \right)  \\
\int_{v_1}^{v(u)}  r B^2  dv &\leq&
 C F_B^Y \left(\{ u \} \times\Big[v_1,
 v(u)\Big] \right) \textrm{\ \ \ for all $ u \geq u_1$ \ \ \ \ } 
\end{eqnarray}
\end{proposition}
\begin{proof}
This follows from the definitions (\ref{bF1}) and (\ref{bF2})
\end{proof}
\begin{proposition} \label{hozest}
Recall the basic dyadic regions $\mathcal{D}^{r^\star_{cl},u_J}_{[t_1,t_2]}$
for the vectorfield $X$ (\ref{basicXreg}) and erect the 
characteristic rectangle 
\begin{equation}
\mathcal{R} = \left[u_1,{u_{hoz}}\right] \times \left[v_1,v_2\right]
\end{equation}
associated with such a region as depicted in the figure.
More precisely, let $u_1=t_1-r^\star_{cl}$, 
$v_1=t_1+r^\star_{cl}$, $v_2=t_2+r^\star_{cl}$.
\begin{figure}[h!]
\[
\input{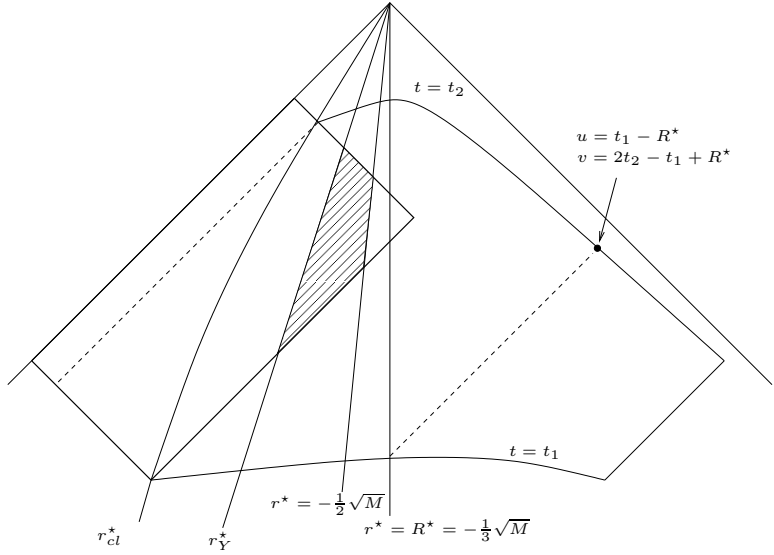} 
\]
\caption{The horizon estimate.} \label{hozestfig}
\end{figure}
\newline 
Let also
\begin{equation}
\mathcal{T} = \{ r^\star \geq r^\star_Y \} \cap \{ v \leq v_2 \} \cap \{ u \geq
u_1 \} \, ,
\end{equation}
and recall that $R^\star = -\frac{1}{3}\sqrt{M}$. We have the inequality
\begin{eqnarray} \label{hozineq}
F^Y_B \left(\{ {u_{hoz}} \} \times \left[v_1,
  v_2\right] \right) + F^Y_B \left(\left[u_1, {u_{hoz}}
  \right] \times \{ v_2 \} \right) + \frac{1}{2}
\tilde{I}^Y_B \left(\mathcal{R} \setminus \mathcal{T}\right) \nonumber
\\ 
\leq C\left(r^\star_{cl},\sigma \right) \Big[ \bar{I}^X_B\left(\mathcal{B}^{r^\star_{cl},R^\star}_{[t_1,t_2]}\right) 
+  m \left(t_2,R^\star\right) - m
  \left(t_2,r^\star_{cl}\right) +  m \left(t_1,R^\star\right) - m
  \left(t_1,r^\star_{cl}\right) \Big]  
\nonumber \\ 
+ C\left(r^\star_{cl}\right)
\Big[m \left(u_1,v_2\right) - m \left(u_1,v_1\right)
  \Big] 
+ F^Y_B\left(\left[u_1, {u_{hoz}} \right] \times \{ v_1 \} \right) \, .
\end{eqnarray}
\end{proposition}
\begin{proof}
Recall the identity (\ref{hozid}):
\begin{eqnarray} 
F^Y_B \left(\{{u_{hoz}}\} \times \left[v_1, v_2\right] \right) + F^Y_B
\left(\left[u_1, {u_{hoz}} \right] \times \{v_2\} \right) + \tilde{I}^Y_B\left(\mathcal{R}\right)
\nonumber \\ = \hat{I}^Y_B
\left(\mathcal{R}\right) + F^Y_B \left(\{u_1\} \times \left[v_1,
  v_2\right] \right)+ F^Y_B \left(\left[u_1, {u_{hoz}} \right] \times
\{v_1\} \right) \, .
\end{eqnarray}
By Proposition \ref{YBcontrol} we control
\begin{equation}
F^Y_B \left(\{ u_1 \} \times \left[v_1,v_2\right] \right) \leq
C\left(r^\star_{cl}\right) \left(m \left(u_1,v_2\right) - m
\left(u_1,v_1\right) \right) \, .
\end{equation}
To establish (\ref{hozineq}) we will show
\begin{eqnarray} \label{fi}
\hat{I}^Y_B \left(\mathcal{R}\right) \leq \frac{1}{2} 
\tilde{I}^Y_B \left(\mathcal{R} \setminus \mathcal{T} \right) \nonumber \\
+ C\left(r^\star_{cl},\sigma \right) \Big[\bar{I}^X_B\left(\mathcal{B}^{r^\star_{cl},R^\star}_{[t_1,t_2]}\right) 
+  m \left(t_2,R^\star\right) - m
  \left(t_2,r^\star_{cl}\right) +  m \left(t_1,R^\star\right) - m
  \left(t_1,r^\star_{cl}\right) \Big] 
\end{eqnarray}
\begin{eqnarray} \label{secondT}
\tilde{I}^Y_B \left(\mathcal{T} \right) \leq C\left(r^\star_{cl},\sigma \right) \Big[\bar{I}^X_B\left(\mathcal{B}^{r^\star_{cl},R^\star}_{[t_1,t_2]}\right) 
 \nonumber \\  +  m \left(t_2,R^\star\right) - m
  \left(t_2,r^\star_{cl}\right) +  m \left(t_1,R^\star\right) - m
  \left(t_1,r^\star_{cl}\right) \Big]
\end{eqnarray}
To see this decompose 
\begin{equation}
\hat{I}^Y_B \left(\mathcal{R}\right) = \hat{I}^Y_B\left(\mathcal{R}
\setminus \mathcal{T} \right) + \hat{I}^Y_B\left(\mathcal{T} \right)
\end{equation}
Since in $\mathcal{R} \setminus \mathcal{T}$ we have by definition 
$r^\star < r^\star_Y$ one can apply (\ref{tildconhat}) to obtain
\begin{equation} 
\hat{I}^Y_B\left(\mathcal{R}
\setminus \mathcal{T} \right) \leq \frac{1}{2} \tilde{I}^Y_B\left(\mathcal{R}
\setminus \mathcal{T} \right) \, .
\end{equation}
On the other hand, in the region $\mathcal{T}$ we have 
\begin{eqnarray} \label{Test}
\hat{I}^Y_B \left(\mathcal{T}\right) &\leq& C \left(r^\star_Y\right) I_B
\left(\mathcal{T} \cap 
\{ r^\star \leq R^\star = -\frac{1}{2}\sqrt{M} \}\right) \, ,
\end{eqnarray}
which follows from the fact that 
$Y$ is only supported for $r^\star \leq -\frac{1}{2}\sqrt{M}$. An
application of Proposition \ref{timecontrol} to the term on the right
hand side of (\ref{Test}) will produce the required second 
term on the right hand side 
of (\ref{fi}). The estimate (\ref{secondT}) is obtained 
completely analogous to (\ref{Test}).
\end{proof}
\begin{proposition} \label{hozestc}
With assumptions and geometry as in Proposition \ref{hozest}
we also have
\begin{eqnarray} 
F^Y_B \left(\{ {u_{hoz}} \} \times \left[v_1,
  v_2\right] \right) + F^Y_B \left(\left[u_1, {u_{hoz}}
  \right] \times \{ v_2 \} \right) + \frac{1}{2}
\tilde{I}^Y_B \left(\mathcal{R} \setminus \mathcal{T}\right) \nonumber
\\ 
\leq C\left(r^\star_{cl},\sigma \right) \Big[
  m\left(u=t_1-R^\star,v=\frac{12}{11}t_2+R^\star\right) - m
  \left(u=t_2-r^\star_{cl},v=t_1+r^\star_{cl}\right) \Big] \nonumber \\
+ \frac{11}{10} F^Y_B\left(\left[u_1,
  {u_{hoz}} \right] \times \{ v_1 \} \right) \nonumber \, .
\end{eqnarray}
\end{proposition}
\begin{proof}
Use the inequality
\begin{equation}
\bar{I}^X_B\left(\mathcal{B}^{r^\star_{cl},R^\star}_{[t_1,t_2]}\right)
\leq \bar{I}^X_B\left({}^{u=t_2-r^\star_{cl}}\mathcal{D}^{r^\star_{cl},t_1-R^\star}_{[t_1,t_2]}\right)
\end{equation}
which is obvious from the positivity of the integrand (compare the 
dashed lines in the previous figure for the regions). Inserting 
the estimate of Proposition \ref{Xenergycor} 
into the inequality (\ref{hozineq}) we obtain the result by 
an appropriate choice of $q$. 
\end{proof}
It is of crucial importance that the constant $C\left(r^\star_{cl}\right)$ 
just depends on the choice of $r^\star_{cl}$ and not on $r^\star_K$.
\subsection{Controlling $F^Y_B$ from $\tilde{I}^Y_B$ and energy, on a
  good slice}
Finally, we are going to control the boundary terms $F^Y$ by $\tilde{I}^Y$ and
the energy on a ``good''  null-slice.
\begin{proposition} \label{goodslice}
With $\mathcal{R}$ and $\mathcal{T}$ as before pick a $\hat{v} \in
[v_1,v_2]$ that satisfies
\begin{equation}
F^Y_B \left(\left[u_1,{u_{hoz}}\right] \times \{ \hat{v} \} \right) =
\inf_{v_1 \leq v \leq v_2} F^Y_B
\left(\left[u_1,{u_{hoz}}\right] \times \{ v \} \right) \, .
\end{equation}
Then
\begin{equation} \label{contest}
F^Y_B \left(\left[u_1,{u_{hoz}}\right] \times \{ \hat{v} \} \right) \leq
C \left(v_2 - v_1 \right)^{-1} \tilde{I}^Y_B \left(\mathcal{R} \setminus
\mathcal{T}\right) + C\left(r^\star_{cl}\right) \left( m \left(u_1, v_2\right) - 
m \left(u_1,v_1 \right) \right) \, .
\end{equation}
\end{proposition}
\begin{proof}
Recall that the expression (\ref{bF2}) is manifestly positive. Set
$u\left(v\right) = v-2r^\star_{cl}$ and estimate 
%\begin{equation} \label{FY}
%F^Y_B \left(\left[u_1,u_{hoz}\right] \times \{ v \} \right) =
%\int_{u_1}^{u_{hoz}} \int_{\mathbb{S}^3} \left(\frac{\alpha}{\Omega^2}
%\left(\partial_u B\right)^2 + \frac{1}{4r^2}
%\left(1-\frac{2}{3}\rho\right) 
%\Omega^2 \beta \right) r^3 du dA_{\mathbb{S}^3}
%\end{equation}
\begin{eqnarray} 
F^Y_B \left(\left[u_1,u_{hoz}\right] \times \{ \hat{v}
\} \right) \leq  \inf_{v_1 \leq v \leq v_2} F^Y_B
\left(\left[u\left(v\right),u_{hoz}\right] \times \{ v \} \right) +
F^Y_B \left(\left[u_1,u\left(\tilde{v}\right)\right] \times \{ \tilde{v} \}
\right) 
\nonumber \\ \leq \frac{1}{v_2-v_1} \int_{v_1}^{v_2} F^Y_B
\left(\left[u\left(v\right),u_{hoz}\right] \times \{ v \} \right) dv +
C\left(r^\star_{cl}\right) \left[m \left(u_1, v_2 \right) - 
m \left(u_1, v_1 \right)\right] \nonumber \, \, .
\end{eqnarray}
where $\tilde{v}$ is the $v$-slice determined by taking the infimum of
$F^Y_B$ in the region $\left[u\left(v\right),u_{hoz}\right]$. 
For the integrand of the first term in the last line we have
\begin{eqnarray} \label{hep}
F^Y_B
\left(\left[v-2r^\star_{cl},u_{hoz}\right] \times \{ v \} \right) \leq 
2\pi^2 \int_{v-2r^\star_{cl}}^{u_{hoz}}du \ \ r^3  4\sqrt{M} 
\Bigg[\frac{\left(\partial_u B\right)^2}{\Omega^2}
 \left(4 \alpha \frac{\Omega_{,v}}{\Omega} 
- \alpha^\prime \right) \nonumber \\ + \beta^\prime
 \left(\partial_v B\right)^2   +\frac{1}{2 r^2} \left(1-\frac{2}{3}\rho\right)
 \left(\frac{\alpha^\prime}{2}  - \frac{\alpha \nu}{r} -
\frac{\beta \lambda \Omega^2}{r} - \frac{1}{2}
\beta^\prime \Omega^2 - 2 \beta
\Omega^2 \frac{\Omega_{,v}}{\Omega}\right) \Bigg] \, ,
\end{eqnarray}
following from the fact that the inequalities (\ref{condY1}), (\ref{condY2}), (\ref{condY3}) hold in $r^\star \leq r^\star_{cl}$. Comparing (\ref{hep}) with (\ref{posY}) 
produces the first term in (\ref{contest}). 
\end{proof}
We will also need a related version of the previous Proposition, which
provides one with a good energy-slice instead of a good $F^Y$-slice:
\begin{proposition}  \label{goodslice2}
With $\mathcal{R}$ and $\mathcal{T}$ as before pick a $\hat{v} \in
[v_1,v_2]$ that satisfies
\begin{equation}
E \left(\left[u_1,{u_{hoz}}\right] \times \{ \hat{v} \} \right) =
\inf_{v_1 \leq v \leq v_2} E
\left(\left[u_1,{u_{hoz}}\right] \times \{ v \} \right) \, .
\end{equation}
Then
\begin{equation} \label{contest2}
E \left(\left[u_1,{u_{hoz}}\right] \times \{ \hat{v} \} \right) \leq
C \left(v_2 - v_1 \right)^{-1} \tilde{I}^Y_B \left(\mathcal{R} \setminus
\mathcal{T}\right) + \left( m \left(u_1, v_2\right) - 
m \left(u_1, v_1\right) \right)
\end{equation} \, .
\end{proposition}
\begin{proof}
Recall that 
\begin{equation} 
E \left(\left[u_1,u_{hoz}\right] \times \{ v \} \right) =
\int_{u_1}^{u_{hoz}} \int_{\mathbb{S}^3} \left(\frac{4\lambda}{\Omega^2}
\left(\partial_u B\right)^2 + \frac{1}{r^2}
\left(1-\frac{2}{3}\rho\right) 
\left(-\nu\right) \right) r^3 du dA_{\mathbb{S}^3} \nonumber
\end{equation}
is manifestly positive. With $u\left(v\right) = v - 2r^\star_{cl}$ estimate
\begin{eqnarray} 
E \left(\left[u_1,u_{hoz}\right] \times \{ \hat{v}
\} \right) \leq  \inf_{v_1 \leq v \leq v_2} E
\left(\left[u\left(v\right),u_{hoz}\right] \times \{ v \} \right) +
E \left(\left[u_1,u\left(\tilde{v}\right)\right] \times \{ \tilde{v} \}
\right) 
\nonumber \\ \leq \frac{1}{v_2-v_1} \int_{v_1}^{v_2} E
\left(\left[u\left(v\right),u_{hoz}\right] \times \{ v \} \right) dv +
 \left[m \left(u_1, v_2\right) - 
m \left(u_1, v_1 \right)\right] \nonumber \, \, .
\end{eqnarray}
The integrand of the first term in the last line can be controlled by
\begin{eqnarray} \label{hep2}
E\left(\left[v-2r^\star_{cl},u_{hoz}\right] \times \{ v \} \right) \leq 
2\pi^2 \int_{v-2r^\star_{cl}}^{u_{hoz}}du \ \ r^3  4\sqrt{M} 
\Bigg[\frac{\left(\partial_u B\right)^2}{\Omega^2}
 \left(4 \alpha \frac{\Omega_{,v}}{\Omega} 
- \alpha^\prime \right) \nonumber \\ + \beta^\prime
 \left(\partial_v B\right)^2   +\frac{1}{2 r^2} \left(1-\frac{2}{3}\rho\right)
 \left(\frac{\alpha^\prime}{2}  - \frac{\alpha \nu}{r} -
\frac{\beta \lambda \Omega^2}{r} - \frac{1}{2}
\beta^\prime \Omega^2 - 2 \beta
\Omega^2 \frac{\Omega_{,v}}{\Omega}\right) \Bigg]  \, ,
\end{eqnarray}
following from the fact that the inequalities 
(\ref{condY1}), (\ref{condY2}), (\ref{condY3}) hold in $r^\star_{cl}$.
%
%
\begin{comment}
which follows from the fact that $r^\star_{cl}$ is chosen
close enough to the horizon for the inequalities
\begin{eqnarray} 
4\alpha \frac{\Omega_{,v}}{\Omega} 
- \alpha^\prime &\geq& \frac{\lambda}{\sqrt{M}} \label{adcon1b} \, ,
\\
\beta^\prime &\geq& 0 \label{adcon2b} \, , \\
4\sqrt{M} \left(\frac{1}{2} \alpha^\prime - \frac{\alpha \nu}{r} -
\frac{\beta \lambda \Omega^2}{r} - \frac{1}{2}
\beta^\prime \Omega^2 - 2 \beta
\Omega^2 \frac{\Omega_{,v}}{\Omega}\right) &\geq&  -2\nu \label{adcon3b}
\end{eqnarray}
to hold in $r^\star \leq r^\star_{cl}$. 
This is again possible by the properties of
the functions $\alpha$ and $\beta$ chosen in section \ref{bootstrap}. 
\end{comment}
%
%
Comparing (\ref{hep2}) with (\ref{posY}) produces 
the first term in (\ref{contest2}).
\end{proof}
The results of this section are already sufficient to obtain a pointwise 
decay bound for the quantity $\frac{\zeta}{\nu}$. For reasons of presentation 
this is postponed to section \ref{pwEY} but the reader impatient to see the 
argument can turn to the latter section at this point.
\section{The vectorfield K} \label{vecKsec}
\subsection{The basic identity}
The vectorfield $K$ is defined as
\begin{equation}
K = \frac{2}{M} \left(u+a\right)^2 \partial_u + \frac{2}{M} \left(v-a\right)^2 \partial_v \, .
\end{equation}
It is the analogue of the Morawetz vector field in 
four dimensions. In particular, it is 
conformally Killing in five dimensional Minkowski space.\footnote{The
  vectorfield field has been shifted by $a$ for reasons which will
  become apparent later.} We note
\begin{equation}
K^u = \frac{2}{M}\left(u+a\right)^2 \textrm{ \ \ \ } K^v = \frac{2}{M}\left(v-a\right)^2 \textrm{ \ \ \ \ \ } K_u =
-\frac{\Omega^2}{M} \left(v-a\right)^2 \textrm{ \ \ \ } K_v =
-\frac{\Omega^2}{M} \left(u+a\right)^2  
\end{equation}
and
\begin{equation}
u = t - r^\star \textrm{ \ \ \ \ \ } v = t + r^\star \textrm{ \ \ \ \
  \ } \left(v-a\right)^2-\left(u+a\right)^2 = 4t\left(r^\star-a\right)
  \, .
\end{equation}
From (\ref{basicintegrand}) we compute the identity
\begin{eqnarray} \label{sptK}
M\left(-T_{\mu \nu} \pi^{\mu \nu} - \nabla^\beta T_{\beta \delta} K^\delta \right)=  \frac{3}{2r} \left(\nu \left(u+a\right)^2 + \lambda \left(v-a\right)^2\right) \Box B^2 \nonumber \\
+ 32\frac{B^2}{r^2} \left(t-\frac{1}{2r} \left(\nu \left(u+a\right)^2 + \lambda \left(v-a\right)^2
\right) + \frac{1}{4\Omega^2} \left(\left(\Omega^2\right)_{,u} \left(u+a\right)^2 +
\left(\Omega^2\right)_{,v} \left(v-a\right)^2\right)\right) \nonumber \\ 
+ \frac{\varphi_1\left(B\right)}{\Omega^2 r^2} \left(\left(\Omega^2\right)_{,v}
\left(v-a\right)^2 + \left(\Omega^2\right)_{,u} \left(u+a\right)^2
\right) + \frac{4t}{r^2} \varphi_1\left(B\right) \nonumber \\
+ \frac{3}{r^3} \left(\nu \left(u+a\right)^2 +
\lambda \left(v-a\right)^2 \right) \left(\varphi_1\left(B\right) + \varphi_2\left(B\right)\right) - \frac{2}{r^3} \varphi_1\left(B\right) \left(\nu
\left(u+a\right)^2 + \lambda \left(v-a\right)^2 \right) \nonumber \\
\end{eqnarray}
with $\varphi_1$ and $\varphi_2$ defined in (\ref{deltaB}). 
We shall apply the basic vectorfield identity 
in the region (cf.~Figure \ref{figkerror}) 
\begin{equation}
\mathcal{D}^K_{[t_0,\tilde{T}]} = {}^{\tilde{T}-r^\star_K}\mathcal{D}^{r^\star_K,u_0}_{[t_0,\tilde{T}]}
\end{equation}
for any $\tilde{T} < T$ producing the identity
\begin{equation}
\widehat{I}^K_B \left(\mathcal{D}^K_{[t_0,\tilde{T}]}\right)=
\widehat{F}^K_B\left(\tilde{T}\right) - \widehat{F}^K_B\left(t_0\right) +
\hat{H}^{K}_{u_H=\tilde{T}-r_K^\star} + 0
\end{equation}
where
\begin{eqnarray} \label{IXKB}
\widehat{I}^K_B\left(\mathcal{D}^K_{[t_0,\tilde{T}]}\right) = \int_{\mathcal{D}^K_{[t_0,\tilde{T}]}} \Big(-T_{\mu \nu} \pi^{\mu
  \nu}-\nabla^\beta T_{\beta \delta} K^\delta \Big) dVol \, ,
\end{eqnarray}
\begin{eqnarray}
\frac{\widehat{F}^K_B\left(t\right)}{2\pi^2} &=& \frac{1}{M} \int_{r^\star_K}^{t-u_0}
 \Bigg(\left(\partial_u B\right)^2 2\left(u+a\right)^2 + \left(\partial_v B\right)^2
 2\left(v-a\right)^2  \\ 
&& \phantom{XXXX} +\left(\left(u+a\right)^2+\left(v-a\right)^2\right) \frac{\Omega^2}{2r^2}
 \left(1-\frac{2}{3}\rho\right) \Bigg) r^3 dr^\star \nonumber \\
 &+& \frac{1}{M}\int_{t-r_{K}^\star}^{\tilde{T}-r^\star_K}\left[2\left(u+a\right)^2
 r^3 \left(\partial_{u}B\right)^2 + \frac{r \Omega^2}{2}
 \left(v-a\right)^2 \left(1-\frac{2}{3}\rho \right)
 \right]\left(u,t+r^{\star}_{K}\right)du \nonumber
\end{eqnarray}
and
\begin{eqnarray}
\frac{\hat{H}^{K}_{\tilde{T}-r_K^\star}}{2 \pi^2} =\frac{1}{M}
\int_{t_0+r_{K}^{\star}}^{\tilde{T}+r_{K}^{\star}}\left[2\left(v-a\right)^2
  r^3 \left(\partial_{v}B\right)^2 + \frac{r \Omega^2}{2}
  \left(u+a\right)^2 \left(1-\frac{2}{3}\rho \right)
  \right]\left(\tilde{T}-r^\star_K,v \right)dv \nonumber \, .
\end{eqnarray}
Note that the $J$-term vanishes in view of the assumption of compact
support. We are now going to define the renormalized quantities $I^K_B$ and
$E^K_B$ that arise from an application of Green's theorem to the $\Box
B^2$ term in the spacetime integral (\ref{IXKB}). 
The $D$ in the basic identity (\ref{basicgreen}) is here given
by (cf.~appendix \ref{reggree})
\begin{equation} \label{DKreg}
D = \frac{3}{2} \left(\frac{\nu \left(u+a\right)^2 + \lambda 
  \left(v-a\right)^2}{r}\right) \, .
\end{equation}
We compute
\begin{eqnarray}
\frac{3}{2} r^2 \Box \left(\frac{\nu \left(u+a\right)^2 + \lambda
  \left(v-a\right)^2}{r}\right) = t \left(-24
  r \frac{r_{,uv}}{\Omega^2} - 12 \frac{\lambda
  \nu}{\Omega^2} \right) \nonumber \\
+ t\left(\frac{r^\star-a}{r}\right) \Bigg(12r \frac{\lambda}{\Omega^2} r_{,uu}
+12 r r_{,uv} \frac{\lambda}{\Omega^2} 
-24\frac{r^2}{\Omega^2}\left(r_{,uv}\right)_{,v} 
+ 24\frac{\nu \lambda^2}{\Omega^2} \Bigg)
\nonumber  \\
+\left(u+a\right)^2 \Bigg(\left[\lambda+\nu\right]
  \Big(\frac{3r_{,uv}}{\Omega^2} + \frac{6\nu \lambda}{\Omega^2
  r}\Big) -\frac{6r}{\Omega^2} 
\left(\left(r_{,uv}\right)_{,v} + \left(r_{,uv}\right)_{,u}
  \right)\Bigg)
\nonumber  \\
+\left(v-a\right)^2 \Bigg(-\frac{3\lambda}{\Omega^2} r_{,uu} -
  \frac{3\nu}{\Omega^2} r_{,vv} \Bigg)
\end{eqnarray}
and define the bulk term
\begin{eqnarray} \label{Kbulk}
I^K_B \left(\mathcal{D}^K_{[t_0,\tilde{T}]}\right) = \frac{1}{M} \int \int \frac{1}{2} r^3 \Omega^2 du
dv \frac{B^2}{r^2} \Bigg\{ t \left[32 - 24 r \frac{r_{,uv}}{\Omega^2} - 12\frac{\lambda \nu}{\Omega^2} +4\frac{\varphi_1\left(B\right)}{B^2}\right]
\nonumber \\ 
+ t\left(\frac{r^{\star}-a}{r}\right) \Bigg[-64\lambda + 24\frac{\lambda^2
    \nu}{\Omega^2} + 12 \frac{r r_{,uv} \lambda}{\Omega^2} -
  24\frac{r^2}{\Omega^2}\left(r_{,uv}\right)_{,v} +12 r
  \frac{\lambda}{\Omega^2} r_{,uu} 
   \nonumber \\ 
 - 64 \frac{\Omega_{,u}}{\Omega}r +\frac{\varphi_1\left(B\right)}{B^2}
  \left(-8r \frac{\Omega_{,u}}{\Omega} + 4\lambda\right) +12 \lambda \frac{\varphi_2\left(B\right)}{B^2} \Bigg] \nonumber \\
+ \left(u+a\right)^2 \Bigg[\left(\lambda+\nu\right)
\Big(3\frac{r_{,uv}}{\Omega^2} +
\frac{6\nu \lambda}{\Omega^2 r} -\frac{16}{r} +\frac{1}{r}
\frac{\varphi_1\left(B\right)+3\varphi_2\left(B\right)}{B^2}\Big)
\phantom{XXXXXX} \nonumber \\ 
 -6\frac{r}{\Omega^2}
  \left(\left(r_{,uv}\right)_{,v} + \left(r_{,uv}\right)_{,u} \right)
   \Bigg] \nonumber \\
 + \left(v-a\right)^2 \left[\left(16+2\frac{\varphi_1\left(B\right)}{B^2}\right) \left(\frac{\Omega_{,v}}{\Omega}
  +\frac{\Omega_{,u}}{\Omega}\right) -
  3\frac{\lambda}{\Omega^2}r_{,uu} -
  3\frac{\nu}{\Omega^2}r_{,vv}  \right]\Bigg\} \, . \phantom{XXX} 
\end{eqnarray}
In order for the identity (cf.~equation (\ref{finidg}))
\begin{equation} \label{Kid}
I^K_B \left(\mathcal{D}^K_{[t_0,\tilde{T}]}\right) = F^K_B \left(\tilde{T}\right) -  F^K_B \left(t_0\right) + H^K_{\tilde{T}-r_K^\star}
\end{equation}
to hold, the boundary terms have to be
\begin{eqnarray} \label{Kbing}
\frac{F^K_B \left(t\right)}{2\pi^2} &=& \frac{\hat{F}^K_B \left(t\right)}{2\pi^2} + \frac{1}{M} \int_{r^\star_K}^{t-u_0} 2B\left(\partial_t B\right) \left(\frac{3}{2r} \left(\nu \left(u+a\right)^2 + \lambda \left(v-a\right)^2 \right)\right)r^3 \left(t,r^\star\right)dr^\star \nonumber \\
&&\phantom{X} -  \frac{1}{M} \int_{r^\star_K}^{t-u_0} B^2 \partial_t \left(\frac{3}{2r} \left(\nu \left(u+a\right)^2 + \lambda \left(v-a\right)^2 \right)\right)r^3 \left(t,r^\star\right)dr^\star \nonumber \\
&&\phantom{X}+ \frac{1}{M}\int_{t-r^\star_K}^{\tilde{T}-r^\star_K} 2B\left(\partial_u
B\right) \left(\frac{3}{2r} \left(\nu \left(u+a\right)^2 + \lambda
\left(v-a\right)^2 \right)\right) r^3 \left(u,t+r^\star_K\right) du
\nonumber \\
&& - \frac{1}{M}\int_{t-r^\star_K}^{\tilde{T}-r^\star_K} B^2 \partial_u \left(\frac{3}{2r} \left(\nu \left(u+a\right)^2 + \lambda \left(v-a\right)^2 \right)\right)r^3 \left(u,t+r^\star_K\right)du
\end{eqnarray}
and
\begin{eqnarray}
\frac{H^K_{\tilde{T}-r^\star_K}}{2\pi^2} &=&
\frac{\hat{H}^K_{\tilde{T}-r^\star_K}}{2\pi^2} -
\frac{1}{M}\int_{t_0+r^\star_K}^{\tilde{T}+r^\star_K} B^2 \partial_v \left(\frac{3}{2r}
\left(\nu \left(u+a\right)^2 + \lambda \left(v-a\right)^2
\right)\right)r^3 \left(\tilde{T}-r^\star_K,v\right) dv \nonumber \\
&+& \frac{1}{M}\int_{t_0+r^\star_K}^{\tilde{T}+r^\star_K} 2B\left(\partial_v
B\right) \left(\frac{3}{2r} \left(\nu \left(u+a\right)^2 + \lambda
\left(v-a\right)^2 \right)\right) r^3
\left(\tilde{T}-r^\star_K,v\right) dv \, .
\end{eqnarray}

\subsection{The Spacetime integral}
Let us turn to an analysis of the integral (\ref{Kbulk}). Besides 
the formulae (\ref{rvvruu}), (\ref{omegaevol}) and (\ref{revol}) the 
following identities will be useful
%\begin{eqnarray}
%r_{,vv} = \frac{2 \Omega_{,v}}{\Omega} \lambda - \frac{2}{r^2} \theta^2
%\end{eqnarray}
%\begin{eqnarray}
%\frac{r_{,uv}}{\Omega^2} = - \frac{\mu}{2r} - \frac{1}{3r} \left(\rho-
%\frac{3}{2} \right)
%\end{eqnarray}
%\begin{eqnarray}
%r_{,uu} = \frac{2 \Omega_{,u}}{\Omega} \nu - \frac{2}{r^2} \zeta^2
%\end{eqnarray}
\begin{eqnarray}
\frac{\left(r_{,uv}\right)_{,v}}{\Omega^2} = &&- \frac{
  \Omega_{,v}}{\Omega} \frac{\mu}{r} + \frac{3 \lambda \mu}{2r^2}
  - \frac{1}{r^3} \left(\frac{\theta^2}{\kappa} + r \lambda
  \left(1-\frac{2}{3}\rho \right) \right) \nonumber \\ &&
- \frac{\Omega_{,v}}{\Omega}
  \frac{2}{3r} \left(\rho-\frac{3}{2}\right) + \frac{\lambda}{3r^2}\left(\rho-\frac{3}{2}\right) + \frac{4}{3 r^{\frac{5}{2}}} \theta
  \left(e^{-2B} - e^{-8B}\right) \, , \nonumber
\end{eqnarray}
\begin{eqnarray}
\frac{\left(r_{,uv}\right)_{,u}}{\Omega^2} = &&- \frac{
  \Omega_{,u}}{\Omega} \frac{\mu}{r} + \frac{3 \nu \mu}{2r^2}
  - \frac{1}{r^3} \left(-4\frac{\lambda}{\Omega^2}\zeta^2 + r \nu
  \left(1-\frac{2}{3}\rho\right) \right) \nonumber \\ &&
- \frac{\Omega_{,u}}{\Omega}
  \frac{2}{3r} \left(\rho-\frac{3}{2}\right) + \frac{\nu}{3r^2}
  \left(\rho-\frac{3}{2}\right) + \frac{4}{3 r^{\frac{5}{2}}} \zeta
  \left(e^{-2B} - e^{-8B}\right)  \, . \nonumber
\end{eqnarray}
%\begin{eqnarray}
%\partial_u \frac{\Omega_{,v}}{\Omega} = \frac{3 \Omega^2}{2r^4} m
%+ \frac{\Omega^2}{2r^2} \left(\rho-\frac{3}{2}\right) - 3
%\frac{\zeta \theta}{r^3} 
%\end{eqnarray}
The bulk integral can be written
\begin{equation} \label{finform}
I^K_B\left(\mathcal{D}^K_{[t_0,\tilde{T}]}\right) = I^K_{B,main}\left(\mathcal{D}^K_{[t_0,\tilde{T}]}\right)  + I^K_{B,error}\left(\mathcal{D}^K_{[t_0,\tilde{T}]}\right)
\end{equation}
with
\begin{eqnarray} \label{IKBmain}
I^K_{B,main} \left(\mathcal{D}^K_{[t_0,\tilde{T}]}\right)= \frac{1}{M} \int \int \frac{1}{2} r^3 \Omega^2 du dv \frac{B^2}{r^2} \Bigg\{ t
\left[35 + 9\mu + 4\frac{\varphi_1\left(B\right)}{B^2} + 8 \left(\rho-\frac{3}{2}\right) \right] 
\nonumber \\ 
+ t\left(\frac{r^{\star}-a}{r}\right) \Big[24\mu r
  \frac{\Omega_{,v}}{\Omega} - 64r  \frac{\Omega_{,u}}{\Omega} +
  \left(1-\mu\right) \left[-70\kappa - 36\kappa \mu -6r
    \frac{\Omega_{,u}}{\Omega} \right] + P\left(B\right)\Big]\Bigg\}
\end{eqnarray}
\begin{eqnarray} \label{IKBerror}
I^K_{B,error} \left(\mathcal{D}^K_{[t_0,\tilde{T}]}\right)=\frac{1}{M} \int \int \frac{1}{2} r^3
  \Omega^2 du dv \frac{B^2}{r^2}\Bigg\{ \frac{\left(u+a\right)^2}{2}
  \Bigg( Q\left(B\right) \nonumber \\
  + \left(\frac{\Omega_{,v}}{\Omega} +
  \frac{\Omega_{,u}}{\Omega}\right)
  \left[12\mu+ 8 \left(\rho-\frac{3}{2}\right) \right]\nonumber \\ 
 + \frac{\left(\lambda+\nu\right)}{r}
 \left[-35-18\mu-14\left(\rho-\frac{3}{2}\right)+2\frac{\varphi_1\left(B\right) + 3 \varphi_2 \left(B\right)}{B^2}\right] \Bigg)
\nonumber \\
+\left(v-a\right)^2 \Bigg(  \left[\frac{\Omega_{,v}}{\Omega} +
  \frac{\Omega_{,u}}{\Omega}\right] \left(16 +
  \frac{3}{2}\left(1-\mu\right) + 2\frac{\varphi_1\left(B\right)}{B^2}
  \right) + 6 \frac{\lambda \zeta^2}{r^2 \Omega^2}  + 6 \frac{\nu \theta^2}{r^2 \Omega^2} \Bigg) \Bigg\}
\end{eqnarray}
where
\begin{eqnarray}
P\left(B\right) = -8r
  \frac{\Omega_{,u}}{\Omega} \frac{\varphi_1\left(B\right)}{B^2}+
  4\kappa\left(1-\mu\right)
  \frac{\varphi_1\left(B\right)+3\varphi_2\left(B\right)}{B^2} +\frac{24\theta^2}{\kappa r} -24 \frac{\lambda \zeta^2}{\Omega^2 r} \nonumber \\
 - \frac{32}{\sqrt{r}}
\theta \left(e^{-2B}-e^{-8B}\right) + \left(\rho-\frac{3}{2}\right)
\left[-28 \kappa \left(1-\mu\right) + 16\frac{\Omega_{,v}}{\Omega} r 
  \right]
\end{eqnarray}
and
\begin{eqnarray}
Q\left(B\right) = \frac{12\theta^2}{r^2 \kappa} - \frac{48\lambda}{r^2
  \Omega^2} \zeta^2 - 16\left(e^{-2B}-e^{-8B}\right) \frac{\theta+\zeta}{r^\frac{3}{2}} \, .
\end{eqnarray}
Note that
\begin{equation}
|P \left(B\right) | \leq C\left(\epsilon\right)\frac{\sqrt{M}}{r}
\end{equation}
by the pointwise bounds of section \ref{pointwiseE}.

\subsubsection{Estimating $I^K_{B,main}$}
We start with the observation that $I^K_{B,main}$ has a good sign near
the horizon and near infinity:
\begin{lemma} \label{goodsignK}
One can find $\hat{R}^\star$ such that the integrand 
of $I^K_{B,main}$ is negative for $r^\star \geq \hat{R}^\star$. It is also negative for $r^\star \leq r^\star_{cl}$.
\end{lemma}
\begin{proof}
The second statement is a consequence of (\ref{condY3}) and the inequality 
$\frac{\Omega_{,u}}{\Omega} < 0$ which follows from Proposition \ref{omucent}. 
For large $r^\star$ on the other hand, we can expand the 
integrand of (\ref{IKBmain}) in powers of $\frac{1}{r}$ using the
results of section \ref{pointwiseE}:
\begin{equation}
|B| \leq C\left(\epsilon\right) \frac{\sqrt{M}}{r} \textrm {\ \ \ and \ \ \ } \kappa =
 \frac{1}{2} + \mathcal{O}\left(\frac{1}{r^2}\right) \textrm{\ \ \ and
 \ \ \ } 
\frac{r\Omega_{,u}}{\mu \Omega} \approx  -\frac{1}{2} \textrm {\ \ 
 and \ \ }  \frac{\Omega_{,v}}{\Omega} = \mathcal{O}\left(\frac{1}{r^2}\right) \nonumber
\end{equation}
and (cf.~identification (\ref{rstarr}))
\begin{equation}
\frac{r^\star}{r} \sim 1-\frac{\tilde{p} \pm \epsilon}{r} +
\mathcal{O}\left(\frac{1}{r^2}\right) \textrm { \ \ \ where \ \ \ } \tilde{p} = \sqrt{\frac{M_A}{2}} p 
=  \sqrt{\frac{M_A}{2}} \left[2\sqrt{2} + \log \left(\frac{2-\sqrt{2}}{2+\sqrt{2}}\right)\right] \nonumber
\end{equation}
to find
\begin{equation}
I^K_{B,main} \left(\mathcal{D}^{K}_{[t_0,\tilde{T}]}\right) = \frac{1}{M} \int \int \frac{1}{2} r^3 \Omega^2
du dv \frac{B^2}{r^2} t \Bigg\{ \frac{35\left(a+\tilde{p}\right) +
  \epsilon}{r} + \mathcal{O}\left(\frac{1}{r^2}\right)\Bigg\} \nonumber \, .
\end{equation}
With the chosen centre $a$ of the $K$ vector field ($a=-\tilde{p}-1$ by 
equation (\ref{pdef})), the integrand will be 
negative in $r^\star \geq \hat{R}^\star$ for some suitably chosen
$\hat{R}^\star$.\footnote{Note that the rest terms are all controlled 
by $\frac{C\left(r_{cl},c\right)}{r^2}$.}
\end{proof}

{\bf Remark:} In particular, we will choose $t_0$ so large 
that $\hat{R}^\star \leq \frac{9}{10}t_0$ holds. \\ 

The idea in estimating the spacetime integral 
$I^K_{B,main}\left(\mathcal{D}^K_{[t_0,\tilde{T}]}\right)$ is 
to decompose the region of integration into dyadic pieces (cf.~footnote \ref{dyadicexplain})
\begin{equation}
\mathcal{D}^{K}_{[t_0,\tilde{T}]} = \sum_{j=0}^{N-1}
\hat{\mathcal{D}}^{K}_{[t_j,t_{j+1}]} \textrm{ \ \ \ with $t_N=\tilde{T}$}
\end{equation}
\begin{equation}
\hat{\mathcal{D}}^{K}_{[t_j,t_{j+1}]} =
\mathcal{D}^{K}_{[t_0,\tilde{T}]} \cap \{ t_j \leq t \leq t_{j+1} \}
\, .
\end{equation}
For each piece
$\hat{\mathcal{D}}^{K}_{[t_j,t_{j+1}]}$ we can control the 
bulk term $I^K_{B,main}$ by the bulk term $\bar{I}^X_B$ 
losing a power of $t$:
\begin{proposition} \label{IKBmaincontrol}
In the region $\mathcal{D}^K_{[t_0,\tilde{T}]}$ we have, for each dyadic piece
\begin{eqnarray} \label{Kconxi}
I_{B,main}^K \left(\hat{\mathcal{D}}^{K}_{[t_j,t_{j+1}]} \right) &\leq& \frac{1}{\sqrt{M}} C\left(r^\star_{cl},\hat{R}^\star,\sigma\right) t_{j+1} \bar{I}^X_B
\left(\mathcal{B}_{[t_j,t_{j+1}]}^{[r^\star_{cl},\hat{R}^\star]} \right)
\nonumber \\
&\leq& \frac{1}{\sqrt{M}} C\left(r^\star_{cl},\hat{R}^\star, \sigma\right) t_{j+1} \bar{I}^X_B
\left({}^{u=\tilde{T}-r^\star_{cl}}\mathcal{D}_{[t_j,t_{j+1}]}^{r^\star_{cl},
  u=\frac{1}{11}t_{j+1}]} \right) \, .
\end{eqnarray}
\end{proposition}
\begin{proof}
By the previous Lemma it suffices to show (\ref{Kconxi}) 
with $\hat{\mathcal{D}}^{K}_{[t_1,t_2]}$ replaced by
$\mathcal{B}_{[t_1,t_2]}^{[r^\star_{cl}, \hat{R}^\star]}$
because the integrand of $I_{B,main}^K$ admits a good sign to the left 
of $r^\star_{cl}$ and to the right of $\hat{R}^\star$. 
For the compact $r^\star$-interval the first part of inequality
(\ref{Kconxi}) follows from Proposition \ref{B2bound}, the second
from $\hat{R}^\star \leq \frac{9}{10}t_0$ and the positivity of $\bar{I}^X_B$. 
\end{proof}

\subsubsection{Estimating $I^K_{B,error}$}
In this subsection we are going to show that the contribution of the 
integral $I^K_{B,error}$ can be made as small as we may wish for late times 
by choosing $r^\star_K$ sufficiently close to the horizon and the initial data sufficiently small. 
To achieve this we shall split the integration into different
regions $\mathcal{U}$,$\mathcal{V}$, $\mathcal{W}$ and
$\mathcal{Z}$ defined as follows 
\begin{equation}
\mathcal{U} = \mathcal{D}^{K}_{[t_0,\tilde{T}]} \cap \{ r^\star \leq
 r^\star_K \} \cap \{ u \geq 2v -4r^\star_K \} \, ,
\end{equation}
\begin{equation}
\mathcal{V} = \mathcal{D}^{K}_{[t_0,\tilde{T}]} \cap \{ r^\star \leq
 r^\star_K \} \cap \{ u \leq 2v -4r^\star_K \} \, ,
\end{equation}
\begin{equation}
\mathcal{W} = \mathcal{D}^{K}_{[t_0,\tilde{T}]} \cap \{ r^\star_K \leq
  r^\star \leq \frac{9}{10}t \} \, ,
\end{equation}
\begin{equation}
\mathcal{Z} = \mathcal{D}^{K}_{[t_0,\tilde{T}]} \cap \{ r^\star \geq
\frac{9}{10}t \} \, .
\end{equation}
\begin{figure} 
\[
\input{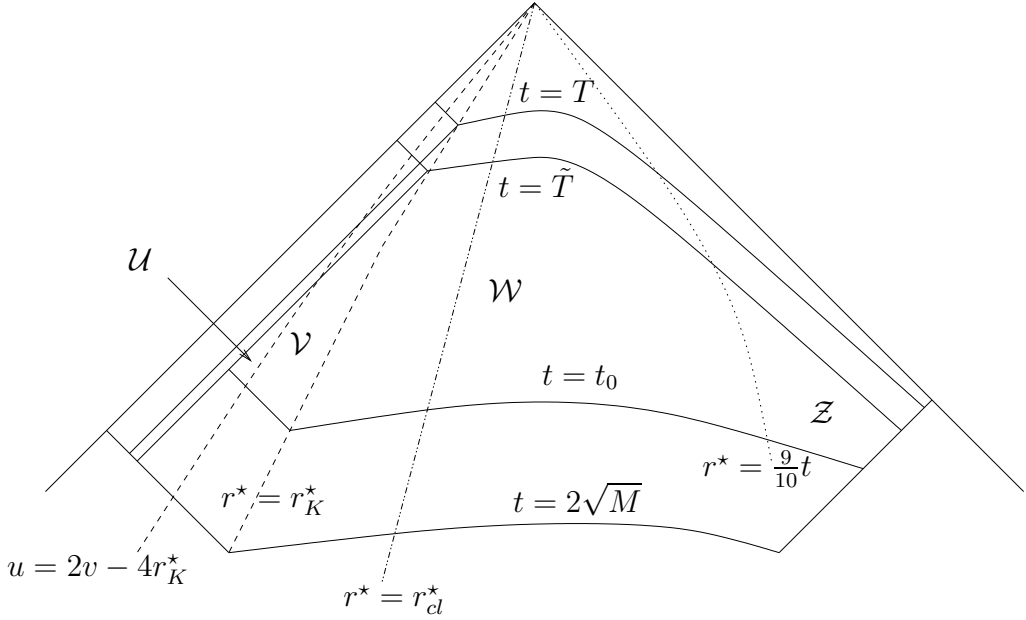} 
\]
\caption{Different regions to control the error-terms of $K$.} \label{figkerror}
\end{figure}
An immediate observation is
\begin{lemma} \label{kapgam}
In all regions we have
\begin{equation}
\kappa + \gamma \geq 0 \, .
\end{equation}
\end{lemma}
\begin{proof}
This is a consequence of the choice of coordinates and the monotonicity of
$\kappa$ in $u$ and of $\gamma$ in $v$ manifest in equations (\ref{kapevol})
and (\ref{gammaevol}).
\end{proof}
The next lemma establishes appropriate bounds to control the
error-terms of $I^K_{B,error}$ in equation (\ref{IKBerror}).
\begin{lemma} \label{bootcons}
Recall that by Propositions \ref{pointBcent}, \ref{pointthcent} and
\ref{pointBthhoz} the bound
\begin{equation} \label{aregW}
|B| + M^{-\frac{1}{4}} |\theta| \leq C\left(r^\star_{cl}\right) \frac{\sqrt{M}}{t} 
\end{equation}
holds in $\mathcal{W}$ and 
\begin{equation} \label{aregUV}
|B| + M^{-\frac{1}{4}} |\theta| \leq C\left(r^\star_{cl}\right) \frac{\sqrt{M}}{v} 
\end{equation}
holds in $\mathcal{U} \cup \mathcal{V}$. Assume also
\begin{equation} \label{zetnuas}
\Big|\frac{\zeta}{\nu}\Big| \leq C\left(r^\star_{cl}\right)\frac{M^\frac{3}{4}}{t}
\textrm{\ \ in $\mathcal{W}$ \ \ \ \ \ \ and \ \ \ \ \ \ } \Big|\frac{\zeta}{\nu}\Big|
\leq C\left(r^\star_{cl}\right)\frac{M^\frac{3}{4}}{v} \textrm{\ \ \ in \ $\mathcal{U} \cup \mathcal{V}$.}
\end{equation}
Then we have the following estimates
\begin{itemize} 
\item In region $\mathcal{W}$
\begin{equation} \label{Wbound2}
|Q\left(B\right)| + \Big|6 \frac{\lambda \zeta^2}{r^2 \Omega^2} \Big| +  \Big|6 \frac{\nu \theta^2}{r^2 \Omega^2} \Big| \leq C\left(r^\star_K,c\right) \frac{M^\frac{7}{4}}{r^\frac{3}{2} t^2} \, ,
\end{equation}
\begin{equation} \label{Wbound1}
\Big|\frac{\Omega_{,v}}{\Omega} + \frac{\Omega_{,u}}{\Omega}\Big| \leq C\left(r^\star_K,c\right) \frac{\sqrt{M}}{t^2} \, .
\end{equation}
\item In region $\mathcal{V}$
\begin{equation} \label{Vbound1}
|Q\left(B\right)| +  \Big|6 \frac{\lambda \zeta^2}{r^2 \Omega^2} \Big|
 +  \Big|6 \frac{\nu \theta^2}{r^2 \Omega^2} \Big| \leq \frac{C\left(r^\star_{cl},c\right)}{\sqrt{M} u^2}
 \, ,
\end{equation} 
\begin{equation} \label{Vbound2}
\frac{\Omega_{,v}}{\Omega} + \frac{\Omega_{,u}}{\Omega} +
\frac{C\left(r^\star_{cl},c\right)\sqrt{M}}{u^2} \geq 0 \, .
\end{equation}
\item In region $\mathcal{U}$
\begin{equation} \label{expdec}
-\nu \leq d_1 \exp \left(-\frac{d_2}{2\sqrt{M}}u\right)
\end{equation}
for positive constants $d_1>0$, $d_2>0$.
\item{In region $\mathcal{Z}$}
%\begin{equation} \label{Zbound1}
%|Q\left(B\right)| +  \Big|6 \frac{\lambda \zeta^2}{r^2 \Omega^2} \Big| +  \Big|6 \frac{\nu \theta^2}{r^2 \Omega^2} \Big| \leq \frac{\epsilon}{r^3}
%\end{equation}
\begin{equation} \label{Zbound2}
\Big|\frac{\Omega_{,v}}{\Omega} + \frac{\Omega_{,u}}{\Omega}\Big| \leq C\left(r^\star_{cl},c\right)\frac{\sqrt{M}}{r^2} \, .
\end{equation}
\end{itemize}
\end{lemma}

\begin{proof}
$\phantom{XXXX}$
\paragraph{The region $\mathcal{W}$} The bound (\ref{Wbound1}) is the
statement of Proposition \ref{omuomvcent}. 
The bound (\ref{Wbound2}) follows directly from the
decay properties (\ref{aregW}) and (\ref{aregUV}).
\paragraph{The region $\mathcal{V}$.} 
From Proposition \ref{omuomvcent} we derive the bound
\begin{equation} \label{omv}
\Big| \frac{\Omega_{,v}}{\Omega} \left(u,v\right) - \frac{m}{r^3}
\left(u,v\right)\Big| \leq  C\left(r^\star_{cl},c\right) \frac{\sqrt{M}}{u^2}
\end{equation}
by observing that $u$ is like $v$ in the region $\mathcal{V}$. \\ \\
The quantity $\frac{\Omega_{,u}}{\Omega}$ is obtained by integrating
(\ref{omegaevol}) written as
\begin{equation}
\partial_v \left(\frac{\Omega_{,u}}{\Omega}\right) = \gamma \left(6
m \frac{\lambda}{r^4} + \frac{2\lambda}{r^2}
\left(\rho-\frac{3}{2}\right) + 3 \frac{\theta}{\kappa}
\frac{\zeta}{\nu} \frac{\lambda}{r^3} \right)
\end{equation}
from the set $L=\{ \{t=T\} \cap \{r^\star_K \leq r^\star \leq
r^\star_{cl} \} \} \cup \{r^\star=r^\star_{cl}\}$
downwards. On $L$ itself we have by Proposition \ref{omuomvcent}
\begin{equation}
\Big| \frac{\Omega_{,u}}{\Omega} \left(u,v\right) - \frac{m}{r^3}
\left(u,v\right)\Big| \leq  C\left(r^\star_{cl},c\right) \frac{\sqrt{M}}{u^2} \, .
\end{equation}
Since $\gamma \leq \frac{1}{2}$ in $\mathcal{V}$ by monotonicity and
moreover $\frac{\Omega_{,u}}{\Omega}$ will always be negative, we can derive
the bound  
\begin{eqnarray} \label{omu}
\frac{\Omega_{,u}}{\Omega} \left(u,v\right) &=&
\frac{\Omega_{,u}}{\Omega}\left(u,v_R\right) - \int_v^{v_R} \gamma  \left(6
m \frac{\lambda}{r^4} + \frac{2\lambda}{r^2}
\left(\rho-\frac{3}{2}\right) + 3 \frac{\theta}{\kappa}
\frac{\zeta}{\nu} \frac{\lambda}{r^3} \right) \left(u,\bar{v}\right)
d\bar{v} \nonumber \\ 
&\geq& -\frac{m}{r^3}
\left(u,v_R\right) -  C\left(r^\star_{cl},c\right) \frac{\sqrt{M}}{u^2} + m\left(u,v_R\right)
\left(\frac{1}{r^3\left(u,v_R\right)} -
\frac{1}{r^3\left(u,v\right)}\right) \nonumber \\ 
&\geq& -\frac{m}{r^3}
\left(u,v\right) -  C\left(r^\star_{cl},c\right) \frac{\sqrt{M}}{u^2}
\end{eqnarray}
in $\mathcal{V}$, where in the last step we used that in $\mathcal{V}$ 
the Hawking mass decays like $\frac{1}{u^2}$. 
Putting the bounds (\ref{omu}) and (\ref{omv}) 
together yields (\ref{Vbound2}). The bound (\ref{Vbound1}) on
$Q\left(B\right)$ follows directly from the pointwise bound (\ref{aregUV}).
\paragraph{The region $\mathcal{U}$}
Integrating the quantity $\nu=\gamma \left(1-\mu\right)$ from the spacelike
$t=T$ curve downwards to any point in the region $\mathcal{U}$ we obtain
\begin{eqnarray}
-\nu\left(u,v\right) = \frac{1}{2} \left(1-\mu\right) \left(u_T,v_T\right) \exp \left(-\int_v^{v_T}
\tilde{f}\left(u,v\right) \left(u_T,\bar{v}\right) d\bar{v} \right) \nonumber \\
= \frac{1}{2} \left(1-\mu\right) \left(u_T,v_T\right) \exp \left(-\int_{v_{r^\star_K}}^{v_T} \tilde{f}\left(u,v\right)
\left(u,\bar{v}\right) d\bar{v} - \int_{v}^{v_{r^\star_K}}
\tilde{f}\left(u,v\right) \left(u,\bar{v}\right) d\bar{v} \right) 
\end{eqnarray}
with
\begin{equation}
\tilde{f}\left(u,v\right) = \frac{4\kappa}{r^3} m + \frac{4}{3}
\frac{\kappa}{r} \left(\rho-\frac{3}{2}\right) \, .
\end{equation}
In both regions $\mathcal{V}$ and $\mathcal{U}$ the quantity $\tilde{f}$ is
clearly positive, bounded below by some $d_2>0$. 
We can estimate, for a point $(u,v)$ in region $U$
\begin{eqnarray}
-\nu\left(u,v\right) \leq \frac{1}{2} \left(1-\mu\right) \left(u_T,v_T\right) \exp \left( - \int_{v}^{v_{r^\star_K}}
\tilde{f}\left(u,v\right) \left(u,\bar{v}\right) d\bar{v} \right) \nonumber \\ \leq \frac{1}{2} \left(1-\mu\right) \left(u_T,v_T\right) \exp\left(-d_2
 \left(v_{r^\star_K}-v\right)\right) \leq \frac{1}{2}
 \left(1-\mu\right) \left(u_T,v_T\right) \exp\left(-\frac{d_2}{2\sqrt{M}} u
 \right) \, .
\end{eqnarray}
Here we have used that $v \leq v_{r^\star_K} - \frac{1}{2}u$ by definition of
the region $\mathcal{U}$. 
\paragraph{The region $\mathcal{Z}$} 
The estimate is the statement of Corollary (\ref{omuomvdecr}). 
\end{proof}
With the necessary bounds in place we can prove the following
\begin{proposition} \label{IKBercontrol}
Assume (\ref{zetnuas}) holds. Then the error-term $I_{B,error}^K$ satisfies
\begin{eqnarray}
I_{B,error}^K \left(\mathcal{D}^{r^\star_K}_{[t_0,\tilde{T}]} \right) &\leq&
 \frac{1}{\sqrt{M}} C\left(r_{cl}^\star,\sigma\right) \sum_{j=0}^{N-1} t_{j+1} \left[ \bar{I}^X_B
\left({}^{u=t_{j+1}-r^\star_{cl}}\mathcal{D}^{r^\star_{cl},u=\frac{1}{11}t_{j+1}}_{[t_j,t_{j+1}]}\right)\right]
 \nonumber \\  &+& M\tilde{\epsilon}\left(r^\star_K\right) 
+ M \tilde{\delta} \left(t_0\right)
\end{eqnarray}
with
\begin{equation}
\lim_{r^\star_K \rightarrow -\infty}
 \tilde{\epsilon}\left(r^\star_K\right) = 0
\textrm{ \ \ \  as well as \ \ \ } \lim_{t_0
  \rightarrow \infty} \tilde{\delta} \left(t_0\right) = 0 \, .
\end{equation}
\end{proposition}
\begin{proof}
By Lemma \ref{kapgam} we have $\lambda +\nu \geq 0$. Hence the term
multiplying $\left(\lambda+\nu\right)$ in $I^K_{B,error}$ has a good 
(negative) sign in all regions and can be ignored. 
For the other two terms we look at the different regions:
\paragraph{Region $\mathcal{W}$:}
Note that $u$ and $v$ are controlled by $t$ in this region. 
We insert (\ref{Wbound1}) and (\ref{Wbound2}) into 
the integral $I^K_{B,error}$. The resulting term, which has to be
controlled is
\begin{equation} \label{borp}
\sqrt{M} C \int_{\mathcal{W}} du \ dv \ \Omega^2 r^3 \frac{B^2}{r^2} \, .
\end{equation}
We split the region of integration into $\mathcal{W}^1= \mathcal{W} \cap \{
r^\star \geq r^\star_{cl} \}$ and $\mathcal{W}^2= \mathcal{W} \cap \{
r^\star \leq r^\star_{cl} \}$. 
The region $\mathcal{W}^1$ is partitioned 
into dyadic slices as was the bulk term:
\begin{equation}
I^K_{B,error} \left(\mathcal{W}^1\right) = \sum_{j=0}^{N-1}
I^K_{B,error} \left(\mathcal{W}^1_{[t_j,t_{j+1}]}\right) \, .
\end{equation} 
We can control each dyadic tube by $\bar{I}^X_B$ losing a power of $t$
(arising from a missing power of $r$ in (\ref{borp}))
\begin{eqnarray}
I^K_{B,error} \left(\mathcal{W}^1_{[t_j,t_{j+1}]}\right)  &\leq&
 \frac{1}{\sqrt{M}} C\left(r_{cl}^\star,\sigma\right) \sum_{j=0}^{N-1} t_{j+1} \left[ \bar{I}^X_B
\left(\mathcal{W}^1_{[t_j,t_{j+1}]}\right) \right] \nonumber \\
&\leq& \frac{1}{\sqrt{M}} C\left(r_{cl}^\star,\sigma\right) \sum_{j=0}^{N-1} t_{j+1} \left[ \bar{I}^X_B
\left({}^{u=t_{j+1}-r^\star_{cl}}\mathcal{D}^{r^\star_{cl},u=\frac{1}{11}t_{j+1}}_{[t_j,t_{j+1}]}\right) \right] \nonumber
\end{eqnarray}
%By Corollary ... 
%\begin{equation}
%\bar{I}^X_B
%\left({}^{u=t_{j+1}-r^\star_{cl}}\mathcal{D}^{r^\star_{cl},u=\frac{1}{11}t_{j+1}}_{[t_j,t_{j+1}]}\right)
%\leq \frac{C}{t_{j+1}^2}
%\end{equation}
%Summing up we obtain
%\begin{equation}
%I^K_{B,error} \left(\mathcal{W}^1\right) \leq \frac{C}{\sqrt{t_0}}
%\sum_{j=0}^{N-1} \frac{1}{\sqrt{t_{j+1}}} \leq \frac{C}{\sqrt{t_0}}
%\sum_{j=0}^{N-1} \frac{1}{\sqrt{1.1^n \ t_0}} \leq \frac{C}{t_0}
%\end{equation} 
In the region $\mathcal{W}^2$ we can ignore the factors of $r$. It
suffices to estimate $B^2 \leq \frac{C\left(r^\star_K,c\right)}{t^2}$ from 
Proposition \ref{pointBcent} and hence
\begin{equation}
\int_{\mathcal{W}^2} dt \ dr^\star \ \Omega^2 r^3 \frac{B^2}{r^2} \leq
\int_{t_0}^{\tilde{T}} dt \frac{C\left(r^\star_K,c\right)}{t^2}
\int_{r^\star_{cl}}^{r^\star_K} dr^\star \frac{4 \kappa \gamma}{\kappa
  + \gamma} \frac{\partial r}{\partial r^\star} \leq C\left(r^\star_K,c\right) \frac{M^2}{t_0} \, .
\end{equation}
\paragraph{Region $\mathcal{V}$:}
In this region $u$ is like $v$ at late
times. We insert the bound (\ref{Vbound1}) into $I^K_{B,error}$ and
estimate the resulting term ($v_0 = t_0 + r^\star_{cl}$, $\tilde{V} =
\tilde{T} + r^\star_{cl}$)
\begin{equation}
\frac{1}{\sqrt{M}} C \int_{\mathcal{V}} dv \ du \Omega^2 B^2 \leq \int_{v_0}^{\tilde{V}}
\frac{1}{\sqrt{M}} \frac{C}{v^2} \int_{t_0-r^\star_K}^{2v-4r^\star_K} du \left(-4 \kappa
\nu\right) \leq C \frac{M}{v_0} \, .
\end{equation}
For the remaining terms, i.e.~those containing a
$\left(\frac{\Omega_{,v}}{\Omega}+\frac{\Omega_{,v}}{\Omega}\right)$-factor, 
we insert the one-sided bound (\ref{Vbound2}) to control the error
term 
\begin{equation} \label{critV}
 \int_{\mathcal{V}} dt dr^\star r^3 \frac{B^2}{r^2} u^2 \left(\partial_t \Omega^2\right)
\end{equation}
for large times (in particular $t_0 + r_K^\star \geq 1$) as follows ($u_H =
\tilde{T}-r^\star_K$). First define
\begin{equation}
r^\star_+ = -\frac{t_0}{2}+\frac{3}{2}r^\star_K \textrm{ \ \ \ \ and \
  \ \ \ } r^\star_- = r^\star_K - \frac{\tilde{T}-r^\star_K}{4}  
\end{equation}
and
\begin{equation}
\bar{t}\left(r^\star\right) = \left\{ \begin{array}{ll}
t_0 + r^\star_K - r^\star & \textrm{ for $r^\star \geq r^\star_+$} \\
-3r^\star + 4r^\star_K & \textrm{ for $r^\star \leq r^\star_+$} \, .
\end{array} \right.
\end{equation}
Then the term (\ref{critV}) can be estimated using Proposition \ref{pointBthhoz}
\begin{eqnarray} %\label{error2}
\frac{1}{M} \int_{r^\star_-}^{r^\star_K} dr^\star
  \int_{\bar{t}\left(r^\star\right)}^{\tilde{T}-r_K^\star+r^\star} dt
  \ r^3
\frac{B^2}{r^2} u^2 \left(\partial_t \Omega^2 + C \Omega^2\frac{\sqrt{M}}{u^2}\right) 
\nonumber \\
\leq C_L C\left(c\right) \sqrt{M}  \int_{r^\star_-}^{r^\star_K} dr^\star
  \int_{\bar{t}\left(r^\star\right)}^{\tilde{T}-r_K^\star+r^\star} dt
  \ \left(\partial_t \Omega^2 + C \Omega^2\frac{\sqrt{M}}{u^2}\right)
\nonumber \\
\leq C_L C\left(c\right) \sqrt{M} \int_{r^\star_-}^{r^\star_K} dr^\star
 \Omega^2 \left(\tilde{T}-r_K^\star+r^\star,r^\star\right)  +
  C_L C\left(c\right) M
  \int_{t_0+r^\star_K}^{\tilde{T}+r^\star_K} dv
  \int_{v-2r^\star_K}^{\min \left(u_H, 2v-4r^\star_K\right)} du
\frac{\Omega^2}{u^2}
 \nonumber \\
\leq C_L C\left(c\right) \sqrt{M} \Big[r\left(\tilde{T},r_K^\star\right) -
  r \left(\tilde{T}-r^\star_K+r^\star_-, r^\star_-\right)\Big]  +
  C_L C\left(c\right) M
  \int_{t_0+r^\star_K}^{\tilde{T}+r^\star_K} dv \frac{1}{v^2}
  \int_{v-2r^\star_K}^{u_H} du \Omega^2 
\nonumber \\
\leq C_L C\left(c\right) M \cdot \epsilon\left(r^\star_K\right) +
   C_L C\left(c\right) M\frac{\sqrt{M}}{t_0}
  \epsilon\left(r^\star_K\right) \, . \nonumber
\end{eqnarray}
where in the first step we have used that the round bracket in the first line is positive. 
The constant $C\left(r^\star_{cl}\right)$ may have different values
in each line. We also used that
\begin{equation}
\partial_{r^\star} r = \lambda - \nu = \frac{1}{4} \Omega^2
\left(\frac{1}{\gamma} + \frac{1}{\kappa}\right)
\end{equation}
and therefore
\begin{equation}
\Omega^2 = 4 \partial_{r^\star}r \frac{\gamma \kappa}{\gamma+\kappa}
\leq 4 \partial_{r^\star}r \, 
\end{equation}
holds. In summary, smallness for this error-term arises from the smallness
of the $r$-difference between any two points in the 
region $r^\star \leq r^\star_K$. The crucial point is that 
only $C\left(r^\star_{cl}\right)$ enters the above estimate, such that
the $r$-difference can ``beat'' the constant.
\paragraph{Region $\mathcal{U}$:}
To control the error-terms in region $\mathcal{U}$
estimate the curly bracket of $I^K_{B,error}$ by some 
constant times $u^2$ and $B^2$ by something small (cf.~Corollary \ref{Bcor}). 
The resulting integral can be controlled via (\ref{expdec})as follows:
\begin{eqnarray}
|I^K_{B,error} \left(\mathcal{U}\right)| \leq \frac{1}{M}C\left(\epsilon\right)
\int_{t_0+r^\star_{K}}^{\frac{\tilde{T}}{2} + \frac{3}{2}r_K^\star} dv
\int_{2v-4r^\star_{K}}^{\tilde{T}-r_K^\star} du \left(-\nu\right) u^2
\left(u,v\right) \nonumber
 \\ \leq \frac{1}{M}C\left(\epsilon\right)\int_{t_0+r^\star_{K}}^{\frac{\tilde{T}}{2} + \frac{3}{2}r_K^\star} dv
\int_{2v-4r^\star_{K}}^{\tilde{T}-r_K^\star} du
\exp\left(-\frac{d}{2\sqrt{M}}u\right) u^2 
\leq M C\left(\epsilon\right) C e^{-\frac{d}{2\sqrt{M}} t_0} \, . \nonumber
\end{eqnarray}
\paragraph{The region $\mathcal{Z}$}
On the one hand, we have to establish smallness for
\begin{eqnarray} \label{regZc}
\frac{1}{M} \int_\mathcal{Z} du dv r^3  \Omega^2 \frac{B^2}{r^2} \left[\left(u+a\right)^2 +
\left(v-a\right)^2 \right] \left[Q\left(B\right) + 6\frac{\lambda
    \zeta^2}{r^2 \Omega^2} + \frac{3 \theta^2}{2 \kappa r^2} \right]
\nonumber \\ \leq C \sqrt{M} \int_\mathcal{Z} du dv \left[Q\left(B\right) + 6\frac{\lambda
    \zeta^2}{r^2 \Omega^2} + \frac{3 \theta^2}{2 \kappa r^2} \right]
\, ,
\end{eqnarray}
where we used that $r$ controls $v$ and $u$ in the region under consideration
and Proposition \ref{decayinr2}. From Proposition \ref{simpsmall} it
is apparent that the critical term to control is
\begin{equation}
 \int_\mathcal{Z} \frac{1}{2} \frac{1}{r^2} \zeta^2 du dv \leq
C \int dv \frac{1}{v^2} \left(\int \zeta^2 du\right) \leq
C\left(\epsilon\right) \frac{M}{t_0} \, .
\end{equation}
Namely, the remaining terms in the square bracket 
of (\ref{regZc}) all decay like $\frac{\epsilon}{r^3}$ 
by Proposition \ref{simpsmall} such that direct
integration will already lead to a smallness factor. 

The other critical term to control is
\begin{equation}
\frac{1}{M} \int du dv \frac{1}{2} r^3 \Omega^2 \frac{B^2}{r^2}
\left(\left(u+a\right)^2 + \left(v-a\right)^2\right) \left[\frac{\Omega_{,v}}{\Omega} +\frac{\Omega_{,u}}{\Omega} \right]
\end{equation}
which upon inserting (\ref{Zbound2}) and using the fact that $r$ controls $u$ and $v$ in the region under consideration reduces to controlling the term
\begin{equation} \label{linehep}
C\left(r^\star_{cl},c\right) \frac{1}{\sqrt{M}} \int_{t_0}^{\tilde{T}} dt \int_{\frac{9}{10}t}^{\tilde{T}-u_0} dr^\star B^2 r \left(-\nu\right) \, .
\end{equation}
Using that $r \sim t$ in region $\mathcal{Z}$ 
we can estimate (\ref{linehep}) by
\begin{eqnarray}
\leq C\left(r^\star_{cl},c\right) \frac{1}{\sqrt{M}} \int_{t_0}^{\tilde{T}} dt \frac{1}{t^2} \int_{\frac{9}{10}t}^{\tilde{T}-u_0} dr^\star B^2 r^3 \left(-\nu\right) \nonumber \\ \leq C\left(r^\star_{cl},c\right) \sqrt{M} \int_{t_0}^{\tilde{T}} dt \frac{1}{t^2} E^K_B \left(t\right) \leq C\left(r^\star_{cl},c\right) \frac{M^\frac{3}{2}}{t_0} 
\end{eqnarray}
where we have used bootstrap assumption (\ref{Kass}).
%
%
%
%
\begin{comment}  %% If \Omega_v had a sign, we could do this
 We are going to
  control the bigger term
\begin{eqnarray}
\frac{1}{M}\int_\mathcal{Z} du dv \frac{1}{2} r^3 \Omega^2 \frac{B^2}{r^2}
\left(\left(u+a\right)^2 + \left(v-a\right)^2\right)
\left[\frac{\Omega_{,v}}{\Omega} -\frac{\Omega_{,u}}{\Omega} \right] \nonumber \\
    \leq  C \sqrt{M} \int_{t_0}^{\tilde{T}} dt
    \int_{\frac{9}{10}t}^{\tilde{T}-u_0} dr^\star \partial_{r^\star}
    \Omega^2  \leq C\sqrt{M} \int_{t_0}^{\tilde{T}} dt \left(\Omega^2 \left(t,\tilde{T}-u_0
 \right)-\Omega^2 \left(t,\frac{9}{10}t
 \right) \right) \, .
\end{eqnarray}
Since in the region $\mathcal{Z}$ we have $r \sim t$ and hence 
the estimate
\begin{equation}
|\Omega^2 - \left(1-\mu\right)| \leq C\frac{M}{r^2} \leq C \frac{M}{t^2}
\end{equation}
we can conclude
\begin{eqnarray}
\int_\mathcal{Z} \frac{1}{2} r^3 \Omega^2 du dv \frac{B^2}{r^2}
\left(\left(u+a\right)^2 + \left(v-a\right)^2\right)
\left[\frac{\Omega_{,v}}{\Omega} -\frac{\Omega_{,u}}{\Omega} \right] \nonumber \\
    \leq  C \sqrt{M} \int_{t_0}^{\tilde{T}} dt
     \left(\frac{2m}{t^2} \left(t,\tilde{T}-u_0
 \right)- \frac{2m}{t^2} \left(t,\frac{9}{10}t \right)+ C\frac{M}{t^2}
  \right) \leq C \frac{M^\frac{3}{2}}{t_0} \, .
\end{eqnarray}
\end{comment}
%
%
%
%
\end{proof}
\subsection{The Boundary Terms}
We write the boundary-terms (\ref{Kbing}) as
\begin{equation} \label{rewr}
F^K_B \left(t\right) = F^K_{B,main} \left(t\right) 
+ F^K_{B,errorarc} \left(t\right) + F^K_{B,errorline} \left(t\right)
\end{equation}
where
\begin{eqnarray} \label{FKBmain}
\frac{F^K_{B,main} \left(t\right)}{2\pi^2} = \frac{1}{M}\int_{r_{K}^\star}^{t-u_0} \left( -12 t \frac{r^\star-a}{r} \nu  B \partial_t
B + 6B^2 \nu \frac{r^\star-a}{r} \right) r^3 dr^\star \nonumber \\
+\frac{1}{M}\int_{r_{K}^\star}^{t-u_0}
 \Bigg(\left(\partial_u B\right)^2 2\left(u+a\right)^2 + \left(\partial_v B\right)^2
 2\left(v-a\right)^2  \nonumber \\ + \left(\left(u+a\right)^2+\left(v-a\right)^2\right) \frac{\Omega^2}{2r^2}
 \left(1-\frac{2}{3}\rho\right) \Bigg) r^3 dr^\star  \nonumber \\
+ \frac{1}{M}\int_{t-r_{K}^\star}^{\tilde{T}-r_K^\star}\left[2\left(u+a\right)^2 r^3
   \left(\partial_{u}B\right)^2 + \frac{r \Omega^2}{2}
   \left(v-a\right)^2 \left(1-\frac{2}{3}\rho \right)
   \right]\left(u,t+r^{\star}_{cl}\right)du \, ,
\end{eqnarray}
\begin{eqnarray} \label{FKBerarc}
\frac{F^K_{B,errorarc} \left(t\right)}{2\pi^2} &=& - \frac{1}{M}\int_{r_{K}^\star}^{t-u_0} \frac{3}{2}B^2 \left(\frac{\left(u+a\right)^2}{r} \left(r_{,uu}
+ r_{,uv} \right) + \frac{\left(v-a\right)^2}{r} \left(r_{,vv}+r_{,uv} \right)
\right) r^3 dr^\star  \nonumber \\
&+& \frac{1}{M}\int_{r_{K}^\star}^{t-u_0} \left(+3B \partial_t B \frac{\left(v-a\right)^2}{r}
\left(\lambda+\nu\right) - 3 \frac{v-a}{r}B^2 \left(\lambda+\nu\right)
\right) r^3 dr^\star  \nonumber \\
&+& \frac{1}{M}\int_{r_{K}^\star}^{t-u_0} \frac{3}{2}B^2\left(\frac{\lambda+\nu}{r^2}
\left(\nu \left(u+a\right)^2 + \lambda \left(v-a\right)^2\right)
\right) r^3 dr^\star 
\end{eqnarray}
and
\begin{eqnarray}
\frac{F^K_{B,errorline} \left(t\right)}{2\pi^2} &=& \frac{1}{M}\int_{t-r_{K}^\star}^{\tilde{T}-r_K^\star} 2B\left(\partial_u B\right) \left(\frac{3}{2r} \left(\nu \left(u+a\right)^2 + \lambda \left(v-a\right)^2 \right)\right) r^3 \left(u,t+r^\star\right) du \nonumber \\
&-& \frac{1}{M}\int_{t-r_{K}^\star}^{\tilde{T}-r_K^\star} B^2 \partial_u
 \left(\frac{3}{2r} \left(\nu \left(u+a\right)^2 + \lambda \left(v-a\right)^2 \right)\right)r^3
 \left(u,t+r^\star\right)du 
\end{eqnarray}
and 
\begin{eqnarray} \label{fhoze}
\frac{H^K_{\tilde{T}-r^\star_K}}{2\pi^2} =
  \frac{1}{M}\int_{t_0+r_{K}^{\star}}^{\tilde{T}+r_{K}^{\star}}\left[2\left(v-a\right)^2
  r^3 \left(\partial_{v}B\right)^2 + \frac{r \Omega^2}{2}
  \left(u+a\right)^2 \left(1-\frac{2}{3}\rho \right)
  \right]\left(T-r^\star_K,v \right)dv \nonumber \\ 
 - \frac{1}{M}\int_{t_0+r^\star_K}^{\tilde{T}+r^\star_K} B^2 \partial_v \left(\frac{3}{2r}
\left(\nu \left(u+a\right)^2 + \lambda \left(v-a\right)^2
\right)\right)r^3 \left(T-r^\star_K,v\right) dv \nonumber \\ + \frac{1}{M}\int_{t_0+r^\star_K}^{\tilde{T}+r^\star_K} 2B\left(\partial_v B\right) \left(\frac{3}{2r} \left(\nu \left(u+a\right)^2 + \lambda \left(v-a\right)^2 \right)\right) r^3 \left(T-r^\star_K,v\right) dv 
\end{eqnarray}
Note that $F^K_{B,errorline}\left(\tilde{T}\right)=0$ and that the last term
of $F^K_{B,main}$ also vanishes for $t=\tilde{T}$.

\subsubsection{Estimating $F^K_{B,main} \left(t\right)$}
We are going to show that the boundary term $F^K_{B,main}
\left(t\right)$ comes with a sign. This is obviously the case 
for the integral in $u$, so it remains to establish non-negativity of
the spacelike integrals. Define
\begin{equation}
S = \left(v-a\right) \partial_v + \left(u+a\right) \partial_u \, ,
\end{equation}
\begin{equation}
\underline{S} = \left(v-a\right) \partial_v -  \left(u+a\right)
\partial_u \, .
\end{equation}
Note that
\begin{equation}
\left(SB\right)^2 + \left(\underline{S}B\right)^2 = 2 \left(u+a\right)^2
\left(\partial_u B\right)^2 + 2\left(v-a\right)^2 \left(\partial_v B\right)^2
\end{equation}
and
\begin{equation}
S = t\partial_t + \left(r^\star-a\right) \partial_{r^\star} \textrm{ \ \ \ \ \ \ } 
\underline{S} = t\partial_{r^\star} + \left(r^\star-a\right) \partial_t
\end{equation}
respectively
\begin{equation}
t\partial_t = S - \left(r^\star-a\right) \partial_{r^\star} \textrm{ \ \ \ \ \ \ } 
t\partial_t = \frac{t}{\left(r^\star-a\right)} \underline{S} - \frac{t^2}{\left(r^\star-a\right)}
\partial_{r^\star} \, .
\end{equation}
We can insert these expressions into the boundary term
(\ref{FKBmain}) and integrate the second term by parts using $S$:
\begin{eqnarray}
 \int_{r_{K}^\star}^{\tilde{T}-u_0} B t \partial_t B \frac{1}{r} \left(-2\nu\right) \left(r^\star-a\right)
r^3 dr^\star  \nonumber \\ =  \int_{r_{K}^\star}^{\tilde{T}-u_0} \frac{-2\nu}{r} \left(r^\star-a\right) B
\Big(\left(SB\right) - \left(r^\star-a\right) \partial_{r^\star} B \Big) r^3
dr^\star \nonumber \\  = \nu r^2 \left(r^\star-a\right)^2 B^2 \Big|^{r^\star=\tilde{T}-u_0}_{r^\star=r^\star_{K}}  
 + \int_{r_{K}^\star}^{\tilde{T}-u_0} \left(-2\nu\right) r^3
\Bigg[\frac{r^\star-a}{r} B \left(SB\right) \nonumber \\ + B^2 
 \left(\frac{r^\star-a}{r}  +
 \frac{\left(r^\star-a\right)^2}{r^2}\left[\kappa - \left(\nu + r
 \frac{\Omega_{,u}}{\Omega}\right) + P_1\left(B\right) \right] \right)\Bigg]dr^\star 
\end{eqnarray}
where 
\begin{equation}
P_1 \left(B\right) = \frac{\zeta^2}{r \nu} +
\frac{2}{3} \kappa \left(\rho-\frac{3}{2}\right) \approx C\left(\epsilon\right)
\end{equation}
is very small by Proposition \ref{simpsmall}. Note that the 
boundary term near infinity vanishes in view of the
assumption of compact support on the initial data and the domain of
dependence. The term at $r^\star = r^\star_{K}$ is manifestly
positive since $\nu < 0$ in the integration region. 

Alternatively, using $\underline{S}$, we can obtain
\begin{eqnarray}
\int_{r^\star_{K}}^{\tilde{T}-u_0} B t \partial_t B \frac{1}{r} \left(-2\nu\right) \left(r^\star-a\right)
r^3 dr^\star \nonumber \\ = \int_{r^\star_{K}}^{\tilde{T}-u_0} \frac{-2\nu}{r} \left(r^\star-a\right) B
\left(\frac{t}{\left(r^\star-a\right)} \underline{S}B - \frac{t^2}{\left(r^\star-a\right)}\partial_{r^\star}B\right) r^3
dr^\star =  
\nu r^2 t^2 B^2 \Big|^{r^\star=\tilde{T}-u_0}_{r^\star=r^\star_{K}} \nonumber \\+ \int_{r^\star_{K}}^{\tilde{T}-u_0} \left(-2\nu \right) r^3 \left(\frac{t}{r}
\left(\underline{S}B\right) B + \frac{t^2}{r^2}B^2 \left[\kappa - \left(\nu + r
 \frac{\Omega_{,u}}{\Omega}\right) + P_1\left(B\right) \right]
\right)dr^\star \, . \nonumber
\end{eqnarray}
Again the boundary term has a positive sign at $r^\star = r^\star_{K}$.

If we split the relevant term in (\ref{FKBmain}) into two equal
pieces and collect terms we can write 
\begin{eqnarray} %\label{lexp}
\frac{1}{2\pi^2} F^K_{B,main}\left(\tilde{T}\right) = \frac{1}{M} \int_{r^\star_{K}}^{\tilde{T}-u_0}
 \Bigg[\left(\partial_u B\right)^2 2\left(u+a\right)^2 + \left(\partial_v B\right)^2
 2\left(v-a\right)^2 \nonumber \\ + \left(\left(u+a\right)^2+\left(v-a\right)^2\right) \frac{\Omega^2}{2r^2}
 \left(1-\frac{2}{3}\rho\right) -12 t \frac{r^\star-a}{r} \nu  B \partial_t
B + 6B^2 \nu \frac{r^\star-a}{r} \Bigg] r^3 dr^\star \nonumber \\
  \geq
 \frac{1}{M}\int_{r^\star_{K}}^{\tilde{T}-u_0} dr^\star r^3 \Bigg[
  \left(1+2\nu\right) \left(\left(SB\right)^2 + \left(\underline{S}B \right)^2 \right)
  \nonumber \\
  + \left(-2\nu \right) \Bigg(\left(SB + \frac{3}{2} \frac{r^\star-a}{r} B
  \right)^2 + \frac{\left(r^\star-a\right)^2}{r^2}B^2 \left(\frac{23}{4} +
   \delta + 3\left[\kappa - \left(\nu + r
 \frac{\Omega_{,u}}{\Omega}\right) \right] \right)
  \Bigg) \nonumber \\
  + \left(-2\nu \right) \Bigg(
 \left(\underline{S}B + \frac{3}{2} \frac{t}{r} B
  \right)^2 +  \frac{t^2}{r^2}B^2 \left(\frac{23}{4} +
  \delta + 3\left[\kappa - \left(\nu + r
 \frac{\Omega_{,u}}{\Omega}\right) \right] \right) \Bigg) \Bigg] r^3 \nonumber
\end{eqnarray}
which is manifestly positive. Furthermore
\begin{equation} \label{borrowEKB}
\frac{1}{2\pi^2} F^K_{B,main}\left(\tilde{T}\right) \geq
  E^K_B\left(\tilde{T}\right) + \frac{4}{M}\int_{r^\star_{K}}^{\tilde{T}-u_0} \left(-\nu\right)
 \frac{B^2}{r^2} \left(t^2 + \left(r^\star-a\right)^2 \right)r^3 dr^\star
\end{equation}
with $E^K_B$ being the quantity appearing in bootstrap assumption \ref{boot2}.
\subsubsection{Estimating $F^K_{B,errorarc} \left(\tilde{T}\right)$}
\begin{lemma} \label{FKBerarccontrol}
Under the assumption (\ref{zetnuas}) we have that
\begin{equation} \label{sumuparc}
\frac{1}{2\pi^2} F^K_{B,main}\left(\tilde{T}\right) +  F^K_{B,errorarc} \left(\tilde{T}\right) \geq E^K_B\left(\tilde{T}\right)
  + M \hat{\epsilon} 
\end{equation}
with the $\hat{\epsilon}$ arising from the fact that $\tilde{T}$ is large.
\end{lemma}
\begin{proof}
We have
\begin{equation}
|\lambda + \nu| \leq C \left(r^\star_K,c\right) \frac{M}{\tilde{T}^2} \textrm{ \ \ \ \ \ for \ \
 $r^\star_K \leq r^\star
 \leq \frac{9}{10}\tilde{T}$}
\end{equation}
by Proposition \ref{kgom} and
\begin{equation}
|\lambda + \nu| \leq C \left(r^\star_K,c\right)\frac{M}{r^2} \textrm{ \ \ \ \ \ for \ \ $r^\star
 \geq \frac{9}{10}\tilde{T}$ }
\end{equation}
by Corollary \ref{kapgamdecr}. Recalling the bound (\ref{omuomvcent}) we can
estimate the expression
\begin{equation}
r_{,vv} + r_{,uv} = 2\lambda \frac{\Omega_{,v}}{\Omega} -
\frac{2}{r^2}\theta^2 + 2 \kappa \frac{\mu}{r} \nu + \frac{4\kappa \nu}{3r}\left(\rho-\frac{3}{2}\right) 
\end{equation}
by
\begin{equation}
| r_{,vv} + r_{,uv} | \leq \frac{\mu}{r}\left(\lambda+ \nu\right) + C \left(r^\star_K,c\right)
  \frac{\sqrt{M}}{\tilde{T}^2} \leq  C \left(r^\star_K,c\right)\frac{\sqrt{M}}{\tilde{T}^2} \,
\end{equation}
and similarly
\begin{equation}
| r_{,uu} + r_{,uv} | \leq C \left(r^\star_K,c\right)\frac{\sqrt{M}}{\tilde{T}^2} \, ,
\end{equation}
both in the region $r^\star_K \leq r^\star \leq \frac{9}{10}\tilde{T}$. 
Analogously, we obtain
\begin{equation}
| r_{,vv} + r_{,uv} | + | r_{,uu} + r_{,uv} | \leq \frac{C}{r^2}
  \textrm{ \ \ \ in the region $r^\star \geq \frac{9}{10}\tilde{T}$}.
\end{equation}
Inserting these estimates into (\ref{FKBerarc}) it becomes clear that
we have to establish smallness for the terms
\begin{equation} \label{remKco}
\int_{r^\star_K}^{\tilde{T}-u_J} \left[\sqrt{M} B^2 r^2  + M^\frac{3}{2} \frac{\theta^2}{r} +
M^\frac{3}{2} \frac{\zeta^2}{r} \right] \left(\tilde{T},r^\star\right) dr^\star \, .
\end{equation}
We split the integral into the region $r^\star_K \leq r^\star \leq
\frac{9}{10}\tilde{T}$ and the region $r^\star \geq
\frac{9}{10}\tilde{T}$. In the first region the derivative terms of 
(\ref{remKco})are manifestly controlled by the energy, decaying like
$\frac{1}{\tilde{T}^2}$ by Proposition \ref{decayfromK}. 
The $B^2$ term can be estimated as
\begin{equation}
\int B^2 r^2 dr^\star \leq C \tilde{T} \int B^2 r dr^\star \leq C
\tilde{T}
\left[m\left(\tilde{T},\frac{9}{10}\tilde{T}\right)-m\left(\tilde{T},r^\star_K\right)
  \right] \leq C \frac{M^2}{\tilde{T}} \,.
\end{equation}
In both cases smallness is obtained from the fact that $t_0$ is chosen
very large. In the region $r^\star \geq \frac{9}{10}t$ on the other
hand, the derivative terms in (\ref{remKco}) can be controlled by
pulling out the $\frac{1}{r}$ as a smallness factor and use the energy
estimate for the rest. For the $B^2$ term we have to borrow from 
the good last term of (\ref{borrowEKB}):
\begin{equation}
\int_{\frac{9}{10}\tilde{T}}^{\tilde{T}-u_J} B^2 r^2 dr^\star \leq
\frac{10}{9\tilde{T}} \int_{\frac{9}{10}\tilde{T}}^{\tilde{T}-u_0} B^2
r^3 dr^\star \, .
\end{equation}
Hence a tiny contribution from the last term of (\ref{borrowEKB})
will control this term and we finally arrive at (\ref{sumuparc}).
\end{proof}

\subsubsection{Estimating $H^K_{\tilde{T}-r^\star_K}$}
\begin{lemma} \label{HKBercontrol}
Under the assumption (\ref{zetnuas}) we have
\begin{equation} \label{Khozb}
-H^K_{u=T-r^\star_K} \leq M \tilde{\epsilon}\left(r^\star_K\right) 
\end{equation}
where $\tilde{\epsilon}\left(r^\star_K\right) \rightarrow 0$ for $r^\star_K
\rightarrow \infty$.
\end{lemma}
\begin{proof}
The first term of (\ref{fhoze}) is clearly positive and can be neglected. For
the other terms split the integration
$\mathcal{I}=\left[v_1,v_2\right] = \left[t_0+r^\star_K,
  \tilde{T}+r_K^\star\right]$ into the part which lies in
$\mathcal{U}$ (where we can use the estimate 
(\ref{expdec})) and the part in $\mathcal{V}$ (where we are going to
exploit the fact that the $r$-difference is small). See Figure
\ref{figkerror}. Following this line of thought we estimate 
the (negative of the) second term of (\ref{fhoze}) 
\begin{eqnarray} 
\int_{\mathcal{I}} B^2
\partial_v \left(\frac{3}{2r}
\left(\nu \left(u+a\right)^2 + \lambda \left(v-a\right)^2
\right)\right)r^3 \left(\tilde{T}-r^\star_K,v\right) dv \nonumber \\
\leq \int_{\mathcal{I}} \frac{3B^2}{2} \left(r\left(-\lambda\right) 
\nu \left(u+a\right)^2 + r^2 r_{,vv} \left(v-a\right)^2 + 2 r^2 \lambda
\left(v-a\right)\right) \left(\tilde{T}-r^\star_K,v\right) dv \nonumber \\
\leq \left[\int_{\mathcal{I} \cap \mathcal{U}} + \int_{\mathcal{I} \cap \mathcal{V}}\right]\frac{3B^2}{2} \left(r\left(-\lambda\right) 
\nu \left(u+a\right)^2 \right)\left(\tilde{T}-r^\star_K,v\right) dv \nonumber
\\ + M^2 C\left(r^\star_{cl},c\right) \int_{\mathcal{I}} \left[\frac{\lambda}{v_0} + r_{,vv}\right]
\left(\tilde{T}-r^\star_K,v\right) dv \nonumber \\
\leq C e^{-\frac{d}{2}u}\left(u+a\right)^2 \int_{\mathcal{I} \cap
  \mathcal{U}} B^2 r \lambda\left(\tilde{T}-r^\star_K,v\right) dv + C M^\frac{3}{2}
\int_{\mathcal{I} \cap  \mathcal{V}} \lambda \left(-\nu\right)
\left(\tilde{T}-r^{\star}_{K},v\right) dv \nonumber \\
+M^2 C\left(r^\star_{cl},c\right) \epsilon \left(r^\star_K\right) \leq
M^2 \tilde{\epsilon} \left(r^\star_K\right) \, ,
\end{eqnarray}
where we used that $r_{,uv} \leq 0$, that $u$ is like $v$ in
region $\mathcal{V}$, and the assumptions (\ref{aregUV}). 
For the (negative of the) third term we obtain ($C$ just depends on $r^\star_{cl}$)
\begin{eqnarray}
-\int_{t_0+r^\star_K}^{\tilde{T}+r^\star_K} 2B\left(\partial_v B\right)
\left(\frac{3}{2r} \left(\nu \left(u+a\right)^2 + \lambda
\left(v-a\right)^2 \right)\right) r^3 \left(\tilde{T}-r^\star_K,v\right) dv \nonumber \\
\leq \left[\int_{\mathcal{I} \cap \mathcal{U}} + \int_{\mathcal{I}
    \cap \mathcal{V}}\right]\frac{3}{2}r^2 \left(\frac{B^2}{r} +
r \left(B_{,v}\right)^2\right)  \left(-\nu\right)\left(u+a\right)^2 \left(\tilde{T}-r^\star_K,v\right) dv
 \nonumber \\ + \int_\mathcal{I}
\frac{3}{2}r^2 \left(\frac{B^2}{r} + r \left(B_{,v}\right)^2\right) \lambda
\left(v-a\right)^2 \left(\tilde{T}-r^\star_K,v\right) dv
\nonumber \\ \leq C e^{-\frac{d}{2}u}\left(u+a\right)^2 \frac{M^\frac{3}{2}}{v_0} +
C M^\frac{3}{2}\int_{\mathcal{I} \cap \mathcal{V}} \lambda
\left(\tilde{T}-r^\star_K,v\right) dv + C M^\frac{3}{2} \int_{\mathcal{I}} \lambda
\left(\tilde{T}-r^\star_K,v\right) dv \nonumber \\ \leq \tilde{\epsilon}\left(r^\star_K\right)
\end{eqnarray}
where we again used that $u$ is like $v$ in
region $\mathcal{V}$, and the inequality (\ref{aregUV}), as well as
the fact that $\left(-\nu\right) \leq \lambda$ (cf.~Lemma \ref{kapgam}).
These estimates together yield (\ref{Khozb}).
\end{proof}

\subsubsection{Estimating $F^K_{B,errorline} \left(t\right)$}
Clearly $F^K_{B,errorline} \left(\tilde{T}\right)=0$
since there is no upper null-boundary for the region in which we
apply $K$. Hence we only have to estimate 
$F^K_{B,errorline} \left(t_0\right)$. This is done in the same manner
as for the horizon term: Splitting the integral into a part lying in 
$\mathcal{V}$ and a part in $\mathcal{U}$, using the estimate (\ref{zetnuas})
in the former and applying (\ref{expdec}) in the latter region. 

\subsection{Summary}
We have shown the following
\begin{proposition} \label{Kbndtermcontrol}
Assume (\ref{zetnuas}) holds. It follows that
\begin{equation}
E^K_B \left(\tilde{T} \right) \leq  \frac{1}{\sqrt{M}} C\left(r^\star_{cl},\sigma\right) \sum_{j=0}^{N-1} t_{j+1}
\bar{I}^X_B\left({}^{u=t_{j+1}-r^\star_{cl}}\mathcal{D}_{[t_j,t_{j+1}]}^{r^\star_{cl},\frac{1}{11}t_{j+1}}\right) 
+ F^K_{B}\left(t_0\right) + M \hat{\epsilon}\left(r^\star_K, t_0\right) \nonumber
\end{equation}
and $\hat{\epsilon}$ can be made arbitrarily small by
choosing both $-r^\star_K$ and then $t_0$ sufficiently large.
\end{proposition}
\begin{proof}
Write (\ref{Kid}) as 
\begin{eqnarray}
F^K_{B,main} \left(\tilde{T}\right) + F^K_{B,errorarc}
\left(\tilde{T}\right) = I^K_{B,main} 
\left(\mathcal{D}_{[t_0,\tilde{T}]}^{r^\star_K,u_0}\right) + I^K_{B,error}
\left(\mathcal{D}_{[t_1,t_2]}^{r^\star_K,u_0}\right)\nonumber \\ + F^K_B
\left(t_0\right) - H^K_{u_H=\tilde{T}-r^\star_K}   
\end{eqnarray}
and apply the estimates of Lemmata  \ref{FKBerarccontrol} and
\ref{HKBercontrol},  as well as Propositions \ref{IKBmaincontrol} and \ref{IKBercontrol}. 
\end{proof}
\section{Closing the bootstrap} \label{closboot}
With the required estimates now in place we are in a position to prove the 
closed-part of Theorem \ref{setS}, i.e.~to improve the 
remaining bootstrap assumptions.\footnote{Recall that the first two have been improved already in Corollaries \ref{boot2impr} and \ref{boot1impr}.}

\begin{comment}
Before presenting the main argument we 
briefly summarize the different choices of constant $r^\star$-curves 
involved in the argument so far to clarify in advance where 
the smallness exploited in the bootstrap derives from.
%
%
%
\subsection{Sources of smallness}  \label{sosmall}
 The order is as follows
\begin{itemize}
\item Determine $r^\star_Y$ introduced in section \ref{Ysection}
  satisfying conditions (\ref{condY1}) and
  (\ref{condY2}).
\item Determine $r^\star_{ki}$ of Lemma \ref{goodsignK}.
\item Choose $r^\star_{cl}$ to the left of both $r^\star_Y$ and
  $r^\star_{ki}$, and subject to the conditions
  (\ref{adcon1a}-\ref{adcon3a}) and (\ref{adcon1b}-\ref{adcon3b}). 
\item Choose $r_K$ and (hence $r^\star_K$) very close to the horizon
  to get the smallness factor required for controlling the error-terms
  of $K$ (Lemmata \ref{FKBerarccontrol} and \ref{HKBercontrol}).
\item Choose $t_0$ large enough to make terms like
  $\frac{C\left(r_K\right)}{t_0}$ small providing another 
smallness factor, which has been exploited at many stages of the paper.
\item Choose the initial-data so small that Cauchy stability holds up
  to the slice $\Sigma_{t_0}$.
\end{itemize}
\end{comment}
%
%
%
%
%
%
%
We start with the observation that
the $X$-bulk-term decays.
\begin{proposition} 
We have
\begin{equation} \label{Xbulkdecay}
\bar{I}^X_{B}
  \left(\mathcal{D}_{t_i,t_{i+1}}^{r^\star_{cl},t_i-R^\star} \right)
  \leq \bar{I}^X_{B}
  \left(\mathcal{D}_{t_i,t_{i+1}}^{r^\star_{cl},\frac{1}{11}t_{i+1}}
  \right) \leq  M^2 \frac{C\left(r^\star_{cl}\right)}{t_i^2} \, .
\end{equation}
\end{proposition}
\begin{proof}
Apply Proposition \ref{Xenergycor} in combination with
 Proposition \ref{decayfromK} and the bootstrap assumption (\ref{intebound2}).
\end{proof}
With the help of the Propositions proven in section \ref{XconY} 
we can derive the pointwise bound (\ref{zetnuas}),
which was assumed for most of the Propositions established in 
section \ref{vecKsec}.\footnote{The reader is assured that none 
of the results of section \ref{vecKsec} will be used in the following 
subsection. The argument has been placed in this section 
because it is also used to improve the integral bound (\ref{intebound2}).}

\subsection{A pointwise estimate for $\frac{\zeta}{\nu}$ using $Y$} \label{pwEY}

\begin{proposition} \label{Step1}
In the region $\mathcal{A}\left(T\right) \cap \{ r^\star \leq
r^\star_{cl}  \} \cap \{v \geq t_0+r^\star_{cl}\}$ we have
\begin{equation} \label{hozBY}
\Big|\frac{\zeta}{\nu}\Big| \leq C\left(r^\star_{cl},c\right)
 \frac{M^\frac{3}{4}}{v} 
\end{equation}
and in $\mathcal{A}\left(T\right) \cap \{ r^\star \geq
r^\star_{cl}  \} \cap \{r^\star \leq \frac{9}{10}t \}$ the estimate
\begin{equation} \label{centBY}
\Big|\frac{\zeta}{\nu}\Big| \leq
 C\left(r^\star_{cl},c\right)
 \frac{M^\frac{3}{4}}{t}  \, .
\end{equation}
\end{proposition}
\begin{proof}
Starting from the slice $\Sigma_{t_0}$ (cf.~definition \ref{slicedef}) 
erect the characteristic rectangle to any
$\Sigma_{t}$, $t_0 \leq t \leq T$. By Cauchy stability (Proposition \ref{Cauchystab}), 
we have that
\begin{equation}
\frac{1}{M} F^Y_B \left(\left[u_1,u_{hoz}\right] \times \{ v_0 = t_0-r^\star_{cl}
\} \right) \leq \tilde{\delta}
\end{equation}
and hence an application of Proposition \ref{hozestc} together with
(\ref{Xbulkdecay}) immediately yields
\begin{equation}
\frac{1}{M} F^Y_B \left(\left[u_1,u_{hoz}\right] \times \{ v = t-r^\star_{cl} \}
\right) \leq \frac{11}{10}\tilde{\delta} + \tilde{\epsilon}
\end{equation}
for any $t \leq T$. This estimate and Proposition \ref{hozest} 
immediately imply the uniform estimate
\begin{equation}
\frac{1}{M} \tilde{I}^Y_B \left(\mathcal{R}_i \setminus \mathcal{T}_i\right) \leq \frac{11}{10}\tilde{\delta} + \tilde{\epsilon}
\end{equation}
for the region $\mathcal{R}_i \setminus \mathcal{T}_i$ of any
characteristic dyadic rectangle. Next we apply 
Proposition \ref{goodslice} to each dyadic rectangle to
find a slice $\hat{v}$ satisfying
\begin{equation}
\frac{1}{M} F^Y_B \left(\left[u_i,u_{hoz}\right] \times \{ \hat{v} \} \right) \leq
C \frac{\tilde{\epsilon}\sqrt{M}}{v_{i+1}-v_i} + C \frac{\sqrt{M}}{\left(v_i\right)^2}
\leq \frac{C\sqrt{M}}{t_i} \, .
\end{equation}
Proposition \ref{hozest} applied to the rectangle enclosed by the 
good slice in $[v_i,v_{i+1}]$ and $v=v_{i+1}$ yields, again using
(\ref{Xbulkdecay}) 
\begin{equation}
\frac{1}{M} F^Y_B \left(\left[u_i,u_{hoz}\right] \times \{ v_{i+1} \} \right) \leq
C \frac{\sqrt{M}}{v_{i+1}} \, .
\end{equation}
Having exported the better decay to all late slices in this fashion, 
we can erect the characteristic rectangle again and apply Proposition
 \ref{hozest}, which produces the uniform decay estimate
\begin{equation}
\frac{1}{M}I^Y_B \left(\mathcal{R}_i \setminus \mathcal{T}_i \right) \leq
C\frac{\sqrt{M}}{v_i} \, .
\end{equation}
One may repeat the procedure, i.e.~apply Proposition
\ref{goodslice} again, which now provides one with 
a good slice (with the $Y$-flux decaying like
$\frac{1}{\left(v_{i+1}\right)^2}$). After application 
of Proposition \ref{hozest} this leads to the decay
\begin{equation}
\frac{1}{M}F^Y_B \left(\left[u_i,u_{hoz}\right] \times \{ v_{i+1} \} \right) \leq
C\left(r^\star_{cl},c\right) \frac{M}{\left(v_{i+1}\right)^2}
\end{equation}
on any late slice $v_i$. Finally one may 
export the decay to any $v$-slice by choosing
appropriate regions:
\begin{equation} \label{Yubound}
\frac{1}{M}F^Y_B \left(\left[u\left(r^\star_{cl}\right),u_{hoz}\right] 
\times \{ v \} \right) \leq C\left(r^\star_{cl},c\right) \frac{M}{v^2} \, ,
\end{equation}
\begin{equation} \label{Yvbound}
\frac{1}{M}F^Y_B \left(u \times \left[v,\hat{v} \right] \right) \leq
C\left(r^\star_{cl},c\right) \frac{M}{v^2} \, ,
\end{equation}
\begin{equation}
\frac{1}{M}\tilde{I}^Y_B\left(\mathcal{R} \setminus \mathcal{T}\right) \leq C\left(r^\star_{cl},c\right) \frac{M}{v^2}
\end{equation}
everywhere. Note that we have assumed a better bound than
(\ref{Yubound}) in the bootstrap assumption (\ref{intebound2}),
however the bound (\ref{Yvbound}) is new and essential to derive the
pointwise bound for $\frac{\zeta}{\nu}$.
Namely, integrating (\ref{zetanuueq}) upwards in a characteristic
rectangle yields
\begin{eqnarray}
\frac{\zeta}{\nu} \left(u,v_{i+1}\right) = \frac{\zeta}{\nu}\left(u,v_i\right) e^{-\int_{v_i}^{v_{i+1}} \left[\frac{4\kappa}{r^3} m +
\frac{4\kappa}{3r}\left(\rho-\frac{3}{2}\right) \right]
  \left(u,\bar{v}\right) d\bar{v}} \nonumber \\ +\int_{v_i}^{v_{i+1}} e^{-\int_{\bar{v}}^v \left[\frac{4\kappa}{r^3} m +
\frac{4\kappa}{3r}\left(\rho-\frac{3}{2}\right) \right]
  \left(u,\hat{v}\right) d\hat{v}} \left[-\frac{3}{2} \frac{\theta}{r} -
\frac{4}{3} \frac{\kappa}{\sqrt{r}} \left(e^{-8B}-e^{-2B}\right)
\right]\left(u,\bar{v}\right) d\bar{v} 
\end{eqnarray}
and hence
\begin{eqnarray}
\Big|\frac{\zeta}{\nu} \left(u,v_{i+1}\right)\Big| \leq \Big|\frac{\zeta}{\nu} \left(u,v_{i}\right) \Big|
  e^{-.1\cdot d \cdot v_{i}}  \nonumber \\ +\frac{3}{2}\frac{1}{r_{min}} \sqrt{\int_{v_i}^{v_{i+1}} e^{-\int_{\bar{v}}^v  \left[\frac{3}{2} \frac{4\kappa}{r^3} m \right]
  \left(u,\hat{v}\right) d\hat{v}} \kappa \left(u,\bar{v}\right) d\bar{v}} \sqrt{\int_{v_i}^{v_{i+1}}
    \frac{\theta^2}{\kappa}\left(u,\hat{v}\right) d\hat{v}} \nonumber \\
+C \left(\sup_{r\leq r^\star_{cl}} \frac{1}{\sqrt{\alpha}} \right)\sqrt{\int_{v_i}^{v_{i+1}} e^{-\int_{\bar{v}}^v
    \left[\frac{3}{2} \frac{4\kappa}{r^3} m \right]
    \left(u,\hat{v}\right) d\hat{v}} \kappa\left(u,\bar{v}\right) d\bar{v}} \sqrt{\int_{v_i}^{v_{i+1}}
  \alpha r B^2\left(u,\hat{v}\right) d\hat{v}} \nonumber \\
\leq C\frac{M^\frac{3}{4}}{v_i} +  \Big|\frac{\zeta}{\nu} \left(u,v_{i}\right) \Big|
  e^{-.1\cdot d \cdot v_{i}} \nonumber \, .
\end{eqnarray}
Reiterating from the first to any chosen late rectangle we find for
any $(u,v_i)$ in the region $r^\star \leq r^\star_{cl}$
\begin{equation}
\Big|\frac{\zeta}{\nu} \left(u,v_{i}\right)\Big| \leq C\left(r^\star_{cl},c\right) \frac{M^\frac{3}{4}}{v_i}
\textrm { \ \ \ and hence \ \ \ } \Big|\frac{\zeta}{\nu} \left(u,v\right)\Big| \leq C\left(r^\star_{cl},c\right) \frac{M^\frac{3}{4}}{v}
\end{equation}
which is (\ref{hozBY}). Integrating (\ref{zetanuueq}) from the set
$L=\{r^\star = r^\star_{cl} \} \cup \{ \{t=t_0 \} \cap \{r^\star \geq
r^\star_{cl}\} \}$, where the bound (\ref{centBY}) holds by Cauchy
stability and the estimate just established, we obtain (\ref{centBY})
in the complete region using the energy estimate and the fact that $u
\sim t$ in the region where $r^\star_{cl} \leq r^\star
\leq \frac{9}{10}t$.
\end{proof}
\subsection{Improving assumption (\ref{Kass})}
With the pointwise bound on $\frac{\zeta}{\nu}$ established we can
improve assumption (\ref{Kass}) 
for any late boundary term $E_B^K\left(\tilde{T}\right)$ via 
Proposition \ref{Kbndtermcontrol}. One applies the $K$-estimate 
in the region 
${}^{u=t_N-r^\star_K}\mathcal{D}_{[t_0,t_N]}^{r^\star_K,u_0}$ 
some large $-r^\star_K$, late $t_0$ and $t_N=\tilde{T}$.
Using (\ref{Xbulkdecay}) we have 
\begin{equation}
\bar{I}^X_B \left({}^{u=t_{j+1}-r^\star_{cl}}\mathcal{D}_{[t_j,t_{j+1}]}^{r^\star_{cl},u=\frac{1}{11}t_{j+1}}\right) \leq C
\left(r^\star_{cl},c\right) \frac{M^2}{\left(t_{j+1}\right)^2} \leq \epsilon\left(t_0\right) \frac{M^\frac{7}{4}}{\left(t_{j+1}\right)^\frac{3}{2}}
\end{equation}
with the $\epsilon$ arising from the fact that $t_0$ can be chosen
as large as we may wish (at the cost of making the data smaller). Consequently
\begin{equation} \label{tIXB}
\sum_{j=0}^{N-1} t_{j+1} \bar{I}^X_B
\left({}^{u=t_{j+1}-r^\star_{cl}}\mathcal{D}_{[t_j,t_{j+1}]}^{r^\star_{cl},u=\frac{1}{11}t_{j+1}}\right)
\leq \epsilon\left(t_0\right) M^\frac{3}{2} \sum_{j=0}^{N-1} \frac{1}{\sqrt{1.1}^{j}} \leq \epsilon\left(t_0\right) M^\frac{3}{2}
\end{equation}
in view of the finiteness of the geometric series
\begin{equation}
\sum_{n=0}^{\infty} \frac{1}{\sqrt{1.1}^{n}} \leq K \, .
\end{equation}
Combining (\ref{tIXB}) with the fact that
$F^K_B\left(t_0\right)$ is small by Cauchy stability, 
Proposition \ref{Kbndtermcontrol} yields
\begin{equation} \label{Kimprove}
E^K_B \left(t_N\right) \leq M \epsilon \left(r^\star_K, t_0, \tilde{\delta}\right)
\end{equation}
for any late $t_N$, which improves assumption (\ref{Kass}).
\subsection{Improving assumptions (\ref{hozass}), (\ref{ceilass}) and (\ref{intebound2})}
We apply Proposition \ref{decayfromK} again, inserting the
better bound for the $K$-boundary terms (\ref{Kimprove}) at late times
to improve the decay of the energy on any arc-part of late
slices, $\Sigma_{t} \cap \{\frac{10}{11}t \geq
r^\star \geq r^\star_{cl} \}$. From Proposition \ref{pointBcent} we
also obtain improved decay for the field $B$ in the 
region $r^\star \geq r^\star_{cl}$. \\
Proposition \ref{goodslice2} applied in \emph{each} 
characteristic dyadic rectangle produces after 
inserting the better energy decay
in the region $r^\star \geq r^\star_{cl}$   
a slice $\hat{v}_i=\tau_i + r^\star_{cl}$ 
\begin{figure}[h!]
\[
\input{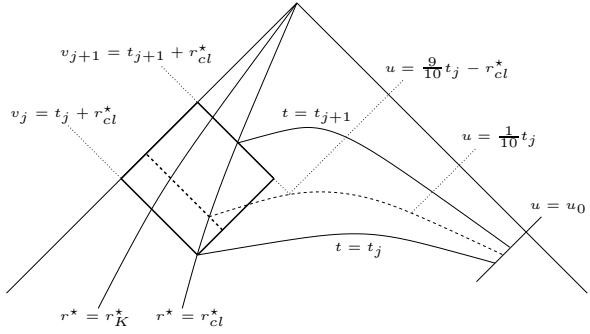}
\]
\caption{Closing the bootstrap.} \label{clobofin}
\end{figure}
with improved energy-flux decay, $\frac{c}{\left(t_i\right)^3} +
\frac{\epsilon}{\left(t_i\right)^2}$. By the domain of dependence property
 the decay of energy flux is improved on the horizon piece $v
\in \left[\hat{v}_i, v_{i+1}\right]$ and on the
ceiling-part of the characteristic region, as indicated in Figure
\ref{clobofin}. This retrieves in particular assumptions (\ref{hozass}) 
and (\ref{ceilass}) with a better constant. 
In view of the energy flux decaying now like $\frac{\epsilon}{v^2}$ 
on all achronal slices in the region $\mathcal{A}\left(T\right) \cap
\{ r^\star \leq \frac{9}{10}t \}$, we can apply Proposition 
\ref{goodslice} to find a good $F^Y_B$-slice
in each characteristic rectangle. Proposition \ref{hozestc} exports
this good decay of the $F^Y_B$-term to all constant $v$-slices and hence
the last outstanding bootstrap assumption (\ref{intebound2}) is finally
retrieved with a better constant.

What we have shown is that $\overline{A} = A$, so $A$ is closed. 
This completes the proof of Proposition \ref{setS}. The set $A$ 
must therefore constitute the entire $[0,\infty]$ and hence the decay rates 
of Theorem \ref{asymptoticstab} are 
proven in the entire $\mathcal{D}$, albeit in a different coordinate system than the 
one stated in the Theorem. The final subsection shows 
that the coordinate systems used in the bootstrap are indeed 
close to the null-coordinate system defined in Theorem \ref{asymptoticstab}. 
\subsection{Convergence of Coordinate Systems}
What we have already shown in section \ref{StabCord} is that the 
coordinates of a region 
$\mathcal{A}\left(\vartheta\left(\tilde{\tau}_\bullet\right)\right)
\cap \{ r^\star \geq r^\star_K \}$ (a-priori
defined in the coordinate system $\mathcal{C}_{\tilde{\tau}_\bullet}$), are
uniformly bounded in any coordinate system
$\mathcal{C}_{\tilde{\tau}}$ for $\tilde{\tau} \geq
\tilde{\tau}_\bullet$.\footnote{Note again 
that $r=r^\star_K$ may change its location in the different 
coordinate systems but remains always close to the geometrically 
defined curve $r=r_K$ of constant area radius.} It is important to
observe that the $u$-coordinate in the region $r^\star \leq r^\star_K$
is not uniformly close between the different coordinate systems. 
Indeed, in the coordinate system of Theorem $\ref{asymptoticstab}$ 
the horizon is located at $u=\infty$, whereas in any coordinate 
system $C_{\tilde{\tau}}$ it generically resides at a 
finite $u$ value (eventually converging 
to $u \rightarrow \infty$ for $\tilde{\tau} \rightarrow \infty$).

We finally establish the relation of the $C_{\tilde{\tau}}$ to the
coordinate system defined in Theorem \ref{StabCord}. First recall that
we have already shown that the geometrically defined point $R$ of Theorem
\ref{asymptoticstab} (which features as an "origin" of the coordinate
system) has coordinates uniformly close to $(\sqrt{M},\sqrt{M})$ in
any coordinate system $C_{\tilde{\tau}}$, cf.~section \ref{stabcorsec}. 
In the second step we
compare the scaling of the coordinates between the coordinate 
systems  $C_{\tilde{\tau}}$ and  the one asserted by 
Theorem \ref{asymptoticstab}. For this pick a point $P$ on
null-infinity. The value of $\gamma$ at this point in the 
coordinate system $\mathcal{C}_{\tilde{\tau}}$ can (for large enough
$\tilde{\tau}$) be estimated by integrating (\ref{gammaevol}) 
from $t=T\left(\tilde{\tau}\right)$ along a line of constant $u$:
\begin{equation}
\frac{1}{2} \leq \gamma \left(P\right) \leq \frac{1}{2} + \frac{C\left(\epsilon\right)}{r_N^2}
\end{equation}
where $r_N$ is the area radius at the intersection of the $\nabla r$
integral curve defining the coordinate system
$\mathcal{C}_{\tilde{\tau}}$ and the null line
$u\left(P\right)$.\footnote{For this estimate only the smallness of
  $\theta$ of Proposition \ref{simpsmall} is used.} 
In the limit $\tilde{\tau} \rightarrow \infty$ we have $r_N
\rightarrow \infty$ and hence 
$\gamma\left(P\right) \rightarrow \frac{1}{2}$. It follows that 
the scaling of the $u$-coordinate of 
$\mathcal{C}_{\tilde{\tau}}$ indeed converges to the 
one defined in Theorem \ref{asymptoticstab}. 

The function $\kappa$ on the other hand satisfies 
$|\kappa - \frac{1}{2}| \leq C\left(\epsilon\right)$
on $\mathcal{D}$ in both the coordinate systems
$\mathcal{C}_{\tilde{\tau}}$ (cf.~Proposition \ref{simpsmall}) 
and the one of Theorem \ref{asymptoticstab}. It is easy to show that
with this bound holding on the null curve $u=\sqrt{M}$, the $v$
coordinate of any two coordinate systems always satifies $v \sim
\bar{v}$ for $v \geq \sqrt{M}$, which is all what is needed to generalize
decay statements in $v$ to all coordinate systems. Namely integrating
from the point $W$ where the initial data intersect the null-line
$u=\sqrt{M}$ ($v \approx \sqrt{M}$ there by previous remarks) to a 
point $Q$ we have
\begin{equation}
v_Q = v_W + \int_{v_W}^{v_Q} dv \leq v_W +
\left(2+C\left(\epsilon\right)\right) \sup_{u=\sqrt{M}}
\frac{1}{1-\mu} \int_{v_W}^{v_Q} r_{,v} dv \leq  \frac{3}{2}\sqrt{M} +
4\left(r_Q - r_W\right) \nonumber
\end{equation}
and
\begin{equation}
v_Q = v_W + \int_{v_W}^{v_Q} dv \geq v_W + \left(2-C\left(\epsilon\right)\right)
\inf_{u=\sqrt{M}} \frac{1}{1-\mu} \int_{v_W}^{v_Q} r_{,v} dv \geq
\frac{\sqrt{M}}{2} + \frac{3}{2}\left(r_Q - r_W\right) \, . \nonumber
\end{equation}
Hence $v \sim \bar{v}$ for any two coordinate systems. 

We have shown that the limit of the coordinate 
systems $C_{\tilde{\tau}}$ is a coordinate system in which the origin 
is slightly shifted compared to the one of 
Theorem \ref{asymptoticstab} and whose $v$ scaling 
may be stretched or squeezed. It is now apparent that 
the decay rates stated also hold in the coordinate system 
of Theorem \ref{asymptoticstab}.
\section{Final Comments and Open Questions}
Theorem \ref{asymptoticstab} leaves room for generalizations. An
obvious one is the treatment of the \emph{triaxial case}, which at least
conceptually is not expected to pose any difficulty. In fact the same
vectorfields are expected to produce the required 
estimates for the fields $B$ and $C$ when contracted with an 
appropriate tensor $T_{\mu \nu}$ -- with the only
additional catch coming from the coupling of $B$ and $C$. A much more 
challenging problem is the derivation of better decay 
rates than the ones established here. As mentioned previously, 
in the context of compatible currents, the maximal 
decay rate is limited by the weights appearing in the 
$K$-vectorfield. It is an interesting question whether an additional
vectorfield (or an entirely different idea) can extract 
stronger decay, which might be expected from the
four-dimensional case \cite{DafRod}. An even more ambitious problem
concerns the large data regime of the five-dimensional Bianchi IX
model. The numerical studies of \cite{Bizon} suggest that a 
similar result to the one proven here should hold. In fact 
it may be possible to find an elaborate refinement of the ideas 
in \cite{DafRod}, which will allow an analysis of the 
large data regime within the symmetry class.

Finally, there should exist various applications of the techniques 
to four-dimensional problems. As already mentioned in the
introduction, the present paper may serve as a blueprint to obtain
a small-data version of \cite{DafRod} for the self-gravitating 
scalar field. For genuinely novel results, the case of a conformally 
coupled scalar field could be investigated. 
\section{Acknowledgements}
I would like to thank Mihalis Dafermos, who introduced me to the
problem and provided continuous support and encouragement along the
way. I am also grateful to EPSRC and Studienstiftung des deutschen 
Volkes for financial support.

\appendix
\section{Regularity and Green's identity} \label{reggree}
It was remarked in section \ref{Coordinates} that the coordinate systems
$\mathcal{C}_{\tilde{\tau}}$ are $C^1$. More precisely it was shown
  that they are piecewise $C^2$ with a discontinuity 
in $\frac{\Omega_{,v}}{\Omega}$ spreading along the 
null-line $v\left(B\right)$ and a discontinuity in 
$\frac{\Omega_{.u}}{\Omega}$ along $u\left(B\right)$. This
discontinuity could be avoided by the introduction of a smooth
interpolating function in the region around the cusp at the point
$B$ (cf.~Figure \ref{codiag}). However, as this would burden the 
notation even further, we will 
show here that the regularity is sufficient to carry out the
calculations involving the vectorfields.

Observing that the quantity $P^\alpha$ defined in (\ref{palpha}) 
is continuous and $\nabla_\alpha
P^\alpha$ at least piecewise continuous (cf.~(\ref{basicintegrand})), 
the basic identity (\ref{bvfi}) is valid for the vectorfields 
$X$,$Y$,$K$ in the coordinate systems $\mathcal{C}_{\tilde{\tau}}$. 

For the vectorfields $X$ and $K$ we also make use of Green's identity
(\ref{basicgreen}) in a region $\phantom{}^{u_H}
\mathcal{D}^{r^\star_g, u_J}_{[t_1,t_2]}$. 
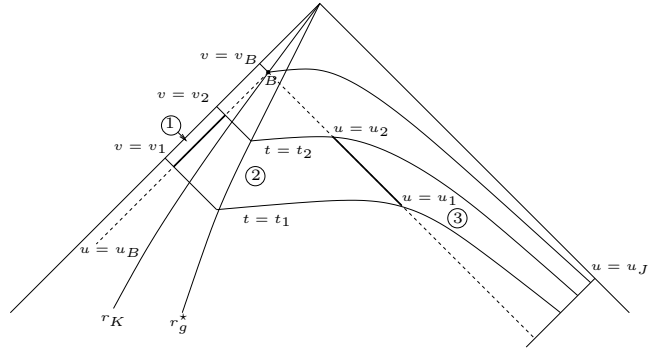
\[
\input{green.pstex_t}
\]
As depicted, the region may contain part of the null-line
$v\left(B\right)$ along which $\frac{\Omega_{,v}}{\Omega}$ could be
discontinuous and part of the null-line $u\left(B\right)$ along which
$\frac{\Omega_{,u}}{\Omega}$ could be discontinuous. The functions $D$
for which (\ref{basicgreen}) is applied are given by (\ref{cgfd}) and
(\ref{DKreg}). In both cases, $D\left(u,v\right)$ is seen to be
piecewise differentiable and such that $\Box D$ is 
piecewise continuous. To derive the identity
(\ref{basicgreen}) for these cases in our coordinate system, one should split 
the integration region $\phantom{}^{u_H}
\mathcal{D}^{r^\star_g, u_J}_{[t_1,t_2]}$ into three pieces, along the
null lines $u\left(B\right)$ and $v\left(B\right)$, introducing 
additional boundary terms from the bold lines.  
Green's identity is then clearly valid in each
subregion because all functions admit appropriate regularity
there, i.e.~in particular $D$ is differentiable and $\Box D$ 
is continuous in the interior. 
The integrand of the additional boundary 
term along the null-line $v\left(B\right)$ however
\begin{equation}
\int \left[B^2 \partial_u D - D \partial_u B^2\right]r^3 du
\end{equation}
is continuous because the $u$ derivative of $D$, which involves only
the term $\frac{\Omega_{,u}}{\Omega}$ (but not its $v$-analogue!), 
is continuous there. It is also bounded and the above integral will 
appear with a different sign for the two subregions. Analogously, 
the integrand of the other boundary term
\begin{equation}
\int \left[B^2 \partial_v D - D \partial_v B^2\right]r^3 dv
\end{equation}
is continuous because the $v$ derivative of $D$ is continuous there.
Hence adding up the three subregions the additional boundary terms
cancel and the identity (\ref{basicgreen}) indeed holds as stated.

\section{Different curves of constant $r^\star$}
\begin{table}[h!]
\begin{tabular}{c l}
$r^\star_K$ & very large and negative (close to the horizon), \\
& features as a source of smallness in the bootstrap \\
$r^\star_{cl}$ & $r^\star_{cl} = r^\star_Y - 2\sqrt{M}$ \\
$r^\star_{Y}$ & negative, chosen in section \ref{rstarcldef}
  to make a certain bulk-term \\
& of the $Y$ vectorfield positive in the region $r^\star \leq r^\star_{Y}$ \\
$-\frac{1}{2}\sqrt{M}$ & functions $\alpha$ and $\beta$ are supported 
in $r^\star \leq -\frac{1}{2}\sqrt{M}$ only \\
$R^\star$ & $R^\star = -\frac{1}{3}\sqrt{M}$ defined in Proposition
\ref{B2cor} \\
$r^\star_{zero}$ & defined in section \ref{rstarzero}, 
$-\frac{1}{6} \sqrt{M} \leq r^\star_{zero} \leq -\frac{1}{10} \sqrt{M}$ \\
$0$ & $r^2 \approx 4M$ (photon sphere for $5$dim.~Schwarzschild) \\
$\tilde{R}^\star$ & squashing field on initial data is not supported
for $r^\star \geq \tilde{R}^\star$ \\
$\hat{R}^\star$ & defined in Lemma \ref{goodsignK},equips a certain
integrand with a sign  \\
&  in a particular region
\end{tabular}
\end{table}

\newpage

\section{Glossary}
\begin{table}[h!]
\begin{tabular}{c l}
$\alpha$ & function depending on $r^\star$, used in the definition of
the vectorfield $Y$ \\
$\beta$ & function depending on $r^\star$, used in the definition of
the vectorfield $Y$ \\
$B$ & squashing field \\
$\gamma$ & defined in (\ref{kapgamdef}) \\
$\mathcal{D}$ & defined in (\ref{calD}) \\
$\delta$, $\tilde{\delta}$ & smallness parameters \\ 
$\epsilon$, $\tilde{\epsilon}$ & smallness parameters \\
$\zeta$ & $\zeta = r^\frac{3}{2} B_{,u}$ \\
$\eta$ & smallness parameter (cf.~Corollary \ref{rcurvcol} and 
Proposition \ref{Cauchystab}) \\
$\theta$ & $\theta = r^\frac{3}{2} B_{,v}$ \\
$\vartheta$ & function used for the definition of the 
coordinate systems $C_{\tilde{\tau}}$, cf.~(\ref{thetamap}) \\
$\kappa$ & defined in (\ref{kapgamdef}) \\
$\lambda$ & $\lambda = r_{,v}$ \\
$\mu$ & $\mu = \frac{2m}{r^2}$ \\
$\nu$ & $\nu = r_{,u}$ \\
$\xi$ & function depending on $r^\star$ defined in (\ref{xidef}) \\
$m$ & Hawking mass (\ref{Hawkmass}) \\
$M_f$ & final Bondi mass \\
$M_A$ & Hawking mass at the point $A$, cf.~section \ref{Coordinates} \\
$r$ & $r\left(u,v\right)$ area radius \\
$\rho$ & defined in (\ref{scalcurv}) \\
$\tilde{S}_{r_K}$ & defined in (\ref{sigrK}) \\
$\Sigma_t$ & defined in (\ref{slicedef}) \\
$\sigma$ & parameter, chosen in the section on the vectorfield $X$,
cf.~(\ref{xshift}) \\
$\tau$ & affine parameter along $r^2 = 4M_A$, section
\ref{Coordinates} \\
$\tilde{\tau}$ & affine parameter along $r^2 = 4M_f$, section
\ref{Coordinates} \\ 
$\varphi_1, \varphi_2$ & defined in (\ref{deltaB}) \\
$\chi, \tilde{\chi}$ & smooth interpolating functions,
cf.~(\ref{alphacond}) and Proposition (\ref{B2cor}) \\
$\psi$ & smallness parameter, section \ref{rstarcldef} \\
$\Omega^2$ & metric function \\
\end{tabular} 
\end{table}

\end{document}